%% file: main.tex
\documentclass[twocolumn,prx,aps,amsfonts,showpacs,nofootinbib,longbibliography,notitlepage,superscriptaddress]{revtex4-2}
\usepackage{ulem}
\usepackage{graphicx} 
\usepackage{amsmath}
\usepackage{amssymb}
\usepackage{amsthm}
\usepackage{braket}
\usepackage{mathrsfs}
\usepackage{extarrows} 
\usepackage{subcaption}
\usepackage{xcolor}
\usepackage{enumerate}
\usepackage{bbm}
\usepackage{comment}
\usepackage{algorithm}
\usepackage{algpseudocode}
\usepackage[pdftex,colorlinks=true,linkcolor=violet,citecolor=blue,urlcolor=cyan]{hyperref}
\usepackage{tikzit}
\input{zxbw.tikzstyles}
\newcommand{\Z}{\mathbb Z}
\newcommand{\F}{\mathbb F}

\newcommand{\cS}{\mathcal{S}}
\newcommand{\cL}{\mathcal{L}}

\newcommand{\ii}{\text{i}}

\newcommand{\G}{\mathcal G}
\newcommand{\M}{\mathcal M}

\newcommand{\tr}{\text{Tr}}

\newtheorem{defn}{Definition}

\newtheorem{thm}{Theorem}
\newtheorem{col}{Corollary}

\usepackage[justification=justified]{caption}
\begin{document}

\title{Fault-tolerant protocols through spacetime concatenation}
\author{Yichen Xu}
\email{yx639@cornell.edu}
\affiliation{Department of Physics, Cornell University, Ithaca, New York 14850, USA}
\author{Arpit Dua}
\email{adua@vt.edu}
\affiliation{Department of Physics, Virgina Tech, Blacksburg, Virginia 24061, USA}

\begin{abstract}
We introduce a framework called spacetime concatenation for fault-tolerant compilation of syndrome extraction circuits of stabilizer codes. This framework enables efficient compilation of syndrome extraction circuits into dynamical codes through structured gadget layouts and encoding matrices, facilitating low-weight measurements while preserving logical information.
This framework uses conditions that are sufficient for fault-tolerance of the dynamical code, including not measuring logical operators and preserving the spacetime distance. This framework goes beyond dynamical weight-reduction methods, which do not explicitly produce Floquet codes with nontrivial logical automorphisms.
Using this framework, we reproduce known Floquet codes such as the honeycomb Floquet code and the ruby lattice Floquet color code with nontrivial automorphisms and a fault-tolerant planar variant of the Floquet toric code. We also construct new explicit examples of dynamical codes, including a dynamical bivariate bicycle code and a dynamical Haah code. 
Beyond constructing examples, we write a restricted equivalence relation which can be used to discuss classification and resource trade-offs of dynamical codes. Lastly, we demonstrate the adaptability of our framework to arbitrary qubit layouts and fabrication defects, making it well-suited for a hardware-first design approach.
\end{abstract}
\maketitle
\tableofcontents

\section{Introduction}

Ensuring the reliability of quantum computations in the presence of environmental noise and operational imperfections is a fundamental challenge in quantum computing. Quantum error correction~\cite{calderbank_good_1996,gottesman_stabilizer_1997,knill_theory_1997,gottesman_introduction_2009} is crucial to combat noise, and significant efforts are being made toward its experimental realization~\cite{Google2021realizing, Semeghini2021probing, Zhao2022realization, Google2022suppressing, Google2022nonabelian, xu2022digital, iqbal2023topological}. Despite this progress, advancing quantum error correction remains a critical priority. Efforts continue to focus on designing error-correcting codes to minimize resource overhead for memory and fault-tolerant operations and to align more effectively with the constraints and capabilities of current quantum hardware.

The most common quantum error-correction strategy is to store information in a subspace of the Hilbert space of physical qubits, called a codespace, which is also an eigenspace of commuting operators, called stabilizers. In active error-correction protocols, these stabilizers are repeatedly measured to detect errors, and the measured eigenvalues form a syndrome that is processed by a classical decoding algorithm. This algorithm guesses the error pattern, enabling the application of corrective operators to return the qubits to the codespace. Recently, following the original work of Hastings and Haah (HH)~\cite{hastings_dynamically_2021}, dynamical codes have emerged as innovative extensions of stabilizer codes~\cite{townsend-teague_floquetifying_2023,bauer_xy_2024,zhang_x-cube_2023,davydova_floquet_2023,kesselring_anyon_2024,dua_engineering_2024,yan_floquet_2024,sullivan2023floquet,fu_error_2024,fuente_xyz_2024,higgott_constructions_2024,tanggara_simple_2025,basak_approximate_2025}. The active QEC protocols associated with these dynamical codes repeating in time have been referred to as Floquet code protocols and were proposed to reduce the reliance on high-weight stabilizer measurements by utilizing sequences of lower-weight measurements, making them more feasible for implementation on certain quantum hardware. A dynamical code can be viewed as a more general compilation of a syndrome extraction circuit (SEC) of a stabilizer code in which the stabilizers associated with a stabilizer code are not necessarily measured directly, but can be measured as products of lower-weight measurements in multiple rounds. Consequently, the associated Floquet code protocol with repeated dynamical codes can be viewed as an alternative compilation of the stabilizer code protocol with repeating rounds of stabilizer measurements~\cite{gidney_fault-tolerant_2021,Gidney2022benchmarkingplanar, Paetznick2023Performance}.

In this work, we construct efficient fault-tolerant protocols through spacetime concatenation, an explicit framework for compilation of SECs of stabilizer codes as efficient Floquet codes. This is done by structuring measurement gadgets using low-weight measurements while ensuring the preservation of logical information. In particular, spacetime concatenation reformulates the notion of a dynamical code associated with a stabilizer code in terms of code concatenation for every qubit (spatial concatenation) and measurements between these codes (temporal concatenation), leading to a temporal evolution of the stabilizer state. In fact, previous constructions in the work~\cite{dua_engineering_2024} by one of the authors used such a principle for particular examples, but a general strategy was lacking.

In general, fault-tolerant protocols can be considered through multiple levels of increasing dynamical complexity. At the most fundamental level, static stabilizer code protocols operate with a fixed instantaneous stabilizer group (ISG) that remains unchanged throughout the process.\footnote{This statement is written only for the time instants where the measurements are completed. The idea of ISGs applies also to the intermediate states during the execution of a standard stabilizer measurement circuit, i.e.\ after any of the gates involved in the circuit. In this picture, ISGs evolve within each step of a static stabilizer code protocol similar to the Floquet code protocols.} Moving beyond this, subsystem code protocols introduce dynamically changing ISGs while maintaining a static set of logical qubits with the same bare logical operators. Finally, Floquet code protocols, simply referred to as Floquet codes, represent the most general approach, allowing both the ISG and the logical operators to evolve over time. This hierarchical perspective clarifies the relationship between different protocols, with Floquet codes offering the broadest definition for fault-tolerant protocols with periodic measurement dynamics. Since all these protocols---static stabilizer codes, subsystem codes, and Floquet codes---are based on instantaneous stabilizer groups at any given time, they can be understood through the unified framework of Floquet code protocols. We illustrate these different classes of fault-tolerant protocols in Fig.~\ref{fig: flowchart}. This perspective allows for a unified understanding of fault-tolerant protocols in terms of the evolution of stabilizer groups and the logical information under measurements.

A central motivation for innovating dynamical codes is the challenge of compiling high-weight stabilizer measurements in quantum architectures. The standard stabilizer measurement method involves entangling gates of data qubits with appropriately initialized ancillary qubits and measurements of those ancillary qubits in an appropriate basis. However, recently developed dynamical weight-reduction techniques~\cite{fuente_dynamical_2024, rodatz_floquetifying_2024} have demonstrated that complex measurements can be systematically decomposed into sequences of lower-weight measurements, while maintaining fault tolerance. Such a compilation is well suited for certain hardware parameters, where low-weight measurements lead to better performance in spite of the increase in the total number of measurements involved.

Our spacetime concatenation framework embeds weight-reduction techniques within a more general structured concatenation scheme, allowing for an effective approach to designing fault-tolerant Floquet codes. It formulates the dynamical code using measurement gadgets that specify a map on Pauli operators according to the connectivity of the qubits in the given stabilizer code. We propose an algorithm (see Algorithm~\ref{alg: gadgetcon}) for designing gadget connectivity which respects the original $k$-locality of the stabilizer code. This ensures that local checks in the static code are implemented by measurements supported on a comparable neighborhood in spacetime.
We refer to this design rule as a strict-$k$ locality-preservation condition (SLPC) and use it as a convenient organizing principle for constructing and classifying gadgets. 

In addition to Floquet code design, spacetime concatenation provides a powerful tool in analyzing the fault-tolerance of the constructed protocols. In particular, to guarantee that the fault-tolerance of the dynamical protocol to be enhanced with increasing number of physical qubits, we propose Macroscopic Logical Support Condition (MLSC), which demands that the image of every logical operator under the dynamical map has support of order at least the static code distance. Our main theorems show that, for any stabilizer code with macroscopic distance, any dynamical code produced by Algorithm~\ref{alg: gadgetcon} that satisfies MLSC automatically has spacetime distance at least $d$. 
In the rest of the paper we often impose SLPC to simplify gadget search and make examples geometrically transparent, but MLSC is the essential condition for preserving fault tolerance.

We emphasize that our construction and fault-tolerance condition generally apply to dynamical code compilations of arbitrary quantum low-density parity-check (qLDPC) codes~\cite{gottesman_fault-tolerant_2014,panteleev_degenerate_2021,panteleev_quantum_2022}. This is in comparison to previous works where fault tolerance was proven on a case-by-case basis, i.e.\ for specific Floquet code examples. Our recipe automatically outputs fault-tolerant Floquet code protocols whenever the MLSC condition is satisfied.

We illustrate the versatility of the spacetime concatenation framework to construct Floquet codes through explicit examples (see Sec.~\ref{sec:examples}), including a Floquet bivariate bicycle (BB) code and a Floquet Haah code using the associated stabilizer codes~\cite{bravyi_high-threshold_2024,haah_local_2011}. Concatenation-based or tensor-network-based schemes~\cite{ferris_tensor_2014,farrelly_tensor-network_2021,cao_quantum_2022,cao_quantum_2024,masot-llima_stabilizer_2024} have been discussed for stabilizer code protocols and carry certain features of our dynamical construction defined in terms of measurement gadgets. Our framework also goes beyond the approach of Ref.~\cite{fuente_dynamical_2024}, which constructs Floquet codes by applying dynamical weight-reduction techniques to the static stabilizer code protocol’s syndrome extraction circuit and can only produce dynamical codes with trivial automorphism. Our framework produces all nontrivial automorphism dynamical codes that obey MLSC.

Although Floquet codes provide a novel and flexible approach to fault-tolerant quantum error correction, the full space of possible Floquet code compilations of SECs of a given stabilizer code is not yet fully understood. The understanding of the classification of Floquet code protocols extends to several fundamental aspects, including the scheduling of measurements and its effect on logical operations. We address the classification of Floquet codes under the notion of spacetime equivalence, which is more constrained than the notion of general ZX-equivalence and uses the equivalence of gadget layout and the number of spacetime stabilizers. We discuss this in Sec.~\ref{sec:classification}.

The understanding of the classification of Floquet code protocols also extends to the efficiency of various Floquet schedules in terms of the overhead of qubits and the number of measurements, circuit depth, and decoding complexity. Practically speaking, it is important to keep in mind the resource overheads of different compilations of SECs for a given stabilizer code, such as static stabilizer code, subsystem code, and Floquet code protocols, and choose a compilation accordingly. Thus, we also discuss in our framework how different resources such as the number of ancillary qubits, the number of low-weight measurements, and the circuit depth of the dynamical code are related to each other; see Sec.~\ref{sec:resource}. These resource overheads are unified into one single number, the number of internal legs that connect the gadgets defined in our framework.

Floquet code protocols can be naturally adapted to fabrication defects in quantum hardware. In fact, Floquet codes have shown promising capabilities in terms of resource overhead to deal with missing or faulty qubits by appropriate modifications~\cite{mclauchlan_accommodating_2024}. Our spacetime concatenation framework is flexible enough to apply to arbitrary quantum hardware layouts and thus easily adapts to fabrication defects. If one already builds a Floquet code protocol on a given layout and fabrication defects are introduced later, then we can discuss how to modify the protocol solution found in our framework. We discuss this in Sec.~\ref{sec:fabrication_defects}.

We now give a brief summary of the sections in this paper. In Sec.~\ref{sec:prior_context}, we discuss in further detail how our framework fits in the context of prior work on Floquet codes and why it is significant. In Sec.~\ref{sec:generalities}, we describe the general principles of our spacetime concatenation framework. In Sec.~\ref{sec:examples}, we discuss how to use the spacetime concatenation framework to construct Floquet codes such as the Floquet toric code, the Floquet bivariate bicycle code, and the Floquet Haah code. In Sec.~\ref{sec:fault_tolerance}, we discuss the fault-tolerance of our code construction. In Sec.~\ref{sec:classification}, we define the notion of spacetime equivalence to make progress toward the classification of dynamical codes. In Sec.~\ref{sec:resource}, we discuss the resource theory for dynamical codes constructed using our spacetime concatenation framework. In Sec.~\ref{sec:fabrication_defects}, we discuss the applicability of our framework to deal with fabrication defects in a general way for dynamical codes. In Sec.~\ref{sec:outlook}, we discuss open questions in the context of Floquet codes and fault-tolerant protocols.

\section{Prior context}
\label{sec:prior_context}
In this section, we position our spacetime concatenation framework in the context of existing approaches to Floquet codes. There are three core subtopics that inform the construction of our framework and we address these below.

\begin{enumerate}
    \item \textbf{Floquet codes from topological order.}  
    The HH Floquet toric code was originally discovered by studying a measurement schedule for the subsystem code associated with Kitaev’s honeycomb model, revealing a dynamically generated logical subspace that evolves coherently across instantaneous stabilizer groups (ISGs). This measurement schedule was later interpreted as a sequence of anyon condensations in a parent color code~\cite{kesselring_anyon_2024}. The notion of dynamic automorphism codes~\cite{davydova_quantum_2024}, where logical gates emerge from measurement dynamics, motivated further works on Floquet codes for topological phases. These include constructions often leveraged physical intuitions such as anyon condensation, foliation structures of topological order, or coupled spin chain models~\cite{zhang_x-cube_2023, yan_floquet_2024, dua_engineering_2024}.

    However, these methods typically rely on specific structures within the ISGs, such as those found in topological stabilizer codes. In contrast, it is not straightforward to exhibit such structures in many quantum LDPC (qLDPC) codes. Our spacetime concatenation framework addresses this gap by providing a first-principles algorithm for constructing Floquet codes from any qLDPC stabilizer code, without assuming any particular ISG intuition or topological interpretation.

    \item \textbf{Floquet codes from ZX-calculus rewrites.} 
    ZX-calculus provides a powerful, graph-theoretic language for representing and manipulating quantum circuits. Several works~\cite{townsend-teague_floquetifying_2023, rodatz_floquetifying_2024, fuente_dynamical_2024} have leveraged this to “Floquetify” standard single-ancilla-syndrome-extraction-circuits (SASECs) by applying ZX rewrite rules. These methods produce circuits with lower-weight measurements and, in principle, can generate Floquet codes for any stabilizer code, since such syndrome extraction circuits are always definable.

    However, by construction, these ZX-rewritten circuits are logically equivalent to the SASECs and therefore yield Floquet codes with only trivial logical automorphisms. In contrast, our approach does not start from a fixed syndrome extraction circuit. Instead, we construct a dynamical map that measures the incoming stabilizer group and generates a stabilizer group on the outgoing legs. This allows our framework to systematically yield Floquet codes with potentially nontrivial logical automorphisms. Moreover, by modifying the number of internal legs in our gadget layout or changing the scale of locality, we can systematically explore the space of all Floquet codes that satisfy the MLSC.

    \item \textbf{Quantum Lego codes.}  
    Because each gadget in our framework functions as a local encoder of a stabilizer code, the spacetime concatenation formalism bears structural resemblance to the tensor-network-based “quantum Lego” codes~\cite{cao_quantum_2022, cao_quantum_2024}. However, the two approaches differ fundamentally in both structure and purpose.

    In quantum Lego codes, the tensor network encodes a stabilizer or subsystem code, with the incoming and outgoing legs representing logical and physical qubits, respectively. The tensor network does not directly specify a physical circuit for implementation
    In contrast, the gadget network in spacetime concatenation encodes an explicit circuit, where the incoming and outgoing legs correspond to the Hilbert space before and after a dynamical transformation. Each gadget is selected such that the full tensor network implements a specific Clifford circuit that measures stabilizers on the incoming legs and generates a stabilizer group on the outgoing legs.

\end{enumerate}

To summarize, prior approaches to constructing Floquet codes have largely relied on structural intuitions, such as anyon condensation in topological codes or on weight reduction methods, such as ZX-calculus rewrites of standard syndrome extraction circuits, which limit their generality. These methods either apply to a narrow class of stabilizer codes or produce only trivial logical automorphisms. In contrast, our spacetime concatenation framework offers a principled, circuit-level construction that applies to arbitrary stabilizer codes, including generic qLDPC codes, and naturally incorporates nontrivial logical automorphisms. While it draws conceptual parallels with quantum Lego codes, our approach differs fundamentally by producing an explicit and local circuit implementation that ensures fault-tolerance through strict locality preservation and logical information retention. This generality and flexibility position our framework as a unifying platform for designing fault-tolerant Floquet protocols across a wide range of code families.
\begin{figure*}
    \centering
    \includegraphics[width=0.9\linewidth]{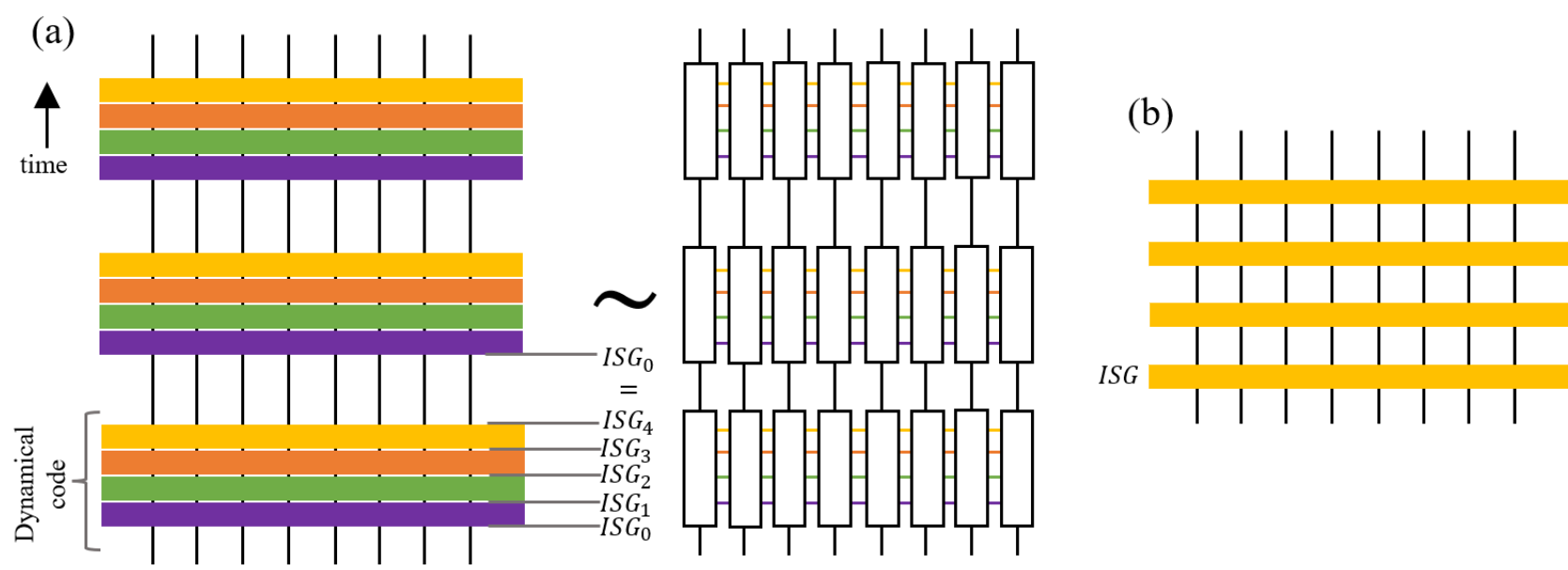}
    \caption{Comparison of fault-tolerant protocols. (a) Floquet codes are implemented via periodic sequences of dynamical codes, with ISGs that vary over time. The dynamical code circuit is constructed using gadgets connected temporally. Subsystem code protocols also involve time-varying ISGs but preserve a fixed logical space. (b) Static stabilizer code protocols maintain both a fixed ISG and logical space throughout the dynamics.}
    \label{fig: flowchart}
\end{figure*}

\section{Generalities}
\label{sec:generalities}
\subsection{Dynamical stabilizer codes: a general perspective}
\label{sec: generalcondition}

We begin by first defining our most general notion of a dynamical stabilizer code, which will be simply referred to as dynamical code in the subsequent discussion. Simply speaking, a dynamical code can be viewed as a SEC of a stabilizer code that measures all stabilizers of a stabilizer code at least once. A Floquet code can then be viewed as a sequential repetition of one or more dynamical code(s). Building on top of these intuitions, we propose a more general notion of dynamical codes that incorporates a large class of Floquet codes that exist in the literature.

Generically, a dynamical code can be defined for any pair of $[[n,k,d]]$ stabilizer codes that are defined on the same sets of qubits; here $n$ is the number of data qubits, $k$ is the number of logical qubits, and $d$ is the code distance. The stabilizer groups of this pair of codes are denoted by $\cS$ and $\overline{\cS}$, and their logical operators form the groups $\cL$ and $\overline{\cL}$, respectively.
In this work, we mainly consider $\cS$ and $\overline{\cS}$ to be Pauli stabilizer codes, so that each incoming and outgoing code stabilizer is a string of Pauli operators. Our definition and construction apply to general Pauli stabilizer codes $\cS$ and $\overline{\cS}$, which are not necessarily Calderbank-Shor-Steane (CSS) codes (i.e. $\cS$ and $\overline{\cS}$ are not necessarily generated by products of either $X$ or $Z$ operators).
We will refer to $\cS$($\overline{\cS}$) and $\cL$($\overline{\cL}$) as the incoming (outgoing) stabilizer group and the logical group for the dynamical code. 
For a dynamical code to preserve all the logical information, the two logical groups must be isomorphic to each other: $\cL\cong\overline{\cL}$. With these inputs, the dynamical code is a linear map $\M$ that acts on the Hilbert space of $n$ qubits as $\M: \mathbb{C}^{2^n}\to \mathbb{C}^{2^n}$ and satisfies the following conditions: 
\begin{enumerate}[(i)]
    \item It measures the incoming stabilizers, that is, it maps any stabilizer $s\in \cS$ of the initial code that is supported on the incoming legs to numbers:
    \begin{equation}\label{eq: ms}
        \M s=c_s\M,
    \end{equation}
    where $c_s$ is a complex phase that depends on $s$. For Clifford dynamical codes, $c\in e^{\frac{\ii \pi}{2}\Z}$.
    \item It generates the outgoing stabilizer group $\overline{\cS}$, which implies that for every $\overline{s}\in\overline{\cS}$, we have
    \begin{equation}\label{eq: m1}
        \overline{s}\M=c_{\overline{s}}\M.
    \end{equation}
    Again,  $c_{\overline{s}}$ is a complex phase that depends on $\overline{s}$.
    \item It is fault-tolerant, i.e. it maps the set of logical operators that is supported on the incoming legs to a set of logical operators that is supported on the outgoing legs, so that the logical information is preserved. That is, for every incoming logical operator $L\in\mathcal{L}$, we have
    \begin{equation}\label{eq: ml}
        \M L=\overline{L}\mathcal{M},
    \end{equation}
    where $\overline{L}\in\overline{\mathcal{L}}$ is an outgoing logical operator.
\end{enumerate}
Some comments on our definition are in order. From condition (i), we can regard any dynamical code as a generalization of a SEC. The conventional approach to stabilizer measurement by introducing one ancilla per each stabilizer, entangling each data qubit that participates in the stabilizer to the ancilla, and measuring the ancilla in the end \cite{divincenzo1996fault}, is a dynamical code once the circuit measures \textit{all the stabilizers} of the code. In this case, $\cS=\overline{\cS}$ and the logical action is apparently trivial, i.e. $L=\overline{L}$ in Eq.~\ref{eq: ml} for every $L\in \mathcal{L}$. We will refer to such a dynamical code as the single-ancilla-syndrome-extraction-circuit (SASEC) (see Appendix \ref{sec: zxrules} for the circuit realization of SASEC). The generalization in our definition is manifested in conditions (ii) and (iii).
In condition (ii), we did not require that the outgoing stabilizer group $\overline{\cS}$ of the outgoing $[[n,k,d]]$ code be exactly the same as the incoming stabilizer group $\cS$. Nevertheless, in conventional quantum error correction (QEC) practice of implementing any fault-tolerant protocol, the same measurement circuit will be repeated. For the measurement circuit $\M$ to be \textit{directly repeatable}, we must have $\cS=\overline{\cS}$. 
Subsequently, when discussing actual Floquet code examples, we will always impose this equality. 
We note that, even if $\M$ is not directly repeatable, it can be followed by a ``dual circuit" $\overline{\M}$ which satisfies
\begin{equation}
    \overline{\M}\overline{s}\propto \overline{\M}\text{ and }s\overline{\M}\propto\overline{\M}
\end{equation}
for any $s\in \mathcal{S}$ and $\overline{s}\in\overline{\mathcal{S}}$, so that the combined circuit, represented by the linear map $\overline{\M}\circ \M$, is repeatable. 

We also note that, if $\cS\neq \overline{\cS}$ and the circuit implementation of $\M$ only contains measurements, the Floquet code obtained by repeating $\overline{\M}\circ \M$ does not have a parent subsystem code. This generalizes the notion of Floquet codes without parent subsystem codes, which are constructed in Ref.~\cite{davydova_floquet_2023}. In fact, if such a Floquet code had a parent subsystem code, both $\mathcal{S}$ and $\overline{\mathcal{S}}$ would need to be part of its stabilizer subgroup, while $\cL\cong\overline{\cL}$ is the logical subspace of this subsystem code. This is true only if all stabilizers in $\mathcal{S}$ and $\overline{\mathcal{S}}$ commute. This would imply that we could merge $\cS$ and $\overline{\cS}$ into a single stabilizer group. However, given the fact that $\cS$ already stabilizes the incoming stabilizer code, the merged stabilizer group would be over-complete unless $\cS=\overline{\cS}$.

Below we will consider codes with $\cS= \overline{\cS}$ because sufficient periods of $\M$ or a combined circuit $\overline{\M}\circ \M$ can always satisfy that condition.

In condition (iii), we do not require that every outgoing logical operator be exactly the same as the incoming logical operator. Rather, $\M$ can act as a unitary operation in the logical subspace, i.e.
\begin{equation}
    P_\mathcal{L}\M P_\mathcal{L}=U_\M P_\mathcal{L},
    \label{eq: pmp}
\end{equation} 
where $P_\mathcal{L}$ is the projector onto the logical subspace and $U_\M\in U(2^k)$ is a unitary operator. This means for every logical operator $L\in \mathcal{L}$, the dynamical code acts as $\M(L)=U_\M LU_\M^\dagger \bar{s}$, where $\bar{s}\in\overline{\cS}$ is some stabilizer in the outgoing stabilizer group.  If $U_\M\neq \mathbbm{1}$, the dynamical code performs a non-trivial logical gate.

We now discuss how the above definitions incorporate some well-known examples of Floquet codes. 
For example, in the HH Floquet code\footnote{The version of this code on the honeycomb lattice has been referred to as the Hastings-Haah code~\cite{hastings_dynamically_2021} or the honeycomb code~\cite{gidney_fault-tolerant_2021}.} proposed in Refs. \cite{hastings_dynamically_2021,haah_boundaries_2022} (see Fig. \ref{fig: tclattice} for an adaptation on the square-octagon lattice),
the physical qubits are measured according to the schedule 0: gZZ; 1: bXX; 2: rYY, where the first letter denotes the color of the edge and the last two letters denote the pairwise measurement of the two qubits on that edge in the corresponding Pauli basis. The four measurement rounds 0, 1, 2, 0 is a dynamical code over the effective qubits formed by round 0 measurements, and $\cS=\overline{\cS}$ are the instantaneous stabilizer groups (ISG) of the 0th round, the toric code on square lattice. We note that one can choose the ISG $\cS_{i}$ ($i=0,1,2$) of any measurement rounds in 0,1,2 to be both the incoming and outgoing stabilizer group, so that the rounds $i,\ i+1,\ i+2, i+3\mod 3$ form a dynamical code with $\cS=\overline{\cS}$ due to the periodicity of the schedule. For the HH Floquet code on a torus, the action on the two logical qubits is, in fact, a Hadamard gate on every logical operator: $U_\M=(H\otimes H)\cdot\text{SWAP}$, see Sec.~\ref{sec: tcclifford}.

Another familiar example is the CSS Floquet toric code proposed in Refs. \cite{davydova_floquet_2023,kesselring_anyon_2024}, which can be implemented in the same lattice layout as in Fig. \ref{fig: tclattice} with a 6 round schedule gZZ, bXX, rZZ, gXX, bZZ, rXX. In this case, the 4 rounds of measurements gZZ, bXX, rZZ and gXX  form a dynamical code over the effective qubits supported on each pair of qubits on the green edges. The incoming stabilizer group $\cS$ and the outgoing stabilizer group $\overline{\cS}$ are both toric codes, but the locations of the stabilizers $X$ and $Z$ are swapped, and the logical space action $U_\M=\mathbbm{1}$ is trivial. Alternatively, the 7 rounds gZZ, bXX, rZZ, gXX, bZZ, rXX and gZZ also form a dynamical code, in which case $\cS=\overline{\cS}$ and $U_\M=\mathbbm{1}$. Repeating the dynamical code yields the Floquet code protocol.

\subsection{Compiling dynamical codes from gadget layout}\label{sec: compile}

A generic measurement circuit involves the introduction of ancillary qubits, the application of unitary gates, and projective measurements. To combine these aspects of the circuit in one unified approach, we use ZX-calculus as a powerful tool in constructing the circuit of dynamical codes. For a general introduction to ZX-calculus and its application in quantum computing, we refer the reader to Refs. \cite{backens_zx-calculus_2014,wetering_zx-calculus_2020}. For convenience, we put together a short list of useful Clifford ZX-calculus rules in Appendix \ref{sec: zxrules} that are directly relevant to our work.

From the viewpoint of ZX-calculus, every dynamical code $\M$ represents blackbox of ZX diagram with $n$ incoming legs and $n$ outgoing legs, see Fig. \ref{fig: instrument}. The ZX diagram produces a linear map between the incoming and outgoing legs that satisfies the rules in Eqs. \eqref{eq: ms}, \eqref{eq: m1} and \eqref{eq: ml}. Since our main focus is to construct Floquet codes using Pauli basis measurements, the corresponding ZX-diagram of $\M$ has to be Clifford, that is, it only consists of $Z$ and $X$ nodes with phases $\frac{\pi}{2}\Z$.

\begin{figure*}
    \centering
    \includegraphics[width=0.8\linewidth]{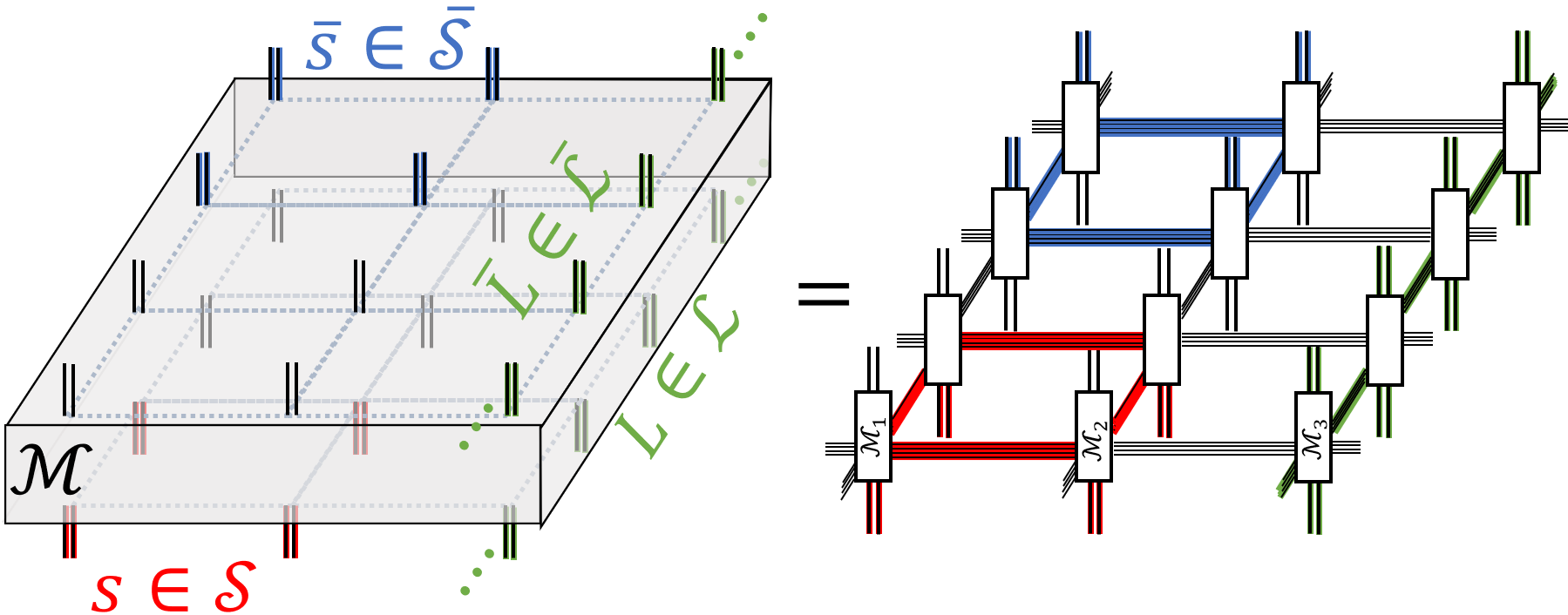}
    \caption{The linear map $\M$ and the layout into gadgets $\M_i$, where $i=1,2,3,\dots$ are lattice site indices. In the example case shown in the figure, the incoming and outgoing stabilizers are located on a 2D square lattice (marked in grey dotted lines) with $m=2$ data qubits per site and each pair of nearest neighbor gadgets are connected by $4$ internal legs. For each incoming/outgoing stabilizer $s$/$\overline{s}$ (marked by vertical red/blue shaded lines) to be measured/generated, the gadgets should map the incoming/outgoing stabilizers to bond operators over the internal legs that connects these gadgets (marked by horizontal red/blue shaded lines). To preserve the logical information, the gadgets should also map incoming logical operator (marked in green) to an outgoing one. The different colors of the shaded legs are used to mark different operators. }
    \label{fig: instrument}
\end{figure*}
\begin{figure}
    \centering
    \includegraphics[width=0.7\linewidth]{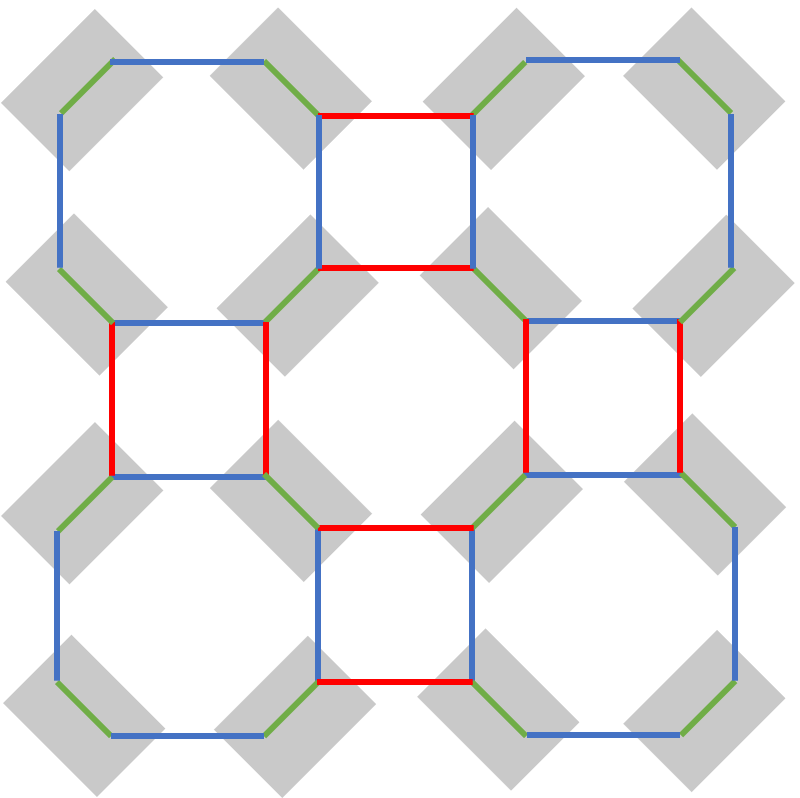}
    \caption{The Floquet toric codes on the square-octagon lattice. Here the red, blue and green bonds are denoted in the measurement schedules as ``r", ``b" and ``g" respectively. The pair of qubits in shaded rectangles form an effective qubit after the pairwise measurements on the green bonds. From the gadget layout perspective, the two qubits is the outcome of the spatial concatenation inside the gadget. 
    }
    \label{fig: tclattice}
\end{figure}

\subsubsection{The gadget layout  with strict-$k$-locality-preservation (SLP)}
We now propose algorithms to construct the linear map $\M$ that takes us from the stabilizers $\cS$ and $\overline{\cS}$. The key strategy is to decompose $\M$ into a set of interconnected gadgets $\M_i$ that are located at the original locations of the data qubits (see Fig. \ref{fig: instrument}). Suppose that there are $m$ data qubits on every site\footnote{This includes the scenario where we have done coarse-graining to bring multiple qubits onto a single site. In general, a gadget can be introduced for every qubit. In some cases, a gadget for a coarse-grained site with multiple qubits provides a simpler representation. For example, see Floquet Haah code in Sec.~\ref{sec: haah} for which the stabilizers are represented on a cube with two qubits per site and the gadget is defined for the site instead of individual qubits, but this is not necessary.} in the original stabilizer codes $\cS$ and $\overline{\cS}$, each gadget $\M_i$ can then be represented by a ZX-diagram with $m$ incoming and $m$ outgoing legs, and additionally a set of $n_L$ internal legs that connect the gadget to other gadgets on neighboring sites. 
In total, the linear map $\M$ is obtained from the gadgets via
\begin{equation}
        \M=\tr_\text{bond}\left[\bigotimes_j\M_j\right],
\end{equation}
where $j=1,2,3,\dots$ are lattice site indices and $\tr_\text{bond}$ denotes the contraction of pairs of internal legs on the bonds which connect the gadgets.  For readers familiar with tensor networks, we essentially turn the linear map $\M$ of the dynamical code into a matrix product operator (MPO)\cite{verstraete_matrix_2004,zwolak_mixed-state_2004}, where the gadgets correspond to the local tensors. Each pair of gadgets connected by $l$ internal legs gives rise to a bond Hilbert space of dimension $2^l$.

For the dynamical code represented by the linear map $\mathcal{M}$ to satisfy the conditions in Eqs.~\eqref{eq: ms} and \eqref{eq: m1}, the Pauli web of every $s\in \mathcal{S}$ or $\bar{s}\in \overline{\mathcal{S}}$ on the incoming (or outgoing) legs needs to close inside $\mathcal{M}$ via the bonds that connect the gadgets, see Fig.~\ref{fig: instrument} (see Appendix~\ref{sec: pauliweb} for details of Pauli webs). We refer to the support of the Pauli web of $s$ on the bonds as \textit{bond operators}.

In order for the dynamical code to respect the underlying $k$-locality of the stabilizer codes $\mathcal{S}$ and $\overline{\mathcal{S}}$, we first determine the connections between gadgets in the gadget layout so that all incoming/outgoing stabilizers are measured/generated with a controlled locality in the graph of gadgets $G_\text{gadget}=(V_\text{gadget},E_\text{gadget})$, where $V_\text{gadget}$ is the collection of all gadgets and $E_\text{gadget}$ is the collection of bonds that connect them. For simplicity, we first assume that $\mathcal{S}=\overline{\mathcal{S}}$ are CSS stabilizer codes, and that the minimum-weight generating set of stabilizers of $\mathcal{S}$ is specified by the Tanner graph $G_\mathcal{S}=(V_\mathcal{S},C^X_\mathcal{S},C^Z_\mathcal{S},E_\mathcal{S})$, where $V_\mathcal{S}$, $C^{X}_\mathcal{S}$ and $C^{Z}_\mathcal{S}$ are collections of vertices that represent the physical qubits, $X$ checks, and $Z$ checks, respectively.\footnote{Note that when $m=1$, we have $V_\mathcal{S}=V_\text{gadget}$.}

Within our framework we impose a \textit{strict-$k$ locality-preserving} design rule for gadget layouts, which we refer to as the strict-$k$ locality-preservation condition (SLPC).\footnote{We also use SLP as an acronym for strict-$k$-locality-preserving when using it as an identifier for the gadget layout or the dynamical code. In the case where SLPC is not obeyed, we will use the nomenclature non-SLP.} Concretely, for every check $c\in C^X_\mathcal{S}$ or $C^Z_\mathcal{S}$ whose participating data qubits are $V_s\equiv N(c)\subseteq V_\mathcal{S}$, we demand that there exists a Pauli web of $s$ from the incoming or outgoing legs that stays within the bonds connecting the gadgets $V_{g|s}$ that replace the data qubits in $V_s$. That is, the internal part of the Pauli web (represented by the coloured bonds in the plane in Fig.~\ref{fig: tclattice}) is supported on a collection of edges $E_{g|s}\subseteq E_\text{gadget}$, such that for every $e\in E_{g|s}$, $\phi(e)\in V_{g|s}$.

We emphasize that SLPC is a structural ingredient of our construction, ensuring that the compiled dynamical code mirrors the locality structure of the underlying stabilizer code and making the gadget layout geometrically transparent. However, SLPC by itself is neither necessary nor sufficient for fault-tolerance. The actual fault-tolerance requirement is captured by the macroscopic logical support condition  introduced in Sec.~\ref{sec: spacetimed}, which demands that logical operators do not map to short bond operators.

We now present an algorithm that, given a Tanner graph for $\mathcal{S}$, constructs a gadget layout obeying SLPC. The construction takes intuition from the fact that each gadget gives rise to an isometry between the incoming/outgoing legs and the internal legs: the bond operators of $X$ and $Z$ on the internal legs between the gadgets must preserve the Pauli algebra of $X$ and $Z$ on the incoming/outgoing legs.

\begin{algorithm}[H]
  \caption{Gadget connectivity with SLPC}
  \label{alg: gadgetcon}
  \textbf{Input:} $G_\cS=(V_\cS,C^X_\cS,C^Z_\cS,E_\cS)$, the Tanner graph of the minimum weight generating set of a CSS code $\cS$.
  
  \textbf{Output:} $G_\text{gadget}=(V_\text{gadget},E_\text{gadget})$, the gadget connectivity graph that satisfies SLPC.
  
   \begin{algorithmic}[1]
   \State For every lattice site with $m$ data qubits, replace all the $m$ data qubits in $V_\cS$ with a gadget vertex in $V_\text{gadget}$. Record the replacement map $R:\ V_\cS\to V_\text{gadget}$.
   \State Initialize $E_\text{gadget}=\{\}$.

   \State For $i=1$ to $n$,
   
   \ \ For $j>i$ to $n$,

   \ \ \ \ If $R(v_i)\neq R(v_j)$, do step 4. Otherwise, continue.

   \State If $\exists c^X\in C^X_\cS, c^Z\in C^Z_\cS$ such that $\{v_i,v_j\}\subseteq N(c^X)\cap N(c^Z)$, add an edge between $R(v_i)$ and $R(v_j)$ to $E_\text{gadget}$ if such an edge does not exist in $E_\text{gadget}$.
   \end{algorithmic}
\end{algorithm}
We note that such an algorithm does not give details such as the bond dimension, i.e. the number of internal legs needed on each bond. This will be determined in the following discussion in Sec. \ref{sec: gadgetencoder}. 

We expect the Algorithm \ref{alg: gadgetcon} to be straightforwardly applicable in the case where $\cS\neq\overline{\cS}$. Furthermore, the algorithm is likely generalizable to constructing gadget layout of non-CSS stabilizer codes,in which case the step $4$ in Algorithm \ref{alg: gadgetcon} is modified so that whenever there exist two checks $c_{1,2}$ such that $\{v_i,v_j\}\subseteq N(c_1)\cap N(c_2)$, draw an edge between $R(v_i)$ and $R(v_j)$ if the support of $c_1$ and $c_2$ on at least one of the two data qubits $v_i$ and $v_j$ does not commute.  
Such a modification for non-CSS stabilizer codes is sufficient to guarantee SLPC. But there may exist more efficient ways with fewer inter-gadget connections. However, those gadget solutions will always be a subset of our solutions for the gadget connection from the algorithm.

\subsubsection{The encoding matrices of the gadget}
\label{sec: gadgetencoder}
Now that we have a gadget layout, we need to solve the bond operators for every incoming and outgoing stabilizer. As we demonstrate below, this problem can be reformulated as the construction of Clifford encoders at the gadget level. 

From the perspective of a lattice site, the $m$ data qubits there participate in multiple stabilizers in $\cS$ and $\overline{\cS}$. This means that, in the dynamical code $\M$, the gadget that replaces the $m$ data qubits involves in Pauli webs of all these stabilizers. The Pauli web of every stabilizer is supported on either the incoming or outgoing legs (but not both), and a collection of internal legs depending on the spatial connection between the gadgets, where its support are the bond operators that we are trying to solve.
As a result, each gadget functions as an encoder of a local stabilizer code that encodes $2m$ logical qubits, as represented by the $m$ incoming and $m$ outgoing legs, to $n_L$ physical qubits that correspond to the internal legs. We note again that we can always choose a gadget for every qubit such that $m$ is 1. Furthermore, the $2m$ logical operators of the local stabilizer codes need to be encoded in the $n_L$ physical qubits as bond operators. Note that this local stabilizer code is not directly concatenated with the incoming stabilizer code throughout the dynamics. Rather, in an actual circuit implementation of $\M$, the $n_L$ internal legs give rise to a temporal sequence of inter-gadget interactions along the spatial direction of the connected internal legs. Therefore, we call this a \textit{spacetime concatenation} where the dynamics induced by the circuit $\M$ is ``concatenated'' to the incoming stabilizer code. We discuss more details in Sec. \ref{sec: fixfloquet}.

Since we only consider $\M$ to be a Clifford circuit and $\cS/\overline{\cS}$ to be Pauli stabilizer codes, every bond operator is a string of Pauli operators whose length is equal to the number of internal legs on the bond. In this case, every gadget is a Clifford encoder that encodes the incoming/outgoing Pauli operators to Pauli operators supported on the internal legs. The encoding map of the gadget can be described by a binary matrix $H=(H_X|H_Z)$.
Every incoming/outgoing stabilizer in which the gadget is involved gives rise to a row $\vec{v}_{X/Z}\in \F_2^{2m+n_L}$ in $H_{X/Z}$. The first $m$ entries of $\vec{v}_{X}$ and the first $m$ entries of $\vec{v}_{Z}$ represent the incoming Pauli operator $\prod_{j=1}^m X^{\vec{v}_{X,j}}_jZ^{\vec{v}_{Z,j}}_j$ on the incoming leg that is to be encoded. Similarly, the following $m$ entries in $\vec{v}_{X/Z}$ represent the outgoing Pauli operator. The remaining $n_L$ rows of $\vec{v}_{X/Z}$ represent the Pauli string acting on the $n_L$ legs of the bonds. In total, if $n_\text{in/out}$ is the number of stabilizers in which the lattice site participates in the incoming/outgoing stabilizer groups, the encoding matrices will have dimensions $H_X, H_Z\in \F_2^{(n_\text{in}+n_\text{out})\times (2m+n_L)}$. In order for the encoding map to preserve Pauli algebra, the two matrices must satisfy the symplectic condition
\begin{equation}\label{eq: cliffordcon}
    H_XH_Z^T+H_ZH_X^T=0\mod 2.
\end{equation}

A simpler encoding scheme for the gadget is CSS encoding, which only works if both the incoming and outgoing stabilizer groups are also CSS codes.
In this case, any gadget encodes incoming/outgoing $X$ ($Z$) operators exclusively to $X$ ($Z$) operators on the internal legs. 
Therefore, each incoming/outgoing $X$/$Z$ stabilizer in which the lattice site is involved gives rise to a row $v_{X/Z}\in \F_2^{2m+n_L}$ in the encoding matrix $H_X$/$H_Z$. The consistency condition is then
\begin{equation}\label{eq: csscon}
    H_XH_Z^T=0\mod 2.
\end{equation}
One benefit of CSS encoding, as we will see in Sec. \ref{sec: fixzx}, is that the entire dynamical code can be represented by a ZX-diagram where every node can only carry phases $0$ or $\pi$, which means that it can be implemented using only $X$/$Z$ basis measurements only.
Since CSS encoding is a strict subset of Clifford encoding, we expect a general Clifford encoding of the gadgets could give rise to a more efficient\footnote{ The number of internal legs $n_L$ as obtained from the solution to the consistency equation quantifies the efficiency for a dynamical code in the following sense. From the ZX-diagram point of view, the number of internal legs $l$ on each bond that connects a pair of gadgets gives the number of entangling gates or measurements between this pair of gadgets. Therefore, the total number of internal legs represents a combined overhead of the dynamical code, which includes the depth of the circuit, the number of physical qubits needed in each gadget, and the non-local connectivity needed in the circuit. In general, with syndrome extraction error, every dynamical code should be repeated $O(d)$ times. Our notion of efficiency is therefore a quantification of the factor in front of the order $d$.} dynamical code compared to a CSS dynamical code; see Sec. \ref{sec:resource} for more details. However, in some examples, we find that the CSS and non-CSS dynamical code solutions of the consistency equations both have the same and the smallest possible number of internal legs $n_L$.

After obtaining solutions to the consistency condition(s), it is important to
check whether the encoding is maximal, i.e. whether the encoding matrices are
full-rank. The rank of the encoding matrix is
\begin{equation}
    r \equiv \mathrm{rank}(H) = \mathrm{rank}(H_X \mid H_Z)
\end{equation}
for a general Clifford encoding and
\begin{equation}
    r \equiv \mathrm{rank}(H_X) + \mathrm{rank}(H_Z)
\end{equation}
for a CSS encoding. If $r < 2m + n_L$, we have to impose additional
stabilizers on the internal legs to completely fix the encoding map of the
gadget. In practice, these additional internal stabilizers are usually easy to
identify with the help of symmetries in the stabilizer code layout.

Finally, we must check whether the dynamical code obtained from the gadgets is
fault-tolerant, i.e. whether it preserves logical information from the incoming
to the outgoing stabilizer code in the sense of Eq.~\eqref{eq: ml}. At the
level of individual examples, this can often be verified explicitly by
tracking the Pauli web of each incoming logical operator through the gadget
layout, which is what we do for most of the constructions in
Sec.~\ref{sec:examples} and the Appendices. In general, in
Sec.~\ref{sec:fault_tolerance} we formalize this as the Macroscopic Logical Support (MLSC) condition, which
requires that the image of every logical under the dynamical map has
macroscopic weight. We prove there that any dynamical code produced by our
algorithm that satisfies MLSC has spacetime distance at least the static code
distance. The SLPC introduced
above is used as a convenient structural constraint in our constructions, but
it is not required as an additional condition for fault-tolerance.

\subsubsection{Constructing the ZX-diagram inside the gadget from encoding maps}\label{sec: fixzx}

After obtaining the maximal encoding maps of the gadgets, we construct the ZX-diagram inside the gadgets accordingly to arrive at the circuit implementation of the dynamical code. To this end, we note that diagrammatic algorithms in ZX-calculus to compile the encoding circuit from CSS encoding maps have already been proposed in Ref. \cite{kissinger_phase-free_2022}, while the generalization to the Clifford encoding is given in Ref. \cite{khesin_universal_2025}. For completeness, we provide a brief summary of the algorithms proposed in these references.

We first discuss the way to compile the encoding circuit for CSS encoding, which, as pointed out by Ref. \cite{kissinger_phase-free_2022}, only requires the knowledge of one of the CSS encoding matrices, $H_X$ or $H_Z$. We will use $H_X$ in the following construction. To begin with, each incoming, outgoing, and internal leg is assigned a $X$ node. For each row $\vec{p}_i$ of $H_X$ ($i=1,2,\dots$ is the row index), we draw a $Z$ node. The $i$-th $Z$ node is then connected to the $X$ nodes of the legs that have a non-zero entry in $\vec{p}_i$ (see Fig. \ref{fig: encode}(a)). This ensures that the corresponding stabilizer $X$ $s_{X,i}=\prod_{j=1}^{2m+n_L}X_j^{(\vec{q}_i)_j}$ is stabilized. At this point, the CSS encoding circuit is already complete.

The ZX-diagrammatical compilation of a general Clifford encoding circuit is a literal implementation of the algorithm that is used to prove the Gottesman-Knill theorem~\cite{aaronson_improved_2004,khesin_universal_2025}. To begin with, we turn every row $(\vec{{v_X}_i}|\vec{{v_Z}_i})$ ($i=1,2,\dots, 2m+n_L$ is the row index) of the encoding matrix $(H_X|H_Z)$ into a Pauli string $\prod_{j=1}^{2m+n_L} X_j^{(\vec{v_X}_i)_j}Z_j^{(\vec{v_Z}_i)_j}$. Since the gadget is maximally stabilized, the set of Pauli strings forms a complete tableau $[s_1,s_2,\dots, s_{2m+n_L}]$. The core of the algorithm is to disentangle the legs one by one from the rest of the legs in the stabilizer tableau using Clifford circuits. For example, starting from any stabilizer $s_1$ that involves the first leg, denoted by subscript 1, we apply a Clifford circuit $U_1$ between it and the rest of the legs, so that $U_1 s_1 U_1^{-1}=Z_1$, that is, the first leg is disentangled from all other legs. For every other stabilizer $s_2$ in the tableau, we check whether $U_1 s_2 U_1^{-1}$ still involves the first leg. Since every pair of stabilizers must commute, if this happens, the involvement can only be $Z_1$. Therefore, the first leg can be disentangled from $s_2$ by multiplying the two stabilizers after the Clifford unitary. The first legs are then measured on the basis of $Z$ and discarded, and the rest of the stabilizers in the tableau will be updated by conjugating with $U_1$. The algorithm continues until all legs are disentangled and measured (see Fig. \ref{fig: encode}(b)). 

\begin{figure*}
    \centering
    \includegraphics[width=0.75\linewidth]{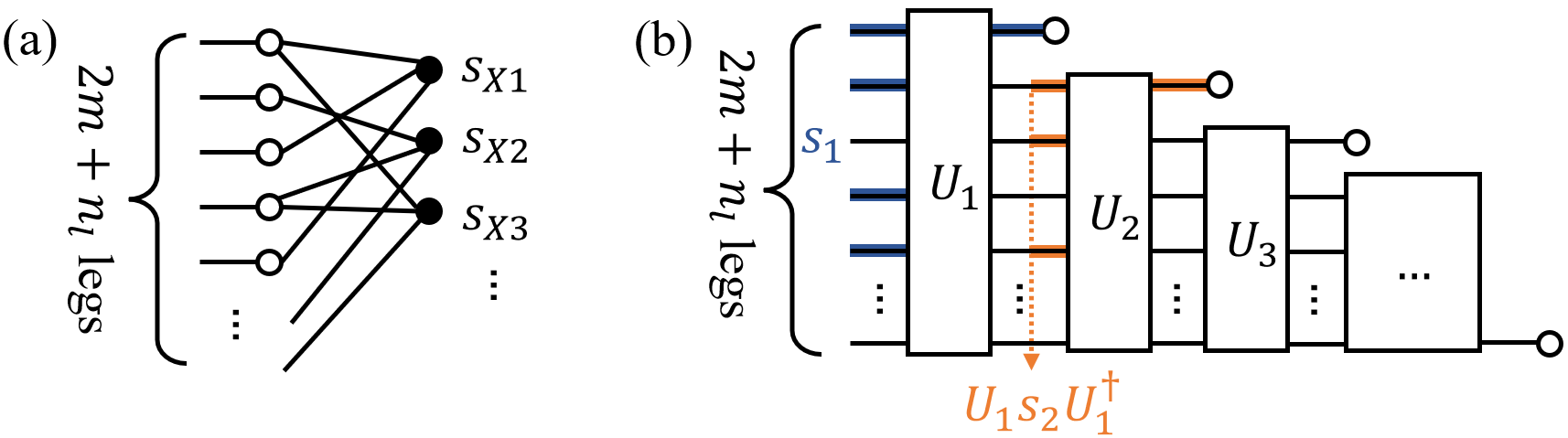}
    \caption{Compiling the encoding circuit for the gadget. (a) The CSS encoder. Every leg of the gadget enters through an $X$ node and every $X$ stabilizer gives a $Z$ node which is connected to the $X$ nodes based on the $X$ stabilizer. (b) The Clifford encoder. Here $U_{1,2,3,\dots}$ are Clifford disentangling circuit, each of them is followed by a projective measurement of the disentangled qubit in $Z$ basis. The first two steps for stabilizers $s_1$ and $s_2$ are explicitly shown in the figure.}
    \label{fig: encode}
\end{figure*}

\begin{figure}
    \centering
    \includegraphics[width=0.75\linewidth]{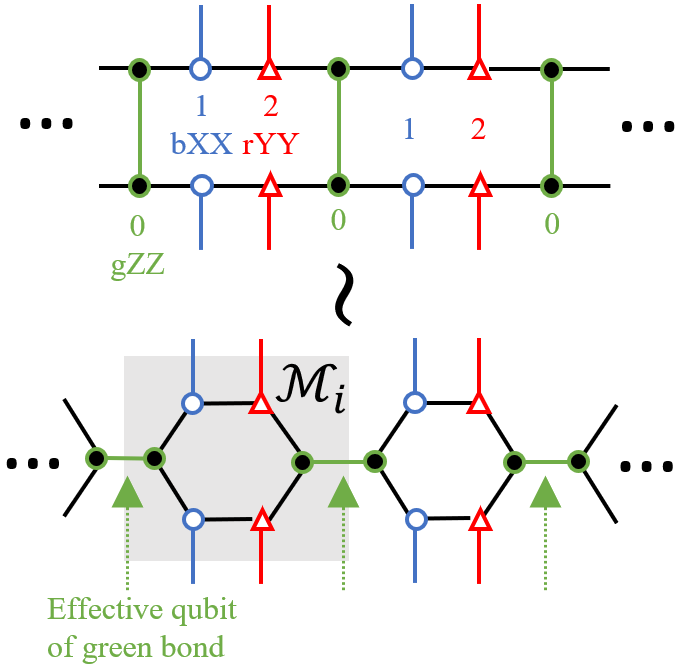}
    \caption{The ZX-diagram of the two qubits connected by a green bond in the HH Floquet code. Here the blue and red lines are connected to nearby qubits, which are not shown in the diagram. The coloring of these lines are mere indications of the actual direction of the leg on the lattice. Using rules of ZX calculus, we can deform the $ZZ$ pairwise measurement so that the effective qubit formed after the 0th round measurment is represented by a single line. The gadget $\M_i$ in this case, marked by the shaded region, has 1 incoming, 2 outgoing and 4 internal legs. The triangle nodes are $Y$ nodes, which are defined in Fig. \ref{fig: tcxyz}.}
    \label{fig: tceffective}
\end{figure}
\begin{figure*}
    \centering
    \includegraphics[width=0.65\linewidth]{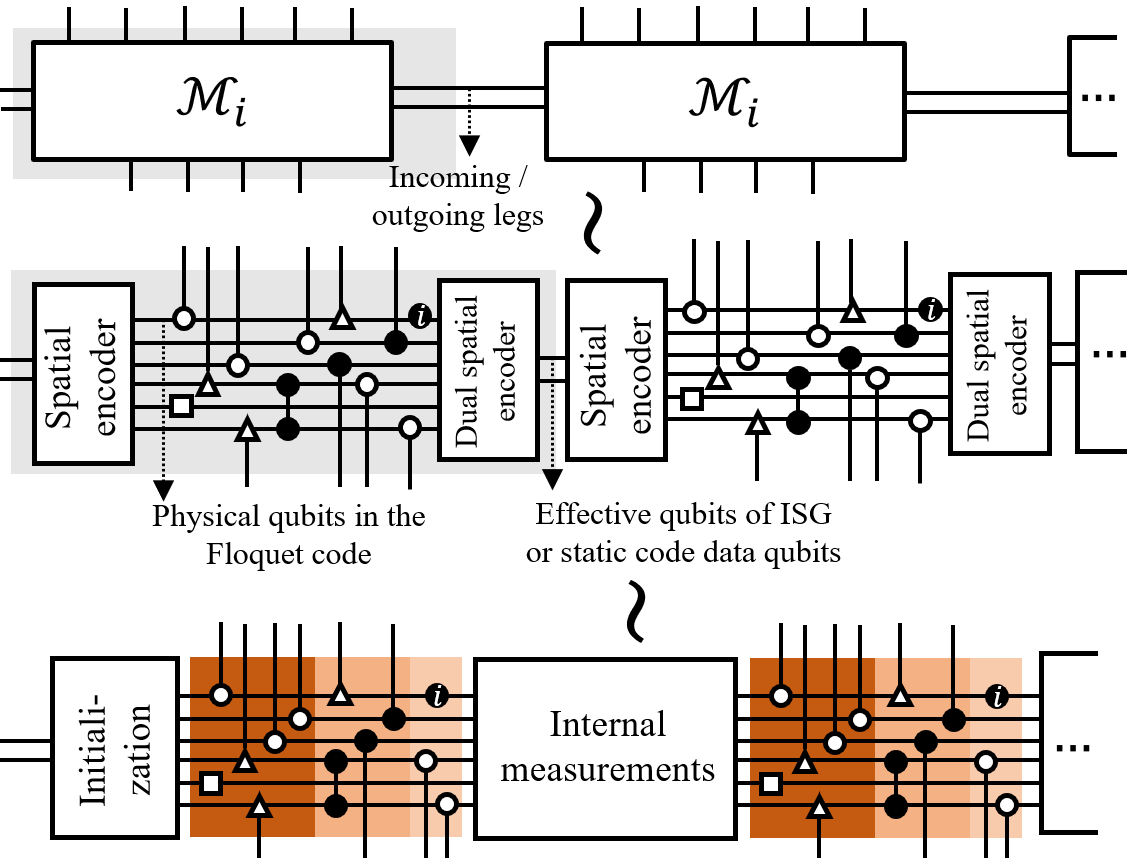}
    \caption{The workflow from the ZX-diagram in the gadget layout to the circuit layout of the Floquet code. The first step is to manipulate the ZX-diagram of the gadget $\M_i$ towards a three-stage layout in the second row, which exposes the spatial and temporal concatenation separately. The second step is to merge dual spatial encoder with the ensuing spatial encoder, which leads to a measurement circuit of all the stabilizers of the code in the spatial concatenation. For clarity, we deferred the time step assignment to the third row. For the example ZX-diagram in the figure, the internal legs can be packed into three time steps. The warm-up of the dynamical code includes the ``initialization" block (where physical qubits will be initialized in a fixed basis based on ZX-diagram of the spatial encoder), the inter-gadget dynamics and the following ``internal measurements" block.}
    \label{fig: floqueteff}
\end{figure*}

\subsubsection{Bringing the gadget layout to the circuit of the dynamical code}
\label{sec: fixfloquet}

Until this point, we have obtained the ZX-diagram of the circuit of the dynamical code. To physically implement the dynamical code, we have to bring the ZX-diagrams of the gadgets to an actual circuit layout. To this end, it is helpful to first gather some intuitions from existing examples of Floquet codes.

After each round of measurement, the physical qubits of a Floquet code are in a stabilizer code state, whose stabilizer group is known as the instantaneous stabilizer group (ISG). 
The ``data qubits" of the ISG are not single physical qubits, but rather effective qubits supported on a set of physical qubits as a result of measurements. In other words, the stabilizer group of interest is concatenated with a local code whose stabilizers are described by the measurements. 
In the ZX-diagram, every effective data qubit of the ISG can be represented via a single leg by deforming the ZX-diagram of the round of measurements. 
Consider again, for example, the HH Floquet code. 
After the measurement round 0 (gZZ), we can deform the pairwise measurement circuit so that each pair of qubits being pairwise measured merge into a single line, which represents the effective qubit created by the measurement (see Fig. \ref{fig: tceffective}). We can then treat the ZX-diagram between two 0 round measurements as a black box which has 1 incoming and 1 outgoing legs, and 4 internal legs which connect to nearby physical qubits. Therefore, we have brought the ZX-diagram of the HH Floquet code between measurement rounds 0, 1, 2 and 0 to a gadget-layout form, where the ISG of round 0 now becomes the incoming and outgoing stabilizer group of the dynamical code. 

Building upon the example of the HH Floquet code, to work our way from the gadget to a circuit layout of the dynamical code, we need to rewrite the ZX-diagram into three consecutive stages (see the second row of Fig. \ref{fig: floqueteff}): 
\begin{enumerate}
    \item The first stage is a spatial encoder that maps the $m$ data qubits in the incoming stabilizer code $\cS$ to the $m_F$ outgoing legs, which will become physical qubits of the final Floquet code. This stage can be understood as a ``spatial concatenation" of the incoming stabilizer code.
    \item The second stage is where $m_F$ physical qubits are engaged in entangling dynamics with physical qubits in other connected gadgets. Of course, there can still be intra-gadget gates and measurements during this stage. Over all, we want to ensure that the ZX-diagram in this stage has an actual \textit{circuit layout}. This means that when internal legs between different gadgets are connected during this stage of the ZX-diagram, they correspond to physically implementable two-qubit Clifford gates or pairwise measurements, whose representation in ZX-calculus are shown in Sec. \ref{sec: zxrules}.
    This stage can be regarded as ``temporal concatenation", since every physical qubit from the spatial concatenation now undergoes inter-gadget temporal dynamics, which measures/generates stabilizers of the stabilizer codes. This can be also seen as a time evolution of the spatial concatenation because it changes across the dynamical code.
    \item The third stage is a dual spatial encoder that map the $m_F$ physical qubits back to $m$ data qubits in the outgoing state code $\overline{\cS}$. 
\end{enumerate}

To reach the circuit layout of the Floquet code, we repeat the dynamical code over time (if it is directly repeatable). In this way, the dual spatial encoder and the ensuing spatial encoder merge into a single circuit block that is completely local to the gadget (see the third row of Fig. \ref{fig: floqueteff}). In fact, merging a dual spatial encoder and a subsequent encoder yields a local SASEC that measures all stabilizers in the local code represented by the spatial encoder\footnote{To see this, we consider the operator map of the merged circuit of the dual spatial encoder and the ensuing encoder, which has $m_F$ incoming and $m_F$ outgoing legs correspoding to the $m_F$ physical qubits in the gadget. This merged circuit maps all the incoming and outgoing stabilizer of the local code to $c\mathbbm{1}$, since they are stabilized the (dual) spatial encoder. It also acts trivially on the $m$ logical operators of the local code which are represented by the $m$ legs between the dual encoder and the encoder. Therefore, from the completeness of ZX-calculus\cite{backens_zx-calculus_2014,wetering_zx-calculus_2020,jeandel_completeness_2020}, the merged circuit must be ZX-equivalent to an SASEC of the local code. } (but not the full spatial code after concatenation). Therefore, we dub it the ``internal measurements" circuit in Fig. \ref{fig: floqueteff}. As such, the ZX-diagram of the internal circuit can be easily brought to a circuit layout where the number of physical qubits stays the same. 

The last subtle but important step towards the final Floquet measurement schedule is to ``synchronize" all pairs of connected gadgets from the final ZX-diagram of the gadget in the circuit layout. The necessity of synchronization arises for the following reason. When the ZX-diagram of the gadget is brought to the circuit layout using ZX-rules, there comes a natural temporal order of internal legs in the second stage where each internal leg can be assigned with a time step. In order for a pair of internal legs from two connected gadgets to be interpreted as a pairwise measurement (or a two-qubit unitary gate), the two internal legs must be at the same time step. However, this is not always guaranteed. One possible way of synchronization is to look for alternative ways that the gadget can be turned to a circuit layout, which is usually not unique, as indicated in Sec. \ref{sec: fixzx}. The temporal order of each internal legs will be different in these other circuit layouts. Therefore, switching one of the two gadgets into this alternative circuit layout may synchronize a connected pair of internal legs that were not synchronized before. However, it is totally possible that no matter how one alters the circuit layout of each gadget, all the internal legs are still not fully synchronized.
If this happens, a simple solution is to introduce one additional ancilla for each unsynchronized pair of legs. 

Fortunately, if the stabilizer code is defined with data qubits on two sublattices, we can always find a gadget layout so that the connected gadgets form a bipartite graph where the gadget on one sublattice only connects to gadgets of the other sublattice. In this case, synchronization is always possible without the need of extra ancillas: once the time steps of the gadgets on one sublattice are fixed by the circuit layout of the ZX-diagram, the schedule of the gadgets on the other sublattice can always be arranged accordingly due to the bipartite connectivity. In most of our examples in this paper, the connected gadgets indeed form bipartite graphs, with the only exception being the Floquet Haah code. 

We note that, in the circuit layout, the warm-up stage of dynamical code corresponds to the combination of the ``initialization" block of spatial concatenation, inter-gadget dynamics and the ``internal measurements" block in the third row of Fig. \ref{fig: floqueteff}. The data qubits, represented by the $n$ incoming legs, can be initialized in an arbitrary initial state, while the logical space will be dynamically generated and stabilized once we go through these three steps. The fault-tolerance of this procedure will be elaborated on in more detail in Sect. \ref{sec: spacetimed}.

Before moving on, we point out that the ZX-diagram of the gadget can be brought to multiple different three-stage layouts, depending on the (dual) spatial encoder we choose. Switching between different spatial encoders may change the number of physical qubits and also the depth of the circuit of the dynamical code. We will analyze these aspects in more detail in Sec. \ref{sec:resource}. Moreover, the spatial encoder and the dual spatial encoder do not have to the ones of the same stabilizer code. This guarantees that every ZX-diagram of the gadget can be turned into the three-stage layout.

\section{Floquet code protocols using spacetime concatenation}
\label{sec:examples}

In this section, we illustrate the versatility of our spacetime concatenation framework through concrete examples of Floquet code constructions. We begin by reproducing all examples of Floquet toric codes, showing how our formalism naturally reproduces these well-known examples. We then present new constructions that extend beyond previously studied cases, such as the Floquet bivariate bicycle code and the Floquet Haah code. These examples highlight the generality of our modular approach in producing Floquet codes aka dynamically generated logical spaces. We also provide, in Appendix~\ref{subsec:steane}, an instructive and minimal example: a Floquet version of the Steane code constructed using gadgets of the Floquet color code in Appendix~\ref{sec: csscc}. In addition to this, we also provide the Floquet checkerboard model in Appendix~\ref{sec: csscb} and a pairwise measurement $\Z_2^{(1)}$ subsystem code in Appendix~\ref{sec: fermion}.
These examples serve as a pedagogical illustration of the principles underlying our framework, particularly how familiar static codes can be adapted to measurement-driven Floquet dynamics.     
For each construction, we explicitly verified that the MLSC is satisfied.

\subsection{Primer: reproducing all Floquet toric codes}

\begin{figure}
    \centering
    \includegraphics[width=0.8\linewidth]{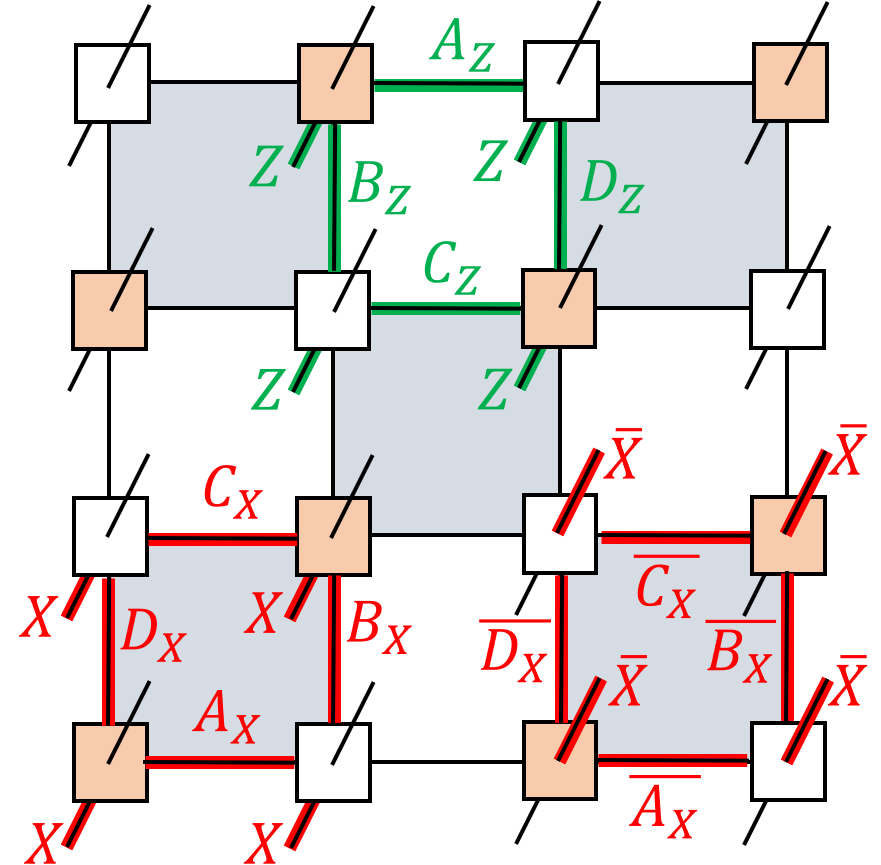}
    \includegraphics[width=0.45\linewidth]{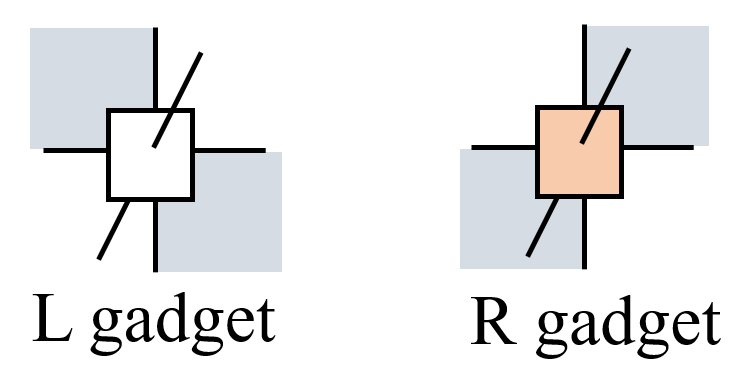}
    \caption{The gadget layout for toric code. Here the L/R gadget are marked by white/orange shaded squares, the incoming/outgoing legs are pointing to the bottom left/upper right, and the internal legs are connected in the plane. The weight-4 $X$/$Z$ stabilizers are located at the shaded/white squares. The incoming/outgoing $X$ and $Z$ operators are mapped to bond operators which are supported on the internal legs, which are denoted by $A_{X/Z}$ to $D_{X/Z}$ and $\overline{A_{X/Z}}$ to $\overline{D_{X,Z}}$.}
    \label{fig: tcgadget}
\end{figure}
\begin{figure*}
    \centering
    \includegraphics[width=0.65\linewidth]{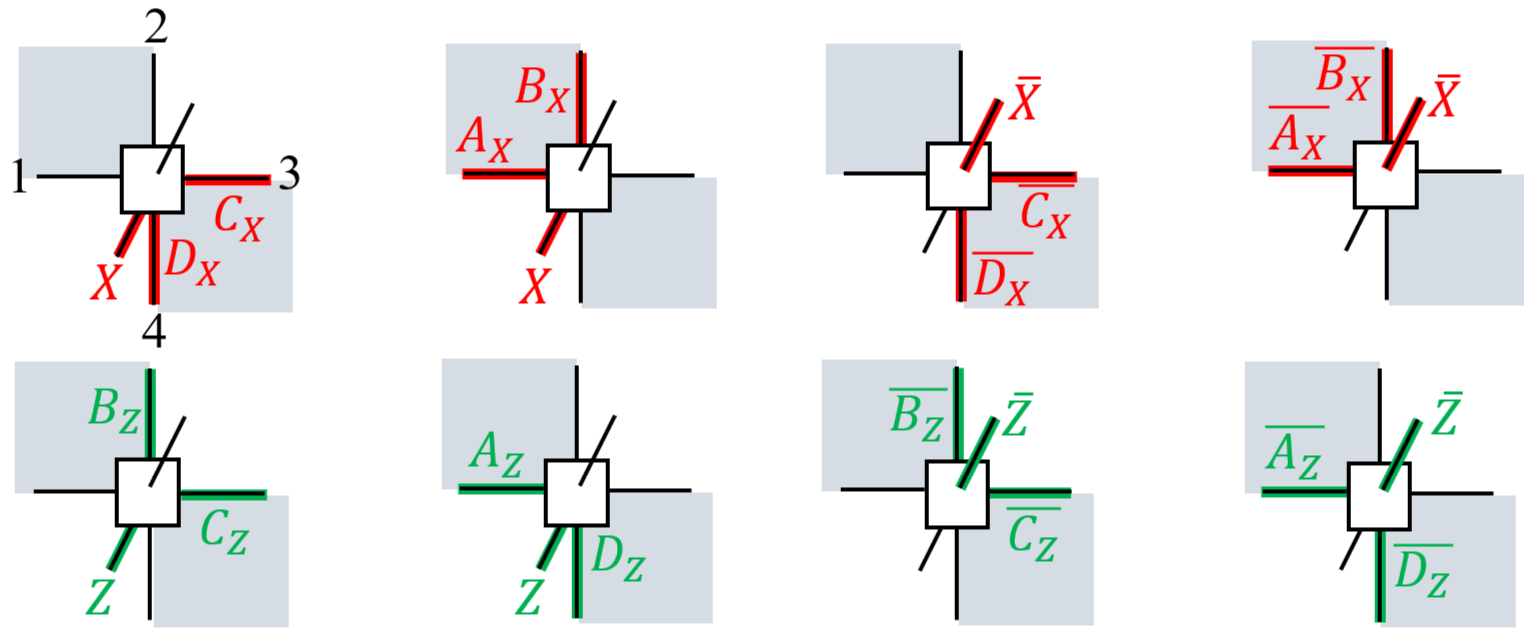}
    \caption{ The encoding map of the incoming/outgoing $X$ and $Z$ operators of the L gadget from the parametrized bond operators in Fig. \ref{fig: tcgadget}. The four bond directions of the internal legs are marked in the first figure from 1 to 4. We use $\overline{X}$ and $\overline{Z}$ to indicate Pauli operators supported on the outgoing leg.}
    \label{fig: tcparam}
\end{figure*}
\begin{figure*}
    \centering
    \includegraphics[width=0.7\linewidth]{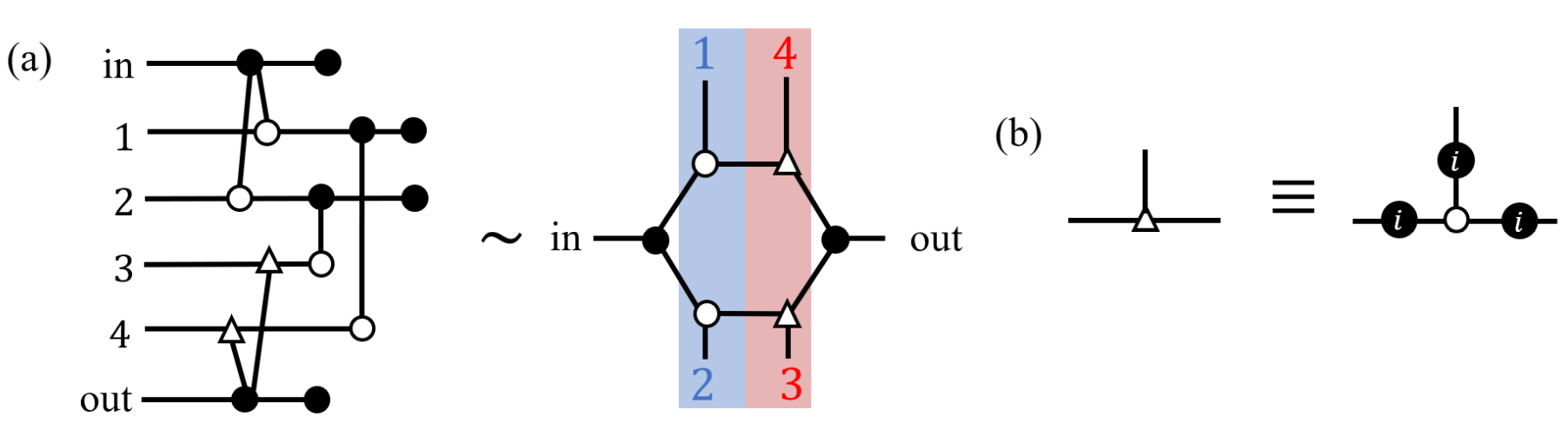}
    \caption{(a) The ZX-diagrammatic Clifford encoder from the complete stabilizer tableau in Eq. \eqref{eq: tctableau} with $P=X$ and $Q=Y$. Here we choose to disentangle the incoming (in) and outgoing (out) legs from the internal legs first. The diagram can be brought to a circuit layout where the spatial concatenation is two-qubit repetition code, and the internal legs form two time steps marked by blue/red shades. (b) We define the Y node, denoted by a triangle on the ZX-diagram, to be an $X$ node whose every leg being conjugated by the $S$-gate, which is a $\pi/2$ phase $Z$ node.}
    \label{fig: tcxyz}
\end{figure*}
\begin{figure}
    \centering
    \includegraphics[width=0.7\linewidth]{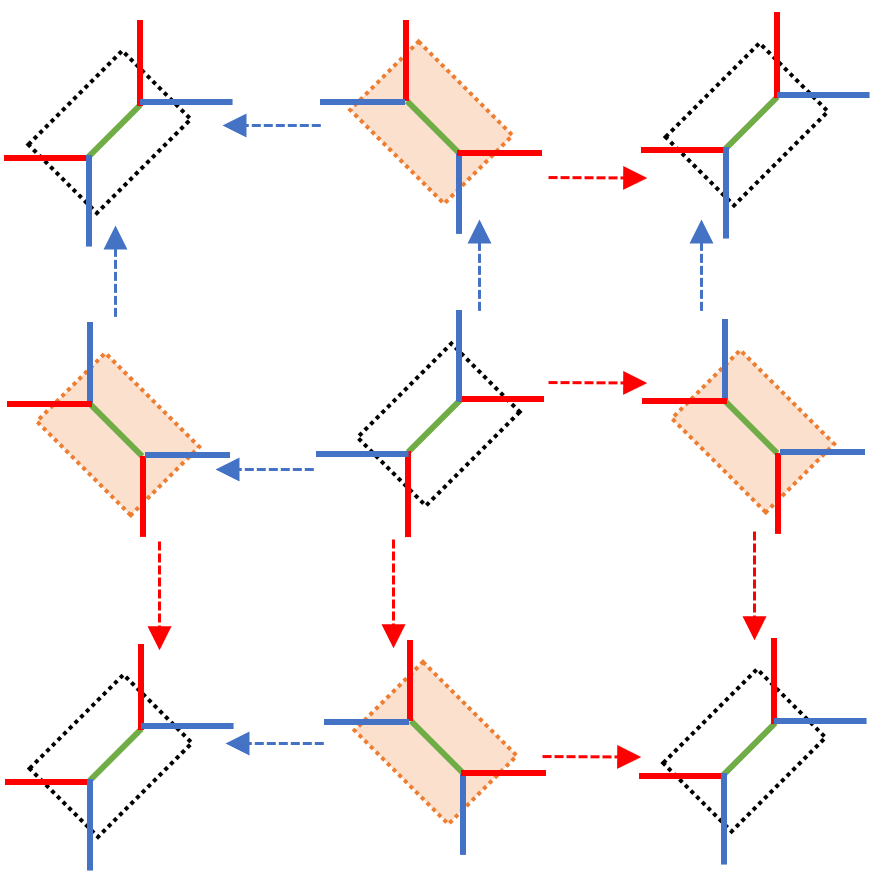}
    \caption{The synchronization between the L and R gadgets for the Floquet toric code in a planar layout, which are marked by dotted rectangles in white and orange, respectively. From the ZX-diagram in Fig. \ref{fig: tcxyz}, the internal legs $1$ and $2$ ($3$ and $4$) of the L gadget are the first (second) time step, which are marked in blue (red) in the figure to be consitent with the coloring of time steps in Fig. \ref{fig: tcxyz}. The orientation of the nearby R gadgets is adjusted according to the L gadget in the middle. Further steps of synchronization are marked by dashed arrows. The outcome is the square-octagon layout of the physical qubits.}
    \label{fig: tcsync}
\end{figure}

As a simple illustration of our approach, we discuss the construction of a dynamical code for 2D toric code with the surface code layout, i.e. the data qubits are located at the lattice sites of a 2D square lattice and the weight-4 $X$ and $Z$ stabilizers are defined on alternating squares, see Fig. \ref{fig: tclattice}. We mainly focus on the case where the incoming and outgoing stabilizer groups are exactly the same toric code, that is, $\cS=\overline{\cS}$.

The lattice sites in this case can be divided into two sublattices, which we will refer to as left (or simply L) and right (or simply R)\footnote{The names will be related to the L/R qubits and L/R gadgets in our discussion in the Floquet BB code in Sec. \ref{sec: bicycle}}.
To perform the gadget layout, we replace each qubit with a gadget with one incoming and one outgoing leg. Gadgets that replace the L/R qubits will be referred to as L/R gadgets.  Given the underlying square lattice of the toric code, we consider the internal leg layout along the square lattice so that each L gadget is connected to four nearby R gadgets; see Fig. \ref{fig: tcgadget}. In this way, we only need to consider the encoding maps of the L (or R) gadgets, after which the encoding map of the R (or L) gadgets will be specified accordingly.

With such an internal leg layout, we show that the most leg-efficient Clifford dynamical code is the HH Floquet code, while the CSS Floquet toric code and the SASEC are both the most efficient CSS dynamical code. In fact, in Sec. \ref{sec:classification}, we prove that the CSS Floquet toric code and the SASEC are equivalent to each other on the level of ZX-diagram, which is also shown in Ref. \cite{bauer_xy_2024}.

\subsubsection{Clifford encoding}\label{sec: tcclifford}
Following the algorithm in Sec. \ref{sec: compile}, we need to first parametrize the bond operators for each incoming and outgoing $X$ and $Z$ stabilizer. From the SLPC, we assume that the mapping of the $X$ Pauli operator is supported on the bonds on which the $X$ stabilizer loop of the incoming stabilizer code is mapped to. We will see below that solutions with such an assumption exist. We label the four bond operators around the loop of the incoming $X$ stabilizer as $A_X$, $B_X$, $C_X$ and $D_X$ in Fig. \ref{fig: tcgadget}. We further assume that the gadgets are translation invariant, so that each incoming $X$ stabilizer is mapped to a translation of the same set of bond operators.
The outgoing $X$, incoming $Z$, and outgoing $Z$ stabilizers are parametrized similarly. From the bond operators, we obtain eight parametrized encoding maps of the L gadget given by the bond operators of the incoming/outgoing $X$ and $Z$ stabilizers, which are shown in Fig. \ref{fig: tcparam}.

The consistency equation can be written down using the encoding maps as follows. 
For example, according to Fig. \ref{fig: tcparam}, the incoming $X$ can be mapped to bond operators $X\to A_{X,1}B_{X,2}$, which means the bond operator $A_{X}$ on the bond along direction 1 together with $B_X$ on the bond along direction 2. Similarly, the incoming $Z$ can be mapped to $Z\to A_{X,1}D_{Z,4}$. We require the algebra of the incoming operators to be preserved under this mapping. Then, to be consistent with the anticommutation relation between the incoming $X$ and $Z$, the two sets of bond operators that are mapped from $X$ and $Z$ must anticommute. Since they only overlap in direction 1, we have $\{A_X,A_Z\}=0$. Repeating this for every pair of maps of Pauli $X$ and $Z$ operators that overlap on at least one direction, we arrive at the following set of equations:
\begin{align}\label{eq: tceqs}\nonumber
    &\{A_X,A_Z\}=\{\overline{A_X},\overline{A_Z}\}=0,\ (A\leftrightarrow B,C,D),\\\nonumber
    &[A_X,\overline{A_Z}]=[\overline{A_X},A_Z]=0,\ (A\leftrightarrow B,C,D),\\\nonumber
    &[A_XB_X,\overline{A_X}\overline{B_X}]=0,\ [C_XD_X,\overline{C_X}\overline{D_X}]=0,\\
    &[A_ZB_Z,\overline{A_Z}\overline{B_Z}]=0,\ [C_ZD_Z,\overline{C_Z}\overline{D_Z}]=0.
\end{align}
The first row of equations come from the anticommutation relation between incoming/outgoing $X$ and $Z$ operators; $(A\leftrightarrow B,C,D)$ means there are three more equations in the first row, where $A$ is replaced by $B,\ C$ and $D$. The second row of equations arise since incoming $X$($Z$) operator needs to commute with outgoing $\overline{Z}$($\overline{X}$) operators; again there are three more equations where $A$ is replaced by $B,\ C$ and $D$. The third and fourth rows are due to the commutation between incoming $X$($Z$) and outgoing $\overline{X}$($\overline{Z}$) operators. Here we use $\overline{X}$ and $\overline{Z}$ to denote outgoing Pauli operators.

The set of consistency equations admits simple solutions where every bond operator is a single-qubit Pauli operator as follows:
\begin{align}\label{eq: tcpaulimat}\nonumber
    A_X=B_X=C_X=D_X=P,\\\nonumber
    A_Z=B_Z=C_Z=D_Z=Q,\\\nonumber
    \overline{A_X}=\overline{B_X}=\overline{C_X}=\overline{D_X}=Q,\\
    \overline{A_Z}=\overline{B_Z}=\overline{C_Z}=\overline{D_Z}=P,
\end{align}
where $P$ and $Q$ are two different Pauli operators in $\{X,Y,Z\}$. Therefore, the most leg-efficient Clifford dynamical toric code only needs one internal leg per bond. In fact, all possible choices of $P$ and $Q$ lead to dynamical code equivalent up to a bond-local Clifford unitary on the internal legs that changes $P$ and $Q$. Choosing these bond operators completely fixes the encoding map of the gadget. To see this, we write down the stabilizer tableau of the gadget using the graphical representation of the encoding map in Fig. \ref{fig: tcparam}, which reads 
\begin{equation}\label{eq: tctableau}
    \left[\begin{array}{c|cccc|c}
        X & P & P & I & I & I  \\
        X & I & I & P & P & I\\
        Z & I & Q & Q & I & I\\
        Z & Q & I & I & Q & I\\
        I & I & I & Q & Q & X\\
        I & Q & Q & I & I & X\\
        I & I & P & P & I & Z\\
        I & P & I & I & P & Z\\
    \end{array}\right].
\end{equation}
Here each row is a spacetime stabilizer, and the first and last columns corresponds to the incoming and outgoing legs, while the second to fifth columns are the four internal legs in the directions from 1 to 4 in Fig. \ref{fig: tcparam}. Since the rank of the tableau is equal to the total number of legs of the gadget $6=2m+n_L=2+4$, the tableau is complete. 

We now show that the solution with $P=X$ and $Q=Y$ leads to the HH Floquet code. Using the method outlined in Sec. \ref{sec: fixzx} for Clifford encoders, we construct the ZX-diagram which produces the encoding maps of the L gadget in Fig. \ref{fig: tcxyz} with the solution $P=X$ and $Q=Y$ in Eq. \eqref{eq: tcpaulimat}. After bringing the diagram to the three-stage form, as shown in Fig. \ref{fig: floqueteff}, we see that the spatial concatenation inside the gadget is a two-qubit repetition code, i.e. the original data qubit is replaced by two physical qubits in the Floquet code. Furthermore, internal legs form two time steps inside the gadget: legs 1 and 2 in the first time step \textcircled{\raisebox{-0.9pt}{1}}, while legs 3 and 4 at the second step \textcircled{\raisebox{-0.9pt}{2}}. 
We now fix the spatial layout using the ZX-diagram of the L gadget and synchronize the R gadgets accordingly, which is illustrated in Fig. \ref{fig: tcsync}.  We first fix the orientation of one L gadget according to the ZX-diagram in Fig. \ref{fig: tcxyz}(a). From Fig. \ref{fig: tcgadget}, we see that the R gadget is essentially the L gadget rotated by $90^\circ$, which can be either clockwise or anticlockwise. To lay the physical qubits and the internal legs in the plane, we choose the direction of rotation of the R gadgets based on the orientation of the L gadget that we fixed. In this way, the orientations of the four nearest-neighbor R gadgets are fixed. The other gadgets can be synchronized similarly. We see that the outcome is exactly the square-octagon layout of the HH Floquet code in Fig. \ref{fig: tclattice}. Since the two internal legs in step \textcircled{\raisebox{-0.9pt}{1}} are both connected to $X$ nodes in the ZX-diagram, time step \textcircled{\raisebox{-0.9pt}{1}} consists of two-qubit pairwise measurements in $X$. SImilarly, step \textcircled{\raisebox{-0.9pt}{2}} represents pairwise measurements in $Y$. The measurement schedule of the Floquet toric code is thus given by gZZ, bXX and rYY where gZZ corresponds to the spatial concatenation, bXX corresponds to the XX measurements in the first time step inside the gadget and rYY corresponds to the YY measurements in the second time step inside the gadget. 

The logical automorphism of the HH Floquet code can also be seen from the gadget layout. We denote the two sets of logical operators in the toric code as as \( (\mathbb{X}_1, \mathbb{Z}_1) \) and \( (\mathbb{X}_2, \mathbb{Z}_2) \), where $\mathbb{X}_1$/$\mathbb{Z}_1$ are products of lines of $d$ Pauli $X$/$Z$ operators in the horizontal/vertical direction of the square lattice, and vice versa for $\mathbb{X}_2/ \mathbb{Z}_2$.  From the stabilizer tableau in Eq. \eqref{eq: tctableau}, we see that the gadget is stabilized by $X|PIPI|Z$. Now consider one of the $X$ logical operators that extends in the bond directions 1 and 3. Given the gadget stabilizer, we see that the incoming logical $X$ operator $\mathbb{X}_1$ is mapped by a line of gadgets to the logical $Z$ operator $\mathbb{Z}_2$ of the logical qubit supported on the outgoing legs, which extends in the same direction (see Fig. \ref{fig: tclogical}). Therefore, the dynamical code representing one period of the HH Floquet code performs a logical Hadamard gate plus swapping the two logical qubits:
\[
\mathbb{X}_1 \to \Z_2, \quad \Z_1 \to \mathbb{X}_2,
\]
\[
\mathbb{X}_2 \to \Z_1, \quad \Z_2 \to \mathbb{X}_1.
\]
This corresponds to the logical Hadamard on each qubit, combined with a SWAP between the logical qubits:
\[
(H \otimes H) \cdot \text{SWAP}.
\]

\begin{figure}
    \centering
    \includegraphics[width=0.8\linewidth]{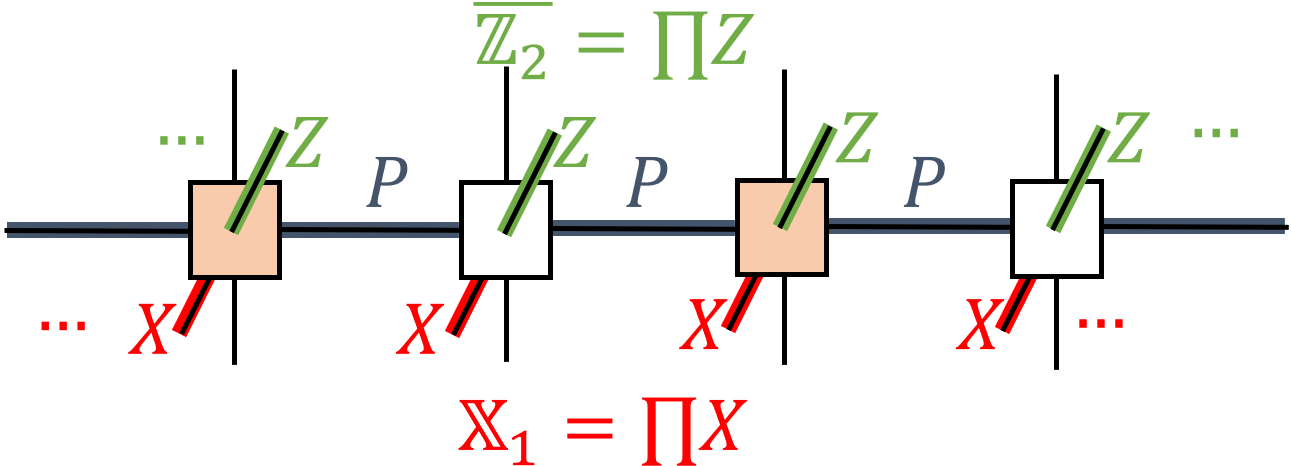}
    \caption{The logical automorphism of the HH Floquet code from gadgets. From the encoding map of the gadget, the incoming $X$ logical is mapped to the outgoing $Z$ logical. }
    \label{fig: tclogical}
\end{figure}

\subsubsection{CSS encoding}\label{sec: tcccs}

We now consider CSS encodings, where bond operators defined in Fig. \ref{fig: tcparam} are strings of $X$ or $Z$ Pauli operators, depending the subscript of the bond operator. Now the equations in Eq. \eqref{eq: tceqs} cannot be solved with every bond operator being a single Pauli operator. Instead, we find the following solution with weight-2 bond operators:
\begin{align}\label{eq: tccss}\nonumber
    A_X=B_X=C_X=D_X=XI,\\\nonumber
    A_Z=B_Z=C_Z=D_Z=ZI,\\\nonumber
    \overline{A_X}=\overline{B_X}=\overline{C_X}=\overline{D_X}=IX,\\
    \overline{A_Z}=\overline{B_Z}=\overline{C_Z}=\overline{D_Z}=IZ.
\end{align}
The rest of the solutions are all bond-unitary equivalent to the one above. The encoding map with the solution in Eq. \eqref{eq: tccss} is not maximal. To see this, we calculate the rank of the $X$ encoding matrix
\begin{align}\nonumber
    \text{rank}\left(H_X\right)=&\text{rank}\left(\begin{array}{c|cccc|c}
        1 & [A_X] & [B_X] & 00 &00 & 0\\
        1 & 00 & 00&[C_X]&[D_X] & 0\\
        0&[\overline{A_X}]&[\overline{B_X}]&00&00 & 1\\
        0&00&00&[\overline{C_X}]&[\overline{D_X}]& 1
    \end{array}\right)\\
    =&\text{rank}\left(\begin{array}{c|cccc|c}
        1 & 10 & 10 & 00 &00 & 0\\
        1 & 00 & 00&10&10 & 0\\
        0&01&01&00&00&1\\
        0&00&00&01&01&1
    \end{array}\right)=4.
\end{align}
Here, $[A_X]$ denotes the binary vector corresponding to the $X$ Pauli string of $A_X$. Similar to Eq. \ref{eq: tcpaulimat}, the two columns on the left of the vertical line are incoming and outgoing legs, and the four columns on the right are internal legs in the direction 1 to 4 in Fig. \ref{fig: tcparam}. The rank of $Z$ encoding matrix is $\text{rank}(H_Z)=4$ from a similar calculation. Therefore, the rank of the encoding map is $\text{rank}(H_X)+\text{rank}(H_Z)=8<2m+n_L=10$, which means we need two extra stabilizers, one $X$ and one $Z$. There are only two different choices here. The one that preserves the logical infornation is to further stabilize $(XX)_2 (XX)_3$ and $(ZZ)_1(ZZ)_2$\footnote{The other choice, which is to stabilize $(XI)_2(XI)_3$ and $(IZ)_1(IZ)_2$, leads to a gadget that directly measures the incoming $X$ logical operators. We formally avoid such solutions using MLSC stated in Sec. \ref{sec:fault_tolerance}.}.
With these internal stabilizers, the CSS encoder can be compiled using the algorithm in Sec. \ref{sec: fixzx}, as shown in Fig. \ref{fig: tczx}. Following a similar procedure as in Fig. \ref{fig: tcsync}, we see that  the connected L and R gadgets lead to the CSS Floquet toric code on the square-octagon lattice in Fig. \ref{fig: tclattice} with a 6 round schedule gZZ, bXX, rZZ, gXX, bZZ and rXX. Therefore, we recover the CSS Floquet toric code proposed in Refs. \cite{kesselring_anyon_2024,davydova_floquet_2023}.

In Sec. \ref{sec:classification}, we prove that the CSS Floquet toric code is, in fact, equivalent to the SASEC that measures all the stabilizers of the toric code once in the sense that the ZX diagrams of the two circuits are equivalent under ZX-rewrite rules. Furthermore, they are also equivalent to two rounds of the HH Floquet code.

\begin{figure}
    \centering
    \includegraphics[width=\linewidth]{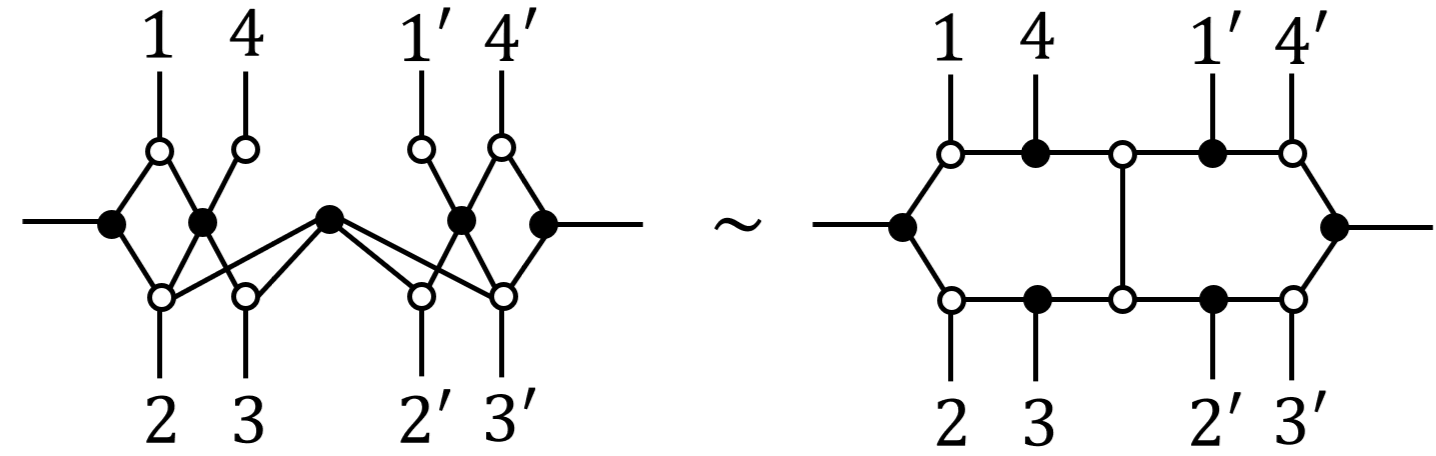}
    \caption{The ZX-diagrams of the gadgets that produces the encoding maps given by the solution in Eq. \eqref{eq: tccss} with two additional stabilizers $(XX)_2 (XX)_3$ and $(ZZ)_1(ZZ)_2$. Following the algorithm for CSS encoder in Sec. \ref{sec: fixzx},the ZX-diagram is constructed using only the $X$ encoding maps and stabilizers. The unprimed and primed legs represent the first and second legs in each direction.}
    \label{fig: tczx}
\end{figure}

\subsection{Floquet BB code}\label{sec: bicycle}

\begin{figure}
    \centering
    \includegraphics[width=0.75\linewidth]{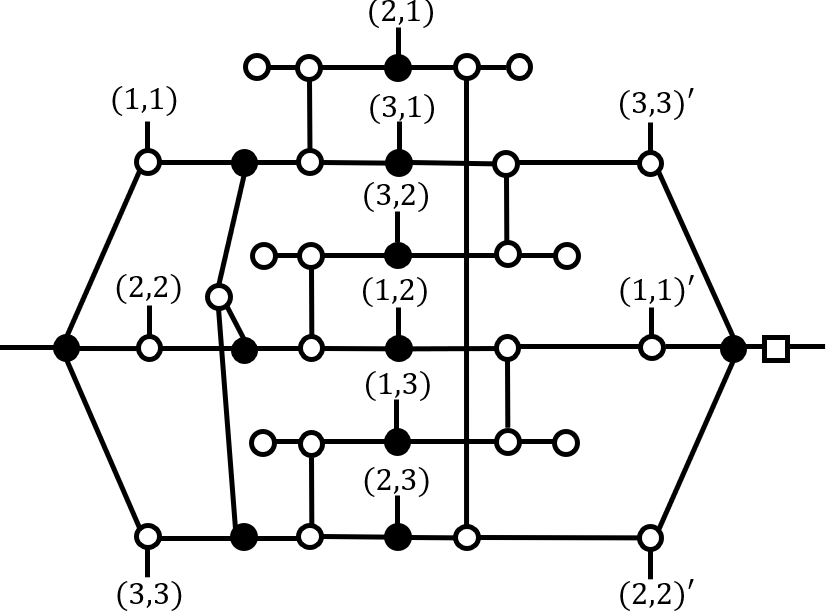}
    \caption{The ZX-diagram of the gadget that produces the encoding map given by Eq. \eqref{eq: bbsol} in the circuit layout with 6 physical qubits. See Fig. \ref{fig: bbzx4} for derivation.}
    \label{fig: bbzx}
\end{figure}

\begin{figure*}
    \centering
    \includegraphics[width=0.75\linewidth]{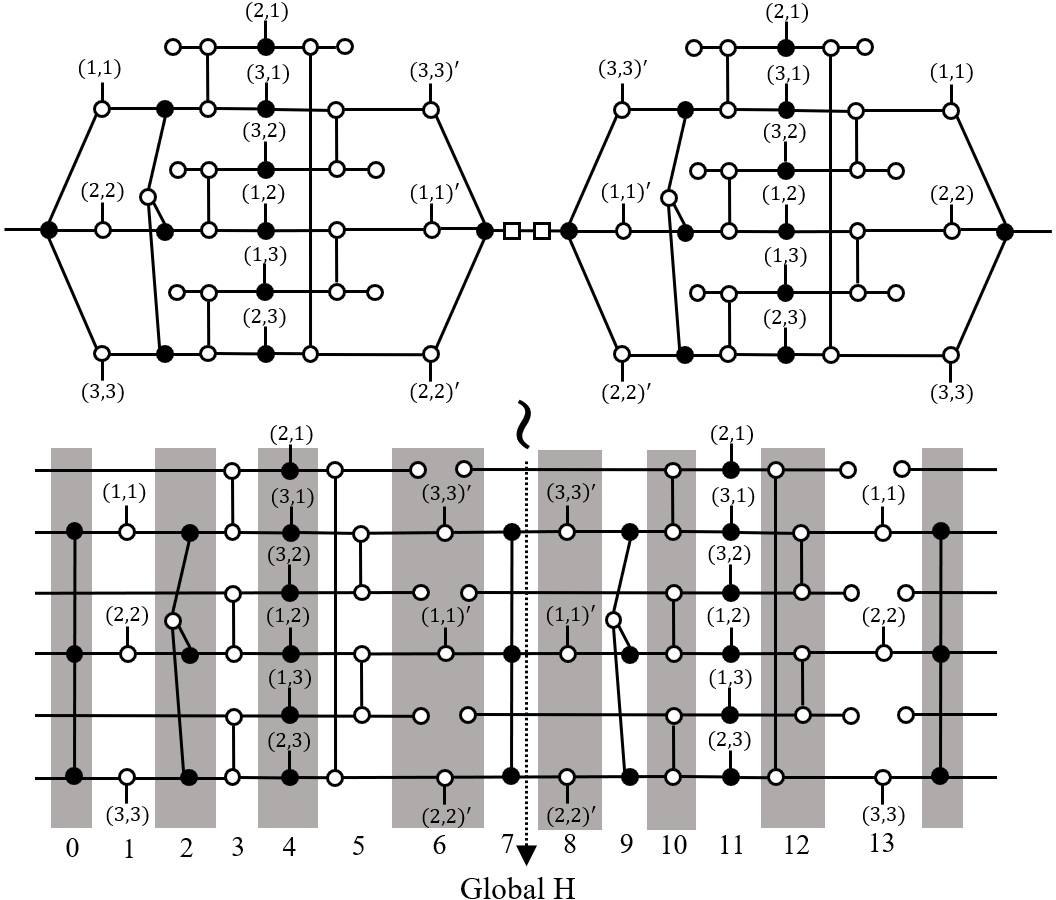}
    \caption{From the ZX-diagram in Fig. \ref{fig: bbzx} of the Floquet BB code followed by the dual gadget, the two round schedule schedule can be packed into 14 step, which are alternatively shaded/unshaded in the figure for clarity. The steps 0 and 7 are pairwise measurement in $Z$ basis. The steps 1, 6, 8 and 13 are inter-gadget pairwise measurements of $X$, while the three of the six qubits are measured in $Z$ basis at steps 6 and 13. Steps 2 and 9 are measurements of the three body $ZZZ$ operator. Steps 3, 5, 10 and 12 are intra-gadget pairwise measurements in $X$. And steps 4 and 11 are inter-gadget pairwise measurements in $Z$.  The cancellation of the Hadamard gate through rewinding means the ISG after step 7 is in the BB stabilizer code conjugated by a global Hadamard gate, i.e. a ``dual" BB code where the $X$ and $Z$ stabilizers switch places. }
    \label{fig: bbschedule}
\end{figure*}
\begin{figure}
    \centering
    \includegraphics[width=\linewidth]{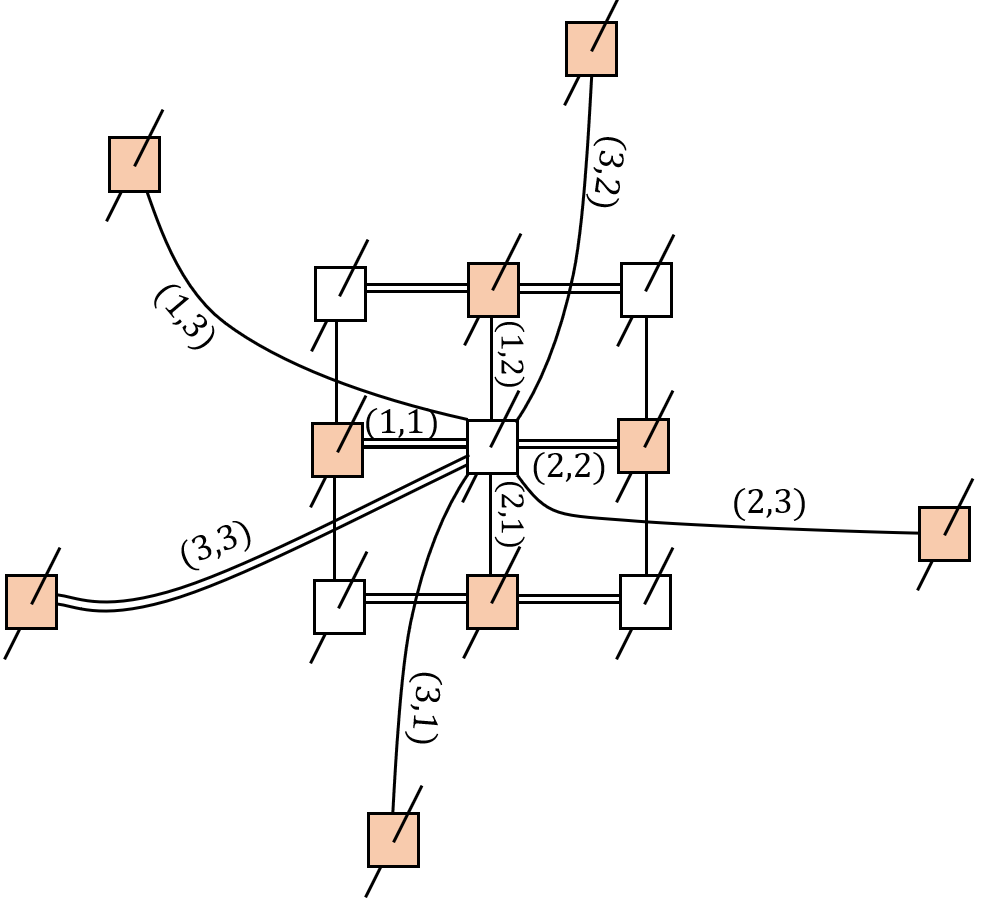}
    \caption{Connection between L and R gadgets on a toric layout of the Floquet BB code, where the four directions $(\alpha,\beta)$ with $\alpha,\beta=1,2$ form bonds on the square lattice, and the rest of the bond directions are long-range connections. The number of internal legs in the figure reflects the actual number of internal legs as given by Eq. \eqref{eq: bbsol}. Note that the square lattice has been rotated by $45^\circ$ comparing to the toric layout of the Tanner graph in Ref. \cite{bravyi_high-threshold_2024}.}
    \label{fig: bbconnection}
\end{figure}

The bivariate bicycle (BB) code, proposed in Refs. \cite{bravyi_high-threshold_2024,panteleev_degenerate_2021}, represents a large class of qLDPC codes
with a high encoding rate $k/n$ and a high error threshold. To demonstrate the capability and flexibility of spacetime concatenation, we construct an efficient Floquet BB code that only contains pairwise measurements of Pauli $X$ and $Z$ operators while maintaining the underlying graph-connectivity of the BB stabilizer code to satisfy SLPC.

We first briefly review the construction of the BB stabilizer code. We define the cyclic permutation matrix $S_m$ of dimension $m$ as $(S_m)_{i,j}=\delta_{j,i+1\mod m}$, where $i,j$ are the indices of the rows and columns, respectively. We then define $x=S_p\otimes \mathbbm{1}_m$ and $y=\mathbbm{1}_p\otimes S_m$, where $\mathbbm{1}_m$ is the identity matrix of dimension $m$. The parity check matrices are then defined as (to avoid confusion, we use $\mathcal{H}_{X,Z}$ to denote check matrices of stabilizer codes)
\begin{equation}\label{eq: bbcheck}
    \mathcal{H}_X=(\mathcal{A}\vert \mathcal{B}),\ \mathcal{H}_Z=(\mathcal{B}^T\vert \mathcal{A}^T),
\end{equation}
where
\begin{equation}
    \mathcal{A}=\mathcal{A}_1+\mathcal{A}_2+\mathcal{A}_3,\ \mathcal{B}=\mathcal{B}_1+\mathcal{B}_2+\mathcal{B}_3
\end{equation}
and $\mathcal{A}_{1,2,3},\mathcal{B}_{1,2,3}$ are monomials of permutation matrices $x$ and $y$ of the form $x^q y^r$, where $q=0,1,\dots,p-1$ and $r=0,1,\dots,m-1$. In this way, each row $\vec{v}$ of $\mathcal{H}_{X}$ represents an $X$ stabilizer $\prod_{j=1}^{2pm}X_j^{v_j}$, the left $pm$ columns are $pm$ data qubits, which we will refer to as ``left qubits" (or simply L qubits), while the right $pm$ columns are similarly referred to as ``right qubits" (or simply R qubits). The total number of data qubits is thus $n=2pm$. The number of logical operators $k$ depends on the choice of the six monomials in $\mathcal{A}$ and $\mathcal{B}$, and it is computationally costly to search for a good combination of $p,m$ and the monomials that yields a high encoding rate. An example is given in Ref. \cite{bravyi_high-threshold_2024} is $\mathcal{A}=x^3+y+y^2$ and $\mathcal{B}=y^3+x+x^2$ with $p=2m=12$, which encodes $k=12$ logical qubits. 

Meanwhile, our construction of Floquet BB code applies with any  choice of the monomials\footnote{Although our construction does not rely on a specific form of the monomials, we do assume that the Tanner graph of the BB code is fully connected, i.e. the monomials $\mathcal{A}_\alpha^{-1}\mathcal{A}_\beta$ and $\mathcal{B}_\alpha^{-1}\mathcal{B}_\beta$ ($\alpha,\beta=1,2,3$) should generate the entire group $\Z_p\times\Z_m$ of the monomials. This ensures that the gadget layout also forms a connected graph. } . Before deploying our algorithm, we need to understand the connectivity between the data qubits in the BB code. 
Given the form of the parity check matrices in \eqref{eq: bbcheck}, each $X$ and $Z$ stabilizer consists of three L and three R qubits, and each L or R qubit participates in three $X$ and three $Z$ stabilizers. Each pair of L and R qubits shares either no or two common stabilizers, one $X$ and one $Z$, while each pair of two L qubits shares at most a single common stabilizer $X$ or $Z$. Therefore, to replace the L/R qubits with L/R gadgets, we need a bipartite connection so that each L gadget is connected with nine R gadgets, as given by the Algorithm \ref{alg: gadgetcon}. Their relative positions on the graph are represented by monomials $\mathcal{A}_\alpha^{-1}\mathcal{B}_\beta$ ($\alpha,\beta=1,2,3$). Therefore, we label the corresponding nine bond directions as $(\alpha,\beta)$. 

We now write down the Clifford encoding map from \eqref{eq: bbcheck} and the gadget connectivity above. Due to the translation symmetry given by the cyclic permutation matrix, we assume every incoming or outgoing $X$/$Z$ stabilizer is encoded to the same set of bond operators between the three L and three R gadgets. Furthermore, since the gadget connection is bipartite, once the encoding map of the L gadget is fixed, so is the map for every R gadget. Therefore, we mainly focus on the L gadget, where the three incoming and three outgoing $X$ stabilizers, labeled by $\alpha=1,2,3$, needs to be encoded onto three bond directions $(\alpha,1)$, $(\alpha,2)$ and $(\alpha,3)$, while each $Z$ stabilizer labeled by $\beta=1,2,3$ need to be encoded onto the three bond directions $(1,\beta)$, $(2,\beta)$ and $(3,\beta)$. Denoting the nine bond operators for incoming/outgoing $X$/$Z$ operators as $A^{X/Z}_{(\alpha,\beta)}$ and $\overline{A^{X/Z}_{(\alpha,\beta)}}$, the encoding maps can be written as
\begin{align}\label{eq: bbencode}\nonumber
X\to \prod_{\beta=1}^3A^X_{(\alpha,\beta)},\ \overline{X}\to \prod_{\beta=1}^3\overline{A^X_{(\alpha,\beta)}}\ (\alpha=1,2,3),\\
    Z\to \prod_{\alpha=1}^3 A^Z_{(\alpha,\beta)},\ \overline{Z}\to \prod_{\alpha=1}^3 \overline{A^Z_{(\alpha,\beta)}}\ (\beta=1,2,3).
\end{align}
The consistency equations that the encoding maps need to satisfy can be directly expressed via the commutation relation between the bond operators:
\begin{align}\label{eq: bbeqs}\nonumber
    &\left[ \prod_{\beta=1}^3A^X_{(\alpha,\beta)},\prod_{\beta=1}^3\overline{A^X_{(\alpha,\beta)}} \right]=0\ (\alpha=1,2,3),\\\nonumber
    &\left[ \prod_{\alpha=1}^3A^Z_{(\alpha,\beta)},\prod_{\alpha=1}^3\overline{A^Z_{(\alpha,\beta)}} \right]=0\ (\beta=1,2,3),\\\nonumber
    &\left\{A^X_{(\alpha,\beta)},A^Z_{(\alpha,\beta)}\right\}=\left\{\overline{A^X_{(\alpha,\beta)}},\overline{A^Z_{(\alpha,\beta)}}\right\}=0,\\
    &\left[A^X_{(\alpha,\beta)},\overline{A^Z_{(\alpha,\beta)}}\right]=\left[\overline{A^X_{(\alpha,\beta)}},A^Z_{(\alpha,\beta)}\right]=0\ (\alpha,\beta=1,2,3).
\end{align}
The set of equations does not admit solutions with every bond operator being a single Pauli operator. We find that the solution with the fewest internal legs for each gadget is $n_L=12$, where Pauli strings of bond operators in the directions $(1,1)$, $(2,2)$ and $(3,3)$ are weight-two while the rest of the directions are weight-one. The solutions to the bond operators are conveniently written in the following matrix format\footnote{We note that the most efficient CSS encoding requires every bond operators to be weight-two, i.e. $n_L=2\times9=18$, in which case the connected gadgets give rise to the SASEC, similar to Fig. \ref{fig: tccnot} for toric code.}:
\begin{align}\label{eq: bbsol}\nonumber
    A^X=\left(\begin{array}{ccc}
        XI & X & X \\\hline
        X & XI & X \\\hline
        X & X & XI\\
    \end{array}\right),\ \overline{A^X}=\left(\begin{array}{ccc}
        IZ & Z & Z \\\hline
        Z & IZ & Z \\\hline
        Z & Z & IZ\\
    \end{array}\right),\\
    A^Z=\left(\begin{array}{c|c|c}
        ZI & Z & Z \\
        Z & ZI & Z \\
        Z & Z & ZI\\
    \end{array}\right), \ \overline{A^Z}=\left(\begin{array}{c|c|c}
        IX & X & X \\
        X & IX & X \\
        X & X & IX\\
    \end{array}\right),
\end{align}
where $A^{X/Z}_{(\alpha,\beta)}$ is located at row $\alpha$ and column $\beta$ of the correspond matrix $A^{X/Z}$. The horizontal/vertical lines are visual guidance for the encoding map in Eq. \eqref{eq: bbencode} where incoming/outgoing $X$ operators are encoded in the rows while the $Z$ operators are encoded in the columns of the corresponding matrices. The encoding maps given by Eq. \eqref{eq: bbsol} are not maximal. To completely fix the gadget, we impose the following two additional stabilizers.
\begin{equation}\label{eq: bbadditional}
    \left(\begin{array}{ccc}
        XX & X & I \\
        X & XX & I \\
        I & I & I\\
    \end{array}\right),\ \left(\begin{array}{ccc}
        ZZ & Z & I \\
        Z & ZZ & I \\
        I & I & I\\
    \end{array}\right).
\end{equation}
Although the encoding map in Eq. \eqref{eq: bbsol} is not a CSS one, it becomes CSS once we conjugate the outgoing leg by a Hadamard gate. Therefore, we can easily construct the ZX-diagram using the algorithm for CSS encoders in Sec. \ref{sec: fixzx} and only the encoding maps of $X$ and $\overline{X}$. 
The diagrammatic algorithm yields multiple possible circuit layouts corresponding to different spatial concatenations and measurement schedules. For example, the circuit layout shown in Fig. \ref{fig: bbzx} corresponds to a spatial concatenation of six physical qubits per data qubit in the original BB code.

In Appendix \ref{sec: bblogical}, we prove that the dynamical code from the gadgets with the solution in Eq. \eqref{eq: bbsol} preserves the logical space and performs Hadamard gates plus SWAP between pairs of logical operators.

To repeat the gadget and obtain the Floquet BB code, the simplest way, from the ZX-diagram point of view, is simply to follow the gadget with its dual gadget, so that the Hadamard gate on the outgoing leg cancels out with the one on the incoming leg of the following gadget, resulting in a rewinding schedule, as shown in the first row of Fig. \ref{fig: bbschedule}. 
We note that such a circuit layout is already synchronized between the L and R gadgets. To see this, we note that from the parity check matrices in Eq. \ref{eq: bbcheck}, the encoding map of the R gadget is the ``transpose" of the encoding map of the L gadget, i.e. the incoming/outgoing $X$ operators are encoded in the directions $(1,\beta)$, $(2,\beta)$, and $(3,\beta)$ ($\beta=1,2,3$). Therefore, the ZX-diagram of the R gadget can be obtained by flipping $(\alpha,\beta)$ in Fig. \ref{fig: bbschedule} to $(\beta,\alpha)$. Since the ``off-diagonal" directions are all at the same time step, the schedules between the L and R gadgets are hence synchronized. Therefore, after going through the procedure in Fig. \ref{fig: floqueteff}, we arrive at a 14 step rewinding schedule as shown in the second row of Fig. \ref{fig: bbschedule}.
We note that after step 7, the ISG of the Floquet BB code is, in fact, a ``dual" BB stabilizer code obtained from conjugating every qubit in the original BB code by a Hadamard gate, hence swapping the definition of $X$ and $Z$ stabilizers compared to the incoming BB code. By rewinding, we effectively remove the Hadamard gate on the outgoing leg on the ZX-diagram in Fig. \ref{fig: bbzx}. 

To obtain a lattice layout of the Floquet BB code, one can go through a similar procedure as we did for Floquet toric code in Fig. \ref{fig: tcsync}. Assuming the BB code has a toric layout, i.e. the L and R data qubits can be laid on a square lattice with periodic boundary conditions, the gadget layout will naturally inherit the square lattice toric layout of the BB stabilizer code, where each L gadget will be connected to four R gadgets located on the nearest neighbor of the square lattice, and five more R gadgets at further locations. We can choose the four nearest-neighbor directions to be $(\alpha,\beta)$ where $\alpha,\beta=1,2$, see Fig. \ref{fig: bbconnection}. Since the $(3,3)$ direction has two internal legs, the total number of long range inter-gadget pairwise measurement is 6 per gadget. In the circuit implementation in Fig. \ref{fig: bbschedule}, the long-range pairwise measurements happen for four of the six physical qubits inside each gadget, while the remaining two only engage in intra-gadget and nearest-neighbor measurements. For details of the measurement schedule, see the caption of Fig. \ref{fig: bbschedule}.

In Appendix \ref{sec: additionalzx}, we construct another Floquet BB code with a much shorter schedule from an alternative circuit layout of the gadget that is equivalent to the one in Fig. \ref{fig: bbzx}. There, each data qubit of the BB code is replaced by 12 physical qubits as the spatial concatenation. In this way, all stabilizers of the BB code will be measured twice in a six-round schedule. These two implementations of Floquet BB code clearly show that circuit depth can be traded with a dynamical qubit overhead in compiling dynamical codes, a crucial idea that will be elaborated more in Sec. \ref{sec:resource}. 

\subsection{Floquet Haah code}\label{sec: haah}

\begin{figure*}
    \centering
    \includegraphics[width=0.7\linewidth]{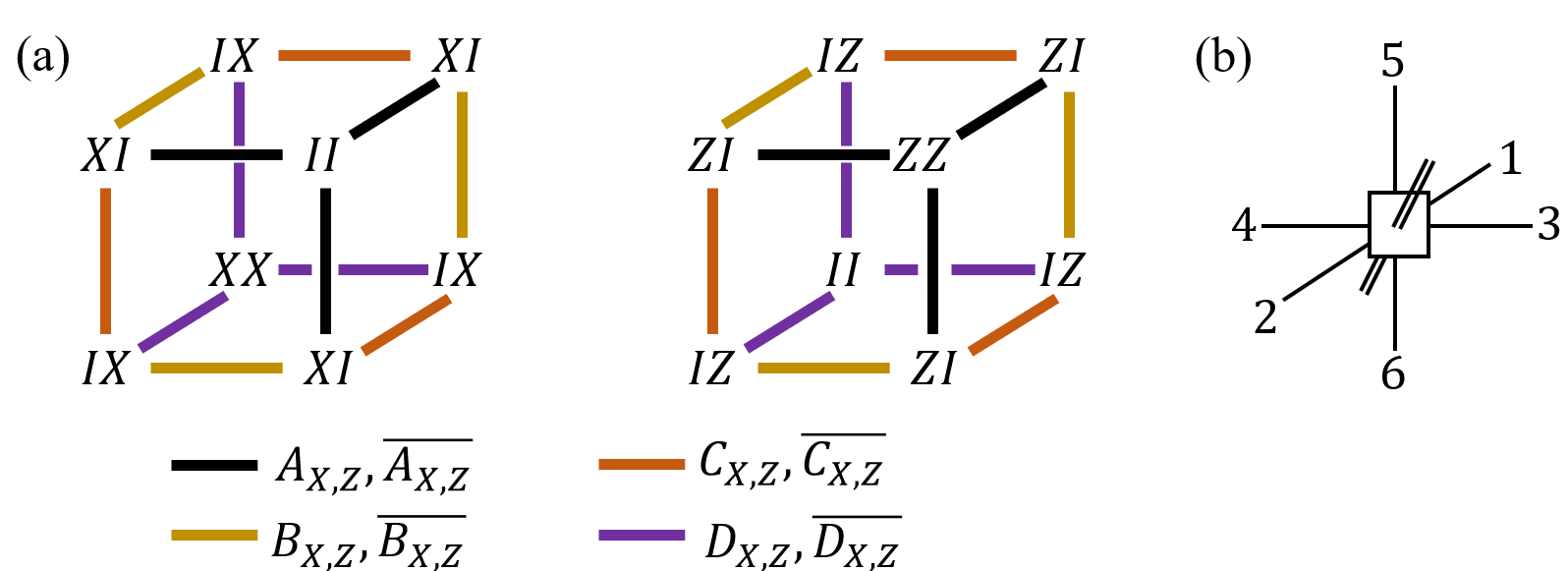}
    \caption{(a) The cubic stabilizer of the Haah code and the parameterization of the bond operators in the encoding matrices in Eq. \eqref{eq: haahmat}. (b) The gadget that replaces each pair of data qubit in the Haah code and the six bond directions.}
    \label{fig: haahparam}
\end{figure*}

We now apply our approach to construct a dynamical code whose incoming and outgoing stabilizer codes are both the Haah's cubic code \cite{haah_local_2011,haah_algebraic_2017} which we will refer to as the Haah code below. Such a dynamical code, to the best of our knowledge, has not yet been constructed in the literature. The Haah code is defined on a coarse-grained 3D cubic lattice with $m=2$ qubits per site, and every cube in the lattice hosts a pair of $X$ and $Z$ stabilizers, whose expressions are shown in Fig. \ref{fig: haahparam}(a). To turn it into a dynamical code, we adopt a gadget layout so that each pair of data qubits in the original Haah code is replaced by a gadget with two incoming and two outgoing legs. From Algorithm \ref{alg: gadgetcon}, each pair of neighboring gadgets on the lattice is connected by internal legs that go along the six lattice directions, labeled from 1 to 6 in Fig. \ref{fig: haahparam}(b). Again, we assume translation symmetry so that the mapping of each stabilizer (incoming $X$, incoming $Z$, outgoing $X$ and outgoing $Z$) on the bond operators is translation invariant. Since the Haah code is three-fold rotation symmetric around the diagonal axis of the cubic lattice in the direction that connects $II$ and $XX$ in the cubic stabilizer, we only consider a homogeneous internal leg layout where each bond direction $i=1,2,\dots,6$ contains $n_{L,i}$ legs. 

Before constructing the encoding matrices, we first prove that the weight of every bond operator of the $X$ and $Z$ cubic stabilizers, $n_{L,i}$, must be at least two. In fact, suppose $n_{L,i}=1$, then every bond operator is a single Pauli operator (or identity). Consider the black and yellow bond operators for the incoming $XX$, incoming $ZI$ and outgoing $ZI$ operators, which we denote by $a_{XX}$, $b_{ZI}$ and $\overline{b_{ZI}}$. Consider the  pair of stabilizers: an $X$ stabilizer and the $Z$ stabilizer on its bottom left, see Fig. \ref{fig: haahproof}. Since they only overlap on one bond, due to translation invariance, the two bond operators $a_{XX}$ and $b_{ZI}$ must anticommute, since they are encoded from anti-commuting incoming Pauli operators $XX$ and $ZI$. Therefore, they are both nontrivial Pauli operators, not identity. Meanwhile, since incoming operators always commute with outgoing operators, we must have $[a_{XX},\overline{b_{ZI}}]=[\overline{a_{XX}},b_{ZI}]=0$. Since $n_{L,i}=1$, we must have $a_{XX}=\overline{b_{ZI}}$, and therefore $\{b_{ZI},\overline{b_{ZI}}\}=0$. 
One can similarly prove that $\{c_{ZI},\overline{c_{ZI}}\}=\{d_{ZI},\overline{d_{ZI}}\}=0$ in Fig. \ref{fig: haahproof}. However, this leads to an apparent contradiction that the incoming $ZI$ and outgoing $ZI$ operators are encoded onto anti-commuting operators, since $\{b_{ZI}c_{ZI}d_{ZI},\overline{b_{ZI}}\overline{c_{ZI}}\overline{d_{ZI}}\}=0$. Therefore, the Pauli weight of one of the three legs that connects to the $ZI$ operator must be at least two, hence the number of legs in each direction $n_{L,i}\geq 2$.

We now demonstrate that $n_{L,i}=2$ is indeed feasible by constructing a CSS Floquet Haah code. We parametrize the bond operators of the 8 body $X$ and $Z$ stabilizers in a 3-fold rotational symmetric way around the diagonal axis that connects $II$ and $XX$ where every bond operator of the incoming/outgoing $X$ and $Z$ stabilizers is represented by a weight-$l$ Pauli string from $A_{X,Z}/\overline{A_{X,Z}}$ to $D_{X,Z}/\overline{D_{X,Z}}$, see Fig. \ref{fig: haahparam}(a). In this way, the CSS encoding matrices read
\begin{widetext}
\begin{align}\label{eq: haahmat}
    H_X=\left(\begin{array}{cc|cccccc|cc}
        0&0 &[A_X]&0&0&[A_X]&0&[A_X] &0&0\\
        1&0&[B_X]&0&[A_X]&0&0&[C_X]&0 &0\\
        1&0&0&[A_X]&0&[C_X]&0&[B_X]&0&0\\
        1&0&[C_X]&0&0&[B_X]&[A_X]&0&0&0\\
        0&1&[D_X]&0&[B_X]&0&[C_X]&0&0&0\\
        0&1&0&[C_X]&0&[D_X]&[B_X]&0&0&0\\
        0&1&[B_X]&[C_X]&0&0&[D_X]&0&0&0\\
        1&1&0&[D_X]&[D_X]&0&[D_X]&0&0&0\\
        0&0 &[\overline{A_X}]&0&0&[\overline{A_X}]&0&[\overline{A_X}] &0&0\\
        0&0&[\overline{B_X}]&0&[\overline{A_X}]&0&0&[\overline{C_X}]&1&0\\
        0&0&0&[\overline{A_X}]&0&[\overline{C_X}]&0&[\overline{B_X}]&1&0\\
        0&0&[\overline{C_X}]&0&0&[\overline{B_X}]&[\overline{A_X}]&0&1&0\\
        0&0&[\overline{D_X}]&0&[\overline{B_X}]&0&[\overline{C_X}]&0&0&1\\
        0&0&0&[\overline{C_X}]&0&[\overline{D_X}]&[\overline{B_X}]&0&0&1\\
        0&0&0&[\overline{B_X}]&[\overline{C_X}]&0&0&[\overline{D_X}]&0&1\\
        0&0&0&[\overline{D_X}]&[\overline{D_X}]&0&[\overline{D_X}]&0&1&1\\
    \end{array}\right),\ H_Z=\left(\begin{array}{cc|cccccc|cc}
        1&1 &[A_Z]&0&0&[A_Z]&0&[A_Z] &0&0\\
        1&0&[B_Z]&0&[A_Z]&0&0&[C_Z]&0&0\\
        1&0&0&[A_Z]&0&[C_Z]&0&[B_Z]&0&0\\
        1&0&[C_Z]&0&0&[B_Z]&[A_Z]&0&0&0\\
        0&1&[D_Z]&0&[B_Z]&0&[C_Z]&0&0&0\\
        0&1&0&[C_Z]&0&[D_Z]&[B_Z]&0&0&0\\
        0&1&0&[B_Z]&[C_Z]&0&0&[D_Z]&0&0\\
        0&0&0&[D_Z]&[D_Z]&0&[D_Z]&0&0&0\\
        0&0&[\overline{A_Z}]&0&0&[\overline{A_Z}]&0&[\overline{A_Z}] &1&1\\
        0&0&[\overline{B_Z}]&0&[\overline{A_Z}]&0&0&[\overline{C_Z}]&1 &0\\
        0&0&0&[\overline{A_Z}]&0&[\overline{C_Z}]&0&[\overline{B_Z}]&1&0\\
        0&0&[\overline{C_Z}]&0&0&[\overline{B_Z}]&[\overline{A_Z}]&0&1&0\\
        0&0&[\overline{D_Z}]&0&[\overline{B_Z}]&0&[\overline{C_Z}]&0&0&1\\
        0&0&0&[\overline{C_Z}]&0&[\overline{D_Z}]&[\overline{B_Z}]&0&0&1\\
        0&0&0&[\overline{B_Z}]&[\overline{C_Z}]&0&0&[\overline{D_Z}]&0&1\\
        0&0&0&[\overline{D_Z}]&[\overline{D_Z}]&0&[\overline{D_Z}]&0&0&0\\
    \end{array}\right).
\end{align}
\end{widetext}
Here, the first two columns denote the two incoming legs, and the last two columns denote the two outgoing legs. The six columns in the middle denote the six directions of the internal legs from 1 to 6 in Fig. \ref{fig: haahparam}(b), and every ``0" over there stands for a string of 0 with length $l$. We solve the consistency equations in Eq. \eqref{eq: csscon} with loosened conditions that the vectors $A_X,\ \overline{A_X},\ D_Z$ and $\overline{D_Z}$ can be zero, since the $X$ or $Z$ operators of the two qubits do not actually participate in the cubic stabilizer in these directions. In this case, we find the following solution at $n_{L,i}=2$:
\begin{align}\label{eq: haahsol}\nonumber
    &[A_X]=00,\ [B_X]=10,\ [C_X]=11,\ [D_X]=01,\\\nonumber
    & [\overline{A_X}]=00,\ [\overline{B_X}]=11,\ [\overline{C_X}]=01,\ [\overline{D_X}]=10,\\ \nonumber
    &[A_Z]=01,\ [B_Z]=10,\ [C_Z]=11,\ [D_Z]=00,\\ \nonumber
    &[\overline{A_Z}]=10,\ [\overline{B_Z}]=11,\ [\overline{C_Z}]=01,\ [\overline{D_Z}]=00.\\
\end{align}
The rest of the solutions are simply permutations of the three nonzero vectors $01,\ 10$ and $11$, and the resulting encoding maps will only differ by CNOT or swap gates over the internal legs. The total rank of the encoding matrices for the solution above is
\begin{equation}
    \text{rank}(H_X)+\text{rank}(H_Z)=16=2m+\sum_{i=1}^6n_{L,i}=2\times 2+2\times 6.
\end{equation}
Therefore, the gadget is completely fixed by the encoding maps. 

The ZX-diagram that produces the encoding map given by the solution in Eq. \eqref{eq: haahsol} is shown in Fig. \ref{fig: haahzx}. Details of the derivation from the encoding map can be found in Appendix \ref{sec: additionalzx}. In particular, we find a compilation of the CSS encoding circuit in the three-stage form as in Fig. \ref{fig: floqueteff} with the spatial encoder representing two three-qubit repetition codes. Therefore, every pair of data qubits in the original Haah code is replaced by 6 physical qubits. Furthermore, the temporal order of the internal legs in the ZX-diagram is already synchronous: all the primed legs are measured before the unprimed ones. Hence the repeated gadget leads to a CSS Floquet Haah code with a 5-round schedule (see Fig. \ref{fig: haahschedule}): bZZ, pXX (measure every pair of $X$ operators on the purple bonds), $CNOT_\text{br}$ (apply CNOT on every brown bond where each brown qubit is controlled by the black one), rZZZ (measure the product of 3 $Z$ operators on each red triangle), gXX. The coloring of the bonds can be found in Fig. \ref{fig: haahlat}. To make this a measurement-only circuit, the $CNOT_\text{Br}$ gates can be further decomposed into pairwise measurements by introducing one ancilla per each CNOT gate. The depth of the circuit will be increased accordingly.

Based on the topological nature of the (static) Haah code\cite{haah_algebraic_2017}, we expect the Floquet Haah code to be fault-tolerant with a preserved logical space based on the argument in Sec. \ref{sec: spacetimed}. However, since the number of logical qubits is macroscopic and highly sensitive to the system size and boundary conditions\cite{aitchison_no_2024}, we leave the task of studying the logical automorphism of the Floquet Haah code to future investigations.

\begin{figure}
    \centering
    \includegraphics[width=\linewidth]{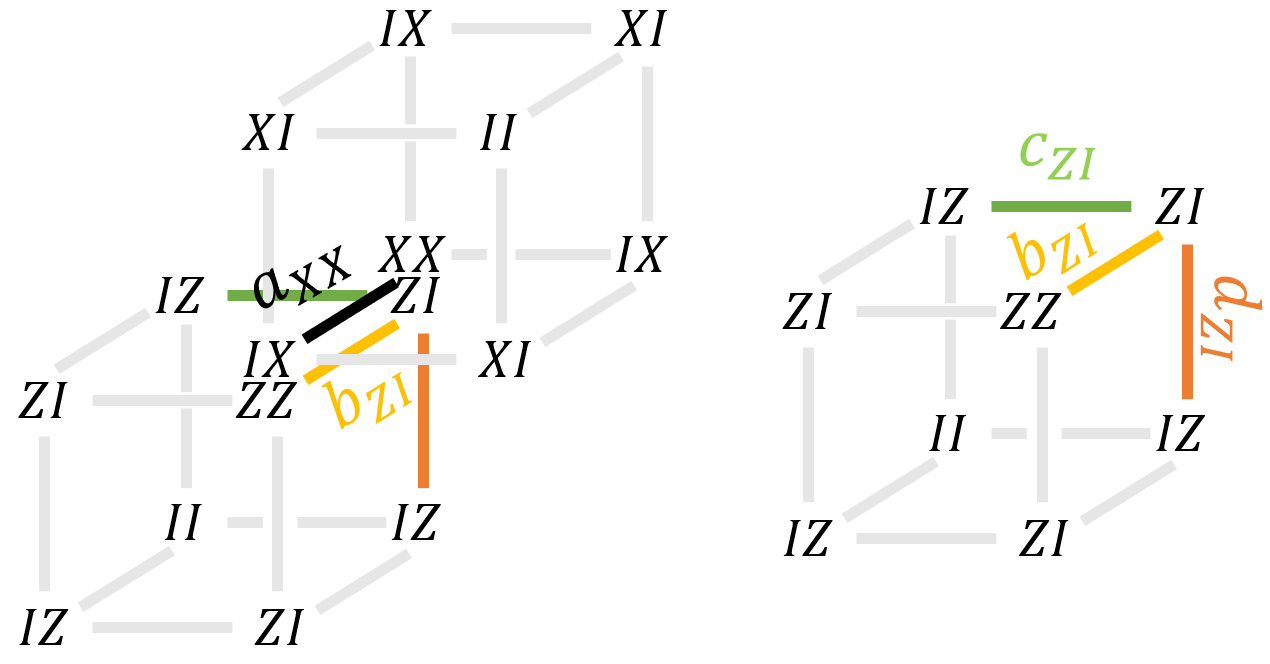}
    \caption{The proof that the Floquet Haah code needs at least $2$ legs per bond. Consider a $Z$ stabilizer on the bottom-left of an $X$ stabilizer. To faithfully encode the Pauli algebra on the incoming and outgoing legs, we must have $\{b_{ZI},\overline{b_{ZI}}\}=0$. This leads to the contradiction that incoming and outgoing $ZI$ operators are encoded by anti-commuting operators $b_{ZI}c_{ZI} d_{ZI}$ and $\overline{b_{ZI}}\overline{c_{ZI}}\overline{d_{ZI}}$.}
    \label{fig: haahproof}
\end{figure}
\begin{figure}
    \centering
    \includegraphics[width=0.5\linewidth]{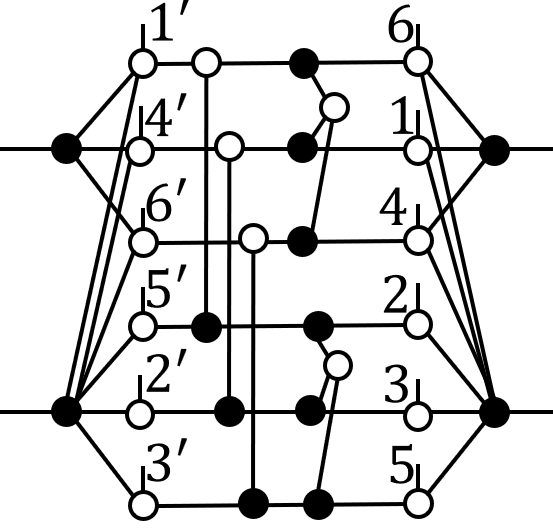}
    \caption{The ZX-diagram of the gadget that produces the encoding map in Eq. \eqref{eq: haahmat} with the solution in Eq. \eqref{eq: haahsol}. The the top(bottom) incoming/outgoing lines corresponds to the first(second) data qubit on each lattice site in the original Haah code. See Fig. \ref{fig: haahcss} for derivation.}
    \label{fig: haahzx}
\end{figure}
\begin{figure*}
    \centering
    \includegraphics[width=0.9\linewidth]{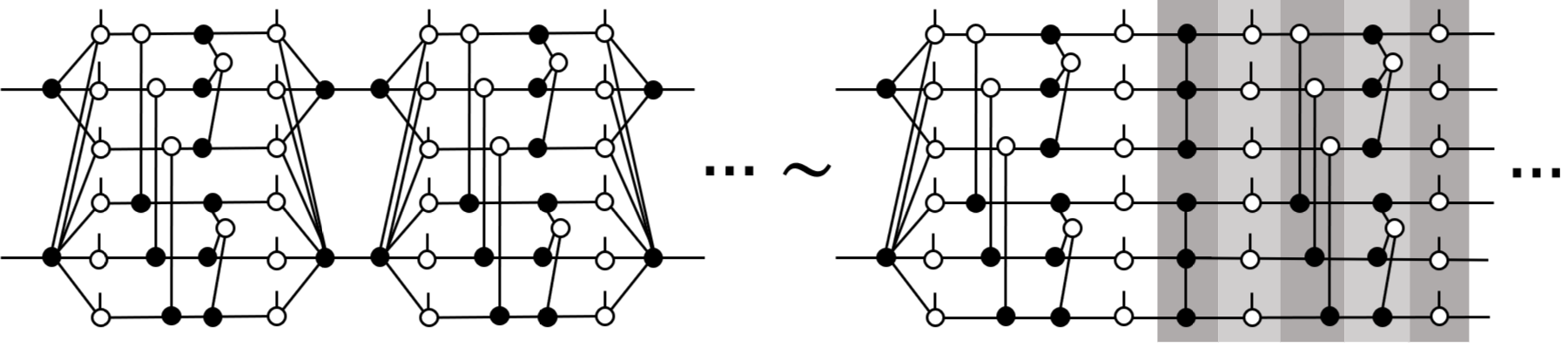}
    \caption{The repeated gadget in Fig. \ref{fig: haahzx} produces a CSS Floquet Haah code with a 5-round measurement schedule. The spatial concatenation for each gadget is two decoupled three-qubit repetition codes, one for each data qubit in the Haah stabilizer code. The five-round schedule of the Floquet Haah code are marked by alternating light/dark gray shades.}
    \label{fig: haahschedule}
\end{figure*}
\begin{figure*}
    \centering
    \includegraphics[width=0.6\linewidth]{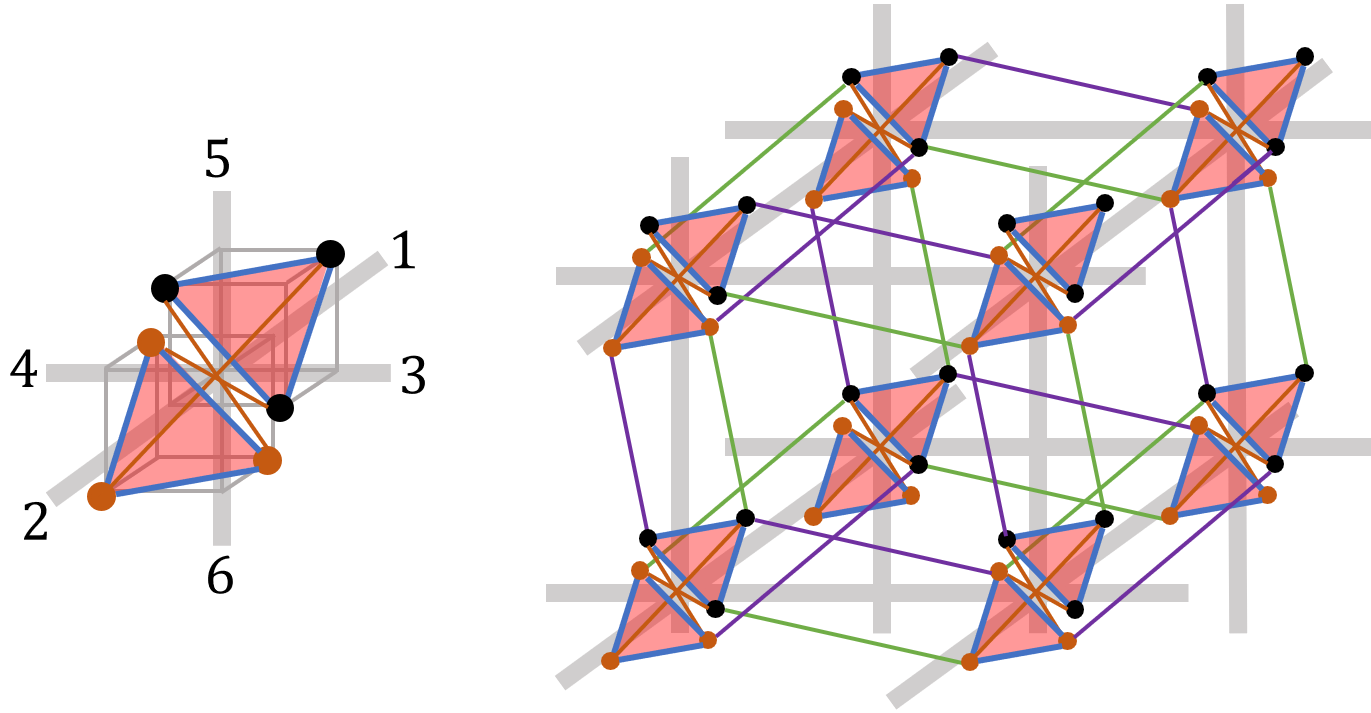}
    \caption{The schedule of the CSS Floquet Haah code from Fig. \ref{fig: haahschedule} and the zoom-in of the 6 physical qubits on each lattice site, which replaces the original 2 data qubits. To avoid visual overlap, the 6 qubits from the top to the bottom of the ZX-diagrams in Fig. \ref{fig: haahschedule} are located on the plane quadrants spanned by the directions 16, 14, 46, 25, 23 and 35, respectively.}
    \label{fig: haahlat}
\end{figure*}

\subsection{Spacetime concatenation for arbitrary LDPC stabilizer codes}

Through the above-mentioned examples of toric code, BB code and Haah code, we have demonstrated the gadget layout as a general recipe for constructing efficient dynamical codes from given stabilizer codes. In Appendices \ref{sec: csscc} and \ref{sec: csscb}, we construct further examples of Floquet codes: CSS Floquet color code, Floquet Steane code and CSS Floquet checkerboard code. We note that Floquet code compilations of the color code and the checkerboard code already exist in the literature. Meanwhile, we find that the CSS Floquet color code we constructed is not equivalent to existing Floquet color codes, including the two examples in Refs. \cite{dua_engineering_2024,fuente_xyz_2024} that require pairwise measurement in $X$, $Y$ and $Z$ Pauli operators, and the one in Ref. \cite{townsend-teague_floquetifying_2023} which is equivalent to the SASEC.

For the $[[7,1,3]]$ Steane code, we demonstrate the flexibility of spacetime concatenation through realizing logical $S$ and Hadamard gates in the dynamical Steane code via pairwise measurement and/or single physical qubit unitaries. 

For the Floquet checkerboard code, we find that the CSS one we construct in Appendix \ref{sec: csscb} is equivalent to the ``fracton Floquet code" constructed in Ref.~\cite{davydova_floquet_2023} between measurement rounds 0 and 5 there, whose ISGs are both checkerboard codes that are defined on the same sublattice of the 3D cubic lattice. Compared to the one in Ref. \cite{davydova_floquet_2023}, the spatial layout of the CSS Floquet checkerboard code that we construct uses half the number of physical qubits and has a slightly longer 8-round schedule. This again highlights a trade-off between the number of physical qubits and the circuit depth, as mentioned in Sec. \ref{sec: fixfloquet}. We will discuss these tradeoffs in more detail in Sec. \ref{sec:resource}.

Finally, in Appendix \ref{sec: fermion}, we construct a pairwise-measurement $\Z_2^{(1)}$ subsystem code whose incoming and outgoing stabilizer groups $\cS=\overline{\cS}$ are not from a stabilizer code, but from the stabilizer group of the $\Z_2^{(1)}$ subsystem code in Ref. \cite{ellison_pauli_2023}. Such a stabilizer group is not CSS, i.e. it is not generated by stabilizers that are products of Pauli $X$ or $Z$.  Crucially, since our construction does not require knowledge of the gauge group $\mathcal{G}$ of the subsystem code, the dynamical code will only stabilize the stabilizer and the logical subspace of the given subsystem code. That is, the dynamical code only measures the syndromes of the spacetime stabilizers that are given by the stabilizer subgroup $\cS$, but not by the gauge operators. Nevertheless, it is possible that, at a given time step, the dynamical code may preserve a larger logical space than just the logical subspace of the given subsystem code.

Beyond these examples, our algorithm for the gadget layout and construction of dynamical codes applies to Pauli stabilizer code and even subsystem codes in general. It is also capable in designing dynamical codes with fixed hardware topology, which can be incorporated through restrictions in inter-gadget connections.

\section{Fault-tolerance and Spacetime Distance}
\label{sec:fault_tolerance}

\begin{figure}
    \centering
    \includegraphics[width=.9\linewidth]{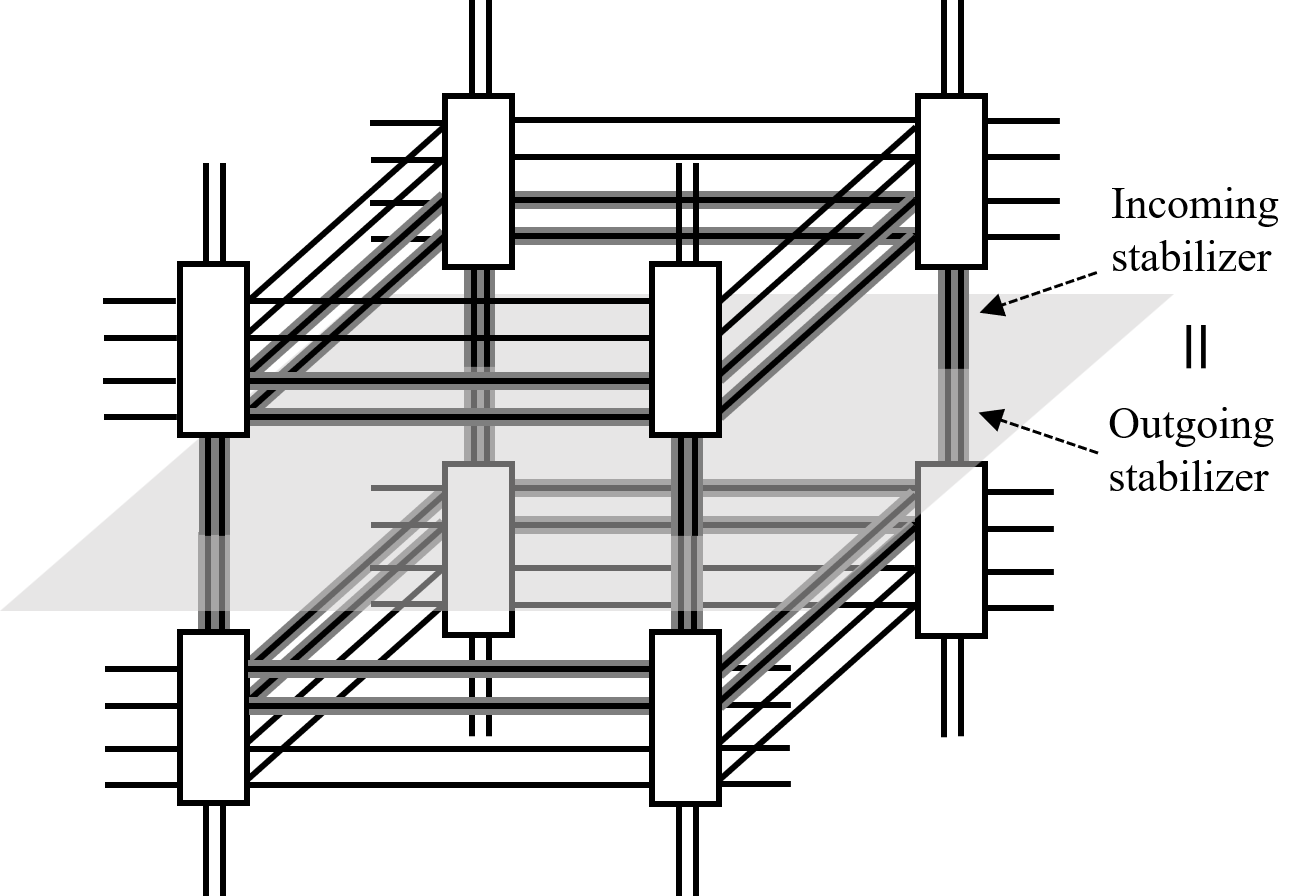}
    \caption{The formation of spacetime stabilizer from repetition of dynamical codes. Here we show that the Pauli web of a code stabilizer form a closed loop in spacetime, whose inter-gadget support are marked in shaded gray lines. }
    \label{fig: spacetimestab}
\end{figure}
\begin{figure}
    \centering
    \includegraphics[width=0.7\linewidth]{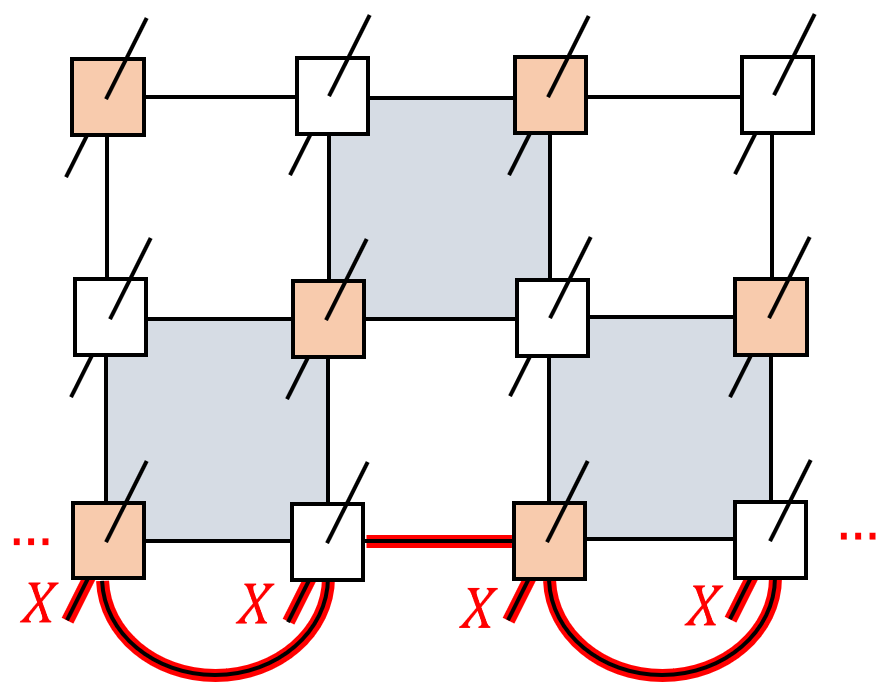}
    \caption{Non-fault-tolerant boundary of the HH Floquet code with cylindrical geometry. Here the incoming and outgoing toric codes are terminated with two-body $Z$ stabilizers on the lower boundary of the lattice which are supported on the two gadgets connected by the arched legs. If one uses the bulk gadget on the boundary, the boundary gadgets will map the incoming $X$ logical operator on the boundary to a set of connected bond operators $P$ on the boundary according to Eq. \eqref{eq: tctableau}, which implies the incoming $X$ logical operator is measured.}
    \label{fig: tcboundary}
\end{figure}

\subsection{Macroscopic Logical Map condition (MLSC)}
\label{subsec:MLSC}
In this section, we discuss fault-tolerance of the dynamical code constructed from the gadget layout. Recall that, in our definition in Sec. \ref{sec: generalcondition}, fault-tolerance of a dynamical code means that it preserves the logical subspace that is shared between the incoming and outgoing stabilizer codes. 

Before delving into fault-tolerance, we briefly discuss the mechanism of error detection in dynamical code. In gadget layout, every incoming/outgoing code stabilizer has a well-defined local mapping onto the bond operators given by the Pauli web of the stabilizer when it enters the dynamical code $\M$ (see Appendix \ref{sec: zxrules} for discussions of Pauli web in ZX-calculus). When $\M$ is repeated (or rewinded with $\overline{\M}$), the Pauli web of every outgoing stabilizer that is generated by the first dynamical code enters and closes in the subsequent one, see  Fig. \ref{fig: spacetimestab}. We refer to these closed Pauli webs as \textit{spacetime stabilizers}. In the literature, especially in the context of circuit-level setting, they are also called spacetime detectors, as they detect circuit-level errors that occur during the dynamical code, including idling error and syndrome extraction error. The mechanism of error-detection is the following. In the ZX-diagram of the circuit, the Pauli web of every spacetime stabilizer involves multiple measurements. Without errors, all these measurements around the spacetime stabilizer will multiply to $1$. Pauli errors that occur in various parts of the circuit can be represented by $\pi$ nodes of different species on the corresponding edges of the ZX-diagram \cite{bombin_unifying_2024}. Examples of circuit-level error on the ZX-diagram are listed in the Appendix \ref{sec: decoding}. If a Pauli error occurs on any edge on which a spacetime stabilizer is also supported, and if it anti-commutes with the action of the Pauli web of the spacetime stabilizer, the error will be detected by the spacetime stabilizer. That is, all the measurement outcomes around the spacetime stabilizer will multiply to $-1$ instead of $1$. 

With the spacetime stabilizer picture, in order to show that the dynamical code is fault-tolerant, we need to verify that the gadgets do not map any incoming logical operator to constant-weight operators on the bonds. This includes the case where a logical operator map to identity on the bond operators, which implies that the logical operator gets measured. We provide an example when this happens, namely an open-boundary HH Floquet code defined on a cylindrical geometry where the stabilizers of the incoming and outgoing toric codes are terminated on the boundary by the two-body $Z$ stabilizers, see Fig. \ref{fig: tcboundary}. If one uses the bulk gadget that is given by the tableau in Eq. \eqref{eq: tctableau} on the boundary in this case, we seemingly have an open-boundary HH Floquet code which satisfies the conditions (i) and (ii) in Eqs. \eqref{eq: ms} and \eqref{eq: m1}. However, in this way, the incoming $X$ logical operator on the boundary will actually be measured, since it is mapped by the boundary gadgets to a set of connected bond operators $P$ on the boundary. Therefore, the dynamical code is not fault-tolerant. We note that this is in fact a manifestation of the anomaly of the measurement quantum cellular automata (MQCA) on the boundary of Floquet toric code \cite{haah_boundaries_2022, aasen_measurement_2023}.

To explicitly calculate the mapping of the logical operators, we consider the combined Hilbert space of the incoming and outgoing legs, and the internal legs of the bonds:
\begin{equation}
    \mathcal{H}_\text{total}=\mathcal{H}_\text{in}\otimes\mathcal{H}_\text{bond}\otimes\mathcal{H}_\text{out}.
\end{equation}
The dimension of $\mathcal{H}_\text{bond}$ is determined by the total number of internal legs $n_\text{bond}$. 
Now, to compute all possible incoming/outgoing operators that will be measured by the gadgets, we collect all the stabilizer tableaux associated with the encoding map of every gadget computed in Sec. \ref{sec: gadgetencoder}. We denote the stabilizer group of the gadget $i$ as $S_i$.  Note that the bond Hilbert space is shared between a pair of connected gadgets. Therefore, for a pair of connected gadgets $i$ and $j$, a stabilizer in $S_i$ may not commute with another stabilizer in $S_j$ of the other gadget $j$. Together, all the gadget stabilizers generate a group
\begin{equation}
    \mathcal{G}_\M\equiv\left\langle\bigcup_i S_i\right\rangle,
\end{equation}
where the angle bracket denote that the group is generated by elements inside the bracket.
In this way, all incoming/outgoing Pauli operators that are measured/generated by $\M$ form a subgroup of $\mathcal{G}_\M$, which we denote as \[\cS_\mathcal{G}\equiv\left\{g\in \mathcal{G}_\M\vert g|_{\mathcal{H}_\text{bond}}=\mathbbm{1}\right\}.\]
From the definition, it is easy to see that $\cS,\overline{\cS}\in \cS_\mathcal{G}$, as the bond operators will cancel out when we multiply the corresponding encoding maps of gadgets in $\mathcal{G}_\M$ for an incoming or outgoing stabilizer.
Meanwhile, $\cS_\mathcal{G}$ must also contain the logical automorphism. That is, for every incoming logical operator $L\in \mathcal{L}$, we have $L\otimes\mathbbm{1}\otimes\mathcal{M}(L)\in \cS_\mathcal{G}$.
To ensure that no logical operator is measured or generated by $\M$, we must have
\begin{equation}
    \mathcal{G}_\M\cap (\mathcal{L}\otimes \mathbbm{1}\otimes \mathbbm{1})=\mathcal{G}_\M\cap(\mathbbm{1}\otimes \mathbbm{1}\otimes\overline{\mathcal{L}})=\varnothing.
\end{equation}
In this way, when restricted to $\mathcal{H}_\text{in}\otimes\mathcal{H}_\text{out}$, $\cS_\G$ can be regarded as the stabilizer group of a stabilizer code with $2n$ physical qubits and $2k$ logical qubits, which describes the automorphism of every incoming logical $X$ and $Z$ operator.

To mathematically verify these conditions, we can compute all elements in $\cS_\mathcal{G}$ from the binary matrix representation of the group $\G_\M$, which we denote by $H_{\G_\M}$. Every row $\vec{r}=(\vec{r}_X|\vec{r}_Z)$ of $H_{\G_\M}$ represents an element in the gadget stabilizer group, where $\vec{r}_{X,Z}\in \F_2^{2n+n_\text{bond}}$ are binary vectors given by embedding every row $(\vec{v}_X|\vec{v}_Z)$ of the gadget's encoding matrices $H_{X,Z}$ (see Sec.~\ref{sec: gadgetencoder}) to the combined Hilbert space $\mathcal{H}_\text{total}$. To find incoming Pauli operators that with trivial support on the bonds, we compute the kernels of $H_{\G_\M}^T$ restricted to the bond subspace:
\begin{equation}
    \ker^\text{in}_\text{bond}(H_{\G_\M}^T)\equiv\left\{w\in \F_2^{n_\text{row}}\left\vert \left.\left(H^T_{\G_\M}w\right)\right\vert_{\mathcal{H}_\text{bond}\cup\mathcal{H}_\text{out}}=0\right.\right\},
    \label{eq:kernel_bond}
\end{equation}
where $n_\text{row}$ is the number of rows in $H_{\G_\M}$. Each $w\in \ker_\text{bond}(H_{\G_\M}^T)$ then gives rise to an element in the group $\cS_\G$ represented by the binary vector $H^T_{\G_\M}w$. The total number of independent stabilizers in $\cS_\G$ that are exclusively supported on the incoming legs is given by the dimension of the kernel in the incoming Hilbert spaces:
\begin{equation}
|\cS_\G|_\text{in}\equiv\dim(\ker^\text{in}_\text{bond}(H^T_{\G_\M})/\ker(H^T_{\G_\M})).
\end{equation}
The number of independent stabilizers on the outgoing legs, $|\cS_\G|_\text{out}$ can be similarly defined. Thus, when no logical operator is measured, we must have the following relation:
\begin{equation}
    |\cS_\G|_\text{in}=|\cS_\G|_\text{out}=n-k.
\end{equation}
We call this the bond-kernel-rank condition (BKRC).

\subsection{Spacetime distance of dynamical codes}
\label{sec: spacetimed}

The fault tolerance of a dynamical code is characterized by its \textit{spacetime code distance} \(d_\text{st}\), defined as the minimal weight of any circuit-level Pauli noise that induces an undetectable logical error. In ZX diagrams, this corresponds the minimal number of \(\pi\)-nodes anywhere in the diagram that anti-commute with the logical Pauli web while commuting with all spacetime stabilizers. 

Although BKRC rules out the case where a logical operator is measured, a logical operator might survive all the measurements but instead shrink to a constant-weight operator in spacetime. In this case, the weight of undetectable logical error will be constant, implying that even if we scale up the dynamical code by increasing $d$ of the incoming/outgoing static codes, the dynamical code itself does not become more fault-tolerant. A well-known example of this situation is the planar Floquet code~\cite{vuillot2021planar}. In our picture, this situation corresponds to an incoming/outgoing logical operator being mapped to a constant weight operator in spacetime. 

To exclude such a case, we need to compute the minimal weight among all possible operators into which an incoming/outgoing logical operator evolves. This can be achieved by using the binary matrix $H_{\G_\M}$ of gadget stabilizers. We first parameterize an incoming logical operator in the incoming stabilizer code by a binary vector $\vec{l}=(\vec{l}_X|\vec{l}_Z)$, where $\vec{l}_{X,Z}\in \F_2^{2n+n_\text{bond}}$ and have nonzero entries only in the first $n$ elements. Then, given the binary matrix $H_{\G_\M}$, all possible dressed logical operators in the combined Hilbert space $\mathcal{H}_\text{total}$ can be represented by $\vec{l}+H_{\G_\M}^Tw$ where $w\in \F_2^{n_\text{row}}$. On the ZX-diagram, every $w$ represents a way to push some of the $\pi$ nodes from the incoming legs to everywhere else in the connected gadgets $\M$. Therefore, the lowest possible weight of any operator that the incoming logical operator $\vec{l}$ can evolve into is given by 
\begin{equation}
    d(\vec{l})\equiv \min_{w\in \F_2^{n_\text{row}}}\left\vert\left\vert \vec{l}+H_{\G_\M}^Tw\right\vert\right\vert,
\end{equation}
where $||\vec{v}||$ denotes the Pauli weight of the operator represented by the binary vector $\vec{v}=(\vec{v}_X|\vec{v}_Z)$, i.e. $||\vec{v}||\equiv\sum_{i}(\vec{v}_X)_i \vee (\vec{v}_Z)_i$. Finally, the spacetime distance in the gadget layout is simply the lowest possible weight of any dressed incoming or outgoing operator:
\begin{equation}\label{eq: dst}
    d_{st}\equiv\min_{\vec{l}\in \mathcal{L},\overline{\mathcal{L}}} d(\vec{l}).
\end{equation}
With this definition, we can exclude the aforementioned case by verifying that 
\begin{equation}
    d_{st}\neq O(1).
\end{equation}
We call this the Macroscopic Logical Support Condition (MLSC).
 
In fact, if the dynamical code is constructed using the Algorithm~\ref{alg: gadgetcon} for gadget connectivity, we can provide a strong statement regarding on the spacetime distance:
\begin{thm}
\label{thm: distance}
Let \(\mathcal{S}\) and \(\overline{\mathcal{S}}\) be the incoming and outgoing stabilizer codes, each with macroscopic code distance \(d\). Then, any dynamical code constructed from Algorithm~\ref{alg: gadgetcon} which satisfies MLSC obeys
\[
d_\text{st} \geq d.
\]
\end{thm}
\begin{proof}
Let $\mathcal{S}$ and $\overline{\mathcal{S}}$ denote the incoming and outgoing stabilizer groups, each with code distance $d$. The ZX diagram of the dynamical circuit $\mathcal{M}$ can be viewed as a tensor network whose edges represent qubits and whose nodes represent either measurements or unitary Clifford operations. Every stabilizer or logical operator of the static code corresponds to a Pauli web in this diagram: a connected set of edges carrying Pauli labels that multiply to identity when traced along the web. Incoming stabilizers define Pauli webs entering from the input boundary, and outgoing stabilizers define webs exiting at the output boundary. When $\mathcal{M}$ is repeated or rewound with $\overline{\mathcal{M}}$, these webs join to form closed spacetime stabilizers that detect local circuit faults.

Now consider a Pauli error $E$ supported on a set of spacetime edges. Such an error is undetectable precisely when it commutes with all spacetime stabilizers, i.e. with the Pauli webs of all measured operators in $\mathcal{M}$. In the ZX picture, from Algorithm~\ref{alg: gadgetcon}, the Pauli webs of the incoming and outgoing stabilizers cover the entire ZX-diagram: every edge participates in at least one Pauli web that is connected to either the input or output boundary. Hence, any Pauli error that commutes with all spacetime stabilizers must be equivalent to a Pauli web that begins and ends on the boundary. In other words, every undetectable error can be deformed (by multiplication with spacetime stabilizers) to an operator acting entirely on the incoming or outgoing code space. The correspondence between such boundary-supported webs and the logical operators of $\mathcal{S}$ or $\overline{\mathcal{S}}$ follows from the centralizer condition $C_P(\mathcal{S}_\mathcal{G}) = \mathcal{L}\otimes\overline{\mathcal{L}}$ described in Sec.~\ref{subsec:MLSC}. Therefore, all undetectable errors are logical operators of one of the static codes.

The MLSC ensures that every logical operator $L\in\mathcal{L}$ is mapped by the gadget encoder to a Pauli operator whose support is extensive in spacetime, with weight at least $d$. Intuitively, the Pauli web representing $L$ may propagate through the diagram but cannot pinch off or shrink locally, since that would violate the macroscopic support requirement. Thus, any undetectable spacetime error, being equivalent to some boundary logical, must also have weight at least $d$. No lower-weight error can commute with all spacetime stabilizers without anti-commuting with at least one of them, which would make it detectable.

Consequently, the minimal weight of an undetectable error, the spacetime code distance $d_\text{st}$, is bounded below by the static code distance $d$:
\[
   d_\text{st} \ge d.
\]
This completes the proof.
\end{proof}

\section{Equivalence relation for dynamical codes}
\label{sec:classification}
\begin{figure*}
    \centering
    \includegraphics[width=0.6\linewidth]{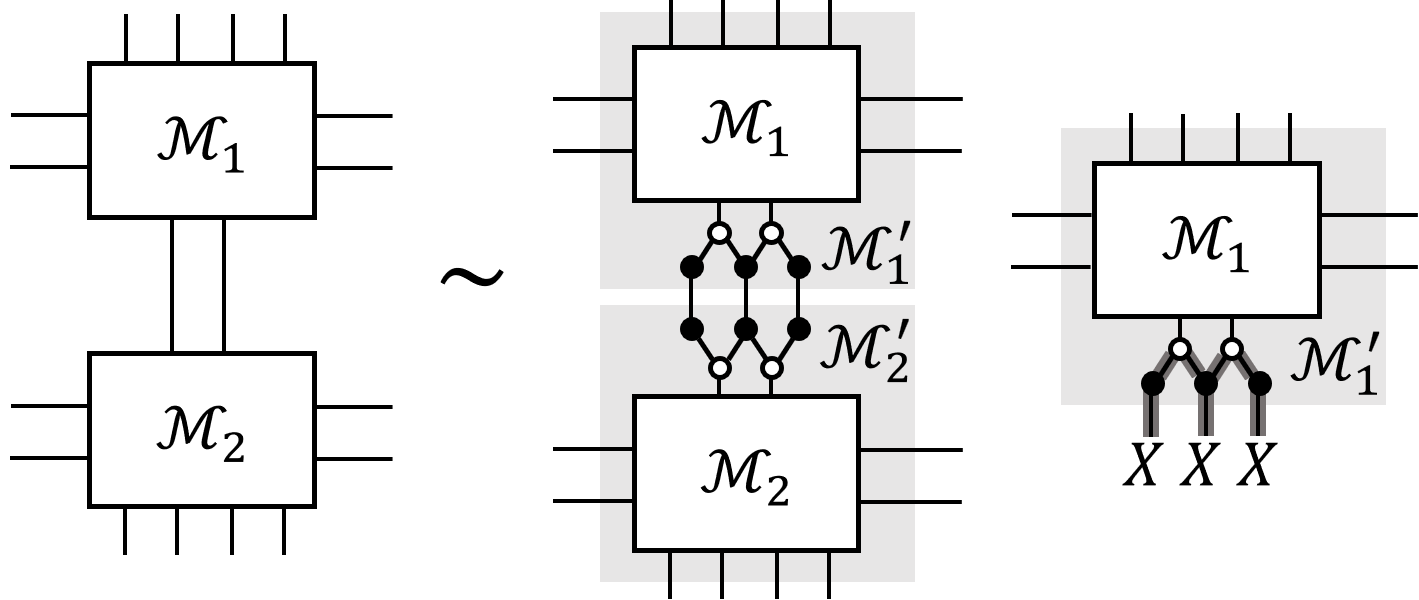}
    \caption{Minimal and non-minimal gadgets. Consider a pair of minimal gadgets $\M_1$ and $\M_2$ which are connected by a bond with two internal legs. A non-minimal gadget can be obtained from a minimal one by encoding the internal legs in the bond direction into more internal legs, which would yield a bond-local internal stabilizer. In the example shown in the figure, the non-minimal gadget $\M_1'$ obtained from $\M_1$ via such a encoding has a bond-local stabilizer $XXX$.}
    \label{fig: minimal}
\end{figure*}

The gadget layout not only provides a useful algorithm for constructing dynamical codes from stabilizer codes, but it also provides a useful setting to discuss the relation between different dynamical codes. In particular, we propose two different notions of equivalence for dynamical codes.
\begin{defn}
    Two Clifford dynamical codes are ZX-equivalent if their ZX-diagrams are equivalent under ZX rewrite rules.
\end{defn}
\begin{thm}
Two Clifford dynamical codes with the same incoming and outgoing stabilizer groups $\cS=\overline{\cS}$ are ZX-equivalent if and only if they have the same incoming and outgoing stabilizer groups and the same logical automorphism.
\label{thm: zxequiv}
\end{thm}
The proof is straightforward due to the completeness of Clifford ZX-calculus \cite{backens_zx-calculus_2014,wetering_zx-calculus_2020,jeandel_completeness_2020}, where the phase of every node is $\frac{\pi}{2}\Z$. From Thm. \ref{thm: zxequiv}, we immediately have:
\begin{col}
Every dynamical code with $\cS=\overline{\cS}$ and trivial logical automorphism is ZX-equivalent to the SASEC of $\cS$.
\end{col}
Therefore, even for a dynamical code $\M$ with $\cS\neq \overline{\cS}$ and/or nontrivial logical automorphism, it is equivalent to the SASEC if we can repeat $\M$ or $\overline{\M}\circ\M$ multiple times to produce trivial logical automorphism.

Crucially, ZX-equivalence does NOT imply that two dynamical codes share the same spacetime distance, as noted in Ref.~\cite{rodatz_floquetifying_2024}. To ensure that two ZX-equivalent dynamical codes share the same fault-tolerance properties, we propose a much stronger equivalence relation between dynamical codes, motivated by the discussions of spacetime detector models in Sec.~\ref{sec: spacetimed}:
\begin{defn}
    Two ZX-equivalent dynamical codes are spacetime equivalent if there exists a set of rewrites between the two ZX-diagrams of the dynamical codes that do not create or destroy any spacetime stabilizers.
\end{defn}
Clearly, this is a stronger notion of equivalence than ZX-equivalence, since using ZX rewrite rules may create or destroy spacetime stabilizers that are completely inside the measurement circuit $\M$, and therefore change the number of total independent spacetime stabilizers. For example, consider a SASEC that is repeated multiple times. No matter how many times the SASEC is repeated, it is always ZX-equivalent to a single SASEC
if no error occurs when the circuit is repeated. However, the repeated circuit clearly has more spacetime stabilizers, and is therefore not spacetime-equivalent to a single SASEC.

To justify the usage of gadget layout in this context, we recall the SLPC from Sec. \ref{sec: compile}:
\begin{defn}\
    A dynamical code $\M$ is SLP if it satisfies the SLPC, that is, if its ZX-diagram can be brought to a gadget layout without using any ZX rewrite rules that generates new spacetime stabilizers, such that every stabilizer $s$ in the incoming/outgoing stabilizer code can be mapped to a set of bond operators that only connect the gadgets associated with the original data qubits that are involved in $s$.
\end{defn}
In the above definition, it is important to note that an arbitrary application of the ZX-rules that generates or eliminates new spacetime stabilizers can potentially affect the SLPC. 
For example, the pairwise measurement surface code proposed in Ref. \cite{gidney_pair_2023} is not SLP according to the original connectivity constraints. This is because, although the ZX-diagram of the pairwise measurement circuit can be brought to a gadget layout by moving the ZX nodes towards the data qubits without merging or adding any node, each incoming stabilizer will be mapped to internal bonds that are connected to gadgets that were not associated with the qubits of that original stabilizer. In fact, there is a decrease in spacetime code distance of the pairwise measurement toric code of Ref. \cite{gidney_pair_2023} from $d$ to $d/2$, in which the SLPC is not maintained. For an explicit demonstration of this on the ZX-diagram, see Appendix \ref{sec: gidney}.

We are now in place to discuss the equivalence relations of dynamical codes from the perspective of gadget layout. In particular, we relate the equivalence of the entire dynamical code to microscopic conditions at the level of the gadgets. To move forward, we make some more definitions as follows. 
\begin{defn}\label{def: gadgetmin}
    A gadget is minimal if it supports no bond-local internal stabilizer.
\end{defn}
Here, the bond-local internal stabilizer refers to any internal stabilizer of the gadget that has support only on internal legs that are part of the same bond.
A diagrammatic understanding of this minimal-ness is shown in Fig. \ref{fig: minimal}.
\begin{defn}\label{def: mmin}
    A gadget layout of an SLP dynamical code $\M$ is minimal if every gadget in the layout is minimal.
\end{defn}

We now state and prove the central result of the relation between gadget layout and the equivalence relation of dynamical codes.
\begin{thm}\label{thm1}
    Two SLP dynamical codes between the same pair of incoming and outgoing stabilizer codes are spacetime equivalent if and only if there exists a pair of minimal gadget layouts of the two codes with the same internal leg layout, and each pair of gadgets at the same spatial location in these two layouts is bond-local unitary equivalent.
\end{thm}
\begin{proof}
    
Here, the internal leg layout refers to the number of internal legs that connects each pair of gadgets, and bond-local unitary equivalence means that the two gadgets can at most differ by a collection of unitary gates that act only on internal legs that are part of the same bond.
The ``if" part of the theorem is straightforward: using the definition of bond-local unitaries above, the Pauli web of every spacetime stabilizer will simply be conjugated by all the bond-local unitaries when going from one dynamical code to the other. Hence, they are spacetime-equivalent. 
The ``only if" part of the theorem can be proven as follows. Consider a pair of spacetime equivalent dynamical codes $\M_A$ and $\M_B$. Suppose that we start from the gadget layout of $\M_A$. Since $\M_A$ and $\M_B$ are spacetime equivalent, we must be able to deform the ZX-diagram of $\M_A$ towards $\M_B$ without affecting the SLPC or introducing new spacetime stabilizers. Therefore, the allowed manipulations are restricted to bond-local unitary gates because they preserve the number of internal legs for each bond. In the case that one or both gadget layouts are non-minimal, we allow operations to bring them to a minimal form, but beyond that, only bond-local unitary gates can be allowed for spacetime equivalence. This completes the proof.
\end{proof}

The above discussion of the equivalence between dynamical codes also has practical importance. In constructing dynamical codes, given enough internal legs $n_L$, the consistency equation of the encoding matrices usually admits multiple different solutions, each leading to a different encoding map of the gadget. To determine whether two encoding maps give rise to the same dynamical code,  we can use Thm. \ref{thm1}: two gadget layouts that are bond-local unitary equivalent yield a pair of spacetime-equivalent dynamical codes. On the flip side, if we find two solutions to the consistency equations of the encoding maps, either Eq. \eqref{eq: csscon} or Eq. \eqref{eq: cliffordcon}, that have the same internal leg layout, but the corresponding two gadgets cannot be transformed into each other under bond-local unitaries, they must lead to neither ZX-equivalent nor spacetime-equivalent dynamical codes.  From Thm. \ref{thm: zxequiv}, this would mean that they have different logical automorphisms for $\cS=\bar{\cS}$. For an example of such a case, see Appendix \ref{sec: csscc}.

As an example, we discuss the relation between the 012 and the CSS Floquet codes. Denoting the entire circuits of the two Floquet codes as $\M_{012}$ and $M_\text{CSS}$, we prove that $\M_\text{CSS}$ is spacetime-equivalent to $\M_{012}\circ \M_{012}$.
In fact, repeating the gadget of the 012 Floquet code in Eq. \ref{eq: tcpaulimat} yields a new gadget with 2 internal legs per bond direction, whose bond operators are  
\begin{align}\label{eq: tc012x2}\nonumber
    A_X=B_X=C_X=D_X=PI,\\\nonumber
    A_Z=B_Z=C_Z=D_Z=QI,\\\nonumber
    \overline{A_X}=\overline{B_X}=\overline{C_X}=\overline{D_X}=IQ,\\
    \overline{A_Z}=\overline{B_Z}=\overline{C_Z}=\overline{D_Z}=IP.
\end{align}
Since one can always find a bond-local unitary on each bond that rotates $PI/QI/IQ/IP$ to $XI/ZI/IX/IZ$, two rounds of the 012 dynamical code lead to the same dynamical code as the CSS Floquet toric code once the internal legs are connected. For example, for the gadget of HH Floquet code where $P=X$ and $Q=Y$, the bond local unitary is $SHS^{-1}\otimes SHS^{-1}$ on the two internal legs of every bond, as is shown in Fig. \ref{fig: tcclifcss}.

\begin{figure}
    \centering
    \includegraphics[width=0.9\linewidth]{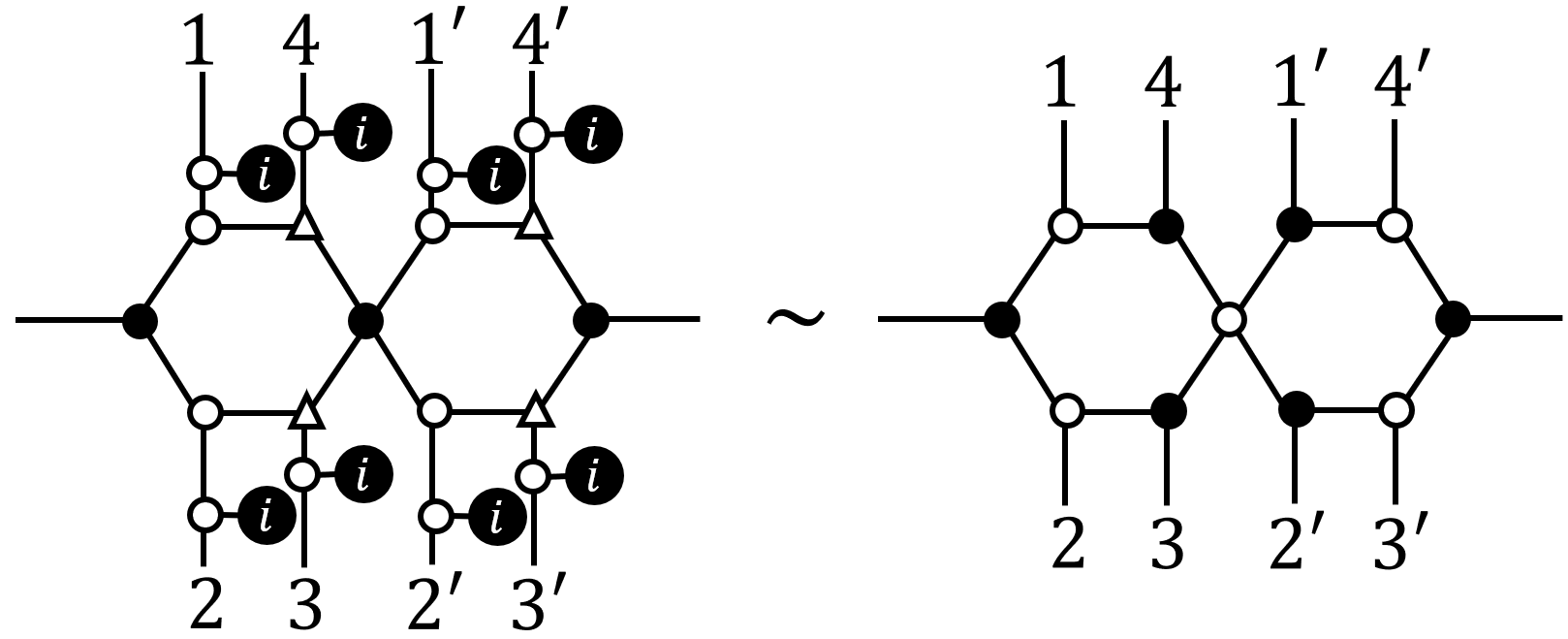}
    \caption{The diagrammatic proof that two rounds of the HH Floquet code is spacetime-equivalent to the CSS Floquet toric code. The unitary gate $SHS^{-1}$ is represented by a phase-less $X$-node connected to a single-leg $\frac{\pi}{2}$ $Z$-node}
    \label{fig: tcclifcss}
\end{figure}

We further demonstrate that CSS Floquet toric code is spacetime-equivalent to the SASEC. To see this, we apply a bond-local unitary $\prod_{i=1}^4C(X_i,Z_{i'})$ to the L gadget in Fig. \ref{fig: tczx}, which results in a tree-shaped ZX-diagram in Fig. \ref{fig: tcequi}\footnote{This new gadget correspond to the solution \begin{align}\nonumber
    A_X=B_X=C_X=D_X=XI,\\\nonumber
    A_Z=B_Z=C_Z=D_Z=ZZ,\\\nonumber
    \overline{A_X}=\overline{B_X}=\overline{C_X}=\overline{D_X}=XX,\\
    \overline{A_Z}=\overline{B_Z}=\overline{C_Z}=\overline{D_Z}=IZ
\end{align}
of Eq. \eqref{eq: tceqs}.}. Connecting the internal legs of the gadgets and merging the $X$ and $Z$ nodes (see Fig. \ref{fig: tccnot}), we see that this is indeed the SASEC. In fact, the gadget in Fig. \ref{fig: tcequi} is exactly the minimal gadget layout of the SASEC, where each weight-4 stabilizer is broken down to a ring of four nodes. We note that a similar ZX-diagram was proposed in Ref. \cite{rodatz_floquetifying_2024} as a code distance-preserving rewrite of the SASEC. 


In Appendix \ref{sec: csscc}, we provide another example in which gadgets of two Floquet color codes, a Clifford one and a CSS one, have the same number of internal legs but are not bond-local equivalent, and the resulting Floquet codes are neither ZX- nor spacetime-equivalent. Nevertheless, they are spacetime-equivalent if both gadgets are repeated 3 times, in which case both dynamical codes have trivial logical automorphism.

\begin{figure}
    \centering
    \includegraphics[width=0.75\linewidth]{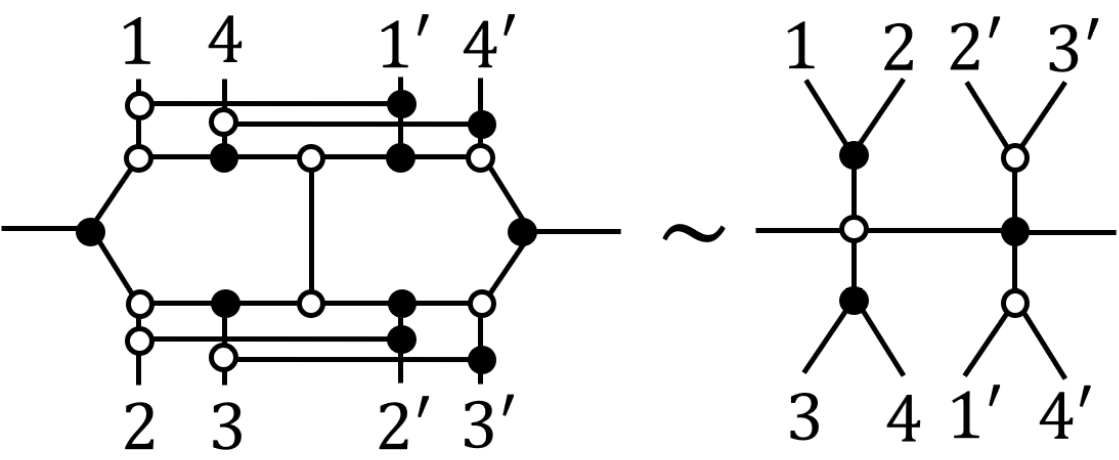}
    \caption{Diagrammatical demonstration that the SASEC of toric code is equivalent to the CSS Floquet toric code.}
    \label{fig: tcequi}
\end{figure}
\begin{figure*}
    \centering
    \includegraphics[width=0.7\linewidth]{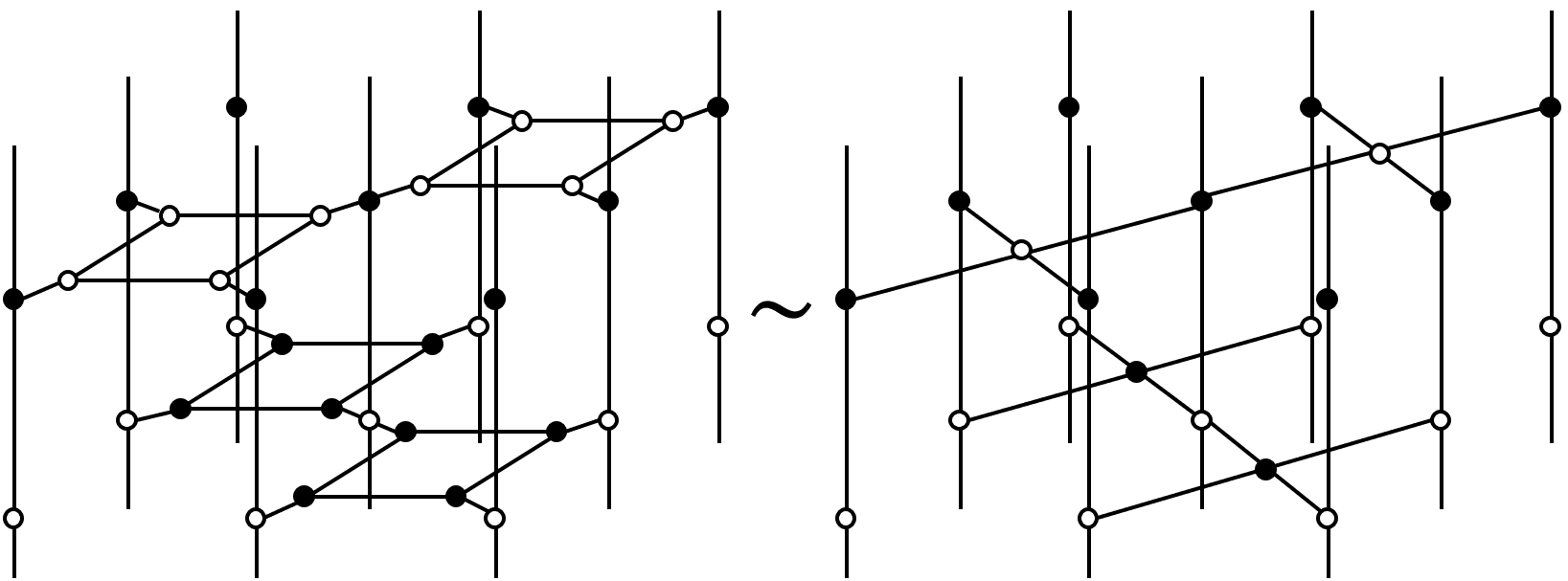}
    \caption{The connected gadgets in Fig. \ref{fig: tcequi} is ZX-equivalent (but not spacetime equivalent) to the SASEC. Here the incoming/outgoing directions are pointing downwards/upwards in the ZX-diagram.}
    \label{fig: tccnot}
\end{figure*}

\section{Resource Trade-offs}
\label{sec:resource}

In addition to being a tool in discussing equivalence relations between dynamical codes, gadget layout provides a natural way to quantify the amount of computational resource on the quantum hardware that is needed to construct a dynamical code given the incoming and outgoing stabilizer codes. 
From the examples that we have constructed, there are clearly three types of resources which are necessary to implement a dynamical code in an actual circuit:
\begin{enumerate}
    \item The circuit depth of the dynamical code
    \item The number of local physical qubits used for spatial concatenation inside each gadget; we denote it as $m_F$. We call the difference $m_F-m$ between the number of physical qubits and the original number of data qubits per gadget, $m$ in the incoming code the \textit{dynamical qubit overhead};
    \item Non-SLP: In terms of gadget layout, this comes from the connectivity of the internal legs between gadgets.
\end{enumerate}
We now show that the gadget composition, together with the diagrammatic description of quantum circuits using ZX-calculus, provides a unified framework for all three types of computational resources of dynamical codes. More precisely, we demonstrate that \textit{given a pair of incoming and outgoing stabilizer codes and the connectivity of the gadget layout, the computational resource needed to implement such a dynamical code can be quantified by the number of internal legs $n_L$ of every gadget in the layout}. Therefore, given the connectivity of the gadget layouts, the most efficient dynamical code is the one with the least number of internal legs $n_L$ per gadget \footnote{Note that we did not constrain the logical automorphism in the statement. In fact, in many examples, having a nontrivial logical automorphism yields a more efficient dynamical code with fewer internal legs.}. 

We first show that the circuit depth and dynamical qubit overhead are inter-convertible resources. To see this, recall from Sec. \ref{sec: fixfloquet} that, given a complete encoding map of a gadget, there are multiple ZX-equivalent ways to bring it to a three-stage layout. In the first stage, the data qubits of the incoming code are spatially concatenated with a local code with $m_F$ physical qubits. In the second stage, each physical qubit undergoes multiple inter-gadget operations, and the total number of inter-gadget operations, represented by the number of internal legs $n_L$, is equal to the number of inter-gadget operators for all the $m_F$ physical qubits on a gadget. Assuming that different inter-gadget operations involving the same physical qubit cannot be parallelized, the circuit depth $d_2$ of the second stage is determined by the most number of inter-gadget operations that a physical qubit is involved in. Therefore, we can roughly estimate that $d_2\sim \frac{n_L}{m_F}$ on average under the assumption of translation invariance. In general, we need a maximum value of $\frac{n_L}{m_F}$ among all gadgets.

Meanwhile, after merging the spatial encoder with the dual encoder, the ``internal measurements" block in Fig. \ref{fig: floqueteff} measures stabilizers of the local spatially concatenated code for the gadget.
Therefore, the additional circuit depth due to internal measurements can be regarded as a constant that does not scale with $n_L$ or $m_F$, hence the total circuit depth for the dynamical code $d_\text{circ}\sim d_2\sim \frac{n_L}{m_F}$. This relation establishes the circuit depth $d_\text{circ}$ and the dynamical qubit overhead $m_F-m$ for a fixed $m$, as tradable resources for the implementation of the dynamical code, as their $d_2 m_F$ is proportional to the number of internal legs $n_L$.

We further demonstrate that the number of internal legs of each gadget, $n_L$, is a tradable resource with connectivity between the gadgets. To illustrate this point, we compare two gadget layout constructions of dynamical code with different inter-gadget connectivity. For example, in constructing the Floquet BB code, instead of directly connecting every $L$ gadget with 9 $R$ gadgets so that every weight-6 incoming stabilizer is mapped to bond operators that directly connect the 6 gadgets, as we did in Sec. \ref{sec: bicycle}, we can instead consider a square lattice connection in the toric layout of the BB code, where every $L$ gadget is only connected to 4 nearest-neighbor $R$ gadgets, same as the connectivity in the Floquet toric code. This is akin to connectivity in the proposals in Ref. \cite{berthusen_toward_2025} in constructing dynamical LDPC codes on 2D planar hardware. Apparently, such a connection is not enough for SLPC, according to Algorithm~\ref{alg: gadgetcon}. In fact, the Pauli web of every incoming/outgoing stabilizer will go across gadgets that are originally not involved in the BB code stabilizer. 
Nevertheless, the spacetime code distance $d_\text{st}$ can, in principle, still match the incoming/outgoing distance $d$, provided that the resulting mapping of logical operators retains macroscopic support. To achieve this, every gadget, as a local encoder that encodes incoming/outgoing Pauli operators to the internal legs, needs to support a set of \textit{independent} stabilizers over the internal legs according to the way the bond operators go through the gadget in the Pauli webs of these extra stabilizers that the gadget was not involved in the incoming/outgoing stabilizer code.
In turn, the total number of internal legs of the gadget, $n_L$, must increase to incorporate these extra stabilizers if we have less inter-gadget connectivity to begin with and to satisfy the SLPC.

Another way to visualize the trade-off between non-SLP and the number of internal legs of each gadget, $n_L$, is to consider a fixed connectivity of the gadget layout, in which case a dynamical code with relaxed connectivity can potentially be constructed using gadgets that have fewer internal legs compared to the one whose gadgets fully respect our notion of SLPC.
To see this, we need to consider spacetime stabilizers formed by connected gadgets that are completely inside the blackbox of $\M$.
These internal spacetime stabilizers are common in our examples of dynamical codes, especially in the CSS ones. For example, the CSS Floquet toric code and the local SASEC in Fig. \ref{fig: tccnot} both feature internal spacetime stabilizers on every square plaquette of the square lattice that do not involve incoming or outgoing legs. In principle, these internal spacetime stabilizers can be removed using rewrite rules of ZX-calculus, since they do not contribute to the mapping of the incoming code stabilizers and logical operators. Doing so will yield a ZX-diagram with fewer spacetime stabilizers, which can potentially be implemented by a circuit that requires a smaller spacetime resource given by $n_L$.
However, such rewrites typically destroy the strict locality pattern in the gadget layouts. For example, in Fig.~\ref{fig: tccnot}, one can contract the loop of four $X$ or $Z$ nodes inside each plaquette, removing the internal spacetime stabilizer formed by the four edges around these loops. This leads to the pairwise-measurement toric code, where the number of internal legs per gadget is $n_L=6$ (see Fig.~\ref{fig: pairwise} in Appendix~\ref{sec: gidney}), compared to $n_L=8$ in the local SASEC. This illustrates how relaxing SLPC can reduce $n_L$, at the cost of a more nonlocal Pauli-web structure and potential shortened spacetime distance $d_{st}$.

\section{Fabrication defects}
\label{sec:fabrication_defects}

In a real hardware device, fabrication defects are inevitably present in the form of defected qubits or connectors. 
Thus, given a fixed hardware layout of qubits with possible fabrication defects, a fundamental question arises: Can we construct a fault-tolerant protocol that best fits the  hardware constraints? Existing approaches to incorporate fabrication defects are mostly discussed in the context of surface code \cite{barrett_fault_2010,stace_thresholds_2009,tang_robust_2016,auger_fault-tolerance_2017,strikis_quantum_2023,debroy_luci_2024}. For Floquet codes, Refs. \cite{aasen_fault-tolerant_2023,mclauchlan_accommodating_2024} demonstrated ways to modify the measurement dynamics in the presence of qubit defects. However, a general approach for a broader class of dynamical codes is lacking.

Our gadget-based approach naturally allows for local modifications, making it well suited for designing error-correcting codes that accommodate hardware imperfections such as missing qubits, broken connectors, and nonuniform qubit connectivity. By systematically studying the effects of qubit dropout and non-local stabilizer measurements, our approach could inform hardware-aware code design, reducing the need for rigid pre-defined layouts. Below, we discuss how to modify a gadget solution for a given incoming stabilizer code $\cS$ after fabrication defects are introduced. More generally, we can look for gadget solutions given a hardware qubit layout and incoming stabilizer code $\cS$, by writing a corresponding gadget layout and finding dynamical code solutions.

\subsection{Gadget layout with broken connectors}
\label{sec: broken}

\begin{figure}
    \centering
    \includegraphics[width=\linewidth]{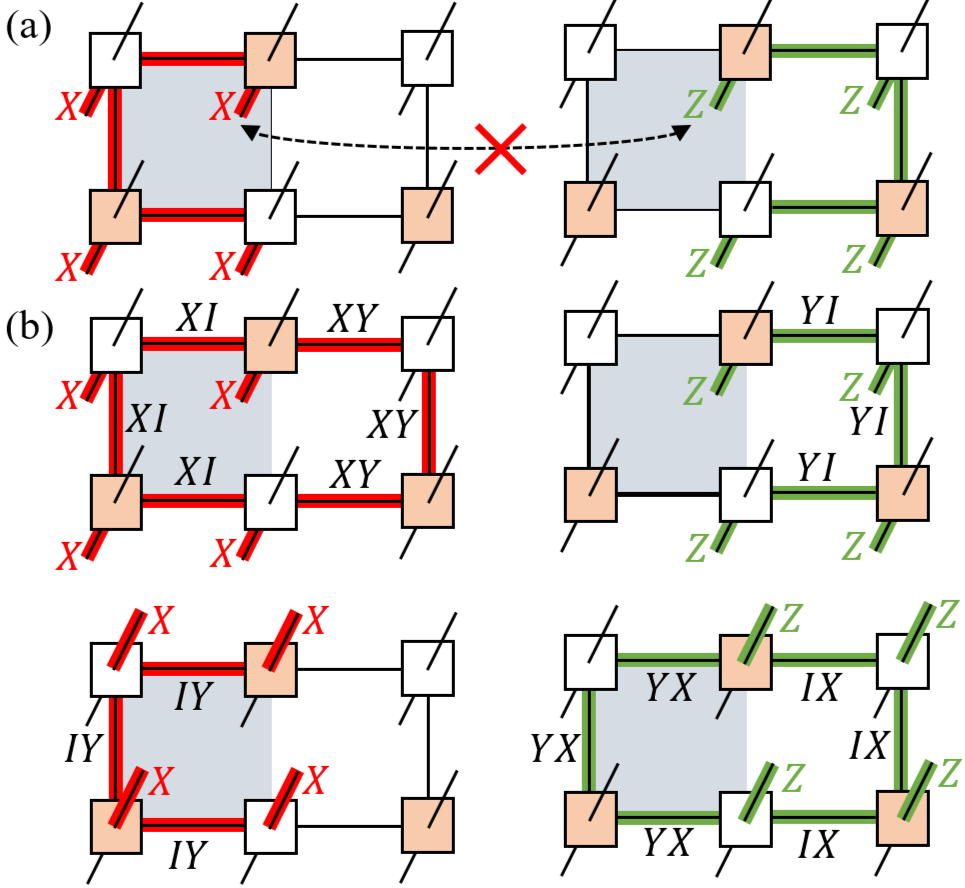}
    \caption{(a)With the presence of a missing bond, the local gadget layout does not meet the condition in Algorithm~\ref{alg: gadgetcon}, since according to the original mapping, the incoming $X$ and $Z$, as marked by the arrows, will be mapped to non-overlapping internal legs and does not preserve the anticommutation algebra. (b) The modified mapping of the incoming/outgoing $X$/$Z$ stabilizer on to the bond operators adjacent to the missing bond. The Pauli web of the incoming $X$ stabilizer enters the nearby gadgets that support the adjacent $Z$ stabilizer, and similarly for the Pauli web of the outgoing $Z$ stabilizer next to the missing bond. }
    \label{fig: brokenconn}
\end{figure}
\begin{figure}
    \centering
    \includegraphics[width=0.9\linewidth]{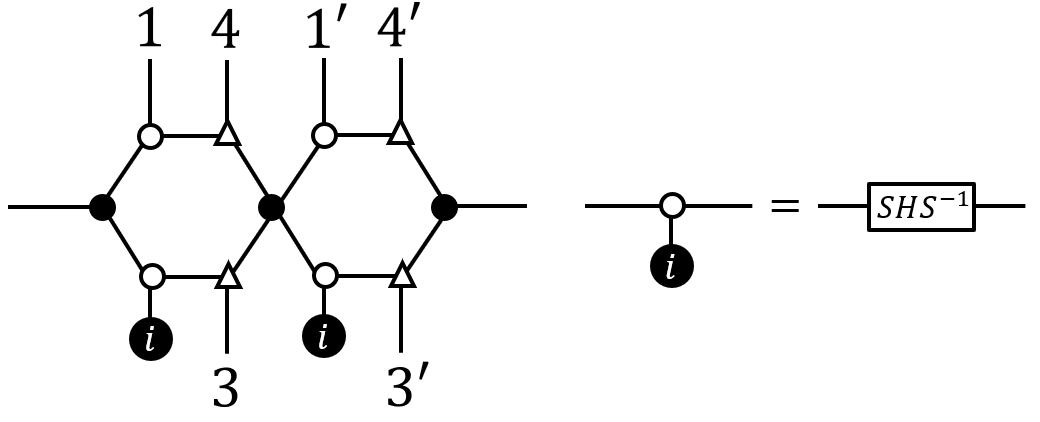}
    \caption{The ZX-diagram of the L gadget on the missing bond from the bond operators in Fig. \ref{fig: brokenconn}(b). We use the convention in Fig. \ref{fig: tcparam} for the four bond directions. The missing leg in direction 2 is replaced by an $X$ node connected to a $\frac{\pi}{2}$ $Z$ node. In circuit implementation it corresponds to the single qubit unitary gate $SHS^{-1}$, which exchanges $Z$ and $Y$.}
    \label{fig: brokenzx}
\end{figure}
\begin{figure}
    \centering
    \includegraphics[width=0.65\linewidth]{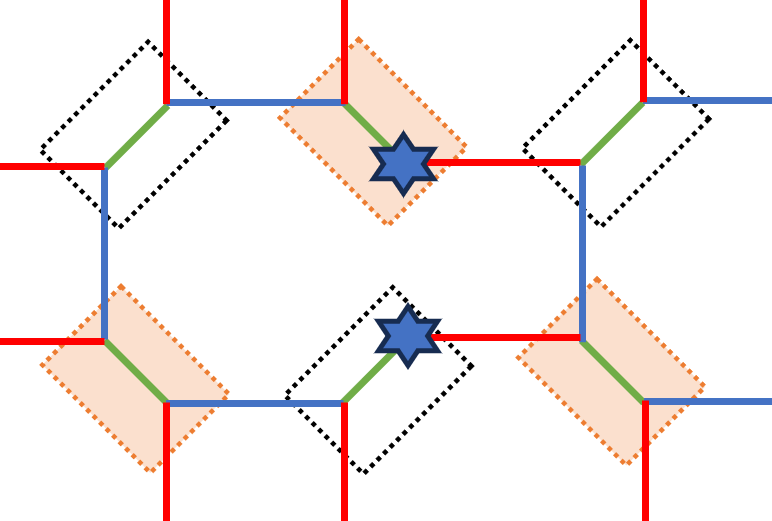}
    \caption{The local circuit implementation of the Floquet toric code with a missing bond. For the HH Floquet code, in the rounds $1$ and $\overline{1}$, we apply $SHS^{-1}$ to the two qubits at the two ends of the missing bond marked by the blue stars. Alternatively, one can measure the single qubit Pauli $X$ operator at the starred qubits at round $1$ and $\overline{1}$.}
    \label{fig: brokenschedule}
\end{figure}

From the viewpoint of gadget layout, to deal with broken connectors between physical qubits, we can consider a more restricted connection between gadgets, where all the internal legs on the bonds that are in the directions of the broken connectors are removed from the gadget layout, since in the circuit implementation, different internal legs on the same bond may use the same connector in the hardware.
In this case, the Pauli web of the incoming stabilizers can no longer satisfy the SLPC. To see this, we consider the gadget layout of the toric code with a missing bond, as shown in Fig. \ref{fig: brokenconn}. We could attempt to satisfy the SLPC by considering the bond operators of the incoming $X$ and $Z$ stabilizers that overlap on the missing bond, as shown in Fig. \ref{fig: brokenconn}(a). However, such a set of bond operators cannot faithfully encode the incoming Pauli operators, since the two gadgets on both ends of the missing bond maps incoming $X$ and $Z$ to non-overlapping internal legs and thus do not preserve the commutation algebra. As a remedy, we relax the SLPC so that the Pauli web of the incoming $X$ stabilizer next to the missing bond now enters the bonds of the adjacent plaquette that hosts the $Z$ stabilizer, and similarly for the outgoing $Z$ stabilizer, as shown in Fig. \ref{fig: brokenconn}(b). This yields a solution that is consistent with the commutation constraints. We note that since the locality is not fully preserved near the missing bond, the spacetime code distance can potentially be shortened.

For simplicity, we aim to preserve the rest of the dynamical code by doing a ``minimal rewrite" of the original dynamical code, i.e. we keep the rest of the gadgets unchanged and solve for encoding maps of the pair of gadgets connected by the missing bond accordingly. Of course, the gadgets of the rest of the system can be repeated in order to search for a nontrivial solution to the encoding map of the two disconnected gadgets if the solution for the disconnected gadgets needs more than one round of the original gadgets.

Let us further investigate the HH Floquet code with a missing bond and focus on the pair of L and R gadgets connected by the missing bond, while keeping the rest of the gadgets the same as the original 012 Floquet code. We consider letting the Pauli operator of the incoming $X$ stabilizer enter the adjacent plaquette that hosts the $Z$ stabilizer, see Fig. \ref{fig: brokenconn}(b). Clearly, the two gadgets on the top right and bottom right of the incoming stabilizer need to be stabilized by only two bond operators. However, we can check from Eq. \ref{eq: tctableau} that such a stabilizer does not exist in the Clifford encoding map of the gadget. Besides these gadgets connected to the broken bond, we can use all other gadgets as in the original dynamical code without broken bond and which produce the desired stabilizers $(XY)_1(XY)_4$ and $(YX)_3(YX)_4$ on the L gadget (we use the convention in Fig. \ref{fig: tcparam} for the four bond directions). These stabilizers come from the outgoing $X$/$Z$ operator originating from the previous gadget and entering the subsequent one. Using these stabilizers, we fix the bond operators as in Fig. \ref{fig: brokenconn} (b) including the ones for the pair of disconnected L and R gadgets. 
We then need to reconsider the encoding map of this pair. Using the bond operators, we find that the ZX-diagram of the L gadget on the missing bond has the form of Fig. \ref{fig: brokenzx}. Therefore, we obtain a minimal modification of the 012 Floquet code, which now has a 6-round schedule $012\overline{0}\overline{1}\overline{2}$, whose lattice layout is shown in Fig. \ref{fig: brokenschedule}. In rounds $1$ and $\overline{1}$, due to the broken connection between the two physical qubits, we instead perform single qubit unitary gates $SHS^{-1}$ on the two qubits at both ends of the broken connector. We warn the reader that, although the ZX-diagram in Fig. \ref{fig: brokenzx} looks like a 3-round schedule, these three rounds do NOT actually form a dynamical code, since the incoming $X$ stabilizer adjacent to the broken connector is not yet measured after these three rounds.

In addition to the choices in Fig. \ref{fig: brokenconn}, we find another solution to the set of bond operators around the missing connector, in which the bond operators of the incoming $Z$ and outgoing $X$ stabilizer will enter the nearby plaquette. In this case, the ZX-diagram of the L gadget on the missing bond is similar to the one in Fig. \ref{fig: brokenzx}, but the $\frac{\pi}{2}$ $Z$-node in the ZX-diagram is replaced by a phaseless $Z$ node. This yields another circuit implementation: at rounds $1$ and $1'$, the starred qubits in Fig. \ref{fig: brokenschedule} are measured in the $X$-basis, rather than being acted on by the unitary gates $SHS^{-1}$.

It is straightforward to verify that the modified HH Floquet code with broken connector is fault-tolerant with a preserved logical space using the method in Sec. \ref{sec:fault_tolerance}. In fact, we can directly visualize the logical automorphism for a logical operator whose Pauli web originally passes through the broken connector, see Fig. \ref{fig: brokenlogical} in Appendix \ref{sec: defectdistance}. There, we further demonstrate that the new HH Floquet code with a missing connector preserves the logical space, and its spacetime code distance is reduced from $d$ to $d-1$ under periodic boundary conditions.

Moving beyond the example of the HH Floquet code, the strategy of using multiple repetitions of the dynamical code $\M$ (either directly repeating or through $\overline{\M}\circ\M$) is a general approach to incorporate broken connectors into dynamical codes. This can be intuitively understood as follows. Suppose that we try to minimally rewrite the dynamical code with a broken connector by simply terminating the Pauli web of every incoming/outgoing stabilizer on single-leg nodes that are stabilized by the bond operator of the broken bond. This is usually impossible since the Pauli webs of a pair of incoming $X$ and $Z$ stabilizers may overlap on one bond with non-commuting bond operators, hence they cannot be stabilized simultaneously. To avoid this problem,
we repeat $\M$ (or $\overline{\M}\circ\M$) multiple times. When $\M$ is repeated $N$ times, the entire dynamical code now contains $N-1$ sets of internal spacetime stabilizers formed through the mechanism in Fig. \ref{fig: spacetimestab}. We can multiply the Pauli web of an incoming/outgoing stabilizer with any of these internal spacetime stabilizers, which yields another Pauli web of the same incoming code stabilizer (albeit potentially being non-local). We increase $N$ from 2 until we find a way to extend the Pauli webs of the incoming code stabilizers so that the overlapping bond operators on the broken bond all commute with each other and can hence be terminated on single-leg nodes. For example, in Fig. \ref{fig: brokenconn}(b), the Pauli webs of the incoming $X$ and outgoing $Z$ stabilizer are extended via internal spacetime stabilizers, so that the bond operators of the four Pauli webs are supported on the missing bond either as $YI$ or $IY$. Therefore, all these Pauli webs can be terminated simultaneously on the broken bond via single leg nodes that are stabilized by $YI$ and $IY$, which are exactly the two $\frac{\pi}{2}$ $Z$ nodes in Fig. \ref{fig: brokenzx}.

\subsection{Gadget layout with qubit dropout}
\label{sec: dropout}
\begin{figure}
    \centering
    \includegraphics[width=0.4\linewidth]{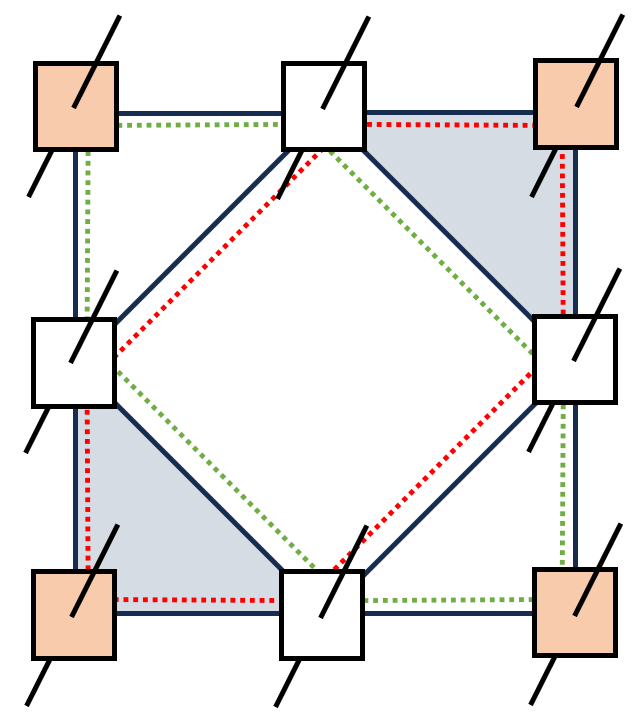}
    \caption{The local subsystem code around a dropout qubit and the corresponding gadget layout with the gadget that contains the dropout qubit removed. Here the gauge operators are weight-3 $X$ operators on the shaded triangles and weight-3 $Z$ operators on the unshaded triangles. The local stabilizers are the weight-6 $X$ operator around the red hexagon and the weight-6 $Z$ operator around the green hexagon marked by the dotted red and green hexagons, respectively. }
    \label{fig: dropout}
\end{figure}
\begin{figure}
    \centering
    \includegraphics[width=0.8\linewidth]{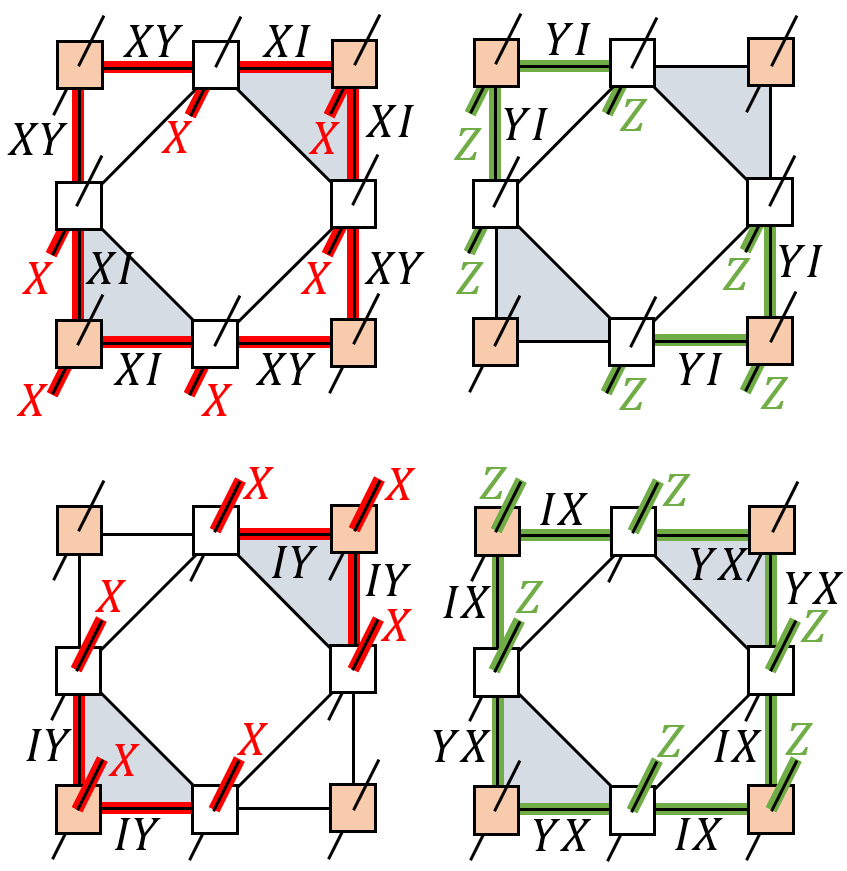}
    \caption{The bond operators of the incoming and outgoing weight-6 stabilizers.}
    \label{fig: dropoutbond}
\end{figure}
\begin{figure}
    \centering
    \includegraphics[width=0.8\linewidth]{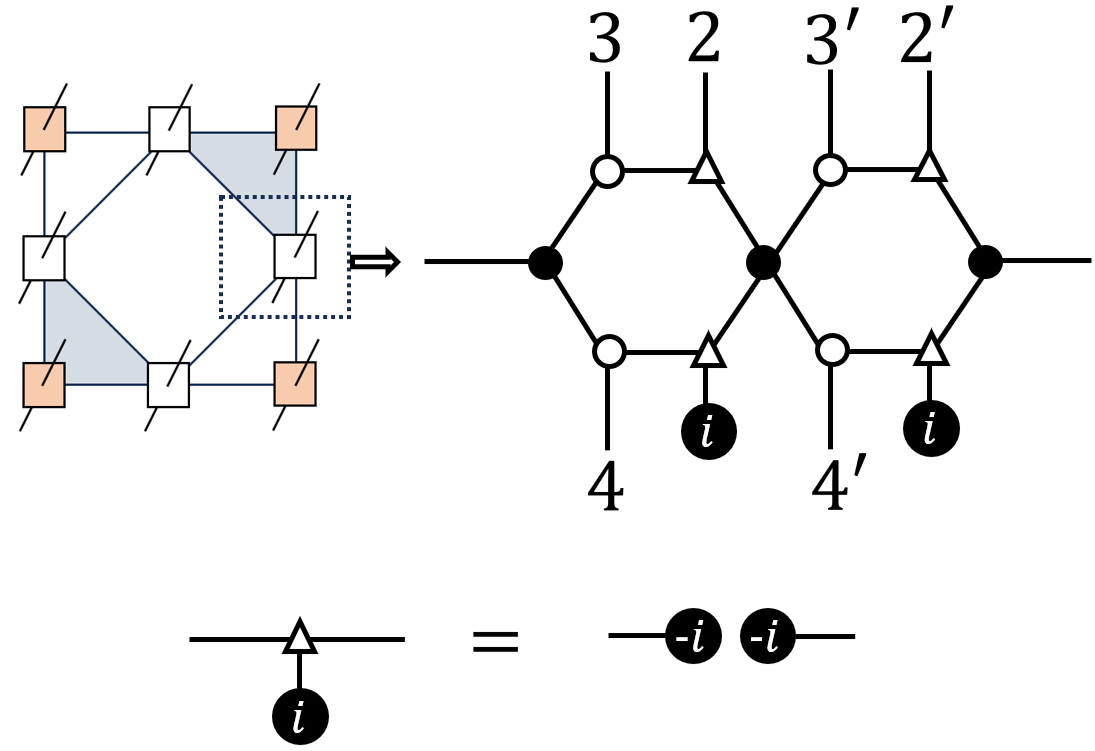}
    \caption{The ZX-diagram of the R gadget on the right of the missing R gadget. The $Y$ node in direction 1 is now replaced with a single-qubit projective measurement in $Y$ basis.}
    \label{fig: dropoutzx}
\end{figure}
\begin{figure}
    \centering
    \includegraphics[width=0.6\linewidth]{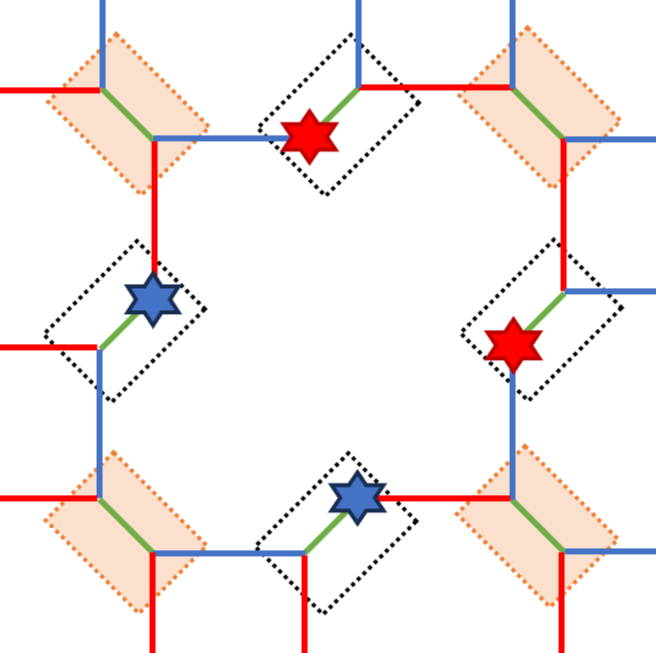}
    \caption{The circuit implementation of the 012 Floquet code with dropout qubits. There are two possible circuit implementations. (A): at rounds $1$ and $\overline{1}$, we apply $SHS^{-1}$ to the blue starred qubits, and at rounds $2$ and $\overline{2}$, we measure the red starred qubits in the $Y$ basis. (B): At rounds $1$ and $\overline{1}$, we measure the blue starred qubits in the $X$ basis, and at rounds $2$ and $\overline{2}$ we apply $H$ gate to the red starred qubits.}
    \label{fig: dropoutlat}
\end{figure}

We now consider another type of fabrication defect, the qubit dropout, i.e. the physical qubit on the spatial location is defective and not present in the circuit. 
Before constructing the new dynamical code with qubit dropout, we have to redesign the incoming and outgoing stabilizer codes according to the spatial concatenation inside the gadget.
To minimally modify the original dynamical code, we can associate a gadget with each qubit of the incoming stabilizer code, and hence remove the gadget associated with the defected qubit and the legs supported on it.
With the defected qubits removed, the stabilizers that overlap on those data qubits become incomplete, and potentially not commuting with each other. A standard approach here is to treat these incomplete stabilizers as gauge operators of a local subsystem code~\cite{auger_fault-tolerance_2017}.  For example, for toric code with a data qubit dropout, the four incomplete stabilizers now form a set of local gauge operators, while the stabilizer subgroup is the weight-6 $X$ and $Z$ operators, shown in Fig. \ref{fig: dropout}. We note that with any redesigned stabilizer codes, the stabilizer code distance will already be potentially reduced. This is a major difference compared to the case with a broken connector, where the incoming stabilizer code remains unchanged.
Similarly to Appendix \ref{sec: fermion}, when constructing dynamical codes from subsystem codes, we only need to consider bond operators of the incoming and outgoing stabilizers, rather than the entire gauge group. Again, we will look for ``minimal rewrites" from already constructed dynamical codes. Furthermore, if possible, we will try to avoid adding more connections between gadgets around the dropout qubit. 

We use the HH Floquet code as an example. Consider a single physical qubit dropout. Since both physical qubits inside the gadget are connected to the incoming and outgoing legs via the (dual) spatial encoder, a physical qubit dropout causes the removal of the data qubit that is represented by the incoming/outgoing legs from the toric code. Therefore, in gadget layout, we need to completely remove the gadget that contains the defective physical qubit. That is, both the defected qubit, the other physical qubit inside the gadget, and all the connections between the defected gadget and nearby gadgets are removed, see Fig. \ref{fig: dropout}. To aim for a minimal rewrite, we assume that all gadgets, except the four L gadgets that originally connect to the defective R gadget, remain the same (up to repetition). 
We find that, with the rest of the gadgets being repeated twice, there exists a set of bond operators for the two weight-6 local stabilizers around the $2\times 2$ square that consistently encode the incoming and outgoing Pauli operators of the four L gadgets, see Fig. \ref{fig: dropoutbond}\footnote{Note that the dynamical code actually performs a switch of gauge in the local subsystem code, where in the incoming stabilizer code the gauge are fixed to be two weight-3 $Z$ stabilizers on the two unshaded triangles on the top left and bottom right in Fig.\ref{fig: dropout}, and the weight-6 $X$ stabilizer on the dotted red hexagon there. The outgoing stabilizer code is stabilized by two weight-3 $X$ stabilizers and the weight-6 $Z$ stabilizer. }. In fact, the L gadget at the bottom produces exactly the same encoding maps as the L gadget in Fig. \ref{fig: brokenconn}(b), therefore its ZX-diagram is still the one in Fig. \ref{fig: brokenschedule}. Meanwhile, the ZX-diagram of the gadget on the right is slightly different, which we show in Fig. \ref{fig: dropoutzx}. We then arrive at the circuit implementation with qubit dropout, as shown in Fig. \ref{fig: dropoutlat}. The schedule is now 6-round, which is denoted by $012\overline{012}$. In rounds $1$ and $\overline{1}$, the qubits marked by the blue stars in Fig. \ref{fig: dropoutlat} undergo unitary gates $SHS^{-1}$, while in rounds $2$ and $\overline{2}$, the two qubits marked by the red stars are measured in the $Y$ basis. 

Similar to the case of missing bond, we find an alternative solution to the bond operators around the missing gadget, where the bond operators of the incoming $Z$ and outgoing $X$ stabilizers will encircle the $2\times2$ region, instead of the incoming $X$ and outgoing $Z$ stabilizers in Fig. \ref{fig: dropoutbond}. In this case, we obtain an alternative 6-round schedule $012\overline{012}$, where in rounds $1$ and $\overline{1}$, we measure the two blue starred qubits in Fig. \ref{fig: dropoutlat} in their $X$ basis, and at rounds $2$ and $\overline{2}$, we apply Hadamard gates to the two red-starred qubits.

Similarly to the broken connector case, we can verify fault-tolerance of the modified HH FLoquet code with dropout qubit using the method in Sec. \ref{sec:fault_tolerance}, see Appendix \ref{sec: defectdistance}. There, we further show that the spacetime distance of the new 012 Floquet code with one qubit dropout is reduced from $d$ to $d-1$. We will also compare our approach to the one proposed in Ref. \cite{mclauchlan_accommodating_2024}.

Our approach can be extended to other dynamical codes following a similar approach for the broken connector discussed at the end of Sec. \ref{sec: broken}. That is, if we repeat $\M$ multiple times, we can find extended Pauli webs for the spacetime stabilizers of the incoming/outgoing stabilizers of the local subsystem code that overlap on the bonds that connect to the dropout qubit, and the bond operators all commute. When the dropout qubit is removed, we simply terminate these bonds on single-leg nodes that are stabilized by the mutually commuting bond operators.

\section{Outlook}
\label{sec:outlook}

Our work introduces the concept of spacetime concatenation, which constructs quantum circuits that effectively measures the stabilizers of a stabilizer code through code concatenation (spatial concatenation) and sequence of small-weight measurements between codes used for concatenation (temporal concatenation). This is achieved by decomposing the overall measurement circuit into a collection of localized gadgets that are interconnected in spacetime through internal legs. We imposed the SLPC on the mapping of the Pauli operators to the internal legs under the action of these gadgets. Together with the MLSC, our construction leads to spacetime distance-preserving low-weight measurement solutions of the dynamical code.
This approach not only recovered well-known examples, such as the Floquet toric codes, but also led to new constructions, including the Floquet BB code and a Floquet Haah code. We discussed a notion of spacetime equivalence, which provides a restricted setting to discuss the classification of dynamical codes. 
Furthermore, we established a resource theory for dynamical codes that quantified overhead in terms of circuit depth, the number of physical qubits, and non-SLP connectivity. Finally, we presented a systematic approach for incorporating fabrication defects in dynamical codes, addressing challenges such as missing connectors and qubit dropouts, ensuring implementation in realistic hardware.

In the future, several key directions emerge for further exploration and optimization.

\paragraph{Hardware-first design}:
Our spacetime-concatenation approach would be instrumental in the construction of hardware-aware fault-tolerant LDPC codes. Specifically, one could start by writing a gadget layout tailored to a given hardware layout and then systematically search for fault-tolerant LDPC code protocols that can be supported on it. If one wants to implement a protocol for a particular incoming/outgoing stabilizer code $S$ on a given hardware layout, then one can check if the associated gadget layout allows for a compatible dynamical code solution or what minimal tweaks could be made to the hardware layout to allow for a solution. Such an approach could lead to practical implementations of high-distance, low-overhead Floquet error correction schemes for various hardware platforms. 
\paragraph{Decoding Strategies for Dynamical Codes}:
Traditional decoding strategies for stabilizer codes rely on spatial syndromes, but dynamical codes introduce an additional temporal structure that could be leveraged for sophisticated error detection and correction. We expect our general approach of constructing Floquet codes through spacetime concatenation to yield insights for developing novel and possibly more efficient spacetime-decoding algorithms, leading to scalable error correction tailored to hardware implementations.

\paragraph{Optimizing Resource Overheads}:
A crucial advantage of Floquet codes is the ability to reduce the weight of stabilizer measurements, at the cost of increased number of measurements and ancillary qubit overhead. There also exists an inherent trade-off between the ancillary qubit overhead and circuit depth. While our constructions show that shorter schedules are achievable through increased dynamical qubit redundancy, a systematic investigation into these trade-offs—both analytically and numerically—would provide deeper insights into the most efficient implementations for practical quantum hardware. Exploring techniques such as tensor network optimization of encoding circuits and compressed ZX-diagram representations may lead to more resource-efficient constructions.

\paragraph{Logical Gates and Universality in Floquet Codes}: 
The dynamical code examples in our work can induce various Clifford unitary actions on the logical subspace while preserving the stabilizers. For example, in the 012 Floquet toric code, one period 012 leads to an action of a logical Hadamard and a swap gate on the two logical qubits. 
A key question is whether dynamical codes can yield non-Clifford logical gates by generalizing our spacetime concatenation framework to codes beyond Pauli stabilizer codes. Our framework is well-suited for such a generalization and we are pursuing this in an ongoing work. 

\paragraph{Generalization to non-LDPC stabilizer codes}:
Our spacetime concatenation framework and the proof of fault-tolerance apply to LDPC stabilizer codes. We conjecture that we can simply generalize the strict-$k$-locality preservation condition to a connectivity condition where we allow maps of the gadget to a possibly large number of legs, again according to the connectivity of the qubit to other qubits via stabilizers. Because we preserve both the connectivity and algebra of operators under the gadget map, we expect that obeying this generalized connectivity condition will lead to fault-tolerant dynamical condition compilations of arbitrary stabilizer codes, including non-LDPC codes.

\section*{Acknowledgements}
 YX thanks Chao-Ming Jian and Zhehao Zhang for helpful comments and Shoham Jacoby for pointing out an inconsistency in the Floquet BB code construction. YX and AD thank Tyler Ellison and Victor Albert for comments. AD also thanks Jeongwan Haah for pointing out the assumption of propagation of errors in our fault-tolerance proof, which we have now proven rigorously. 
 AD thanks Victor Albert also for help with the upcoming error correction zoo entry on space-time concatenation codes. YX acknowledges support from the NSF through OAC-2118310. AD acknowledges support from the start-up fund at Virginia Tech. This research was supported in part by the NSF grant PHY-2309135 to the Kavli Institute for Theoretical Physics (KITP).
This work was performed in part at the Aspen Center for Physics, which is supported by the National Science Foundation grant PHY-2210452.

\bibliography{main}

\appendix

\section{A short introduction to ZX-calculus}
\label{sec: zxrules}
\subsection{ZX rewrite rules}
In this appendix, we give a brief introduction to ZX-calculus and summarize some useful ZX rewrite rules  that are relevant to our work. A more comprehensive introduction to ZX-calculus can be found in Refs. \cite{backens_zx-calculus_2014,wetering_zx-calculus_2020}.

ZX-calculus is a diagrammatic way of representing quantum circuits. Every ZX-diagram is generated by two types of notes, called $X$ nodes and $Z$ nodes. Nodes with three or more legs are also referred to as spiders. The $X$/$Z$ spider represents a phased GHZ state in $X$/$Z$ basis over all the legs around the node\footnote{In the literature, it is more common to put the phase $\alpha$ rather than $e^{i\alpha}$ inside the node. The convention we use here makes measurement outcomes clearer, such as the ones in Eqs. \eqref{eq: xpair2} and \eqref{eq: xstabmeas}.}:
\begin{align}\nonumber
    &\input{znode.tikz}= \ket{+\cdots+} \bra{+\cdots+}+e^{i\alpha}\ket{-\cdots-}\bra{-\cdots-},\\
    &\input{xnode.tikz}= \ket{0\cdots0} \bra{0\cdots0}+e^{i\alpha}\ket{1\cdots1}\bra{1\cdots1}.
\end{align}
Physically, in our work, every node only contains one incoming leg for one outgoing leg, and each pair of incoming and outgoing legs corresponds to one qubit in the projector in the explicit expression above. More generally, a ZX node can have different incoming and outgoing legs such that the map is an isometry.

When the node is phaseless, that is, $\alpha\in2\pi\Z$, we omit the phase inside the node. If a node only has one leg, it represents initializing or projecting a single qubit on the following states:
\begin{align}
    &\input{xinit.tikz}=\ket{+}+e^{i\alpha}\ket{-},\ \input{zinit.tikz}=\ket{0}+e^{i\alpha}\ket{1}.
\end{align}

A node with two legs represents a unitary single-qubit gate. In particular, an $X$ gate is simply a two-leg $X$ node with phase $\alpha=\pi$. Any two-leg node with phase $2\pi$ is identity and can be removed. 
The Hadamard gate is represented in ZX-calculus by the H-node, which can be decomposed using three $\pi/2$ nodes of $X$ and $Z$:
\begin{equation}
    \input{Hadamard.tikz}=\frac{1}{\sqrt{2}}\left(\begin{array}{cc}
        1 & 1 \\
        1 & -1
    \end{array}\right)=e^{-\frac{i\pi }{4}}\input{Hadamard2.tikz}.
\end{equation}
Note that the $\frac{\pi}{2}$ phase $Z$ node represents an $S$ gate.

In ZX-calculus, the two-qubit entangling gates CNOT and CZ can be represented by phase-free diagrams
\begin{equation}
    \input{cnotcz.tikz},
\end{equation}
respectively. In the ZX-diagram of the CNOT gate, the physical qubit connecting to the $Z$ ($X$) node is the control (target) qubit.

Since we only consider Clifford dynamical codes in our work, it is worth mentioning that every Clifford circuit (which potentially involves Pauli basis measurements) can be represented by a ZX-diagram where the phase of every node is an integer multiple of $\frac{\pi}{2}$.

We then list a few important rewrite rules of the ZX-calculus. First, a pair of connected $X$ spiders can be fused into a single $X$ node:
\begin{equation}
    \input{xfusion.tikz},
\end{equation}
and similarly for the $Z$ spider. The two spiders are switched via conjugation of the Hadamard node:
\begin{equation}
    \input{xhard.tikz}.
\end{equation}
Two important rewrite rules of phase-free spiders are the Hopf rule
\begin{equation}
    \input{hopf.tikz}
\end{equation}
and the bialgebra rule
\begin{equation}
    \input{bialgebra.tikz}.
    \label{eq: bialgebra}
\end{equation}
In the RHS of the bialgebra rule, every leg on the left (right) is connected to an $X$ (a $Z$) node, and the $X$ and $Z$ nodes form a bipartite graph where every $X$ node is connected to all the $Z$ nodes and vice versa. 

In understanding fault-tolerance using ZX-calculus, every single-qubit Pauli error is represented by a single $\pi$-node of the corresponding species \cite{bombin_unifying_2024}. In this case, one useful ZX-rule when dealing with Pauli errors is the $\pi$-node commutation rule (here we ignore the global phases of the wave functions):
\begin{equation}\label{eq: pinode}
    \input{pinode.tikz}.
\end{equation}
Alternatively, a single $\pi$ node of $Z$ on one of the legs of the $X$ node will spread to all the other legs of the $X$ node while flipping the phase of the $X$ node. 

\subsection{Pauli web in Clifford ZX-calculus}
\label{sec: pauliweb}
One particular advantage of using ZX-calculus in QEC is its convenience to keep track of the flow of Pauli operators through the circuit. In general, a Pauli operator supported on any leg of the ZX-diagram can travel through the nodes, leaving a web of Pauli operators throughout the ZX-diagram. For example, since the phase-less $X$ node is stabilized by any pair of $X$ operators on two of its legs or the product of  $Z$ operators on all of its legs, the Pauli web around an $X$ node has the following possible forms
\begin{equation}
    \input{xweb.tikz}.
\end{equation}
Unless otherwise stated, the red/green shaded edges in the ZX-diagrams represent an $X$/$Z$ operator supported on the corresponding edge.
Similarly, the Pauli web near a phase-less $Z$ node can have the following forms:
\begin{equation}
    \input{zweb.tikz}.
\end{equation}
The Hadamard gate changes the coloring of the edge:
\begin{equation}
    \input{hweb.tikz}.
\end{equation}
In more general Clifford ZX-diagrams where the nodes can carry phases $\frac{\pi}{2} \Z$, it is still possible to construct Pauli webs. For convenience, we can even introduce the $Y$ operator on the edges. However, it is now important to keep track of the overall phase from the Pauli Web, which will be important for measurement outcomes in circuit implementation of ZX-diagrams. See Ref. \cite{fuente_xyz_2024}.

\subsection{Adaptation of the ZX-rules in actual circuits with measurements}
Now that we have introduced the general rules of ZX-calculus, we now present some useful ZX-rules when dealing with ZX-diagrams in the circuit layout.

The most frequently used ZX-rule in our paper is the conversion from a four-leg spider to pairwise measurement, namely
\begin{equation}
    \input{xpair.tikz}.
\end{equation}
Here, the final circuit layout can be understood as three steps: (1) initialize an ancilla in $\ket{+}$; (2) apply two CNOT gates between the two physical qubits with the ancilla; (3) project the ancilla onto $\ket{+}$. This means that, in an actual circuit, the four-leg $X$ spider corresponds to a pairwise measurement in the $X$ basis and post-selecting the measurement outcome $XX=1$. If one does not post-select, but rather record the measurement outcome of the ancilla, we denote this via an extra outgoing leg:
\begin{equation}\label{eq: xpair2}
    \input{xpair2.tikz}.
\end{equation}
The value of $XX$ after measurement is equal to the measurement outcome of the ancilla in $X$ basis. Note that we switch the equality to ``$\sim$" since the two ZX-diagrams are not completely the same, as the one on the RHS has a potential $\pi$ node.

The conventional way of measuring higher-weight stabilizer, dubbed the ``SASEC" in our work, is represented by the ZX-diagram of the kind shown below: here we show an example of a measurement circuit of the product of four Pauli $Z$ operators of four physical qubits:
\begin{equation}\label{eq: zpair4}
    \input{weight4.tikz}.
\end{equation}
Similar to Eq. \eqref{eq: xpair2}, the ZX-diagram on the LHS represents measuring the product of all the $Z$ operators and post-selecting the measurement outcome $+1$, while the RHS of the equivalence relation represents a typical circuit implementation of the SASEC where one entangles every physical qubit with an ancilla through CNOT gates and records the measurement outcome of the ancilla in $Z$ basis (without post-selection), which yields the eigenvalue of $ZZZZ$ acting on the four physical qubits.

In most of our constructions of dynamical codes, we always assume a post-selected measurement outcome $+1$ for every multi-qubit measurement for convenience in the ZX-diagram. In this way, every incoming stabilizer will literally be mapped by the ZX-diagram of the circuit to $+1$.
In actual circuit implementations, one can always keep track of the measurement outcomes by spelling out the actual measurement circuit, like in Eqs. \eqref{eq: xpair2} and \eqref{eq: zpair4}. In this way, the measurement outcome of every incoming stabilizer can be inferred from a series of measurement outcomes along the Pauli web of the incoming stabilizer. For example, using Eq. \eqref{eq: xpair2}, around a ring of four connected gadgets in the 012 Floquet code which becomes an octagon with green and blue bonds after spatial concatenation, the incoming $X$ stabilizer around  in Fig. \ref{fig: tcsync} is measured by the following ZX-diagram:
\begin{equation}\label{eq: xstabmeas}
    \input{xstab.tikz},
\end{equation}
i.e. the value of the stabilizer $XXXX$ is equal to the product of measurement outcomes of the four pairwise measurements of $XX$ in the ZX-diagram in Eq. \eqref{eq: xstabmeas}.

In applications to Floquet codes, in order to represent the pairwise $Y$ measurement, we introduce the $Y$ node as an $X$ node whose every leg is conjugated by $\pi/2$ $Z$ nodes (i.e. $S$ gates):
\begin{equation}\label{eq: ynode}
    \input{ynode.tikz}.
\end{equation}
A pair of connected three-leg $Y$ nodes therefore represents a pairwise measurement of $YY$ of two physical qubits:
\begin{equation}
    \input{ypair.tikz}.
\end{equation}

\subsection{Circuit-level errors and decoding in ZX-calculus}
\label{sec: decoding}
For the sake of decoding the dynamical code with 
circuit-level errors, we discuss the way to incorporate different types of these errors in ZX-calculus. 

In the ZX-diagram of a dynamical code in the circuit layout, edges of the ZX-diagram are either pointing in spatial or temporal directions. We will call these directions ``space-like" and ``time-like", respectively. Single-qubit Pauli errors during the idling stage of the qubits are represented by $\pi$-nodes of the corresponding species on time-like edges. For example, a single qubit bit-flip and phase flip errors are represented by
\begin{equation}
    \input{xerror.tikz}\text{ and }\input{zerror.tikz},
\end{equation}
respectively. A $Y$ error can be represented by the composition of the above two nodes. 

Meanwhile, readout errors of multi-qubit measurements can be represented by $\pi$ nodes on the space-like edges. For example, a faulty pairwise $XX$ measurement can be represented by a single $\pi$ node of $Z$ between the two $X$ nodes:
\begin{equation}
    \input{pairerror.tikz},
\end{equation}
where the $\mp$ indicates the measured eigenvalue is $\mp1$ when the measurement was supposed to be $\pm1$.
More general multi-qubit measurement errors can be similarly defined by flipping the $\pm 1$ node, which records the measurement result, from $\pm 1$ to $\mp 1$ in the ZX-diagram.

To decode the dynamical code, we gather all the measurement outcomes of the lower-weight Pauli measurements, whose spacetime locations are directly indicated by the ZX-diagram in the way similar to Eq. \eqref{eq: xstabmeas}. These measurement outcomes are then converted to syndromes for the spacetime stabilizers via multiplying all the measurement outcomes that are part of the spacetime stabilizers. The ZX-diagram can also be turned into a decoding (hyper-)graph, where every $-1$ syndrome is represented by a vertex of the graph, and every circuit-level Pauli error is represented by a (hyper-)edge that connects all the syndrome vertices that it flips. The syndromes and the decoding graph can be fed into an appropriate decoder.
\begin{figure}
    \centering
    \includegraphics[width=0.7\linewidth]{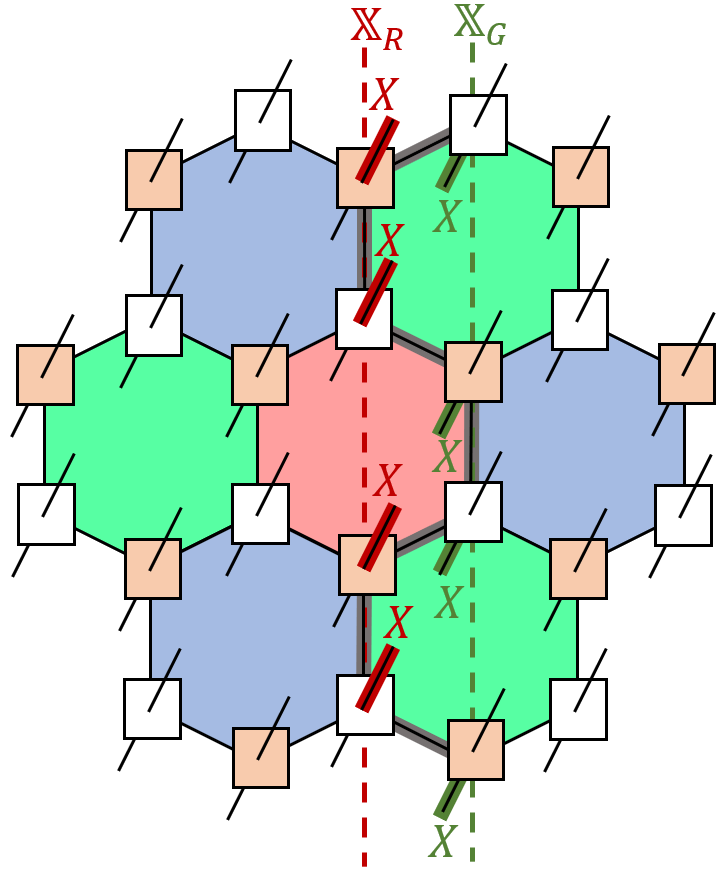}
    \caption{The gadgets for Floquet color code. The shaded lines is a demonstration of the $\Z_3$ logical automorphism given by the bond operators in Eq. \eqref{eq: ccsol}, where the incoming green logical operator $\mathbb{X}_G$ is mapped to the outgoing red logical operator $\mathbb{X}_R$.}
    \label{fig: cclogical}
\end{figure}

\begin{figure}
    \centering
    \includegraphics[width=0.8\linewidth]{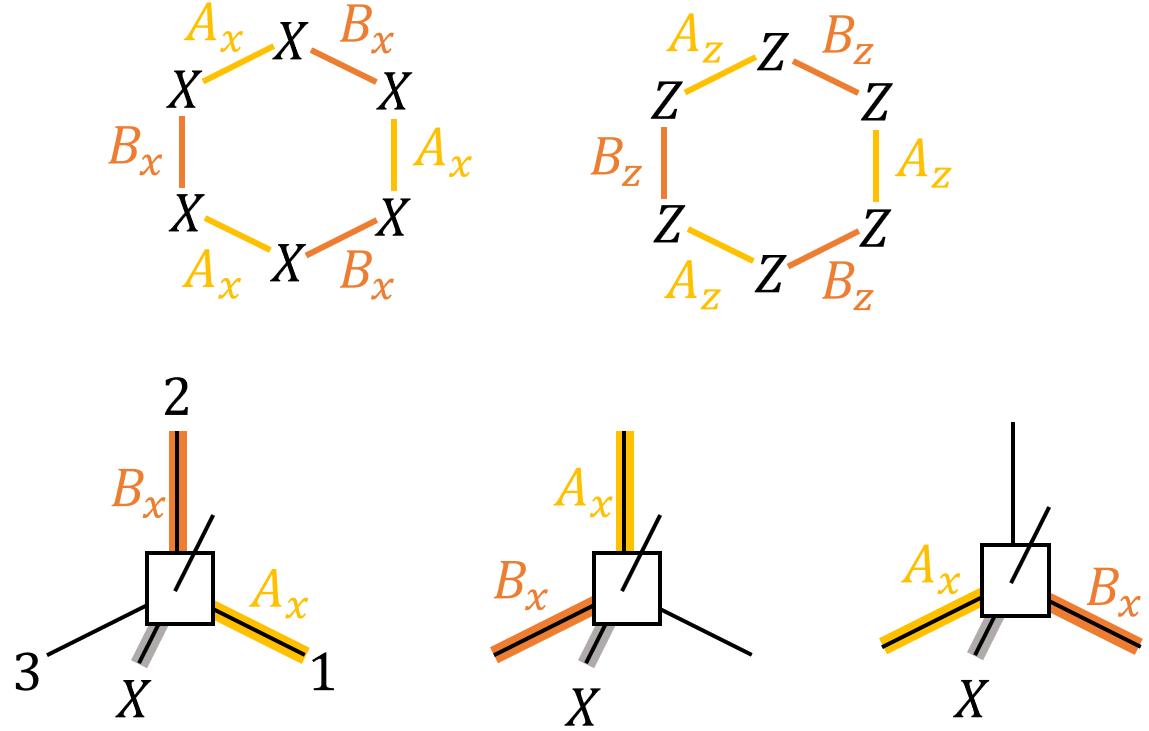}
    \caption{Parametrization of the bond operators for the 2D color code. The parametrization of the incoming $Z$ and outgoing $X/Z$ operators are similar to the one for the incoming $X$ operator. For simplicity only the encoding of the incoming $X$ operator is shown, while other incoming and outgoing operators are encoded similarly.}
    \label{fig: ccparam}
\end{figure}

\section{Floquet Color Code and Floquet Steane Code }
\label{sec: csscc}

\begin{figure*}
    \centering
    \includegraphics[width=0.6\linewidth]{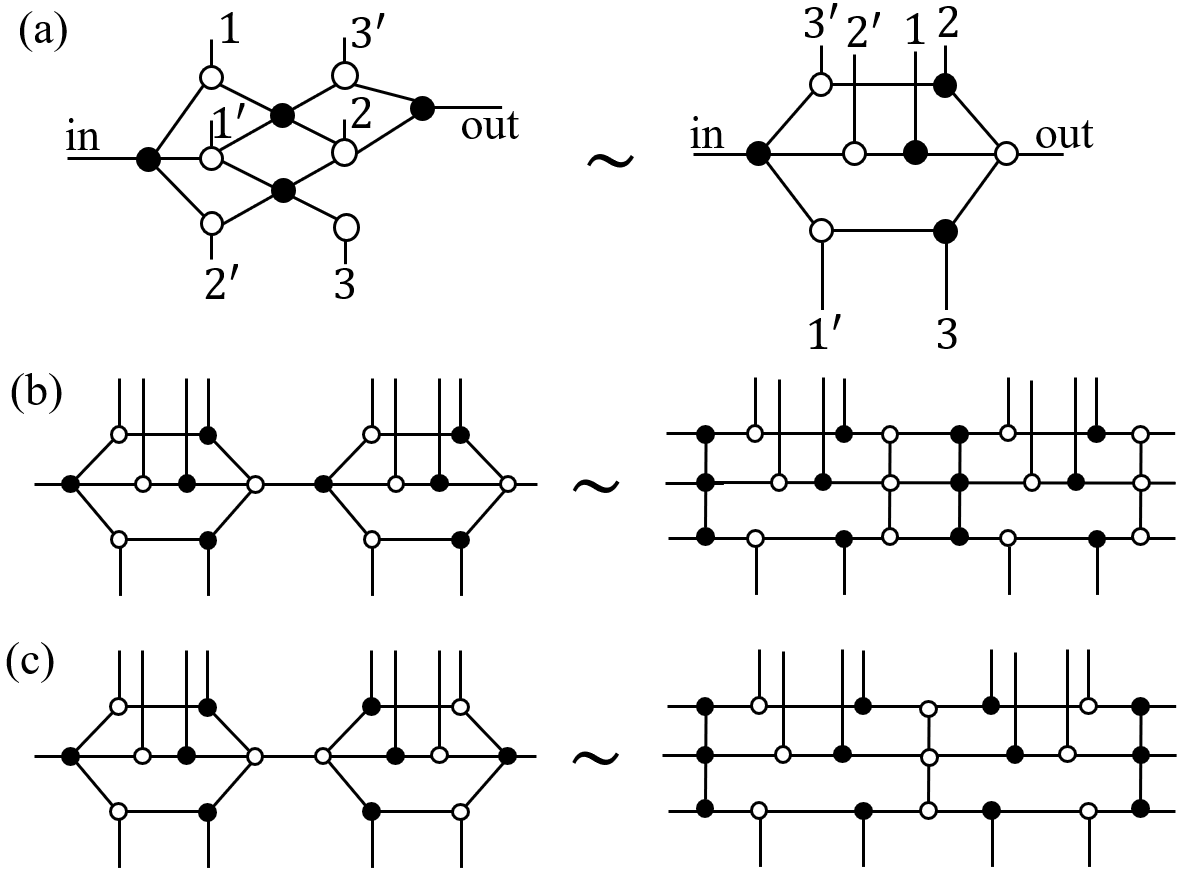}
    \caption{(a) The ZX-diagram of the gadget that produces the encoding map in Eq. \eqref{eq: ccmat} with the solution in Eq. \eqref{eq: ccsol}. (b) Using the Hopf rule of ZX-calculus, repeated gadgets give a CSS Floquet color code where each data qubit is replaced by 3 physical qubits. (c) A more efficient CSS Floquet color code schedule where each stabilizer is measured twice after 7 rounds of measurements.}
    \label{fig: cczx}
\end{figure*}

In this appendix, we first construct a 2D CSS Floquet color code with periodic boundary condition using gadget layout, whose incoming and outgoing stabilizer codes $\cS$ and $\overline{\cS}$ are both 2D color codes on hexagonal lattice, where the incoming/outgoing stabilizers are weight-six $X$ and $Z$ operators around each hexagon. We then construct dynamical code for the smallest open-boundary Color Code, known as the $[[7,1,3]]$ Steane code, and realize a variety of single logical qubit operation via manipulation of the Floquet schedule.

\subsection{Floquet Color Code}

Similar to the constructions in the main text, each data qubit in the 2D color code is replaced by a gadget, see Fig. \ref{fig: cclogical}. The gadgets are connected along the direction of the hexagonal lattice. We consider a homogeneous internal leg layout so that each bond hosts $l$ internal legs. Similar to Fig. \ref{fig: haahproof}, it is easy to see that $l$ should be at least 2 so that the overlapping stabilizers on a gadget can be consistently encoded. To construct the encoding matrix, we parametrize the bond operators in a 3-fold rotational symmetric way, as shown in Fig. \ref{fig: ccparam}. The bond operators $A_{x,z}$, $B_{x,z}$, $\overline{A_{x,z}}$ and $\overline{B_{x,z}}$ are Pauli strings of weight-$l$. Again, we assume translation invariance, so that each type of stabilizer is parametrized by the same set of bond operators.
For simplicity, we consider CSS encoding, so that each bond operator only consists of Pauli $X$ or $Z$, depending on the subscript. Furthermore, since the hexagonal lattice is bipartite, we only need to consider gadgets of one of the two sublattices, for which the $X$ and $Z$ encoding matrices read (again, ``0"s in between the two vertical lines denote a zero vector of length $l$).
\begin{widetext}
\begin{align}\label{eq: ccmat}
    H_X=\left(\begin{array}{c|ccc|c}
        1& [A_x] & [B_x] & 0 &0 \\
        1 &0&[A_x]& [B_x]&0\\
        1&B_x&0&[A_x]&0\\
        0&[\overline{A_x}] &[\overline{B_x}]&0&1\\
        0& 0&[\overline{A_x}]&[\overline{B_x}]&1\\
        0&[\overline{B_x}]&0&[\overline{A_x}]&1
    \end{array}\right),\ H_Z=\left(\begin{array}{c|ccc|c}
        1& [A_z] & [B_z] & 0&0  \\
        1 &0&[A_z]& [B_z]&0\\
        1&[B_z]&0&[A_z]&0\\
        0&[\overline{A_z}] &[\overline{B_z}]&0&1\\
        0& 0&[\overline{A_z}]&[\overline{B_z}]&1\\
        0&[\overline{B_z}]&0&[\overline{A_z}]&1
    \end{array}\right).
\end{align}
\end{widetext}
Here the three columns on the right of the vertical lines are bond operators in the directions from 1 to 3 in Fig. \ref{fig: ccparam}. We find that the consistency condition $H_XH_Z^T=0\mod 2$ can indeed be solved with $n_{L,i}=2$. One representative solution in this case is
\begin{align}\label{eq: ccsol}\nonumber
    [A_x]=11,\ [B_x]=10,\ [A_z]=10,\ [B_z]=01,\\
    [\overline{A_x}]=10,\ [\overline{B_x}]=01,\ [\overline{A_z}]=01,\ [\overline{B_z}]=11.
\end{align}
All the other solutions are related by permutation of the three directions or exchanging $X$ and $Z$, which all lead to the same dynamical code once the gadget are connected. The total rank of the encoding matrices with these solutions is
\begin{equation}
    \text{rank}(H_X)+\text{rank}(H_Z)=8=2m+n_L=2+2\times 3.
\end{equation}
Therefore, the gadget is completely fixed by the encoding map. 

The ZX-diagram that produces the encoding maps in Eq. \eqref{eq: ccmat} with the solutions in Eq. \eqref{eq: ccsol} is shown in Fig. \ref{fig: cczx}, where the unprimed and primed legs denote the first and the second element in the Pauli string of the bond operators. Since the primed legs always come before the unprimed ones in the ZX-diagram of the gadget, the schedule is already synchronized between gadget of the two sublattices. This leads to a 4-round pairwise measurement schedule: gZZ, bXX, rZZ, gXX, where different colored bonds are marked in Fig. \ref{fig: ccschedule}.
\begin{figure}[H]
    \centering
    \includegraphics[width=0.5\linewidth]{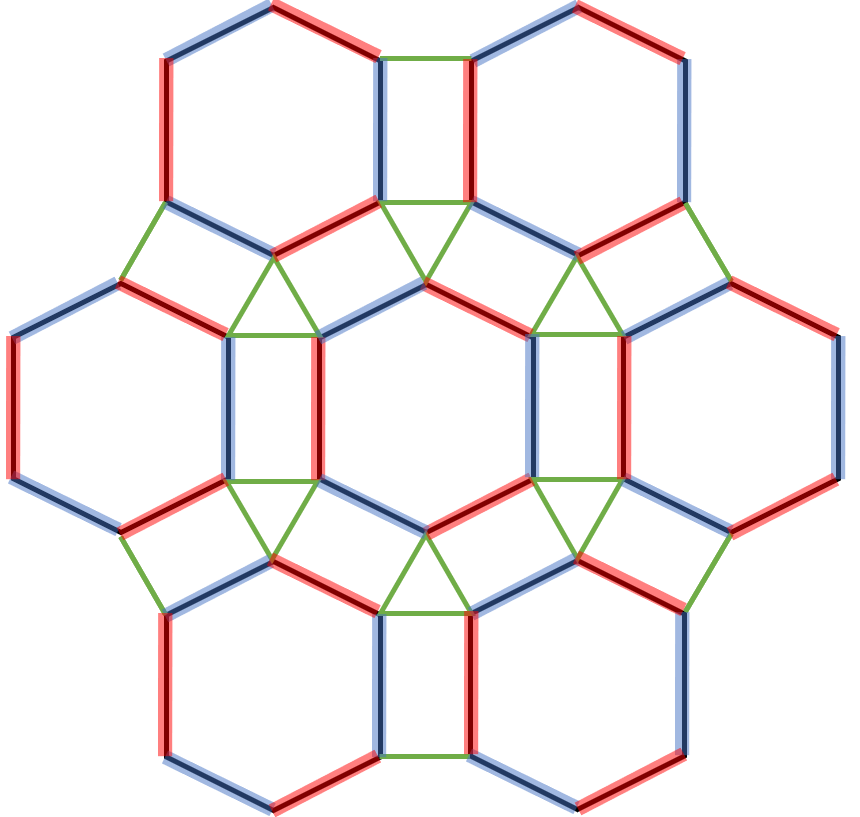}
    \caption{The measurement schedules of the CSS Floquet color code. Spatially each data qubit in the color code is replaced by 3 physical qubits around the gree triangle. The 4-round measurement schedule is gZZ, bXX, rZZ and gXX. An alternative 6-round schedule which measures each stabilizer twice is gZZ, bXX, rZZ, gXX, bZZ and rXX. }
    \label{fig: ccschedule}
\end{figure}

\begin{figure}[H]
    \centering
    \includegraphics[width=0.4\linewidth]{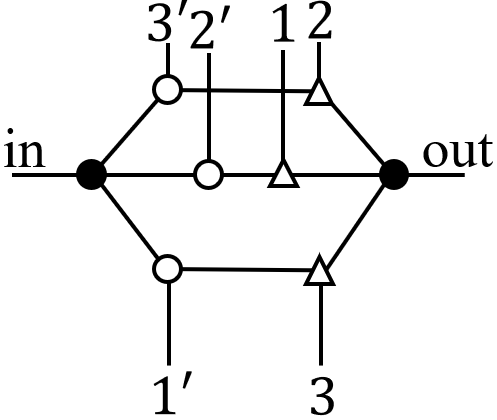}
    \caption{The ZX-diagram of the gadget layout of the XYZ ruby color code in Ref. \cite{fuente_xyz_2024}.}
    \label{fig: ccruby}
\end{figure}

We now demonstrate that the logical subspace is preserved by the gadgets of the CSS Floquet color code. To see this, we notice that, in the solution in Eq. \eqref{eq: ccsol}, $B_x=\overline{A_x}$ and $B_z=\overline{A_z}$. Therefore, given the encoding map of the gadgets, the incoming logical $X$ operator of the green sublattice, $\mathbb{X}_G$, will be mapped to the outgoing logical $X$ operator of the red sublattice $\mathbb{X}_R$, and similarly for the logical $Z$ operators. See Fig. \ref{fig: cclogical}. To summarize, the $\Z_3$ logical automorphism of the CSS Floquet color code can be written as
\begin{equation}
    \M(\mathbb{X}_G)=\mathbb{X}_R,\ \M(\mathbb{X}_R)=\mathbb{X}_B,\ 
    \M(\mathbb{X}_B)=\mathbb{X}_G,
\end{equation}
and similarly for the logical $Z$ operators. Therefore, the logical subspace is preserved by the gadgets.

We note that there is a more efficient measurement schedule where each stabilizer is measured twice in the span of 7 rounds. This can be constructed by following this gadget with another gadget whose incoming leg is connected to an $X$-node, see Fig. \ref{fig: cczx}(c)\footnote{This gadget is represented by the solution $[A_X]=10,\ [B_X]=01,\ [\overline{A_X}]=01,\ [\overline{B_X}]=11,\ [A_Z]=11,\ [B_Z]=10,\ [\overline{A_Z}]=10$ and $[\overline{B_Z}]=01$ to the consistency equation of the encoding matrices in Eq. \eqref{eq: ccmat}.}. The corresponding 6-round measurement schedule is gZZ, bXX, rZZ, gXX, bZZ, rXX. 

The CSS Floquet color code from our construction is notably inequivalent from earlier proposals of Floquet color codes that require pairwise measurements in $X$, $Y$ and $Z$ basis\cite{dua_engineering_2024,fuente_xyz_2024} and CSS proposals that directly decomposes the SASEC\cite{townsend-teague_floquetifying_2023}.
Take the ``XYZ ruby code" in Ref. \cite{fuente_xyz_2024} as an example, which is defined on the same lattice and bond coloring as in Fig. \ref{fig: ccschedule} with a 3-round schedule: gZZ, bXX and rYY. From the ZX-diagram of the local gadget that consists of three physical qubits (see Fig. \ref{fig: ccruby}), we see that the bond operator in this case is
\begin{align}\label{eq: ccrubysol}\nonumber
    A_x=XX,\ B_x=XI,\ A_z=YI,\ B_z=IY,\\
    \overline{A_x}=XY,\ \overline{B_x}=YZ,\ \overline{A_z}=XI,\ \overline{B_z}=IX.
\end{align}
No bond-local unitary can bring this set of bond operators to the CSS ones parametrized in Eq. \eqref{eq: ccsol}. Therefore, from Thm. \ref{thm1}, the CSS Floquet color code is inequivalent to the XYZ ruby color code. In fact, this can also be seen from the logical automorphism. While the CSS Floquet color code performs a $\Z_3$ logical automorphism by switching the color of the $X$ or $Z$ logical operators, the XYZ ruby color code performs a $\Z_3$ logical automorphism by permuting the $X/Y/Z$ logical operators to $Z/X/Y$ logical operators while also permuting the three colos of the logical operators. This can be verified on the level of gadgets from the bond operators in Eq. \eqref{eq: ccrubysol}, where we have $B_x=\overline{A_z}$ and $B_z=\overline{A_x}\times\overline{A_z}$. 

\begin{figure}
    \centering
    \includegraphics[width=\linewidth]{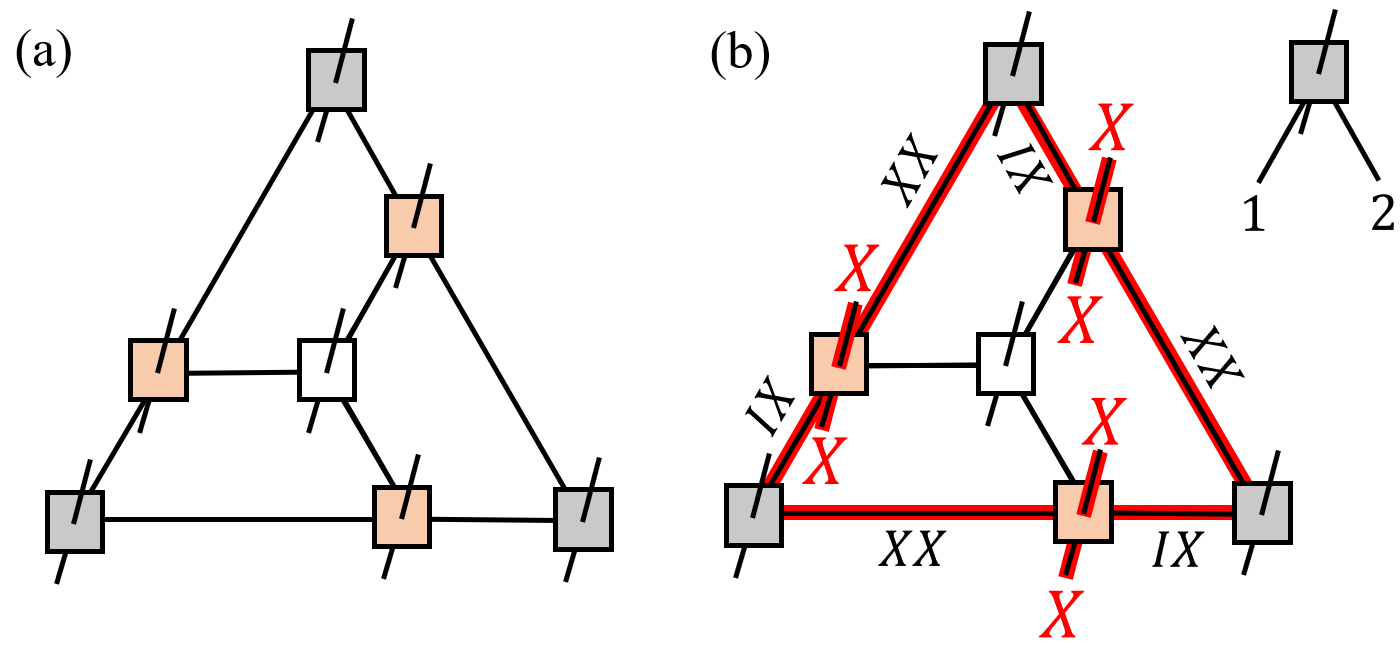}
    \caption{(a) Gadget layout of the $[[7,1,3]]$ Steane code, where each data qubit is replaced by a gadget. The three gadgets on the edge are marked in orange shades, and the three corner gadgets are marked in gray shades. All the gadget encoding maps (apart from the two additional stabilizers on the corner gadgets) can be directly inferred from the gadget in the center, which is the same as the gadget in Fig. \ref{fig: ccparam} for the Floquet color code. (b) The Pauli web of a trivial logical automorphism of incoming $X$ logical operator that is supported on the three edge gadgets. The way such a Pauli web goes through the corner gadgets implies that the corner gadget should be stabilized by $I|XX\ IX|I$, where the two bond directions are marked on the right.  }
    \label{fig: steane}
\end{figure}
\begin{figure*}
    \centering
    \includegraphics[width=0.65\linewidth]{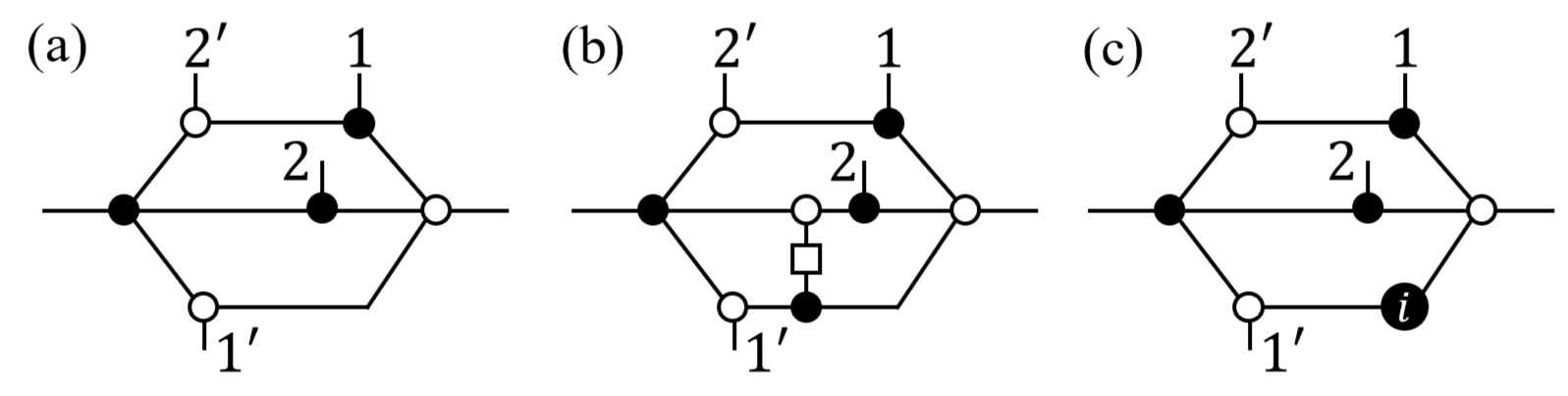}
    \caption{The ZX-diagrams of the corner gadget for different logical automorphisms. (a) The corner gadget that realizes trivial logical automorphism from the tableau in Eq. \eqref{eq: trivialtab}. (b) The corner gadget for logical Hadamard gate from Eq. \eqref{eq: hadamardtab}. Note that an $X$ node and a $Z$ node connected by a Hadamard edge represents a projective measurement of $XZ$ and post-selecting the outcome $XZ=1$. (c) The corner gadget for logical $S$ gate from Eq. \eqref{eq: sgatetab}. }
    \label{fig: steanezx}
\end{figure*}
\begin{figure}
    \centering
    \includegraphics[width=0.65\linewidth]{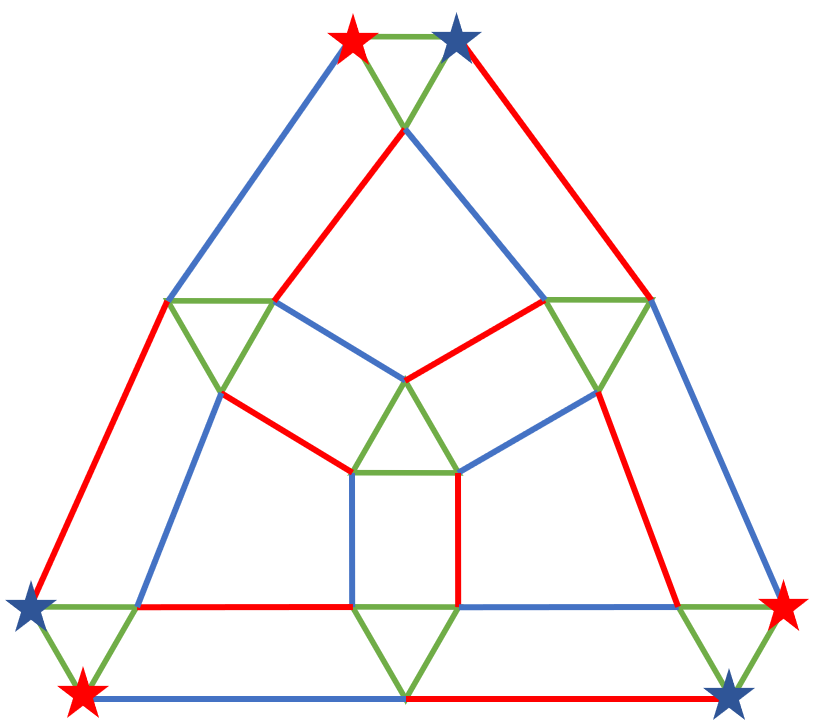}
    \caption{The layout of the physical qubits in the Floquet Steane code. The measurement schedule which produces trivial logical automorphism is the same as the CSS Floquet color code: gZZ, bXX, rZZ, gXX (the starred qubits are simply idling if not involved in a pairwise measurement). For logical Hadamard gate, between the rounds bXX and rZZ, we measure $Z_rX_b$ between the three pairs of red and blue starred qubits. For logical $S$ gate, at round $rZZ$, we apply $S$ gates to the three red starred qubits. }
    \label{fig: steanelat}
\end{figure}

\subsection{Floquet Steane Code and dynamical Clifford logical gates}
\label{subsec:steane}
The $[[7,1,3]]$ Steane code \cite{Steane1996} can be regarded as the smallest open-boundary color code in 2D, where every trapezoidal plaquette hosts a weight-4 $X$ stabilizer and a weight-4 $Z$ stabilizer, see Fig. \ref{fig: steane}. We construct a dynamical Steane code by replacing each data qubit in the original Steane code with a gadget, and the gadgets are connected according to the connectivity of the underlying lattice as in Fig. \ref{fig: steane}(a). Compared to the case of color code, we see that, apart from the three gadgets at the corners of the lattice, the rest of the gadgets should produce exactly the same encoding map as the gadgets in the color code. Therefore, we can directly use the solutions in Eq. \eqref{eq: ccsol} or Eq. \eqref{eq: ccrubysol} for the gadget. Even for the three gadgets on the edge of the triangle, since each of them hosts the Pauli webs of 4 incoming and 4 outgoing stabilizers, the gadget should already be maximally stabilized, given it has 8 total legs: an incoming, an outgoing and 6 internal legs.

The only remaining step towards the dynamical Steane code is to determine the gadget on the three corners, since they are not yet maximally stabilized: each of them only hosts 2 incoming and 2 outgoing stabilizers, but it has 6 legs in total. Therefore, we have to fix two extra stabilizers of the gadget, which must be chosen in a way that is consistent with the BKRC. In fact, different choices of the extra stabilizers correspond to different logical automorphisms performed by the dynamical Steane code. To see this, we choose the incoming/outgoing $X$($Z$) logical operators to be the product of three $X$($Z$) Pauli operators on the three data qubits on the edge of the triangle (see Fig. \ref{fig: steane}(b)). Now consider the Pauli web of the logical operator. Since the encoding map of the gadget in Eq. \ref{eq: ccmat} ensures that both incoming and outgoing Pauli operators can be mapped to the two bonds on the edge of the triangle, the Pauli web of any logical operator can be supported on the six edge bonds around the triangle. Such a Pauli web goes through the corner gadget without any incoming/outgoing Pauli operator, i.e. the corner gadget must stabilize such a Pauli web on the internal legs. This is exactly an extra stabilizer that we need to fix for the corner gadget.

We use the gadget in the CSS Floquet color code as an example, whose encoding map is given by Eqs. \eqref{eq: ccmat} and \eqref{eq: ccsol} and the ZX-diagram in Fig. \ref{fig: cczx}. We first demand a trivial logical automorphism. From the Pauli web of the incoming and outgoing logical $X$ ($Z$) operator around the triangle, the corner gadget must additionally stabilize $I|XXIX|I$ and $I|ZIZZ|I$. In total, the corner gadget is now maximally stabilized by the following stabilizer tableau:
\begin{equation}\label{eq: trivialtab}
    \left[\begin{array}{c|cc|c}
        X & XX & XI & I \\
        Z & ZI & IZ &I \\
        I & XI & IX & X \\
        I & IZ & ZZ & Z\\
        I & XX & IX & I\\
        I & ZI & ZZ & I\\
        \end{array}\right],
\end{equation}
where the two columns between the vertical lines are the two bond directions of the corner gadget. A quick check shows that the dynamical code formed by the connected gadgets satisfies the BKRC and is therefore fault-tolerant. The ZX-diagram of the gadget is simply the one in Fig. \ref{fig: cczx} with the two nodes that connect to the directions $3$ and $3'$ being removed. Therefore, the physical qubits of the dynamical Steane code can be laid down on the open-boundary ruby lattice, and the measurement schedule for the dynamical Steane code with trivial logical automorphism is the same as the CSS Floquet color code, see Fig. \ref{fig: steanelat}.

The dynamical Steane code can also realize a variety of nontrivial logical Clifford operations. For example, to realize a logical Hadamard gate, we write down the Pauli web of an incoming $X$ (resp. $Z$) and an outgoing $Z$ (resp. $X$) logical operator, which implies that the corner gadget has to stabilize $I|YZXY|I$ (resp. $I|IYYI|I$). The corner gadget is now maximally stabilized by the following stabilizer tableau:
\begin{equation}\label{eq: hadamardtab}
    \left[\begin{array}{c|cc|c}
        X & XX & XI & I \\
        Z & ZI & IZ &I \\
        I & XI & IX & X \\
        I & IZ & ZZ & Z\\
        I & YZ & XY & I\\
        I & IY & YI & I\\
        \end{array}\right].
\end{equation}
Again, the entire dynamical code satisfies the BKRC. The ZX-diagram of the gadget is shown in Fig. \ref{fig: steanezx}. Comparing this to the ZX-diagram of the rest of the gadgets given in Fig. \ref{fig: cczx}, we see that the measurement schedule is modified as follows: between the rounds bXX and rZZ, we measure $Z_rX_b$, i.e. the product of $Z$ of the red starred qubit and $X$ of the blue starred qubit at the three tips of the triangle; see Fig. \ref{fig: steanelat}.

Similarly, to realize a logical $S$ gate, we write down the Pauli web of an incoming $X$ (resp. $Z$) and outgoing $Y$ (resp. $Z$) logical operator. In this case, the corner gadget is stabilized by the following tableau:
\begin{equation}\label{eq: sgatetab}
    \left[\begin{array}{c|cc|c}
        X & XX & XI & I \\
        Z & ZI & IZ &I \\
        I & XI & IX & X \\
        I & IZ & ZZ & Z\\
        I & YY & IY & I\\
        I & ZI & ZZ & I\\
        \end{array}\right].
\end{equation}
The ZX-diagram of the gadget is shown in Fig. \ref{fig: steanezx}, which means that the schedule is modified as follows: in the round rZZ, we also perform an $S$ gate on the three physical qubits marked by the red stars in Fig. \ref{fig: steanelat}. Note that we only need 3 physical $S$ gates in total, compared to the transversal logical $S$ gate of the Steane code, which is implemented via 7 $S$ gates on all 7 physical qubits.

We note that, as mentioned in Sec.~\ref{sec: spacetimed}, a naive scale-up of the Floquet Steane code into a Floquet color code on an open-boundary triangular patch will only have constant spacetime distance, no matter what the distance $d$ of the static color code is. This is because the Pauli web of the incoming logical operator will fully enter the three corner gadgets no matter the size of the lattice, which implies the logical operator will evolve into local operators at the three corners in the middle of the measurement schedule.

\section{Floquet Surface Code with Constant Spacetime Distance}
\label{sec: floquetsurface}

In this appendix, we discuss a case of constant space-time distance for a Floquet toric code with planar boundaries, i.e. a Floquet surface code. We show that a naive Floquet code construction in which one uses the bulk gadget on the boundary leads to a catastrophic result in which the spacetime code distance $d_{st}$ remains constant despite the increasing stabilizer code distance $d$. We note that such a phenomenon of constant spacetime distance is already observed in Refs. \cite{haah_boundaries_2022,vuillot2021planar,ellison_pauli_2023}.

For simplicity, we consider the smallest surface code with $d=3$, since longer code distances can be constructed similarly. We choose a gadget layout that resembles the usual surface code lattice. That is, the arches on the boundary now represent extra legs that connect pairs of boundary gadgets. Therefore, the four L gadgets at the centers of the four open boundaries still have four internal legs each. Hence, it seems handy to use the bulk L gadget for these four boundary gadgets. The four R gadgets at the four corners each have three internal legs, or five legs in total. The four Pauli webs of the incoming/outgoing $X$ and $Z$ stabilizers leave one extra stabilizer to be fixed for each corner gadget. We fix the remaining stabilizer on the corner gadget so that the entire dynamical code performs a logical Hadamard gate (see Fig.~\ref{fig:floquetsurface}), following a procedure similar to the Floquet Steane code in Appendix~\ref{subsec:steane}. The resulting Floquet code can be realized on an open-boundary square-octagon lattice (see Fig.~\ref{fig:surfacelat}), where the measurement schedule is modified so that the four physical qubits on the four corners of the lattice undergo single-qubit measurement in the $X$ or $Y$ basis at time steps $1$ and $2$. This closely resembles the construction of planar Floquet code in Ref.~\cite{vuillot2021planar}. 

From the Pauli web of the incoming $X$ logical operator shown in Fig.~\ref{fig:floquetsurface} or using the definition in Eq.~\eqref{eq: dst} in Sec.~\ref{sec: spacetimed}, we find that the spacetime distance of the naively constructed Floquet surface code is always $d_{st}=2$, regardless of the distance $d$ of the surface code. In fact, the incoming $X$ logical operator can be completely shoveled into the corner gadget at the bottom right, which implies that the logical operator evolves into a local operator somewhere in the middle of the measurement schedule.

To avoid this, one can adopt strategies such as rewinding \cite{haah_boundaries_2022, dua_engineering_2024}, which, in the language of gadget decomposition, corresponds to increasing the number of internal legs for each gadget. In this way, we can find a solution in the gadget decompositions which leads to a Floquet surface code whose spacetime distance scales with the static surface code distance.

\begin{figure}
    \centering
    \includegraphics[width=0.5\linewidth]{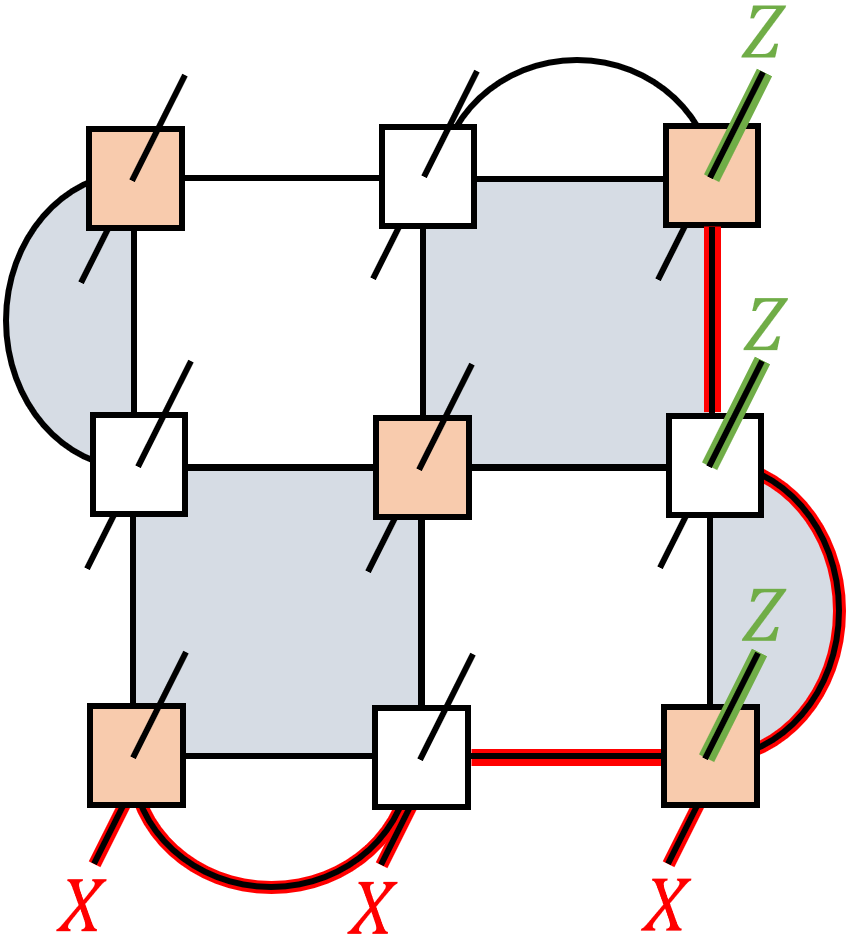}
    \caption{Gadget layout in constructing a dynamical code for a $d=3$ surface code. The corner R gadgets can be fixed once we demand the dynamical code to perform logical Hadamard gate, whose Pauli web is shown explicitly in the figure.}
    \label{fig:floquetsurface}
\end{figure}
\begin{figure}
    \centering
    \includegraphics[width=0.6\linewidth]{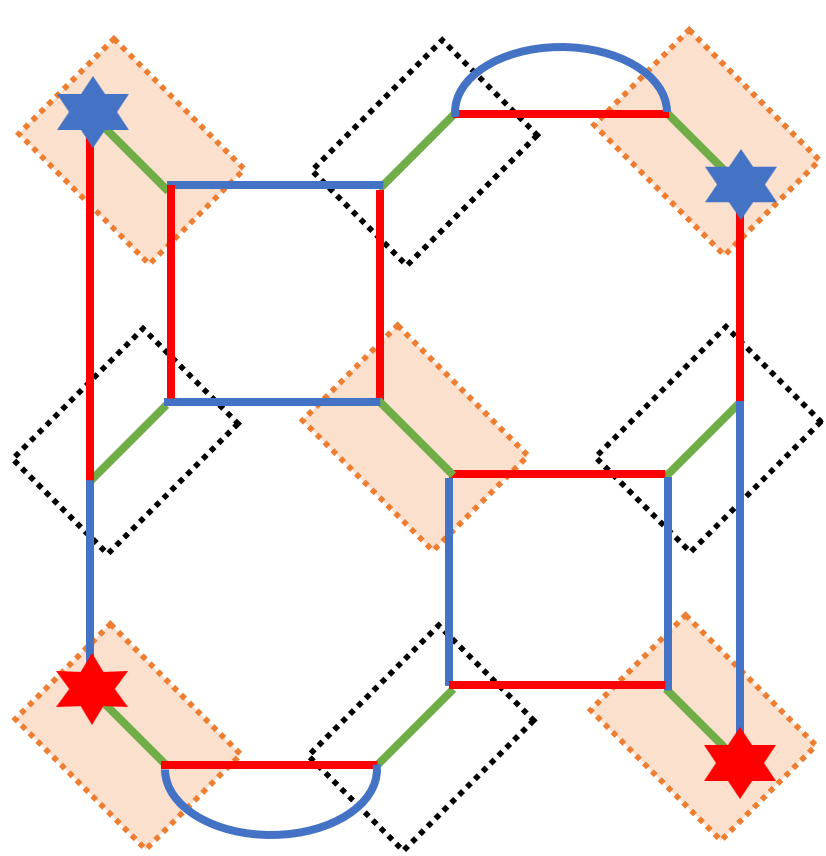}
    \caption{Lattice realization of the Floquet surface code. Here qubits live on vertices of the lattice. Every colored bond still represents a pairwise measurement in the respective basis, similar to the HH Floquet code. The blue/red stars represent single qubit $X$/$Y$ measurements at the time step 1/2, respectively.}
    \label{fig:surfacelat}
\end{figure}

\section{Floquet Checkerboard Code}
\label{sec: csscb}

\begin{figure}
    \centering
    \includegraphics[width=0.9\linewidth]{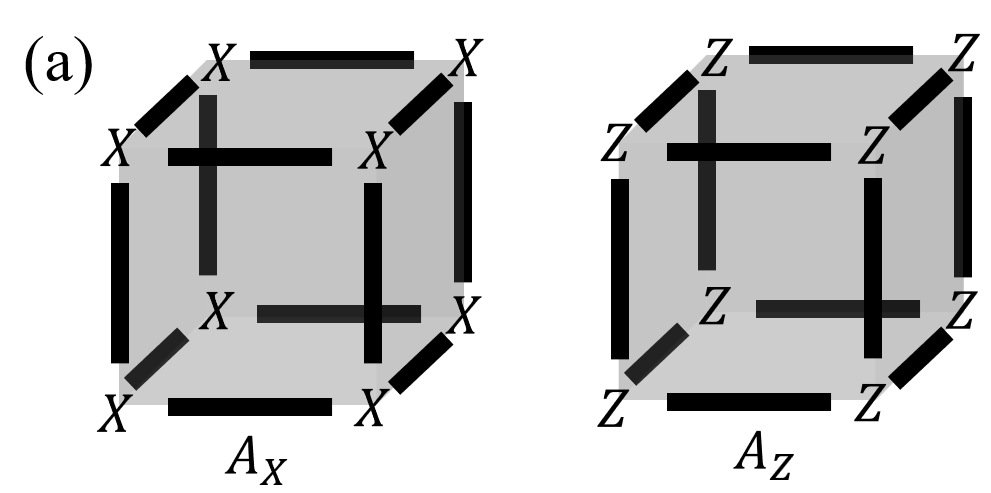}
    \includegraphics[width=0.8\linewidth]{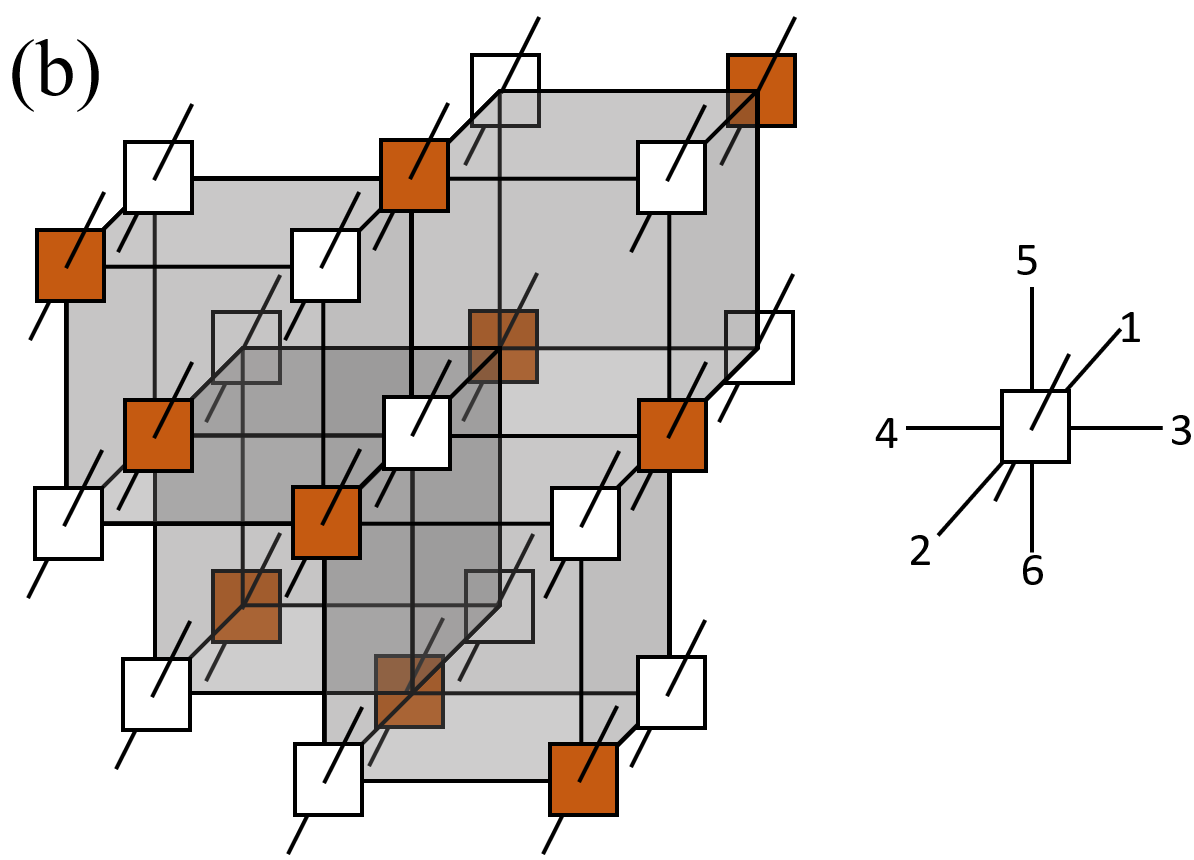}
    \caption{(a) The parametrization of the bond operators for the incoming $X$ and $Z$ cubic stabilizers. The outgoing ones are similarly parametrized. (b) The L and R gadgets of the Floquet checkerboard code. The cubes that host the cubic stabilizers are shaded in grey. The L gadgets are marked in white squares and the R gadgets are marked in orange. The leg directions when constructing the encoding maps in Eq. \eqref{eq: cbmat} are shown on the right. Here the shaded cubes are the ones that host the 8-body $X$ and $Z$ stabilizers.}
    \label{fig: cblattice}
\end{figure}
\begin{figure}
    \centering
    \includegraphics[width=0.5\linewidth]{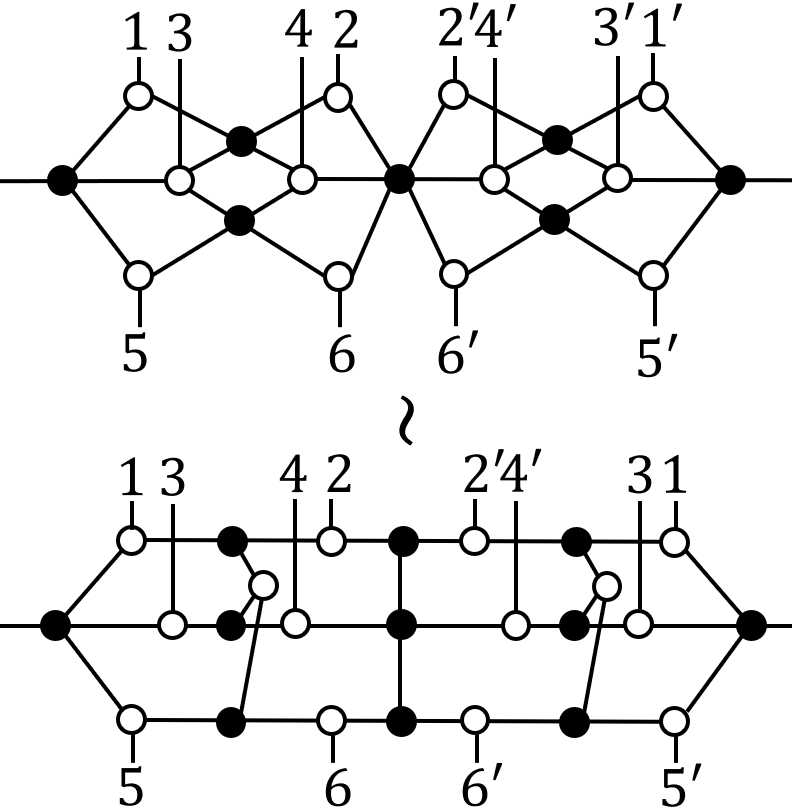}
    \includegraphics[width=\linewidth]{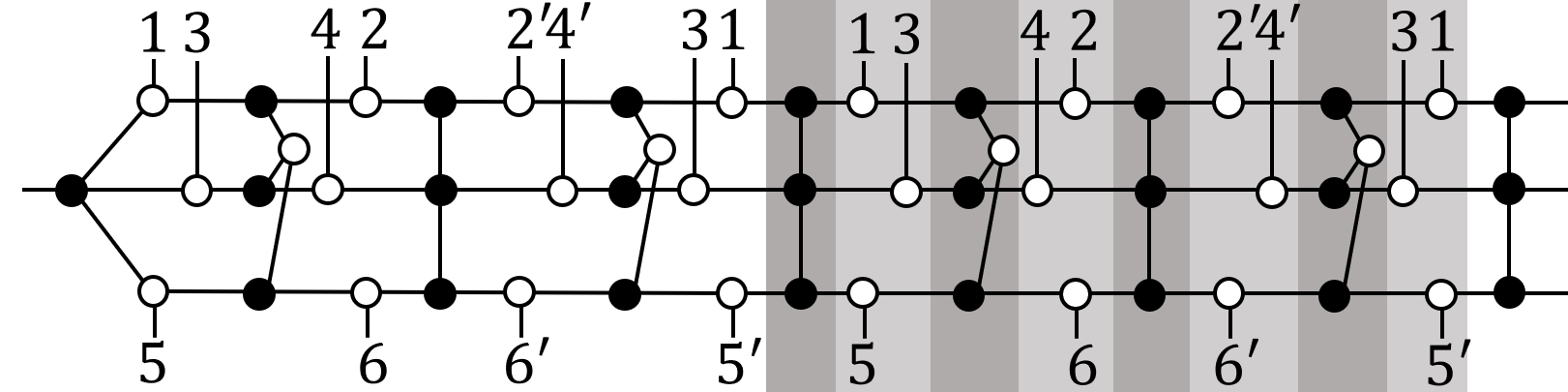}
    \caption{The ZX-diagram of the gadget from the encoding matrix $H_X$ in Eq. \eqref{eq: cbmat} with the solution in Eq. \eqref{eq: cbsol} and extra stabilizers in Eq. \eqref{eq: cbextra}. Here the unprimed and primed legs correspond to the first and second Pauli operator of the bond operators. Repeated gadget yields a 8-round schedule on the ZX-diagram marked by the alternating bright and dark gray shades.}
    \label{fig: cbzx}
\end{figure}

\begin{figure}
    \centering
    \includegraphics[width=0.6\linewidth]{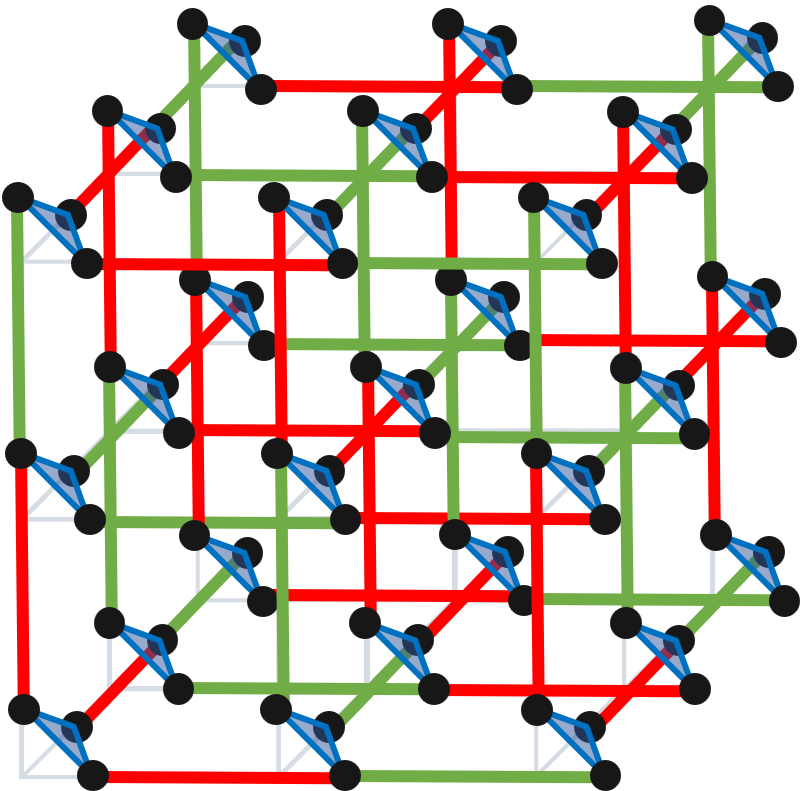}
    \caption{The schedule of the CSS Floquet checkerboard code built from the ZX-diagram in Fig. \ref{fig: cbzx}.}
    \label{fig: cbschedule}
\end{figure}

3D type-I fractonic codes are known to possess an extended number of logical qubits that scales with code distance $d$. One such example is the 3D checkerboard code on a 3D cubic lattice, where each lattice site hosts one data qubit. The stabilizers are eight-body $X$ and $Z$ operators around every other cube of the cubic lattice that forms a checkerboard pattern.

We now construct a dynamical code whose incoming and outgoing stabilizers are both the checkerboard code on the same sublattice of the cubic lattice. The gadget layout converts each data qubit into a gadget that connects to 6 nearby gadgets, following the connectivity of the cubic lattice. We consider a homogeneous internal leg layout so that each pair of connected gadgets share $n_{L,i}$ internal legs, where $i=1,2,\dots,6$ indicates the bond direction in Fig. \ref{fig: cblattice}. Following similar steps as in Fig. \ref{fig: haahproof}, it is straightforward to prove that $n_{L,i}=1$ is impossible to consistently encode the incoming and outgoing Pauli operators. 

We now demonstrate that the most efficient Floquet checkerboard code can achieved at $n_{L,i}=2$ with CSS encoding. We consider the simplest parametrization of the bond operator, where each bond in the incoming/outgoing $X$ and $Z$ cubic stabilizers hosts bond operators $A_{X,Z}$ and $\overline{A_{X,Z}}$, respectively (see Fig. \ref{fig: cblattice}(a)). We also assume the parametrization is translation invariant, so that each type of cubic stabilizer is parametrized in the same way. With the specification of the cubic stabilizers, the lattice sites on the cubic lattice split into two different sublattices, which we again name as L and R (see Fig. \ref{fig: cblattice}(b)). Below we only consider gadgets in the L sites, while the R gadgets can be synchronized accordingly.
The encoding matrices for the L gadget are (the six columns in between the two vertical lines corresponds to the six directions labeling in Fig. \ref{fig: cblattice})
\begin{align}\label{eq: cbmat}\nonumber
    H_X=\left(\begin{array}{c|cccccc|c}
        1&[A_X]&0&[A_X]&0&[A_X]&0 &0 \\
        1&[A_X]&0&0&[A_X]&0&[A_X]&0\\
        1&0&[A_X]&[A_X]&0&0&[A_X]&0\\
        1&0&[A_X]&0&[A_X]&[A_X]&0&0\\
        0&[\overline{A_X}]&0&[\overline{A_X}]&0&[\overline{A_X}]&0&1\\
        0&[\overline{A_X}]&0&0&[\overline{A_X}]&0&[\overline{A_X}]&1\\
        0&0&[\overline{A_X}]&[\overline{A_X}]&0&0&[\overline{A_X}]&1\\
        0&0&[\overline{A_X}]&0&[\overline{A_X}]&[\overline{A_X}]&0&1\\
    \end{array}\right), \\
    H_Z=\left(\begin{array}{c|cccccc|c}
        1&[A_Z]&0&[A_Z]&0&[A_Z]&0 &0 \\
        1&[A_Z]&0&0&[A_Z]&0&[A_Z]&0\\
        1&0&[A_Z]&[A_Z]&0&0&[A_Z]&0\\
        1&0&[A_Z]&0&[A_Z]&[A_Z]&0&0\\
        0&[\overline{A_Z}]&0&[\overline{A_Z}]&0&[\overline{A_Z}]&0&1\\
        0&[\overline{A_Z}]&0&0&[\overline{A_Z}]&0&[\overline{A_Z}]&1\\
        0&0&[\overline{A_Z}]&[\overline{A_Z}]&0&0&[\overline{A_Z}]&1\\
        0&0&[\overline{A_Z}]&0&[\overline{A_Z}]&[\overline{A_Z}]&0&1\\
    \end{array}\right)
\end{align}
The consistency equation $H_XH_Z^T=0\mod 2$ can be simplified to
\begin{align}\nonumber
    [A_X]\cdot [A_Z]=[\overline{A_X}]\cdot[\overline{A_Z}]=1\mod 2,\\
    [A_X]\cdot [\overline{A_Z}]=[A_Z]\cdot[\overline{A_X}]=0\mod 2,
\end{align}
which admits the following solution
\begin{align}\label{eq: cbsol}
    [A_X]=10,\ [A_Z]=10,\ [\overline{A_X}]=01,\ [\overline{A_Z}]=01.
\end{align}
\begin{figure*}
    \includegraphics[width=0.7\textwidth]{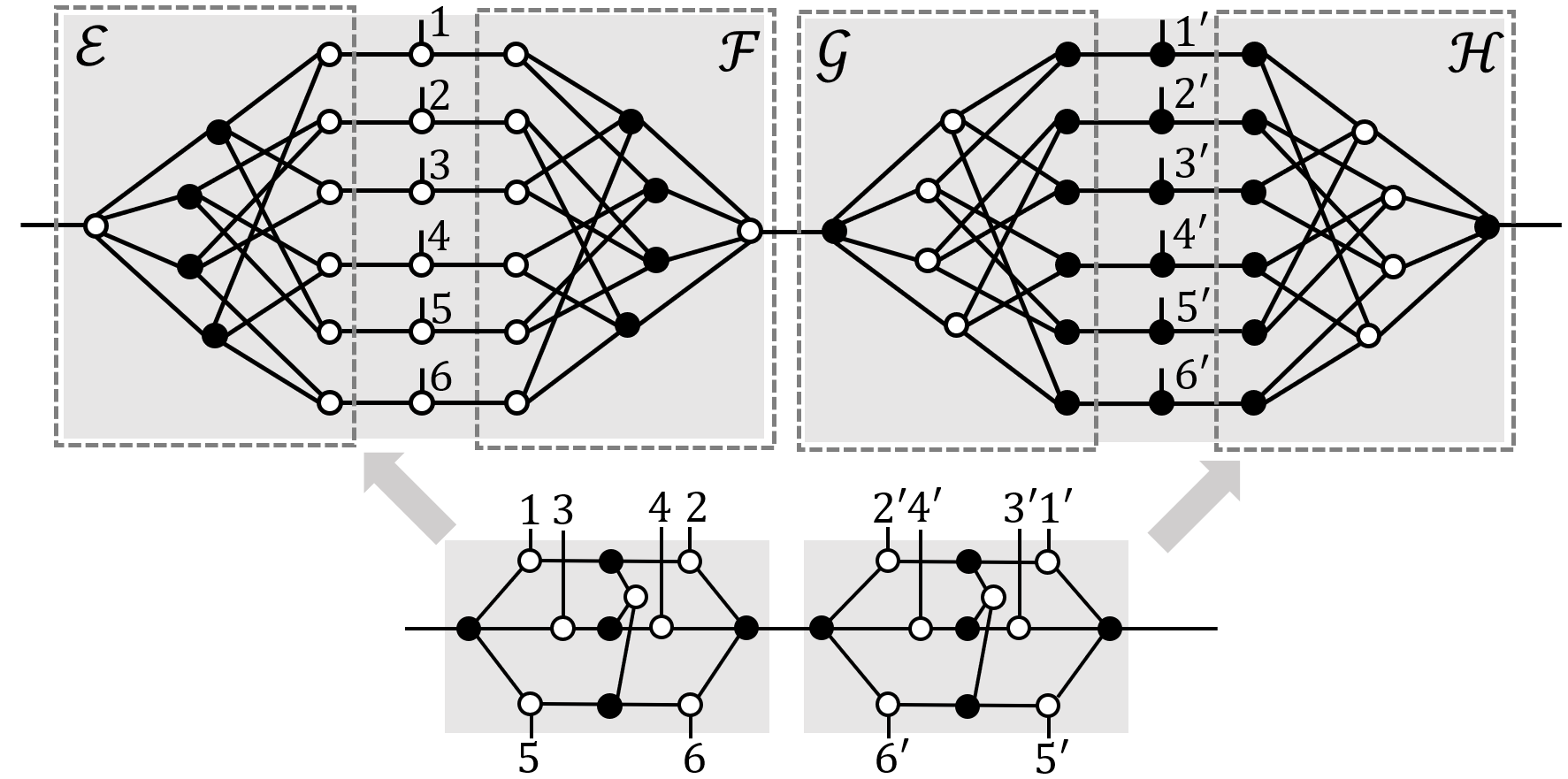}
    \caption{Proof that the fracton Floquet code in Ref. \cite{davydova_floquet_2023} is spacetime-equivalent to CSS Floquet checkerboard code on the level of gadget. The ZX-diagram in Fig. \ref{fig: cbzx} can be recompiled in two stages. The dual encoder $\mathcal{F}$ and encoder $\mathcal{G}$ can be combined into a measurement-only circuit that measures all the $Z$ stabilizers of the encoder $\mathcal{F}$ and all the $X$ stabilizer of the encoder $\mathcal{G}$, which corresponds to the rounds rZZZ and rXXX in Ref. \cite{davydova_floquet_2023}. Similarly, the dual encoder $\mathcal{H}$ and the encoder $\mathcal{E}$ combined gives the rounds bXXX and bZZZ. }
\end{figure*}

The rest of the solutions of the bond operators can be obtained by acting bond-local unitaries on the ones in Eq. \eqref{eq: cbsol}. Therefore they yield equivalent CSS Floquet codes. The total rank of the encoding matrix for both cases is 
\begin{equation}
    \text{rank}(H_X)+\text{rank}(H_Z)=12<2+\sum_{i=1}^6n_{L,i}=14,
\end{equation}
which means we have to specify two extra internal stabilizers. The choice that preserves logical information is to further stabilize
\begin{equation}\label{eq: cbextra}
    (XX)_2(XX)_4(XX)_6\text{ and }(ZZ)_2(ZZ)_4(ZZ)_6.
\end{equation}
The ZX-diagram of the gadget as an CSS encoder can be constructed in Fig. \ref{fig: cbzx}, which shows that the spatial concatenation for the CSS Floquet checkerboard code is a 3-qubit repetition code for each original data qubit. Laying them on the cubic lattice (see Fig. \ref{fig: cbschedule}), the measurement indicated by the ZX-diagram is bZZ, rXX, bZZZ, gXX, bZZ, gXX, bZZZ, rXX, where r and g rounds correspond to pairwise measurement on the red and green bonds in Fig. \ref{fig: cbschedule}. bZZ denotes pairwise measurement of $ZZ$ on the three bonds around the blue triangles , while BZZZ means measuring the product of the 3 $Z$ operators around the blue triangles. We note that, based on the ZX-diagram in Fig. \ref{fig: cbzx}, the incoming stabilizers will be completely measured after the following 5 rounds: bZZ, rXX,  bZZZ, gXX, bZZ, after which we arrive at another checkerboard code but with the cubic stabilizers supported on the other sublattice of the cubic lattice.

Similar to Fig. \ref{fig: tclogical}, it is straightforward to verify that the logical space with the extensive number of logical qubits is preserved by the gadgets. In fact, the entire measurement linear map acts trivially in the logical space, which can be seen from the operator maps
\begin{equation}
    X\to (XX)_1(XX)_2\leftarrow\overline{X}, \ Z\to (ZZ)_1(ZZ)_2\leftarrow \overline{Z}.
\end{equation}
These maps indicate that the $X$ and $Z$ logical operators in the direction of the directions of $1$ and $2$ will be preserved. Logical operator maps in the other two directions can be obtained similarly. Overall the CSS Floquet checkerboard code acts trivially in the logical space.

We now demonstrate that CSS Floquet checkerboard code we have just constructed is in fact spacetime equivalent to the fracton Floquet code between rounds 0 and 5 constructed in Ref. \cite{davydova_floquet_2023}. To see this, we recompile the ZX-diagram of the gadget in Fig. \ref{fig: cbzx}. The recompilation is a literal implementation of the algorithm for CSS encoder in Sec. \ref{sec: fixzx} with the incoming $X$ encoding map in Eq. \eqref{eq: cbmat} and outgoing $Z$ encoding map there. In the middle, the dual encoder $\mathcal{F}$ and encoder $\mathcal{G}$ combined leads to a measurement-only circuit that measures all the $Z$ stabilizers in $\mathcal{F}$ and all the $X$ stabilizers in $\mathcal{G}$, which corresponds to the rZZZ and rXXX rounds in the Ref. \cite{davydova_floquet_2023}. Similarly, the dual encoding map $\mathcal{H}$ and the following encoding map $\mathcal{E}$ form a measurement circuit of all the $X$ stabilizers in $\mathcal{H}$ and all the $Z$ stabilizers in $\mathcal{E}$, which corresponds to the bXXX and bZZZ rounds. Therefore, the two codes on the level of the gadget are completely equivalent. 

From the resource-theoretical perspective, such an equivalence is yet another demonstration that the dynamical qubit overhead and dynamical circuit depth are interconvertible resources in dynamical codes: we can trade more physical qubits with a shorter measurement schedule, or less physical qubits with a longer schedule.

\section{Pairwise-measurement $\Z_2^{(1)}$ subsystem code}
\label{sec: fermion}

In this appendix, we construct a dynamical code whose incoming and outgoing stabilizer group do not come from a stabilizer code, but rather from the stabilizer group of a subsystem code. One particular motivation of this example is the fact that many novel topological states, like chiral topological orders, do not admit commuting-projector parent Hamiltonian on qubits/qudits. Therefore, they cannot be the code state of any Pauli stabilizer code defined on qubits/qudits. One way to circumvent this constraint, as demonstrated in Ref. \cite{ellison_pauli_2023}, is to stabilize these via Pauli subsystem codes, where the stabilizer group $\cS$ is equal to the center of the gauge group $\mathcal{G}$ up to roots of identity. 
Although the relation between Floquet code and subsystem code has been explored in-depth in the literature\cite{hastings_dynamically_2021}, our construction is fundamentally different in the sense that we do not require the knowledge of the gauge group of the subsystem code. The only input information is the lattice layout of the stabilizer group of the subsystem code. In turn, the dynamical code will dynamically stabilize this restricted Hilbert space, rather than a single code state.
\begin{figure}
    \centering
    \includegraphics[width=0.55\linewidth]{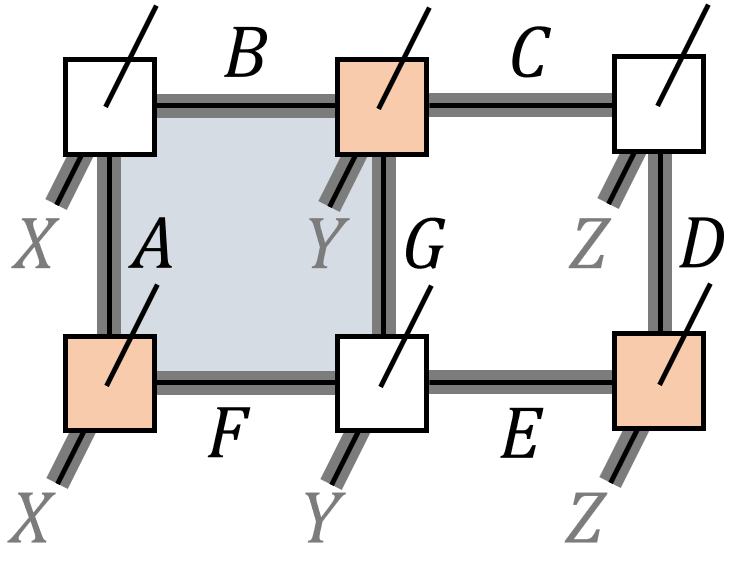}
    \caption{The six-body stabilizer in the $\Z_2^{(1)}$ fermionic subsystem code and the parametrization of the incoming stablizer. }
    \label{fig: fermionparam}
\end{figure}
\begin{figure}
    \centering
    \includegraphics[width=0.8\linewidth]{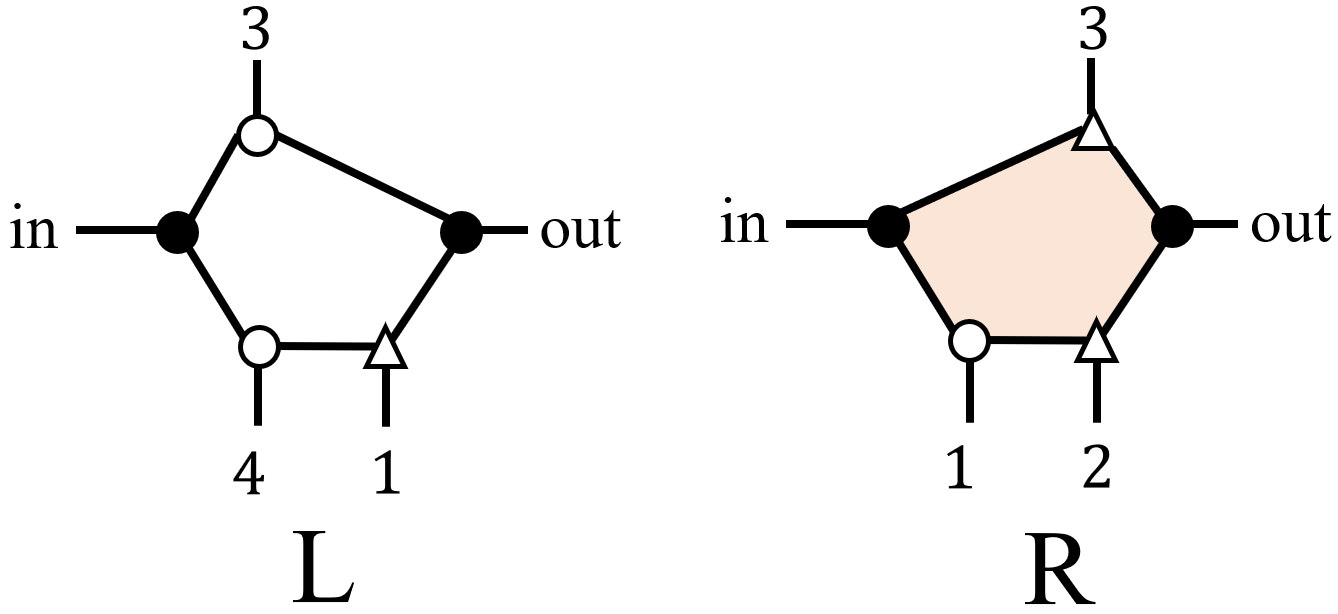}
    \caption{The ZX-diagram of the L gadget that produces the encoding map in Eq. \eqref{eq: fermionencode} with the solution in Eq. \eqref{eq: fermionsol}. The ZX-diagram of the R gadget can be similarly obtained. We use the convention in Fig. \ref{fig: tcparam} for the four direction of the bonds from 1 to 4.}
    \label{fig: fermionzx}
\end{figure}
\begin{figure}
    \centering
    \includegraphics[width=0.75\linewidth]{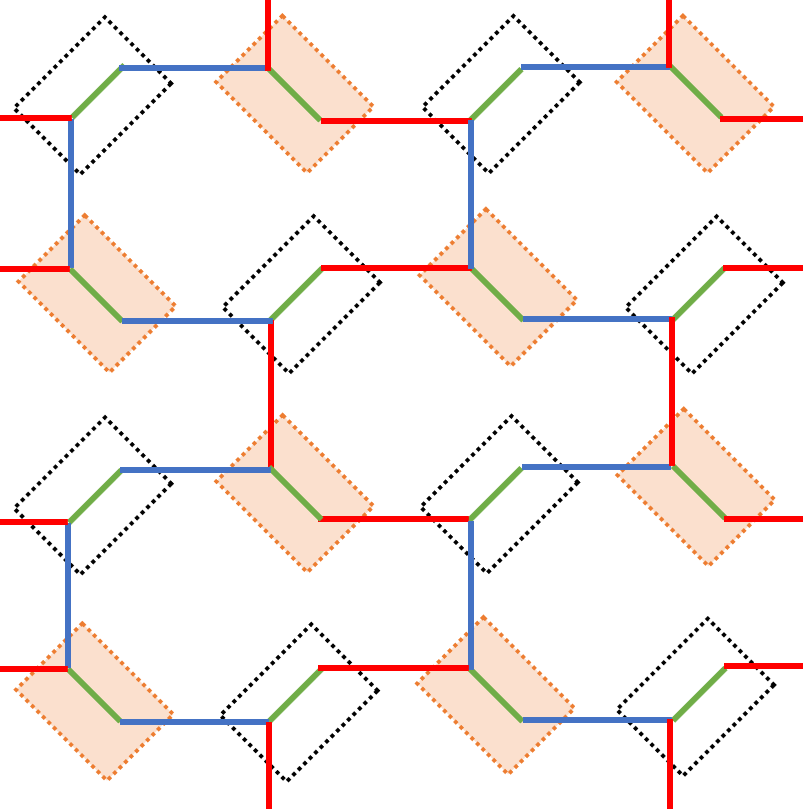}
    \caption{The spatial layout of the Floquet fermion code on a square-hexagon lattice with missing bonds. The measurement schedule is gZZ, rXX and bYY.}
    \label{fig: fermionlat}
\end{figure}

The example we are constructing is called the $\Z_2^{(1)}$ subsystem code on the 2D square lattice in Ref. \cite{ellison_pauli_2023}, whose stabilizer group $\cS$ is generated by the product of a square $X$ stabilizer and the square $Z$ stabilizer on its right(see Fig. \ref{fig: fermionparam}). Hence each stabilizer has Pauli weight six. We note that the restricted Hilbert space stabilized by $\cS$ is in fact fermionic, as the stabilizers are required to bosonize 2D lattice fermions \cite{Chen2018bosonization,obrien_Local_2024}. Therefore, we name the dynamical code that we are constructing the \textit{pairwise-measurement $\Z_2^{(1)}$ subsystem code}.

We consider the same gadget layout as in Fig. \ref{fig: tcgadget} for toric code. The incoming weight-six stabilizer is parametrized translation-invariantly by a set of 7 bond operators from $A$ to $G$, and similarly for the outgoing stabilizers (see Fig. \ref{fig: fermionparam}). Each gadget participates in six stabilizers in total, from which we obtain the encoding maps of the incoming/outgoing $X/Y/Z$ operators of the L gadget (L/R gadget designation and the four bond directions are the same as Figs. \ref{fig: tcgadget} and \ref{fig: tcparam}):
\begin{align}\label{eq: fermionencode}\nonumber
    X\to B_3 A_4,\ Y\to F_1 E_3,\ Z\to  C_1 D_4,\\
    \overline{X}\to \overline{B_3} \overline{A_4}, \ \overline{Y}\to \overline{F_1} \overline{E_3},\ \overline{Z}\to  \overline{C_1}\overline{D_4}.
\end{align}
Since the bond operators $G$ and $\overline{G}$ do not enter any encoding map, they are set to identity. The rest of the bond operators can be solved by the following set of single Pauli operators
\begin{align}\label{eq: fermionsol}\nonumber
    A=B=F=X,\ C=D=E=Y\\
    \overline{A}=\overline{B}=\overline{F}=Y,\ \overline{C}=\overline{D}=\overline{E}=X.\\
\end{align}
Other consistent solutions are equivalent to the one above by bond-local unitaries. The ZX-diagram based on the encoding map and the solution is constructed in Fig. \ref{fig: fermionzx}. The ZX-diagram of the R gadget can be constructe accordingly. Therefore, the Floquet fermion code can be laid down on the same physical lattice as the Floquet toric code, but with some missing bonds (see Fig. \ref{fig: fermionlat}). The measurement schedule is still three-round: gZZ, rXX and bYY, where the coloring of the bonds are shown in Fig. \ref{fig: fermionlat}.  

It is interesting to point out that the physical nature of the pairwise-measurement $\Z_2^{(1)}$ subsystem is actually a mixed-state toric code that is maximally decohered under the ``fermionic noise" \cite{Wang_intrinsic_2025,Sohal_noisy_2025,Ellison_toward_2025,zhang_strong_2025}, which excites simultaneously excites a pair of $e$ and $m$ anyon but not the fermion $\psi$ in the toric code, concatenated with a local 2-qubit repetitions code.

\section{The logical automorphism of the Floquet BB code}
\label{sec: bblogical}
In this appendix, we prove that the Floquet BB code constructed from the gadgets with the solution in Eq. \eqref{eq: bbsol} preserves the logical information of the BB code and acts as Hadamard+SWAP gates between pairs of logical operators. 

As a first step, we need to find the logical operators of the BB stabilizer code. A general way to represent logical operators of the BB code is the parametrized form $X(P_X,Q_X)$ with $P_X,Q_X\in \F_2^{lm}$, which corresponds to the $X$ operator $\prod_{j=1}^{lm}X_{L,j}^{P_{X,j}}X_{R,j}^{Q_{X,j}}$ over the L and R data qubits. In fact, due to translation symmetry of the permutation matrices $x$ and $y$, the matrix $(P_X|Q_X)$ represents a series of logical $X$ operators, where inside the matrix $P_X$ polynomial should be understood as a matrix $P_X=\sum_{p=0}^{l-1}\sum_{q=0}^{m-1}u_{X,ql+p+1}x^py^q$. In order for $X(P_X,Q_X)$ to be a logical $X$ operator, it has to commute with all the $Z$ checks, i.e. every row of $(P_X|Q_X)$ belongs to the kernel of the row space of $H_Z$: $(P_X|Q_X)\in \text{ker}(H_Z)$. Writing explicitly, we have
\begin{equation}\label{eq: logicalcon}
    (P_X|Q_X)H_Z^T=P_X\mathcal{B}+Q_X\mathcal{A}=0\mod 2.
\end{equation}
by which we mean that every monomial in $P_X\mathcal{B}+Q_X\mathcal{A}$ when expanded in powers of $x$ and $y$ has even coefficients. The total number of independent logical operators in $(P_X|Q_X)$ is equal to
\begin{align}\nonumber
    &\text{rank}\left[\left(P_X\vert Q_X\right)/\text{rs}(H_X)\right]\equiv\\
    &\quad\text{rank}\left[\left(\begin{array}{c|c}
        P_X & Q_X \\
        \mathcal{A} & \mathcal{B}
    \end{array}\right)\right]-\text{rank}\left[(\mathcal{A}|\mathcal{B})\right],
\end{align}
which is the rank of $X(P_X|Q_X)$ where two rows of $(P_X|Q_X)$ are equivalent if they differ by any linear combinations of multiple rows in $H_X$, which spans the linear space $\text{rs}(H_X)$. The logical $Z$ operators $Z(P_Z,Q_Z)$ can be similarly defined, where the two polynomials $P_Z$ and $Q_Z$ satisfies
\begin{equation}\label{eq: logicalconz}
    \mathcal{A}P_X^{-1}+\mathcal{B}Q_X^{-1}=0\mod 2.
\end{equation}

We now prove that $\left(P^{(123)}_X|Q^{(123)}_X\right)$ with the following set of polynomials are $X$ logical operators:
\begin{align}\label{eq: pqpoly}\nonumber
    &P^{(123)}_{X}=\sum_{r,s=0}\left(\mathcal{A}_1\mathcal{B}_1^{-1}\mathcal{A}_2^{-1}\mathcal{B}_2\right)^r\left(\mathcal{A}_2\mathcal{B}_2^{-1}\mathcal{A}_3^{-1}\mathcal{B}_3\right)^s,\\
    &Q^{(123)}_X=\mathcal{A}_2^{-1}\mathcal{B}_2P^{(123)}_X,
\end{align}
where the summation of $r$ and $s$ stops when all different monomials generated by $\mathcal{A}_1\mathcal{B}_1^{-1}\mathcal{A}_2^{-1}\mathcal{B}_2$ and $\mathcal{A}_2\mathcal{B}_2^{-1}\mathcal{A}_3^{-1}\mathcal{B}_3$ are summed over.
In fact, this can be checked straightforwardly from the condition in Eq. \eqref{eq: logicalcon}:
\begin{widetext}

    \begin{align}\nonumber
    P^{(123)}_X\mathcal{B}+Q^{(123)}_X\mathcal{A}&=P^{(123)}_X\left(\mathcal{B}+\mathcal{A}\mathcal{A}_2^{-1}\mathcal{B}_2\right)=P^{(123)}_X\left[\mathcal{B}_1+2\mathcal{B}_2+\mathcal{B}_3+\left(\mathcal{A}_1+\mathcal{A}_3\right)\mathcal{A}_2^{-1}B_2\right]\\\nonumber
    &=P^{(123)}_X\left[\mathcal{B}_1\left(1+\mathcal{A}_1\mathcal{B}_1^{-1}\mathcal{A}_2^{-1}\mathcal{B}_2\right)+\mathcal{B}_3\left(1+\mathcal{A}_3\mathcal{B}_3^{-1}\mathcal{A}_2^{-1}\mathcal{B}_2\right)\right]\\
    &=P^{(123)}_X\left(2\mathcal{B}_1+2\mathcal{B}_3\right)=0\mod 2.
\end{align}
\end{widetext}
In the final equality, we used the fact that the polynomial $P^{(123)}_{X}$ is generated by $\mathcal{A}_1\mathcal{B}_1^{-1}\mathcal{A}_2^{-1}\mathcal{B}_2$ and $\mathcal{A}_2\mathcal{B}_2^{-1}\mathcal{A}_3^{-1}\mathcal{B}_3$, so that $\mathcal{A}_1\mathcal{B}_1^{-1}\mathcal{A}_2^{-1}\mathcal{B}_2P^{(123)}_X=\mathcal{A}_2\mathcal{B}_2^{-1}\mathcal{A}_3^{-1}\mathcal{B}_3P^{(123)}_X=P^{(123)}_X$. Meanwhile, $\left(P^{(123)}_X|Q^{(123)}_X\right)$ cannot be generated by products of $X$ checks  as long as the BB code has nontrivial logical space. In fact, any product of $X$ checks must be represented by a polynomial that contains the factor $\mathcal{A}=\sum_{\alpha=1}^3\mathcal{A}_\alpha$, which is not the case for $P_X^{(123)}$. Therefore,  $\left(P^{(123)}_X|Q^{(123)}_X\right)\notin \text{rs}(H_X)$\footnote{We note that, for the [[144,12,12]] code, $\left(P^{(123)}_X|Q^{(123)}_X\right)$ contains 2 different $X$ logical operators.}.

The polynomials in Eq. \ref{eq: pqpoly} also give rise to a set of $Z$ logicals $Z\left(P_Z=P_X^{(123)},Q_Z=Q_X^{(123)}\right)$, which can be verified using Eq. \eqref{eq: logicalconz}. However, we caution the reader that $Z\left(P_X^{(123)},Q_X^{(123)}\right)$ are NOT the logical $Z$ operators of $X\left(P_X^{(123)},Q_X^{(123)}\right)$, but rather logical $Z$ operators of another independent set of $X$ logical operators.

We now show that the Floquet BB code $\mathcal{M}_\text{BB}$ constructed from the gadgets with the solution in Eq. \eqref{eq: bbsol} maps $X\left(P_X^{(123)},Q_X^{(123)}\right)$ to $Z\left(P_X^{(123)},Q_X^{(123)}\right)$, hence thereby performing a logical Hadamard as well as swapping these two sets of logical operators, similar to the case of the HH Floquet code. Recall that the L and R gadgets have bipartite connectivity so that each L gadget is connected to 9 R gadgets in the directions $(\alpha, \beta)$, $\alpha, \beta=1,2,3$, which are represented by monimials $\mathcal{A}_\alpha^{-1}\mathcal{B}_\alpha$.  
From Eq. \eqref{eq: bbsol}, the incoming $X$ (outgoing $Z$) operator can be simultaneously mapped to all 9 entries in the matrix $A_X$ ($\overline{A_Z}$), since $X^3=X$ and $Z^3=Z$. Therefore, we have
\begin{align}
    X\to\left(\begin{array}{ccc}
        XX & I & I \\
        I & XX & I \\
        I & I & XX
    \end{array}\right) \leftarrow\overline{Z}.
\end{align}
This means that any loop of incoming $X$ operators of the L and R data qubits along the two directions $(1,1)\to(2,2)$ and $(2,2)\to(3,3)$ will be mapped to a loop of outgoing $Z$ operators that are supported on the same set of data qubits. Since these two loop directions correspond to the monomials $\left(\mathcal{A}_1^{-1}\mathcal{B}_1\right)^{-1}\mathcal{A}_2^{-1}\mathcal{B}_2$ and $\left(\mathcal{A}_2^{-1}\mathcal{B}_2\right)^{-1}\mathcal{A}_3^{-1}\mathcal{B}_3$, the location of the L data qubits along the loop are exactly represented by the matrix $P_X^{(1,2,3)}$. The R data qubits are located on the loop which is shifted from the L data qubits by any relative direction between them. Choosing such a direction to be $(2,2)$, we arrive at the matrix $Q_X^{(1,2,3)}$ in Eq. \eqref{eq: pqpoly}. Therefore, we have 
\begin{equation}
    \mathcal{M}_\text{BB}\left(X\left(P_X^{(123)},Q_X^{(123)}\right)\right)=Z\left(P_X^{(123)},Q_X^{(123)}\right).
\end{equation}

Similarly, using the additional stabilizer in Eq. \eqref{eq: bbadditional}, these following set of maps of incoming $X$ and outgoing $Z$ operators yield the map of other sets of incoming $X$ logicals to outgoing $Z$ logicals of the BB code:
\begin{align}\nonumber
    X\to\left(\begin{array}{ccc}
        II & I & X \\
        I & XX & I \\
        X & I & II
    \end{array}\right) \leftarrow\overline{Z},\\\nonumber
    X\to\left(\begin{array}{ccc}
        XX & I & I \\
        I & II & X \\
        I & X & II
    \end{array}\right) \leftarrow\overline{Z},\\\nonumber
        X\to\left(\begin{array}{ccc}
        II & X & I \\
        X & II & I \\
        I & I & XX
    \end{array}\right) \leftarrow\overline{Z},\\\nonumber
        X\to\left(\begin{array}{ccc}
        XI & I & X \\
        I & IX & I \\
        I & X & II
    \end{array}\right) \leftarrow\overline{Z},\\\nonumber
        X\to\left(\begin{array}{ccc}
        XI & X & I \\
        I & II & X \\
        I & I & IX
    \end{array}\right) \leftarrow\overline{Z},\\\nonumber
        X\to\left(\begin{array}{ccc}
        II & I & X \\
        X& XI & I \\
        I & I & IX
    \end{array}\right) \leftarrow\overline{Z}.
\end{align}
We note that the underlying logical operators from these operator maps may not be mutually independent. Nevertheless, they together ensure that every nontrivial logical operator of the BB code will undergo a logical Hadamard plus SWAP gate.
For example, using the first operator map above, we define 
\begin{align}\label{eq: pqpoly321}\nonumber
    &P^{(321)}_{X}=\sum_{r,s=0} \left(\mathcal{A}_1\mathcal{B}_3^{-1}\mathcal{A}_2^{-1}\mathcal{B}_2\right)^r\left(\mathcal{A}_2\mathcal{B}_2^{-1}\mathcal{A}_3^{-1}\mathcal{B}_1\right)^s,\\
    &Q^{(321)}_X=\mathcal{A}_2^{-1}\mathcal{B}_2P^{(321)}_X,
\end{align}
then the set of $X$ logicals $X\left(P_X^{(321)},Q_X^{(321)}\right)$ will be mapped by the Floquet BB codes to
\begin{equation}
    \mathcal{M}_\text{BB}\left(X\left(P_X^{(321)},Q_X^{(321)}\right)\right)=Z\left(P_X^{(321)},Q_X^{(321)}\right).    
\end{equation}

\section{Detailed construction of the ZX-diagrams of gadgets}
\label{sec: additionalzx}
\subsection{ZX-diagrams of the gadget in the Floquet BB code}
\begin{figure}
    \centering
    \includegraphics[width=\linewidth]{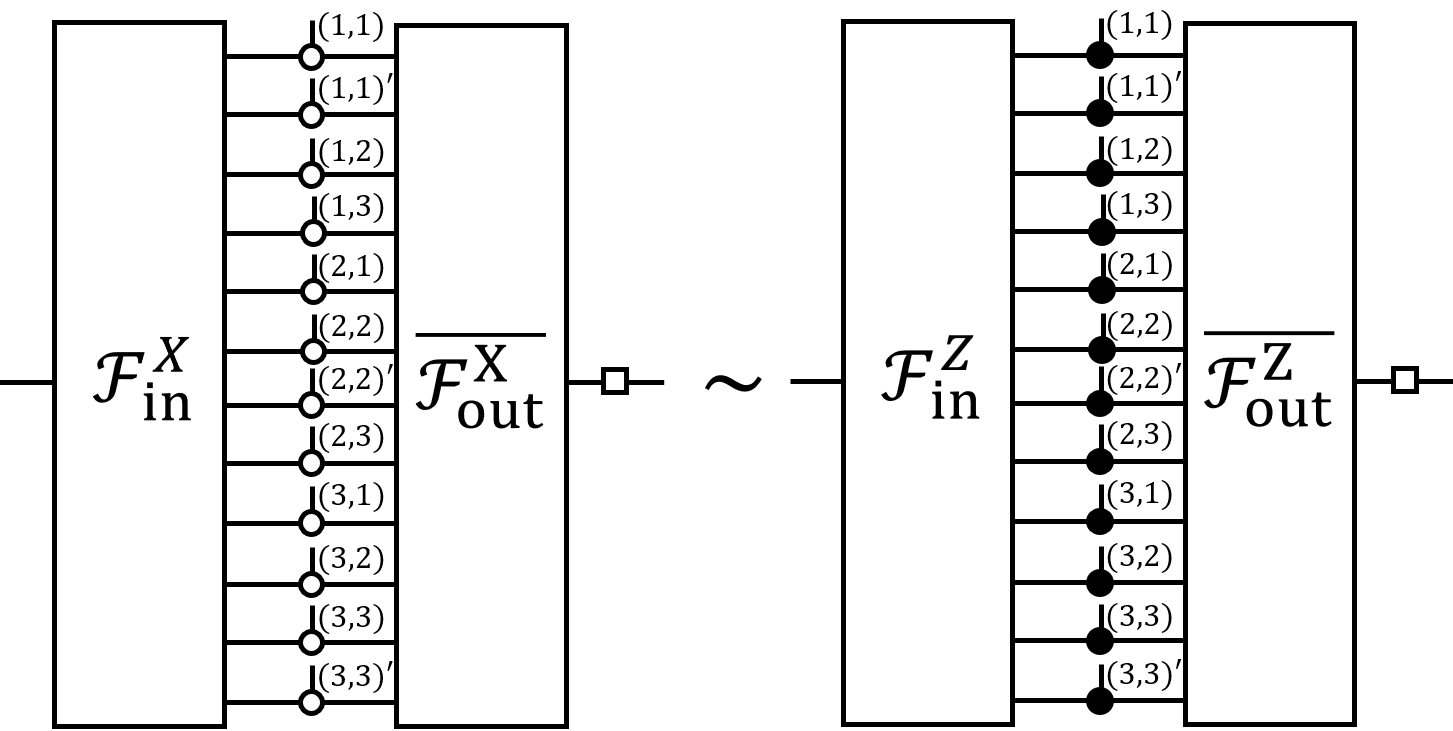}
    \caption{The gadget in the circuit layout with 12 physical qubits. To construct the rewinding schedule, we need both ZX-diagrammatic encoders constructed using the $X$ ecoding maps $\mathcal{F}^X_\text{in}/\overline{\mathcal{F}^X_\text{out}}$ and $Z$ encoding maps $\mathcal{F}^Z_\text{in}/\overline{\mathcal{F}^Z_\text{out}}$. The two compilations are equivalent under ZX-rules.}
    \label{fig: bb12}
\end{figure}
\begin{figure*}
    \centering
    \includegraphics[width=0.7\linewidth]{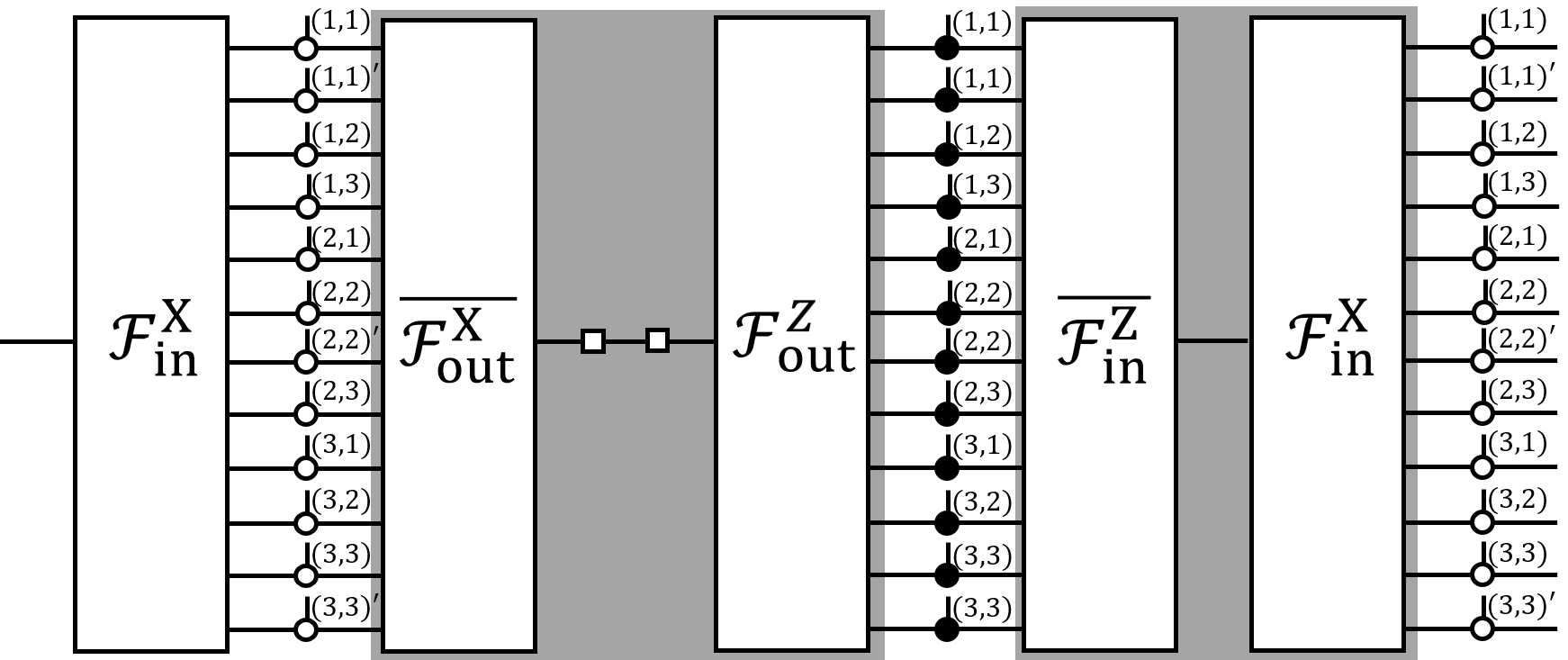}
    \includegraphics[width=0.7\linewidth]{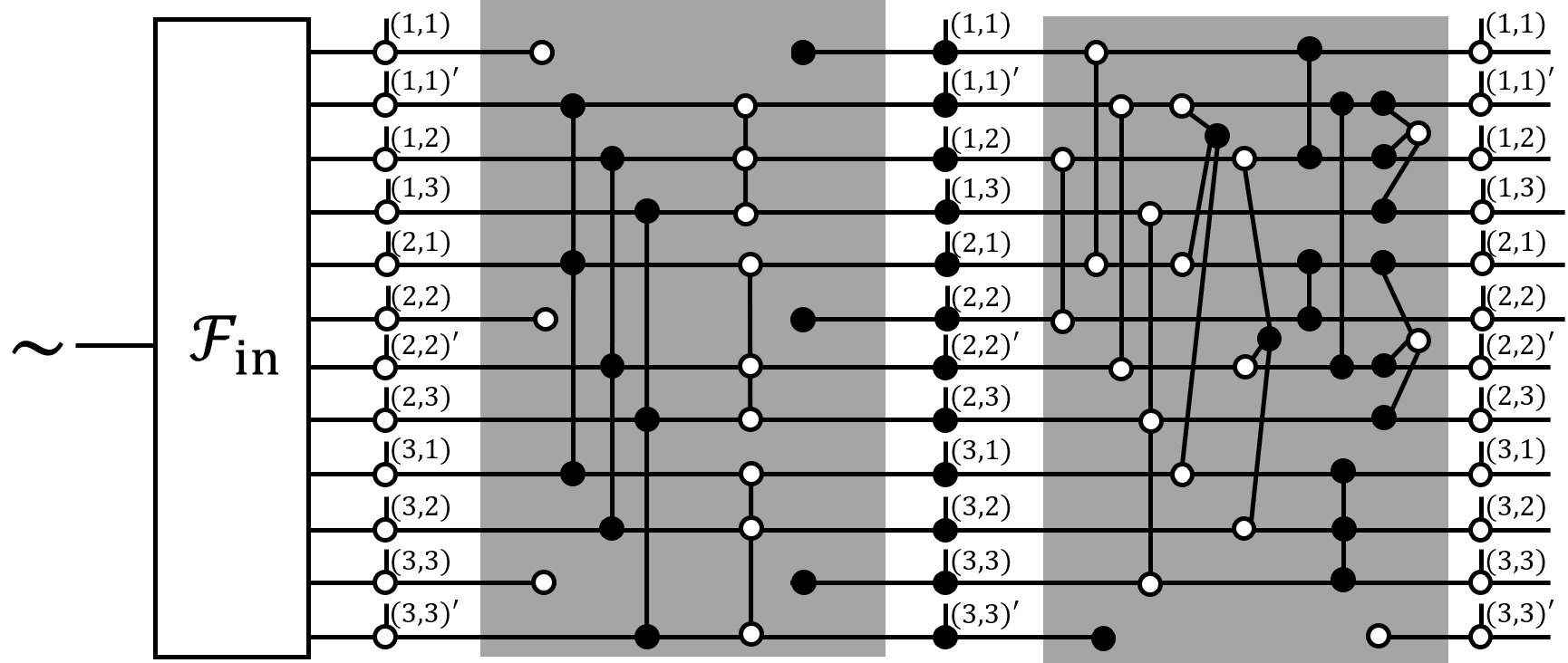}
    \caption{The 12-qubit circuit layout of the gadget for the Floquet BB code constructed via the $X$-basis encoder and the rewinded $Z$ basis encoder in Fig. \ref{fig: bb12}.}
    \label{fig: bb12zx}
\end{figure*}
\begin{figure*}
    \centering
    \includegraphics[width=0.9\linewidth]{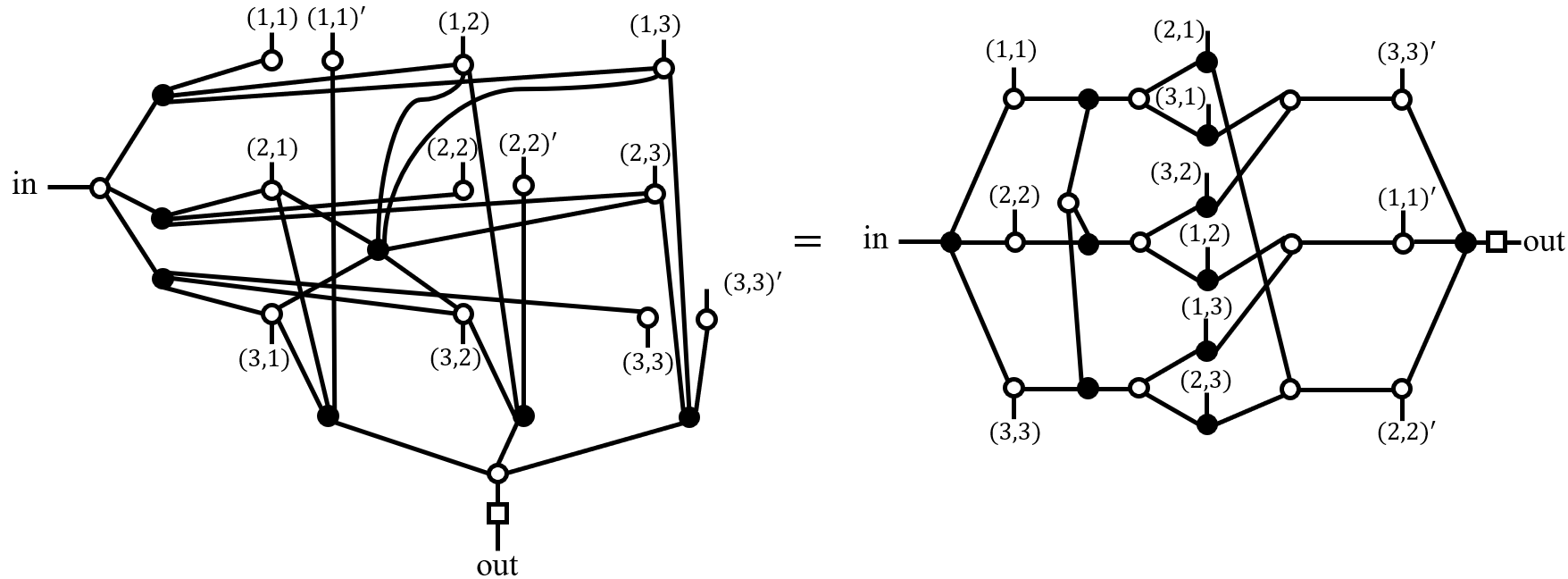}
    \caption{The diagrammatic encoder that yields the circuit layout of the gadget in Fig. \ref{fig: bbzx} with a 6-qubit spatial concatenation.}
    \label{fig: bbzx4}
\end{figure*}

In this appendix, we construct the ZX-diagram from the encoding map of the gadget from Eqs. \eqref{eq: bbencode}, \eqref{eq: bbsol} and \eqref{eq: bbadditional}. Since the encoding maps can be brought to a CSS form once we apply a Hadamard gate to the outgoing leg, we use the diagrammatic algorithm for CSS encoders in Sec. \ref{sec: fixzx}, which only requires the knowledge of either encoding matrices $H_X$ or $H_Z$. We use $H_X$ in the following construction. 

We first construct the encoder with the shortest schedule, which requires 12 physical qubits. The strategy here is to construct a stabilizer code encoder $\mathcal{F}^X_\text{in}$ from the incoming legs using $A^X$ in Eq. \eqref{eq: bbsol}, together with the additional stabilizer $(XX)_{(1,1)}X_{(1,2)}X_{(2,1)}(XX)_{(2,2)}$ in Eq. \eqref{eq: bbadditional}, and a dual stabilizer code encoder $\overline{\mathcal{F}^X_\text{out}}$ using $\overline{A^Z}$ (without the need for the additional stabilizer). This converts the ZX-diagram into a three-stage layout, see Fig. \ref{fig: bb12}. To construct a simple meausurement schedule, we adopt the strategy of rewinding by following such a gadget with the dual gadget compiled using the $Z$ encoding map, whose encoder and dual encoder are denoted by $\mathcal{F}^Z_\text{in}/\mathcal{F}^Z_\text{out}$ and the additional $Z$ stabilizer, $(ZZ)_{(1,1)}Z_{(1,2)}Z_{(2,1)}(ZZ)_{(2,2)}$, is encoded in $\mathcal{F}^Z_\text{in}$. In this way, the dual encoder $\overline{\mathcal{F}^X_\text{out}}$ is followed by the encoder $\mathcal{F}^Z_\text{out}$, see Fig. \ref{fig: bb12zx}. They together form a measurement circuit of all the $Z$ stabilizers that commutes with the three encoding maps in $\overline{A^Z}$, followed by a measurement circuit of all the $X$ stabilizers that commutes with the encoding maps in $\overline{A^X}$. These stabilizers that we need to measure can be chosen as
\begin{align}\label{eq: foutstab}\nonumber
    &Z_{(1,1)'}Z_{(2,1)},\ Z_{(2,1)}Z_{(3,1)}, \ Z_{(1,2)}Z_{(2,2)'},\ Z_{(2,2)'}Z_{(3,2)},\\\nonumber
    &Z_{(1,3)}Z_{(2,3)},\ Z_{(2,3)}Z_{(3,3)'},\ Z_{(1,1)},\ Z_{(2,2)},\ Z_{(3,3)},\\\nonumber
    &X_{(1,1)'}X_{(1,2)},\ X_{(1,2)}X_{(1,3)}, \ X_{(2,1)}X_{(2,2)'},\ X_{(2,2)'}X_{(2,3)},\\
    &X_{(3,1)}X_{(3,2)},\ X_{(3,2)}X_{(3,3)'},\ X_{(1,1)},\ X_{(2,2)},\ X_{(3,3)}.
\end{align}
Here we have abused the notation of internal legs to also label the physical qubits in the circuit layout of the ZX-diagram.
The measurement circuits are shown in the first gray shaded box in Fig. \ref{fig: bb12zx}.
Similarly,  the dual encoder $\overline{\mathcal{F}^Z_\text{in}}$ compiled from the encoding maps in $A^Z$ is followed by $\mathcal{F}^X_\text{in}$, which leads to a measurement circuit for all the $X$ stabilizers that commute with the three encoding maps in $A^Z$ and the additional $Z$ stabilizer, followed by another measurement circuit that measures all the $Z$ stabilizers that commute with the three encoding maps in $A^X$ and the additional $X$ stabilizer. These stabilizers that we need to measure can be chosen as
\begin{align}\label{eq: finstab}\nonumber
    &X_{(1,1)}X_{(2,1)},\ X_{(1,2)}X_{(2,2)},\ X_{(1,1)'}X_{(2,2)'},\ X_{(1,1)'}X_{(2,1)}X_{(3,1)},\\\nonumber
    &X_{(1,2)}X_{(2,2)'}X_{(3,2)},
    X_{(1,3)}X_{(2,3)},X_{(2,3)}X_{(3,3)},X_{(3,3)'}\\\nonumber
    &Z_{(1,1)}Z_{(1,2)},\ Z_{(2,1)}Z_{(2,2)},\ Z_{(1,1)'}Z_{(2,2)'},\ Z_{(1,1)'}Z_{(1,2)}Z_{(1,3)},\\
    &Z_{(2,1)}Z_{(2,2)'}Z_{(2,3)},
    Z_{(3,1)}Z_{(3,2)},Z_{(3,2)}Z_{(3,3)},Z_{(3,3)'}
\end{align}
Combining with the $X$ and $Z$ nodes that are direct connected to the internal legs, the final circuit layout of the ZX-diagram in shown Fig. \ref{fig: bb12}. This corresponds to a Floquet BB code with 12 physical qubits in each gadget, one per each internal leg. The measurement schedule is a four-round rewinding schedule as follows: (0) measure all the $X$ stabilizers followed by all the $Z$ stabilizers in Eq. \eqref{eq: finstab}, (1) pairwise measure all inter-gadget checks in $X$ basis in every bond direction, (2) measure all the $Z$ stabilizers followed by all the $X$ stabilizers in Eq. \eqref{eq: foutstab}, (3) pairwise measure all inter-gadget checks in $Z$ basis for every bond direction. We note that all the stabilizers in the BB stabilizer code will be measured twice in the four measurement rounds. From the resource-theoretical perspective, this corresponds to the situation where we trade the dynamical qubit overhead with the circuit depth.

To compile a circuit layout of the gadget's ZX-diagram with fewer number of dynamical qubit overhead, which is presented in the maintext in Fig. \ref{fig: bbzx}, we again use the diagrammatic algorithm in Sec. \ref{sec: fixzx} for CSS encoders. Using the encoding maps of incoming $X$ and outgoing $Z$ operators in Eq. \ref{eq: bbsol}, together with the additional $X$ stabilizer in Eq. \ref{eq: bbadditional}, we have the ZX-diagram of the encoder in Fig. \ref{fig: bbzx4}, which can be brought to a circuit layout as in Fig. \ref{fig: bbzx} in the main text by repeated usage of the bialgebra rule in Eq. \eqref{eq: bialgebra}.

\subsection{ZX-diagram of the gadget in the CSS Floquet Haah code}

\begin{figure}
    \centering
    \includegraphics[width=0.5\linewidth]{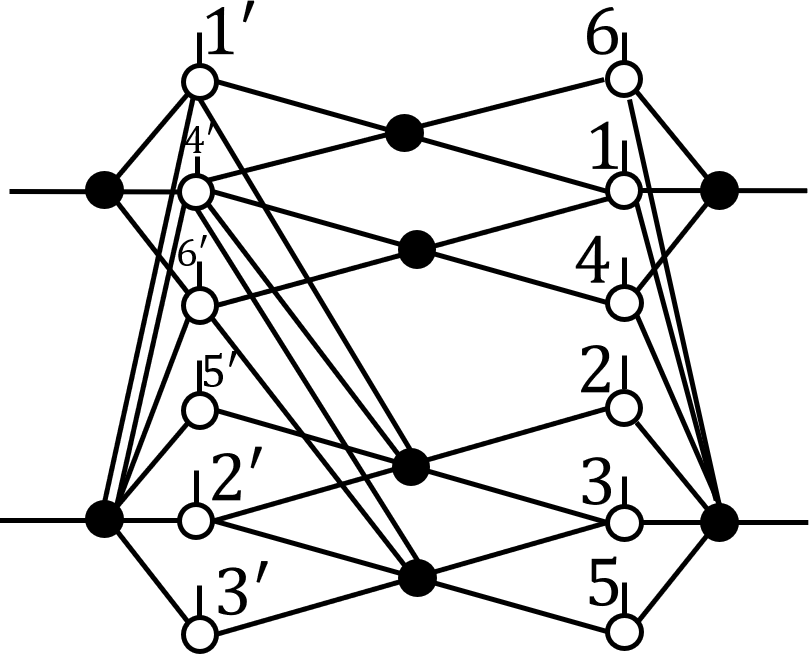}
    \caption{Diagrammatic construction of the CSS encoder for the gadget in CSS Floquet Haah code.}
    \label{fig: haahcss}
\end{figure}
We now construct the ZX-diagram of the gadget in the CSS Floquet Haah code from the $X$ encoding matrix $H_X$ in Eq. \eqref{eq: haahmat} with the solution in Eq. \eqref{eq: haahsol} using the diagrammatic algorithm for CSS encoder in Sec. \ref{sec: fixzx}. Again, we use the strategy that only connects each incoming/outgoing legs to one set of internal legs, while transforming the rest of the encoding maps completely to internal stabilizers. To this end, we see that encoding matrix $H_X$ can be transformed into
\begin{equation}
    \left(\begin{array}{c|cccccc|c}
        10&01&00&00&01&00&01 &00 \\
        01&01&01&01&01&01&01&00 \\
        00&10&00&00&10&00&10&10\\
        00&10&10&10&10&10&10&01\\
        00&11&00&00&01&00&10&00\\
        00&10&00&00&11&00&01&00\\
        00&01&11&10&01&01&00&00\\
        00&00&01&11&01&10&01&00
    \end{array}\right).
\end{equation}
Using these $X$ stabilizers, the diagrammatical encoder is constructed in Fig. \ref{fig: haahcss}, which can be brought to the circuit layout with a six-qubit spatial concatenation in Fig. \ref{fig: haahzx}.

\section{Non-SLP and hook error in the pairwise measurement surface code}
\label{sec: gidney}

\begin{figure}
    \centering
    \includegraphics[width=\linewidth]{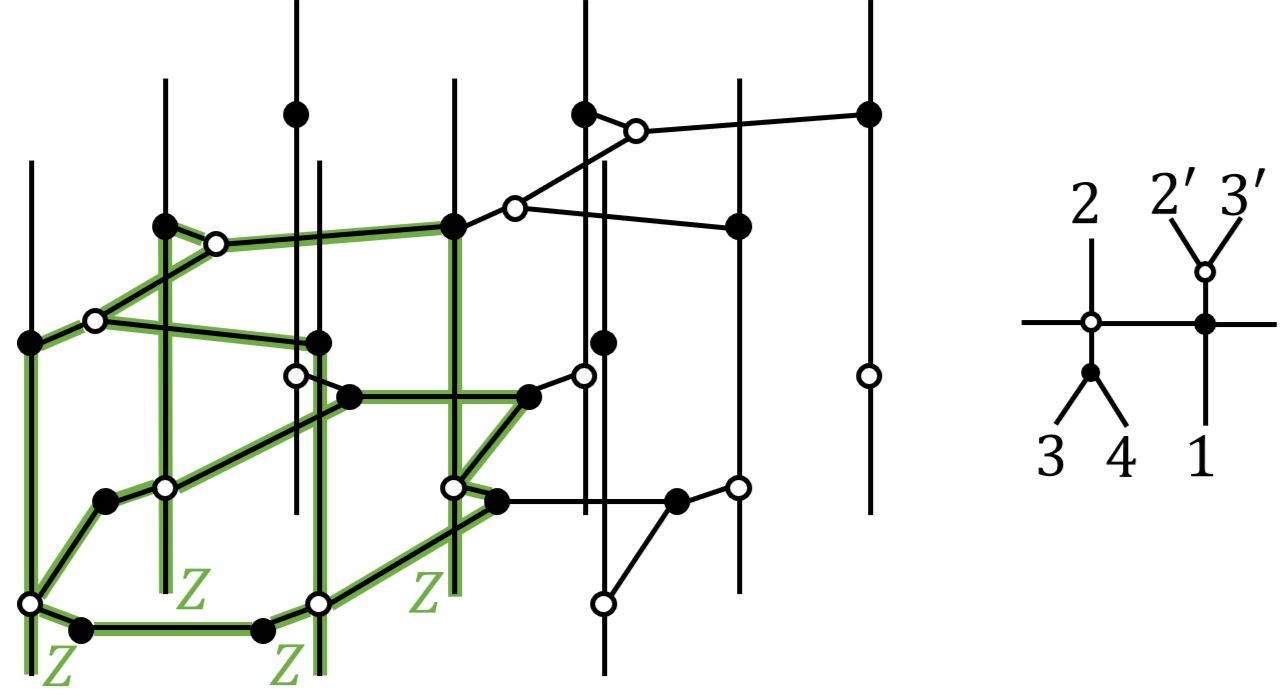}
    \caption{Left: The ZX-diagram of the pairwise measurement surface code in the gadget layout through dragging the intra-plaquette nodes towards the data qubits. The Pauli web of the incoming $Z$ stabilizer is marked in green shaded lines. The Pauli web goes into nearby gadgets that are not part of the $Z$ stabilizer. Right: the ZX-diagram of the L gadget in the gadget layout of the pairwise measurement surface code.}
    \label{fig: pairwise}
\end{figure}
\begin{figure}
    \centering
    \includegraphics[width=0.7\linewidth]{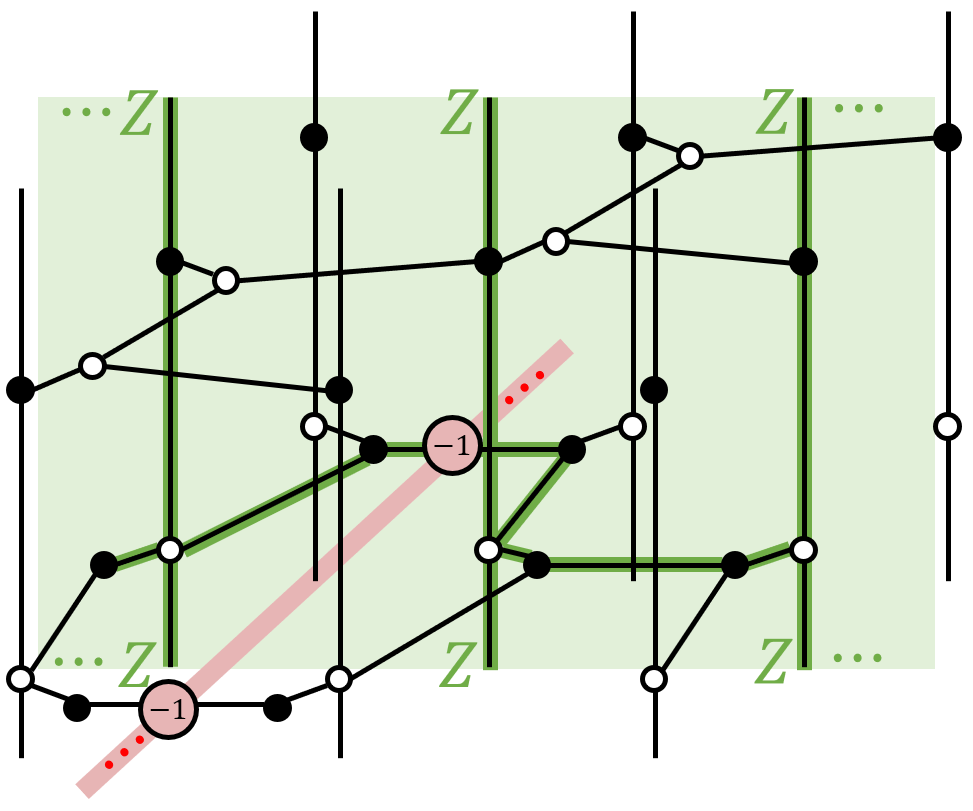}
    \caption{The Hook error in the pairwise measurement surface code. Here the green plane represents the support of the Pauli web of the logical $Z$ operator. The hook error is a line of Pauli $X$ error represented by a line of $X$ nodes with phase $\pi$, which are marked in red. It anti-commutes with the Pauli web of the logical $Z$ operator while commuting with all the spacetime $Z$ stabilizer (for example, the one in Fig. \ref{fig: pairwise}). The weight of the hook error is $d/2$, instead of $d$.  }
    \label{fig: gidneyhook}
\end{figure}
In this section, we use the ZX-diagram of the pairwise measurement surface code proposed in Ref. \cite{gidney_pair_2023} as an example to show that a gadget layout that does not satisfy the SLPC will potentially lead to a reduction in the spacetime code distance from $d$ to $d/2$.

To reach a gadget layout form of the ZX-dagram, we bring the two nodes inside each plaquette to two of the four data qubits around the plaquette (see Fig. \ref{fig: pairwise}).  Therefore, the ZX-diagram near the vertical line that represents a data qubit can be viewed as the ZX-diagram of a local gadget. Now consider the Pauli web of an incoming $Z$ stabilizer. From Fig. \ref{fig: pairwise}, it is easy to see that the Pauli web enters two other gadgets that are not part of the four-body $Z$ stabilizer. The outgoing $X$ stabilizers have similar problems. Therefore, such a gadget layout does not satisfy the SLPC. 

The reduction in space-time distance can also be visualized directly from the Pauli web of the logical $Z$ operator, as shown in Fig. \ref{fig: gidneyhook}. There, we find a chain of $X$ errors that faults the logical $Z$ operator while being undetected, i.e. it commutes with all the spacetime $Z$ stabilizers. Physically, every $X$ error long the line represents a measurement error in pairwise measuring $ZZ$ of the two ancillas.
Since the error only occurs in every other plaquette along the line, the weight of such a logical error is $d/2$, instead of $d$. This is known as a hook error in the QEC literature.  
Similarly, a weight-$d/2$ line of $Z$ errors in the direction perpendicular to the $X$ error chain in Fig. \ref{fig: gidneyhook} will fault the logical $X$ operator.
\begin{figure}
    \centering
\includegraphics[width=\linewidth]{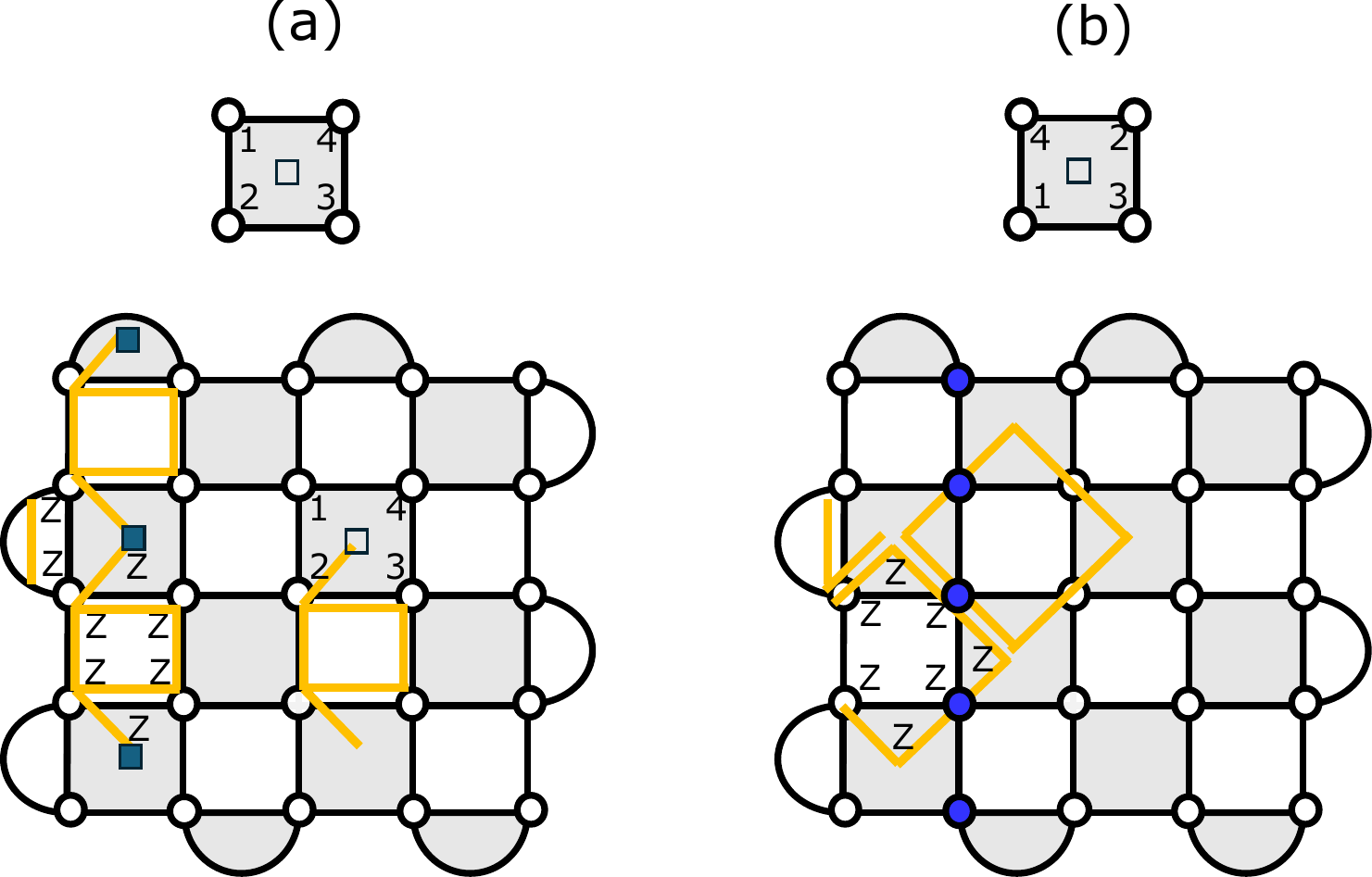}
    \caption{The X-stabilizer supported on data qubits (circles) of the gray plaquettes is measured using CNOTs with an ancillary qubit (shown as square) used as control. The ordering of CNOTs is prescribed using the numbers 1, 2, 3 and 4 on the data qubits. (a) This ordering leads to evolved Z stabilizers (6-qubit terms in the bulk and 2-qubit terms at the boundary, shown using orange lines) for the ISG after the first two CNOTs for each X stabilizer measurement circuit as shown. A smaller X-logical can be supported on the purple square ancillaries due to the asymmetric shape of these evolved stabilizers which has a reduced number of Z stabilizers supported on the ancillary qubits compared to the data qubits. (b) A different ordering that avoids hook errors leads to symmetrically evolved bulk Z stabilizers (8-qubit terms in the bulk, 7-qubit terms and 3-qubit terms at the boundary shown using orange lines) in the middle of the circuit with the same number of Z stabilizers supported on the ancillary qubits as the original data qubits. The logical operator is supported on the data qubits shown in purple.}
    \label{fig:ISGs_CNOT_circuit}
\end{figure}
Nevertheless, as pointed out by Ref. \cite{gidney_pair_2023}, the spacetime distance may still be preserved for surface code, if the $X$ logical operator is perpendicular to the direction that the Pauli web of the incoming $Z$ operator extends (i.e. the $X$ logical operator goes from left to right in Fig. \ref{fig: pairwise}. rather than the direction in Fig. \ref{fig: gidneyhook}). And similarly,
As mentioned at the end of Sec. \ref{sec: spacetimed}, this is exactly the case where non-SLP of the Pauli web does not extend in the direction of the logical operator.

\section{Non-SLP and hook error in the SASECs for the surface code}
\label{sec:NonSLP_CNOTcircuits}
We compare the spacetime distance for the syndrome extraction SASECs for the surface code with and without hook errors, using the instantaneous stabilizer groups (ISGS). In the circuit with hook error, the ordering of CNOTs is such that the X error in the middle of the SASEC propagates to two Xs on the path of an X-logical. In the circuit without hook error, the X error in the middle of the circuit propagates to two Xs along the diagonal such that it can only be part of an X-logical that is longer, and thus does not affect the distance. In terms of ISGs, this can be explained from the structure of the stabilizers in the two cases for the ISG in the middle of the circuit. In one case, the evolved Z stabilizers in the middle of the syndrome extraction SASEC for X stabilizers are asymmetric with respect to the square lattice (which decides the minimum weight logical operators) as shown while in the other case, they are symmetric. In both cases, the evolved Z stabilizers have spread out on the lattice because of the support on the nearby ancillaries. But it is the asymmetry of the evolved Z stabilizers that leads to a reduced static code distance for the ISG in the middle of the circuit in the case with hook errors. In the asymmetric case, there are fewer neighboring stabilizers supported on the ancillary qubits, and thus a smaller-weight logical operator can be supported on them. We illustrate this in Fig.~\ref{fig:ISGs_CNOT_circuit}.

\begin{figure}
    \centering
    \includegraphics[width=0.5\linewidth]{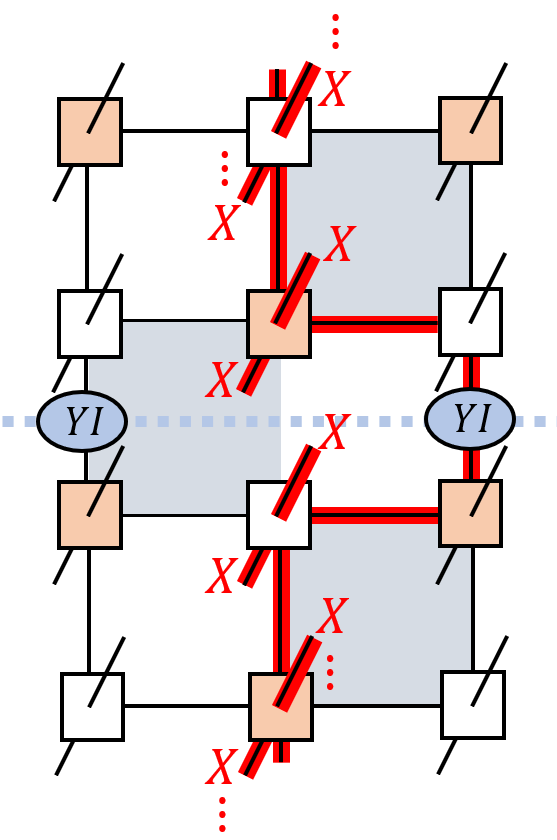}
    \caption{The spacetime distance of the HH Floquet code is reduced from $d$ to $d-1$. Here we show Pauli web of the vertical $X$ logical operator when going through the dynamical code with a missing bond. The string of Pauli errors $YI$ in the direction perpendicular to the missing bond, marked by the dashed line, causes an undetected logical error, since it anti-commutes with the bond operator $XY$ of the vertical $X$ logical operator. Such an error chain has weight $d-1$.}
    \label{fig: brokenlogical}
\end{figure}
\begin{figure}
    \centering
    \includegraphics[width=0.75\linewidth]{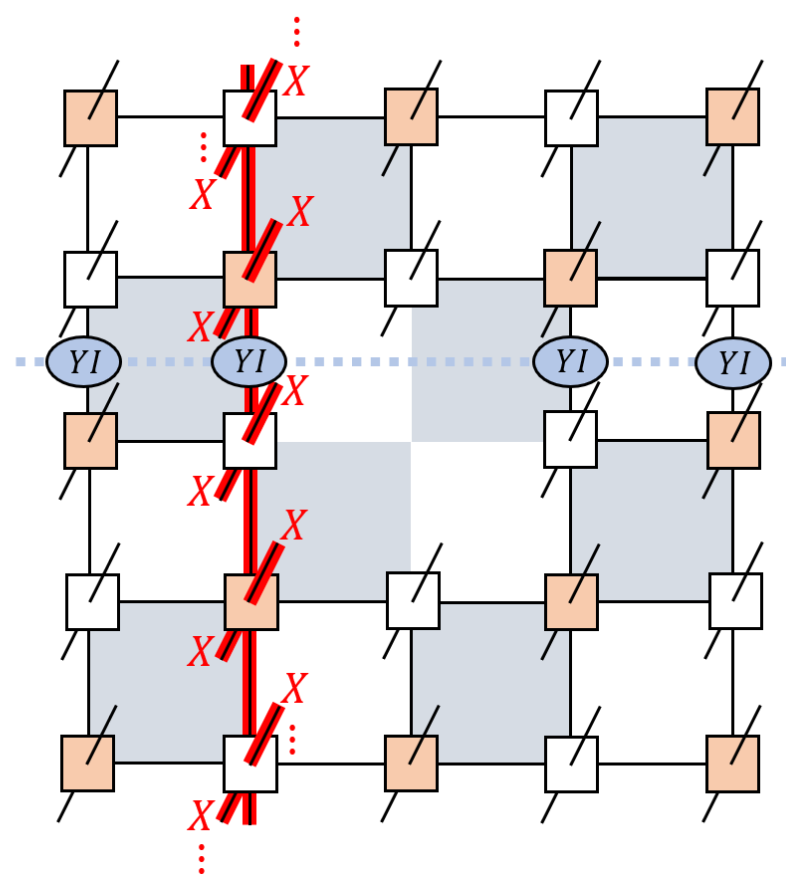}
    \caption{In the modified 6-round schedule in Sec. \ref{sec: dropout}, the spacetime code distance of the HH Floquet code is reduced from $d$ to $d-1$ with the missing gadget. The lowest weight undetected logical error for the vertical $X$ logical operator can be chosen as the string of $d-1$ $YI$ errors on the dotted line. }
    \label{fig: dropoutdistance}
\end{figure}

\section{Spacetime code distance with fabrication defects}
\label{sec: defectdistance}

In this Appendix, we demonstrate that, for both cases of broken connector and qubit dropout in the HH Floquet code, the spacetime code distance is reduced from $d$ to $d-1$ when we implement the modified schedules in Figs. \ref{fig: brokenschedule} and \ref{fig: dropoutlat}. We will also compare our approach to Ref. \cite{mclauchlan_accommodating_2024} when dealing with qubit dropouts.

We first consider the broken connector case. In the gadget layout, consider a chain of Pauli errors $YI$ that occur in the extensive direction perpendicular to the missing bond. With periodic conditions, such an error chain has weight $d-1$. Now consider the vertical $X$ logical operator that goes around the missing bond, see Fig. \ref{fig: brokenlogical}. We see that the chain of $YI$ errors causes a logical error that faults the vertical $X$ logical operator while commuting with all the spacetime stabilizers, since it anti-commutes with the bond operator of the logical operator $XY$. Similarly, the string of Pauli errors $IY$ will fault the vertical logical $Z$ operator.
Therefore, the spacetime code distances of both logical subspaces are reduced by $1$.

Similarly, for the case of qubit dropout, the spacetime code distance of the HH Floquet code is also reduced by $1$ from the modified 6-round schedule around the dropout qubit in Sec. \ref{sec: dropout}. In fact, a string of $d-1$ $YI$ errors will again fault the logical $X$ operator without being detected by any spacetime stabilizers, see Fig. \ref{fig: dropoutdistance}. In comparison, we also adapt the approach in Ref. \cite{mclauchlan_accommodating_2024} to a gadget layout form for the toric code on the square lattice, which further removes connections to the R gadget on the bottom left and top right of the missing gadget. In this case, the spacetime code distance is actually reduced by 2 for all $X$ and $Z$ logical operators in both directions, as demonstrated in Fig. \ref{fig: dropoutdistance2}. The schedule in this case is still 3-round, as shown in Fig. \ref{fig: dropoutschedule2}. We note that the original implementation of the protocol in Ref. \cite{mclauchlan_accommodating_2024} is for the toric code on the hexagonal lattice, in which case the spacetime distance for some of the logical operators is reduced by 1, instead of 2.

\begin{figure}[H]
    \centering
    \includegraphics[width=0.7\linewidth]{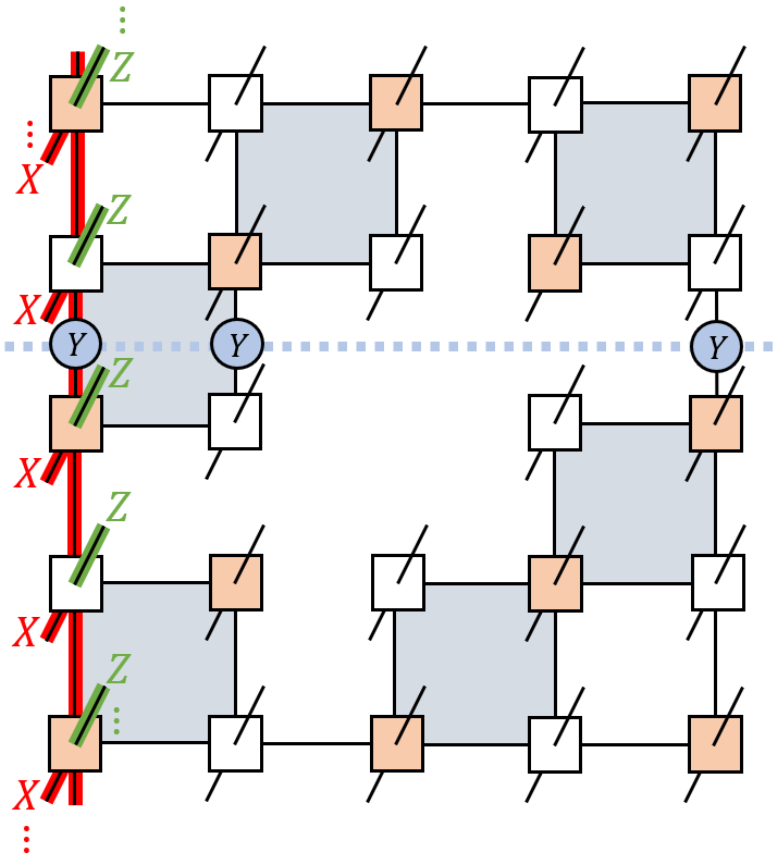}
    \caption{The gadget layout adapted from Ref. \cite{mclauchlan_accommodating_2024} for qubit dropouts. The lowest weight undetected logical error for the vertical logical $X$ operator can be chosen as the string of $Y$ errors, which has weight $d-2$. Distance reduction of other logical operators can be similarly obtained.}
    \label{fig: dropoutdistance2}
\end{figure}

\begin{figure}[H]
    \centering
    \includegraphics[width=0.7\linewidth]{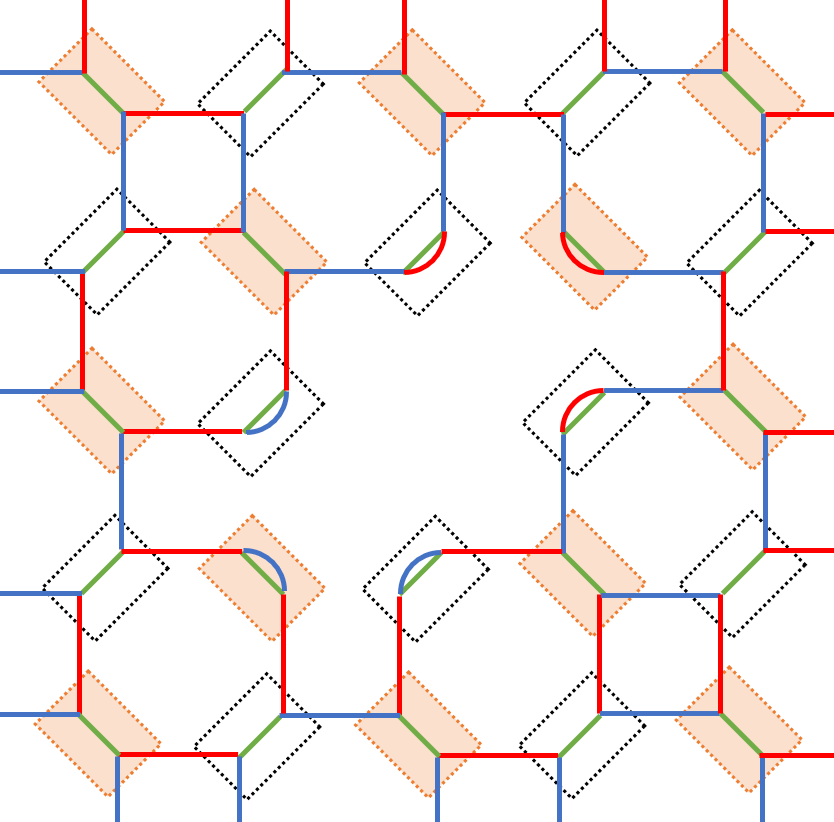}
    \caption{Lattice implementation of the protocol adapted from Ref. \cite{mclauchlan_accommodating_2024} for qubit dropouts, which is still implemented by the 3-round schedule: gZZ, bXX and rYY.}
    \label{fig: dropoutschedule2}
\end{figure}

\end{document}

%% file: zxbw.tikzstyles
\tikzstyle{x node}=[fill=black, draw=black, shape=circle, text=white, minimum size=1mm]
\tikzstyle{z node}=[fill=white, draw=black, shape=circle, minimum size=1mm]
\tikzstyle{hnode}=[fill=white, draw=black, shape=rectangle]
\tikzstyle{y node}=[fill=white, draw=black, shape=isosceles triangle, isosceles triangle apex angle=60, rotate=90]

\tikzstyle{green shaded}=[-, line width=3pt, draw={rgb,255: red,20; green,206; blue,0}]
\tikzstyle{red shaded}=[-, line width=3pt, draw=red]
\tikzstyle{blue shaded}=[-, line width=3pt, draw=blue]

%% file: znode.tikz
\begin{tikzpicture}
	\begin{pgfonlayer}{nodelayer}
		\node [style=z node] (1) at (0, 0) {};
		\node [style=z node] (2) at (0, 0) {};
		\node [style=z node] (3) at (0, 0) {$e^{i\alpha}$};
		\node [style=none] (7) at (1, 2) {};
		\node [style=none] (8) at (-1, 2) {};
		\node [style=none] (9) at (-2, 0) {};
		\node [style=none] (11) at (1, -2) {};
		\node [style=none] (12) at (-1, -2) {};
		\node [style=none] (13) at (1.25, 0) {};
		\node [style=none] (14) at (1.25, 0) {};
		\node [style=none] (15) at (1.5, 0) {$\cdots$};
	\end{pgfonlayer}
	\begin{pgfonlayer}{edgelayer}
		\draw (7.center) to (3);
		\draw (8.center) to (3);
		\draw (9.center) to (3);
		\draw (3) to (12.center);
		\draw (3) to (11.center);
	\end{pgfonlayer}
\end{tikzpicture}

%% file: xnode.tikz
\begin{tikzpicture}
	\begin{pgfonlayer}{nodelayer}
		\node [style=x node] (1) at (0, 0) {};
		\node [style=x node] (2) at (0, 0) {};
		\node [style=x node] (3) at (0, 0) {$e^{i\alpha}$};
		\node [style=none] (7) at (1, 2) {};
		\node [style=none] (8) at (-1, 2) {};
		\node [style=none] (9) at (-2, 0) {};
		\node [style=none] (11) at (1, -2) {};
		\node [style=none] (12) at (-1, -2) {};
		\node [style=none] (13) at (1.25, 0) {};
		\node [style=none] (14) at (1.25, 0) {};
		\node [style=none] (15) at (1.5, 0) {$\cdots$};
	\end{pgfonlayer}
	\begin{pgfonlayer}{edgelayer}
		\draw (7.center) to (3);
		\draw (8.center) to (3);
		\draw (9.center) to (3);
		\draw (3) to (12.center);
		\draw (3) to (11.center);
	\end{pgfonlayer}
\end{tikzpicture}

%% file: xinit.tikz
\begin{tikzpicture}
	\begin{pgfonlayer}{nodelayer}
		\node [style=z node] (0) at (0, 0) {};
		\node [style=z node] (1) at (0, 0) {$e^{i\alpha}$};
		\node [style=none] (2) at (1.5, 0) {};
	\end{pgfonlayer}
	\begin{pgfonlayer}{edgelayer}
		\draw (1) to (2.center);
	\end{pgfonlayer}
\end{tikzpicture}

%% file: zinit.tikz
\begin{tikzpicture}
	\begin{pgfonlayer}{nodelayer}
		\node [style=x node] (0) at (0, 0) {};
		\node [style=x node] (1) at (0, 0) {$e^{i\alpha}$};
		\node [style=none] (2) at (1.5, 0) {};
	\end{pgfonlayer}
	\begin{pgfonlayer}{edgelayer}
		\draw (1) to (2.center);
	\end{pgfonlayer}
\end{tikzpicture}

%% file: hadamard.tikz
\begin{tikzpicture}
	\begin{pgfonlayer}{nodelayer}
		\node [style=none] (0) at (-1, 0) {};
		\node [style=none] (1) at (1, 0) {};
		\node [style=hnode] (2) at (0, 0) {};
	\end{pgfonlayer}
	\begin{pgfonlayer}{edgelayer}
		\draw (2) to (1.center);
		\draw (0.center) to (2);
	\end{pgfonlayer}
\end{tikzpicture}

%% file: hadamard2.tikz
\begin{tikzpicture}
	\begin{pgfonlayer}{nodelayer}
		\node [style=x node] (0) at (0, 0) {$i$};
		\node [style=z node] (1) at (1.5, 0) {$i$};
		\node [style=z node] (2) at (-1.5, 0) {$i$};
		\node [style=none] (3) at (-2.5, 0) {};
		\node [style=none] (4) at (2.5, 0) {};
	\end{pgfonlayer}
	\begin{pgfonlayer}{edgelayer}
		\draw (3.center) to (2);
		\draw (4.center) to (1);
		\draw (1) to (0);
		\draw (0) to (2);
	\end{pgfonlayer}
\end{tikzpicture}

%% file: cnotcz.tikz
\begin{tikzpicture}
	\begin{pgfonlayer}{nodelayer}
		\node [style=z node] (9) at (0, 1) {};
		\node [style=none] (11) at (-1, 1) {};
		\node [style=none] (12) at (-1, -1) {};
		\node [style=none] (13) at (1, -1) {};
		\node [style=none] (14) at (1, 1) {};
		\node [style=x node] (26) at (0, -1) {};
		\node [style=none] (27) at (2, 0) {and};
		\node [style=none] (29) at (3, 1) {};
		\node [style=none] (30) at (3, -1) {};
		\node [style=none] (31) at (5, -1) {};
		\node [style=none] (32) at (5, 1) {};
		\node [style=x node] (33) at (4, -1) {};
		\node [style=x node] (34) at (4, 1) {};
		\node [style=hnode] (35) at (4, 0) {};
	\end{pgfonlayer}
	\begin{pgfonlayer}{edgelayer}
		\draw (11.center) to (9);
		\draw (9) to (14.center);
		\draw (9) to (26);
		\draw (13.center) to (26);
		\draw (26) to (12.center);
		\draw (31.center) to (33);
		\draw (33) to (30.center);
		\draw (29.center) to (32.center);
		\draw (34) to (33);
	\end{pgfonlayer}
\end{tikzpicture}

%% file: xfusion.tikz
\begin{tikzpicture}
	\begin{pgfonlayer}{nodelayer}
		\node [style=z node] (0) at (-1.25, -1) {$e^{i\alpha}$};
		\node [style=z node] (1) at (1, 1.25) {$e^{i\beta}$};
		\node [style=none] (2) at (0, 0) {$\cdots$};
		\node [style=none] (3) at (-3, 0) {};
		\node [style=none] (4) at (-3.25, -1.5) {};
		\node [style=none] (5) at (-1, -2.75) {};
		\node [style=none] (6) at (0.5, 3.5) {};
		\node [style=none] (7) at (2, 3) {};
		\node [style=none] (8) at (3, 1) {};
		\node [style=none] (9) at (2.25, 2) {$\cdots$};
		\node [style=none] (10) at (-2.25, -2) {$\cdots$};
		\node [style=none] (11) at (3.75, 0) {$=$};
		\node [style=z node] (12) at (7, 0) {$e^{i(\alpha+\beta)}$};
		\node [style=none] (14) at (4.5, 1.25) {};
		\node [style=none] (15) at (4.25, -0.75) {};
		\node [style=none] (16) at (7.25, -2.5) {};
		\node [style=none] (17) at (5.5, -1.75) {$\cdots$};
		\node [style=none] (19) at (6.25, 3) {};
		\node [style=none] (20) at (8.5, 2.5) {};
		\node [style=none] (21) at (9.5, -0.5) {};
		\node [style=none] (22) at (9.25, 1) {$\cdots$};
	\end{pgfonlayer}
	\begin{pgfonlayer}{edgelayer}
		\draw [bend right] (1) to (0);
		\draw [bend left=45] (1) to (0);
		\draw (3.center) to (0);
		\draw (4.center) to (0);
		\draw (0) to (5.center);
		\draw (6.center) to (1);
		\draw (7.center) to (1);
		\draw (1) to (8.center);
		\draw (14.center) to (12);
		\draw (15.center) to (12);
		\draw (12) to (16.center);
		\draw (19.center) to (12);
		\draw (20.center) to (12);
		\draw (21.center) to (12);
	\end{pgfonlayer}
\end{tikzpicture}

%% file: xhard.tikz
\begin{tikzpicture}
	\begin{pgfonlayer}{nodelayer}
		\node [style=z node] (12) at (0, 0) {$e^{i\alpha}$};
		\node [style=none] (14) at (-2, 1) {};
		\node [style=none] (15) at (-2, -1) {};
		\node [style=none] (16) at (2, -1) {};
		\node [style=none] (19) at (0, 2) {};
		\node [style=none] (20) at (2, 1) {};
		\node [style=none] (22) at (0, -1.25) {$\cdots$};
		\node [style=hnode, rotate=330] (23) at (1.5, -0.75) {};
		\node [style=hnode] (25) at (0, 1.5) {};
		\node [style=hnode, rotate=330] (26) at (-1.5, 0.75) {};
		\node [style=hnode, rotate=30] (27) at (-1.5, -0.75) {};
		\node [style=hnode, rotate=30] (28) at (1.5, 0.75) {};
		\node [style=none] (30) at (4, 1) {};
		\node [style=none] (31) at (4, -1) {};
		\node [style=none] (32) at (8, -1) {};
		\node [style=none] (33) at (6, 2) {};
		\node [style=none] (34) at (8, 1) {};
		\node [style=none] (35) at (6, -1.25) {$\cdots$};
		\node [style=x node] (36) at (6, 0) {};
		\node [style=x node] (37) at (6, 0) {$e^{i\alpha}$};
		\node [style=none] (39) at (3, 0) {$=$};
	\end{pgfonlayer}
	\begin{pgfonlayer}{edgelayer}
		\draw (14.center) to (12);
		\draw (15.center) to (12);
		\draw (12) to (16.center);
		\draw (19.center) to (12);
		\draw (20.center) to (12);
		\draw (33.center) to (37);
		\draw (34.center) to (37);
		\draw (37) to (32.center);
		\draw (37) to (31.center);
		\draw (37) to (30.center);
	\end{pgfonlayer}
\end{tikzpicture}

%% file: hopf.tikz
\begin{tikzpicture}
	\begin{pgfonlayer}{nodelayer}
		\node [style=x node] (0) at (-0.75, 0) {};
		\node [style=z node] (1) at (0.75, 0) {};
		\node [style=none] (2) at (2, 1) {};
		\node [style=none] (3) at (2, -1) {};
		\node [style=none] (4) at (-2, 1) {};
		\node [style=none] (5) at (-2, -1) {};
		\node [style=none] (6) at (-2, 0) {$\cdots$};
		\node [style=none] (7) at (2, 0) {$\cdots$};
		\node [style=x node] (8) at (4.75, 0) {};
		\node [style=z node] (9) at (5.75, 0) {};
		\node [style=none] (10) at (7, 1) {};
		\node [style=none] (11) at (7, -1) {};
		\node [style=none] (12) at (3.5, 1) {};
		\node [style=none] (13) at (3.5, -1) {};
		\node [style=none] (14) at (3.5, 0) {$\cdots$};
		\node [style=none] (15) at (7, 0) {$\cdots$};
		\node [style=none] (16) at (2.75, 0) {$=$};
	\end{pgfonlayer}
	\begin{pgfonlayer}{edgelayer}
		\draw [bend left=330, looseness=1.25] (1) to (0);
		\draw (0) to (4.center);
		\draw (0) to (5.center);
		\draw (1) to (2.center);
		\draw (1) to (3.center);
		\draw [bend right, looseness=1.25] (0) to (1);
		\draw (8) to (12.center);
		\draw (8) to (13.center);
		\draw (9) to (10.center);
		\draw (9) to (11.center);
	\end{pgfonlayer}
\end{tikzpicture}

%% file: bialgebra.tikz
\begin{tikzpicture}
	\begin{pgfonlayer}{nodelayer}
		\node [style=x node] (0) at (-0.75, 0) {};
		\node [style=z node] (1) at (0.5, 0) {};
		\node [style=none] (2) at (1.75, 1) {};
		\node [style=none] (3) at (1.75, -1) {};
		\node [style=none] (4) at (-2, 1) {};
		\node [style=none] (5) at (-2, -1) {};
		\node [style=none] (6) at (-2, 0) {$\cdots$};
		\node [style=none] (7) at (1.75, 0) {$\cdots$};
		\node [style=none] (10) at (6.75, 1) {};
		\node [style=none] (11) at (6.75, -1) {};
		\node [style=none] (12) at (3.25, 1) {};
		\node [style=none] (13) at (3.25, -1) {};
		\node [style=none] (14) at (3.25, 0) {$\cdots$};
		\node [style=none] (15) at (6.75, 0) {$\cdots$};
		\node [style=none] (16) at (2.5, 0) {$=$};
		\node [style=z node] (17) at (4.25, 1) {};
		\node [style=z node] (18) at (4.25, -1) {};
		\node [style=x node] (19) at (5.75, 1) {};
		\node [style=x node] (20) at (5.75, -1) {};
	\end{pgfonlayer}
	\begin{pgfonlayer}{edgelayer}
		\draw (0) to (4.center);
		\draw (0) to (5.center);
		\draw (1) to (2.center);
		\draw (1) to (3.center);
		\draw (0) to (1);
		\draw (12.center) to (17);
		\draw (13.center) to (18);
		\draw (10.center) to (19);
		\draw (11.center) to (20);
		\draw (19) to (18);
		\draw (17) to (20);
		\draw (19) to (17);
		\draw (20) to (18);
	\end{pgfonlayer}
\end{tikzpicture}

%% file: pinode.tikz
\begin{tikzpicture}
	\begin{pgfonlayer}{nodelayer}
		\node [style=z node] (0) at (0, 0) {$e^{i\alpha}$};
		\node [style=x node] (1) at (0, 2) {-1};
		\node [style=x node] (2) at (-2, 1) {-1};
		\node [style=x node] (3) at (2, 1) {-1};
		\node [style=x node] (4) at (-2, -1) {-1};
		\node [style=x node] (5) at (2, -1) {-1};
		\node [style=none] (6) at (0, 3) {};
		\node [style=none] (7) at (3, 1.5) {};
		\node [style=none] (8) at (3, -1.5) {};
		\node [style=none] (9) at (0, -1.5) {$\cdots$};
		\node [style=none] (10) at (-3, -1.5) {};
		\node [style=none] (11) at (-3, 1.5) {};
		\node [style=none] (12) at (3.75, 0) {$=$};
		\node [style=z node] (13) at (7, 0) {$e^{-i\alpha}$};
		\node [style=none] (19) at (7, -1.5) {$\cdots$};
		\node [style=none] (20) at (9, -1) {};
		\node [style=none] (21) at (5, -1) {};
		\node [style=none] (22) at (5, 1) {};
		\node [style=none] (23) at (7, 2) {};
		\node [style=none] (24) at (9, 1) {};
	\end{pgfonlayer}
	\begin{pgfonlayer}{edgelayer}
		\draw (1) to (0);
		\draw (0) to (3);
		\draw (0) to (5);
		\draw (0) to (4);
		\draw (0) to (2);
		\draw (1) to (6.center);
		\draw (3) to (7.center);
		\draw (5) to (8.center);
		\draw (4) to (10.center);
		\draw (2) to (11.center);
		\draw (13) to (23.center);
		\draw (13) to (24.center);
		\draw (13) to (20.center);
		\draw (13) to (21.center);
		\draw (13) to (22.center);
	\end{pgfonlayer}
\end{tikzpicture}

%% file: xweb.tikz
\begin{tikzpicture}
	\begin{pgfonlayer}{nodelayer}
		\node [style=none] (7) at (0.75, 1.5) {};
		\node [style=none] (8) at (-0.75, 1.5) {};
		\node [style=none] (9) at (-1.75, 0) {};
		\node [style=none] (11) at (0.75, -1.5) {};
		\node [style=none] (12) at (-0.75, -1.5) {};
		\node [style=none] (13) at (1.25, 0) {};
		\node [style=none] (14) at (1.25, 0) {};
		\node [style=none] (15) at (1.25, 0) {$\dots$};
		\node [style=z node] (16) at (0, 0) {};
		\node [style=none] (17) at (4.75, 1.5) {};
		\node [style=none] (18) at (3.25, 1.5) {};
		\node [style=none] (19) at (2.25, 0) {};
		\node [style=none] (20) at (4.75, -1.5) {};
		\node [style=none] (21) at (3.25, -1.5) {};
		\node [style=none] (22) at (5.25, 0) {};
		\node [style=none] (23) at (5.25, 0) {};
		\node [style=none] (24) at (5.25, 0) {$\dots$};
		\node [style=z node] (25) at (4, 0) {};
		\node [style=none] (26) at (8.75, 1.5) {};
		\node [style=none] (27) at (7.25, 1.5) {};
		\node [style=none] (28) at (6.25, 0) {};
		\node [style=none] (29) at (8.75, -1.5) {};
		\node [style=none] (30) at (7.25, -1.5) {};
		\node [style=none] (31) at (9.25, 0) {};
		\node [style=none] (32) at (9.25, 0) {};
		\node [style=none] (33) at (9.25, 0) {$\dots$};
		\node [style=z node] (34) at (8, 0) {};
	\end{pgfonlayer}
	\begin{pgfonlayer}{edgelayer}
		\draw [style=red shaded] (16) to (8.center);
		\draw [style=red shaded] (16) to (12.center);
		\draw (7.center) to (16);
		\draw (16) to (8.center);
		\draw (16) to (9.center);
		\draw (16) to (11.center);
		\draw (16) to (12.center);
		\draw [style=red shaded] (25) to (20.center);
		\draw [style=red shaded] (25) to (21.center);
		\draw (25) to (20.center);
		\draw (25) to (21.center);
		\draw (25) to (19.center);
		\draw (18.center) to (25);
		\draw (25) to (17.center);
		\draw [style=green shaded] (34) to (26.center);
		\draw [style=green shaded] (34) to (27.center);
		\draw [style=green shaded] (34) to (28.center);
		\draw [style=green shaded] (34) to (29.center);
		\draw [style=green shaded] (34) to (30.center);
		\draw (34) to (29.center);
		\draw (34) to (30.center);
		\draw (34) to (28.center);
		\draw (34) to (27.center);
		\draw (34) to (26.center);
	\end{pgfonlayer}
\end{tikzpicture}

%% file: zweb.tikz
\begin{tikzpicture}
	\begin{pgfonlayer}{nodelayer}
		\node [style=none] (7) at (0.75, 1.5) {};
		\node [style=none] (8) at (-0.75, 1.5) {};
		\node [style=none] (9) at (-1.75, 0) {};
		\node [style=none] (11) at (0.75, -1.5) {};
		\node [style=none] (12) at (-0.75, -1.5) {};
		\node [style=none] (13) at (1.25, 0) {};
		\node [style=none] (14) at (1.25, 0) {};
		\node [style=none] (15) at (1.25, 0) {$\dots$};
		\node [style=x node] (16) at (0, 0) {};
		\node [style=none] (17) at (4.75, 1.5) {};
		\node [style=none] (18) at (3.25, 1.5) {};
		\node [style=none] (19) at (2.25, 0) {};
		\node [style=none] (20) at (4.75, -1.5) {};
		\node [style=none] (21) at (3.25, -1.5) {};
		\node [style=none] (22) at (5.25, 0) {};
		\node [style=none] (23) at (5.25, 0) {};
		\node [style=none] (24) at (5.25, 0) {$\dots$};
		\node [style=x node] (25) at (4, 0) {};
		\node [style=none] (26) at (8.75, 1.5) {};
		\node [style=none] (27) at (7.25, 1.5) {};
		\node [style=none] (28) at (6.25, 0) {};
		\node [style=none] (29) at (8.75, -1.5) {};
		\node [style=none] (30) at (7.25, -1.5) {};
		\node [style=none] (31) at (9.25, 0) {};
		\node [style=none] (32) at (9.25, 0) {};
		\node [style=none] (33) at (9.25, 0) {$\dots$};
		\node [style=x node] (34) at (8, 0) {};
	\end{pgfonlayer}
	\begin{pgfonlayer}{edgelayer}
		\draw [style=green shaded] (16) to (8.center);
		\draw [style=green shaded] (16) to (12.center);
		\draw (7.center) to (16);
		\draw (16) to (8.center);
		\draw (16) to (9.center);
		\draw (16) to (11.center);
		\draw (16) to (12.center);
		\draw [style=green shaded] (25) to (17.center);
		\draw [style=green shaded] (25) to (19.center);
		\draw (25) to (20.center);
		\draw (25) to (21.center);
		\draw (25) to (19.center);
		\draw (18.center) to (25);
		\draw (25) to (17.center);
		\draw [style=red shaded] (34) to (26.center);
		\draw [style=red shaded] (34) to (27.center);
		\draw [style=red shaded] (34) to (28.center);
		\draw [style=red shaded] (34) to (29.center);
		\draw [style=red shaded] (34) to (30.center);
		\draw (34) to (29.center);
		\draw (34) to (30.center);
		\draw (34) to (28.center);
		\draw (34) to (27.center);
		\draw (34) to (26.center);
	\end{pgfonlayer}
\end{tikzpicture}

%% file: hweb.tikz
\begin{tikzpicture}
	\begin{pgfonlayer}{nodelayer}
		\node [style=none] (0) at (-1.5, 0) {};
		\node [style=none] (1) at (1.5, 0) {};
		\node [style=hnode] (2) at (0, 0) {};
	\end{pgfonlayer}
	\begin{pgfonlayer}{edgelayer}
        \draw [style=red shaded](2) to (1.center);
		\draw (2) to (1.center);
        \draw [style=green shaded](0.center) to (2);
		\draw (0.center) to (2);
	\end{pgfonlayer}
\end{tikzpicture}

%% file: xpair.tikz
\begin{tikzpicture}
	\begin{pgfonlayer}{nodelayer}
		\node [style=z node] (0) at (-2, 0) {};
		\node [style=none] (1) at (-2.75, 1) {};
		\node [style=none] (2) at (-2.75, -1) {};
		\node [style=none] (3) at (-1.25, 1) {};
		\node [style=none] (4) at (-1.25, -1) {};
		\node [style=none] (5) at (-0.25, 1.25) {};
		\node [style=none] (6) at (-0.25, -1.25) {};
		\node [style=none] (7) at (-3.75, 1.25) {};
		\node [style=none] (8) at (-3.75, -1.25) {};
		\node [style=z node] (9) at (2.25, 1) {};
		\node [style=z node] (10) at (2.25, -1) {};
		\node [style=none] (11) at (1.25, 1) {};
		\node [style=none] (12) at (1.25, -1) {};
		\node [style=none] (13) at (3.25, -1) {};
		\node [style=none] (14) at (3.25, 1) {};
		\node [style=none] (15) at (0.5, 0) {$=$};
		\node [style=z node] (16) at (5.5, 1) {};
		\node [style=z node] (17) at (5.5, -1) {};
		\node [style=none] (18) at (4.5, 1) {};
		\node [style=none] (19) at (4.5, -1) {};
		\node [style=none] (20) at (6.5, -1) {};
		\node [style=none] (21) at (6.5, 1) {};
		\node [style=none] (22) at (3.5, 0) {=};
		\node [style=x node] (23) at (5.5, 0) {};
		\node [style=x node] (24) at (4.5, 0) {};
		\node [style=x node] (25) at (6.5, 0) {};
	\end{pgfonlayer}
	\begin{pgfonlayer}{edgelayer}
		\draw [bend left=15, looseness=0.75] (7.center) to (1.center);
		\draw [bend left=15] (1.center) to (0);
		\draw [bend left=15] (0) to (3.center);
		\draw [bend left=15, looseness=0.75] (3.center) to (5.center);
		\draw [bend right=15] (0) to (4.center);
		\draw [bend right=15, looseness=0.75] (4.center) to (6.center);
		\draw [bend left=15] (0) to (2.center);
		\draw [bend left=15, looseness=0.75] (2.center) to (8.center);
		\draw (11.center) to (9);
		\draw (9) to (14.center);
		\draw (9) to (10);
		\draw (10) to (12.center);
		\draw (10) to (13.center);
		\draw (18.center) to (16);
		\draw (16) to (21.center);
		\draw (16) to (17);
		\draw (17) to (19.center);
		\draw (17) to (20.center);
		\draw (24) to (23);
		\draw (23) to (25);
	\end{pgfonlayer}
\end{tikzpicture}

%% file: xpair2.tikz
\begin{tikzpicture}
	\begin{pgfonlayer}{nodelayer}
		\node [style=z node] (0) at (0, 1) {};
		\node [style=z node] (1) at (0, -1) {};
		\node [style=none] (2) at (1, 1) {};
		\node [style=none] (3) at (1, -1) {};
		\node [style=none] (4) at (-1, -1) {};
		\node [style=none] (5) at (-1, 1) {};
		\node [style=none] (6) at (1.75, 0) {$\sim$};
		\node [style=z node] (7) at (3.75, 1.5) {};
		\node [style=z node] (8) at (3.75, -1.5) {};
		\node [style=none] (9) at (4.75, 1.5) {};
		\node [style=none] (10) at (4.75, -1.5) {};
		\node [style=none] (11) at (2.75, -1.5) {};
		\node [style=none] (12) at (2.75, 1.5) {};
		\node [style=x node] (13) at (2.75, 0) {};
		\node [style=x node] (14) at (3.75, 0) {};
		\node [style=x node] (15) at (5.25, 0) {$\pm1$};
	\end{pgfonlayer}
	\begin{pgfonlayer}{edgelayer}
		\draw (5.center) to (2.center);
		\draw (4.center) to (3.center);
		\draw (0) to (1);
		\draw (12.center) to (9.center);
		\draw (11.center) to (10.center);
		\draw (7) to (8);
		\draw (13) to (15.center);
	\end{pgfonlayer}
\end{tikzpicture}

%% file: weight4.tikz
\begin{tikzpicture}
	\begin{pgfonlayer}{nodelayer}
		\node [style=x node] (0) at (0, 3) {};
		\node [style=x node] (1) at (0, 1) {};
		\node [style=x node] (2) at (0, -1) {};
		\node [style=x node] (3) at (0, -3) {};
		\node [style=none] (4) at (-1.75, 3) {};
		\node [style=none] (5) at (-1.75, 1) {};
		\node [style=none] (6) at (-1.75, -1) {};
		\node [style=none] (7) at (-1.75, -3) {};
		\node [style=none] (8) at (1.75, -3) {};
		\node [style=none] (9) at (1.75, -1) {};
		\node [style=none] (10) at (1.75, 1) {};
		\node [style=none] (11) at (1.75, 3) {};
		\node [style=z node] (12) at (1, 0) {};
		\node [style=x node] (13) at (7.25, 3) {};
		\node [style=x node] (14) at (6.25, 1.5) {};
		\node [style=x node] (15) at (5.25, -1.5) {};
		\node [style=x node] (16) at (4.25, -3) {};
		\node [style=none] (17) at (3.25, 3) {};
		\node [style=none] (18) at (3.25, 1.5) {};
		\node [style=none] (19) at (3.25, -1.5) {};
		\node [style=none] (20) at (3.25, -3) {};
		\node [style=none] (21) at (8.75, -3) {};
		\node [style=none] (22) at (8.75, -1.5) {};
		\node [style=none] (23) at (8.75, 1.5) {};
		\node [style=none] (24) at (8.75, 3) {};
		\node [style=none] (26) at (2.25, 0) {$\sim$};
		\node [style=z node] (27) at (8.75, 0) {$\pm 1$};
		\node [style=z node] (28) at (3.25, 0) {};
		\node [style=z node] (29) at (4.25, 0) {};
		\node [style=z node] (30) at (5.25, 0) {};
		\node [style=z node] (31) at (6.25, 0) {};
		\node [style=z node] (32) at (7.25, 0) {};
	\end{pgfonlayer}
	\begin{pgfonlayer}{edgelayer}
		\draw (4.center) to (11.center);
		\draw (10.center) to (5.center);
		\draw (6.center) to (9.center);
		\draw (7.center) to (8.center);
		\draw [bend left=15] (12) to (1);
		\draw [bend right=15] (12) to (2);
		\draw [bend left=15] (0) to (12);
		\draw [bend left=15] (12) to (3);
		\draw (17.center) to (24.center);
		\draw (23.center) to (18.center);
		\draw (19.center) to (22.center);
		\draw (20.center) to (21.center);
		\draw (28) to (27.center);
		\draw (29) to (16);
		\draw (30) to (15);
		\draw (31) to (14);
		\draw (32) to (13);
	\end{pgfonlayer}
\end{tikzpicture}

%% file: xstab.tikz
\begin{tikzpicture}
	\begin{pgfonlayer}{nodelayer}
		\node [style=x node] (0) at (0, -0.75) {};
		\node [style=z node] (1) at (0.75, -0.25) {};
		\node [style=z node] (3) at (0.75, -1.25) {};
		\node [style=none] (4) at (-1, -0.75) {};
		\node [style=x node] (5) at (0, 1.25) {};
		\node [style=z node] (6) at (1, 1.75) {};
		\node [style=z node] (8) at (0.75, 0.75) {};
		\node [style=none] (9) at (-1, 1.25) {};
		\node [style=x node] (10) at (0, -2.75) {};
		\node [style=z node] (11) at (0.75, -2.25) {};
		\node [style=z node] (13) at (0.75, -3.25) {};
		\node [style=none] (14) at (-1, -2.75) {};
		\node [style=x node] (15) at (0, -4.75) {};
		\node [style=z node] (16) at (0.75, -4.25) {};
		\node [style=z node] (18) at (1, -5.25) {};
		\node [style=none] (19) at (-1, -4.75) {};
		\node [style=none] (20) at (2, 1.75) {};
		\node [style=none] (21) at (2, 0.75) {};
		\node [style=none] (22) at (2, -0.25) {};
		\node [style=none] (23) at (2, -1.25) {};
		\node [style=none] (24) at (2, -2.25) {};
		\node [style=none] (25) at (2, -3.25) {};
		\node [style=none] (26) at (2, -4.25) {};
		\node [style=none] (27) at (2, -5.25) {};
		\node [style=x node] (32) at (5.5, 0) {};
		\node [style=z node] (33) at (6.25, 0.5) {};
		\node [style=z node] (34) at (6.25, -0.5) {};
		\node [style=none] (35) at (4.5, 0) {};
		\node [style=x node] (36) at (5.5, 3) {};
		\node [style=z node] (37) at (6.5, 3.5) {};
		\node [style=z node] (38) at (6.25, 2.5) {};
		\node [style=none] (39) at (4.5, 3) {};
		\node [style=x node] (40) at (5.5, -3.5) {};
		\node [style=z node] (41) at (6.25, -3) {};
		\node [style=z node] (42) at (6.25, -4) {};
		\node [style=none] (43) at (4.5, -3.5) {};
		\node [style=x node] (44) at (5.5, -6.5) {};
		\node [style=z node] (45) at (6.25, -6) {};
		\node [style=z node] (46) at (6.5, -7) {};
		\node [style=none] (47) at (4.5, -6.5) {};
		\node [style=none] (48) at (8.75, 3.5) {};
		\node [style=none] (49) at (8.75, 2.5) {};
		\node [style=none] (50) at (8.75, 0.5) {};
		\node [style=none] (51) at (8.75, -0.5) {};
		\node [style=none] (52) at (8.75, -3) {};
		\node [style=none] (53) at (8.75, -4) {};
		\node [style=none] (54) at (8.75, -6) {};
		\node [style=none] (55) at (8.75, -7) {};
		\node [style=x node] (60) at (6.25, 1.5) {};
		\node [style=x node] (61) at (8.25, 1.5) {$\pm1$};
		\node [style=x node] (62) at (7.5, -1.75) {};
		\node [style=x node] (63) at (8.25, -5) {$\pm1$};
		\node [style=x node] (64) at (9.5, -1.75) {$\pm1$};
		\node [style=x node] (65) at (6.25, -1.75) {};
		\node [style=x node] (66) at (4.5, -1.75) {$\pm1$};
		\node [style=x node] (67) at (6.25, -5) {};
		\node [style=none] (68) at (3, -1.75) {$\sim$};
	\end{pgfonlayer}
	\begin{pgfonlayer}{edgelayer}
		\draw [style=red shaded] (0) to (1);
		\draw (0) to (1);
		\draw [style=red shaded] (0) to (3);
		\draw (0) to (3);
		\draw [style=red shaded] (4.center) to (0);
		\draw (4.center) to (0);
		\draw [style=red shaded] (5) to (6);
		\draw (5) to (6);
		\draw [style=red shaded] (5) to (8);
		\draw (5) to (8);
		\draw [style=red shaded] (9.center) to (5);
		\draw (9.center) to (5);
		\draw [style=red shaded] (10) to (11);
		\draw (10) to (11);
		\draw [style=red shaded] (10) to (13);
		\draw (10) to (13);
		\draw [style=red shaded] (14.center) to (10);
		\draw (14.center) to (10);
		\draw [style=red shaded] (15) to (16);
		\draw (15) to (16);
		\draw [style=red shaded] (15) to (18);
		\draw (15) to (18);
		\draw [style=red shaded] (19.center) to (15);
		\draw (19.center) to (15);
		\draw [style=red shaded] (8) to (1);
		\draw (8) to (1);
		\draw [style=red shaded] (3) to (11);
		\draw (3) to (11);
		\draw [style=red shaded, bend left=15] (6) to (18);
		\draw [bend left=15] (6) to (18);
		\draw [style=red shaded] (13) to (16);
		\draw (13) to (16);
		\draw (13) to (25.center);
		\draw (16) to (26.center);
		\draw (18) to (27.center);
		\draw (11) to (24.center);
		\draw (23.center) to (3);
		\draw (1) to (22.center);
		\draw (8) to (21.center);
		\draw (6) to (20.center);
		\draw (9.center) to (5);
		\draw [style=red shaded] (32) to (33);
		\draw (32) to (33);
		\draw [style=red shaded] (32) to (34);
		\draw (32) to (34);
		\draw [style=red shaded] (35.center) to (32);
		\draw (35.center) to (32);
		\draw [style=red shaded] (36) to (37);
		\draw (36) to (37);
		\draw [style=red shaded] (36) to (38);
		\draw (36) to (38);
		\draw [style=red shaded] (39.center) to (36);
		\draw (39.center) to (36);
		\draw [style=red shaded] (40) to (41);
		\draw (40) to (41);
		\draw [style=red shaded] (40) to (42);
		\draw (40) to (42);
		\draw [style=red shaded] (43.center) to (40);
		\draw (43.center) to (40);
		\draw [style=red shaded] (44) to (45);
		\draw (44) to (45);
		\draw [style=red shaded] (44) to (46);
		\draw (44) to (46);
		\draw [style=red shaded] (47.center) to (44);
		\draw (47.center) to (44);
		\draw [style=red shaded] (38) to (33);
		\draw (38) to (33);
		\draw [style=red shaded] (34) to (41);
		\draw (34) to (41);
		\draw [style=red shaded, in=75, out=-75, looseness=1.25] (37) to (46);
		\draw [style=red shaded] (42) to (45);
		\draw (42) to (45);
		\draw (42) to (53.center);
		\draw (45) to (54.center);
		\draw (46) to (55.center);
		\draw (41) to (52.center);
		\draw (51.center) to (34);
		\draw (33) to (50.center);
		\draw (38) to (49.center);
		\draw (37) to (48.center);
		\draw (39.center) to (36);
		\draw [bend left=15, looseness=1.25] (37) to (46);
        \draw [style=red shaded] (65) to (66);
		\draw (65) to (66);
        \draw [style=red shaded] (62) to (64);
		\draw (62) to (64);
        \draw [style=red shaded] (60) to (61);
		\draw (60) to (61);
        \draw [style=red shaded] (67) to (63);
		\draw (67) to (63);
	\end{pgfonlayer}
\end{tikzpicture}

%% file: ynode.tikz
\begin{tikzpicture}
	\begin{pgfonlayer}{nodelayer}
		\node [style=x node] (1) at (1.5, 0) {$i$};
		\node [style=x node] (2) at (0, 1.5) {$i$};
		\node [style=x node] (3) at (-1.5, 0) {$i$};
		\node [style=none] (4) at (3, 0) {};
		\node [style=none] (5) at (0, 3) {};
		\node [style=none] (6) at (-3, 0) {};
		\node [style=z node] (7) at (0, 0) {};
		\node [style=y node] (8) at (-6.5, 0) {};
		\node [style=none] (9) at (-5, 0) {};
		\node [style=none] (10) at (-6.5, 1.5) {};
		\node [style=none] (11) at (-8, 0) {};
		\node [style=none] (12) at (-4, 0) {$=$};
	\end{pgfonlayer}
	\begin{pgfonlayer}{edgelayer}
		\draw (1) to (4.center);
		\draw (2) to (5.center);
		\draw (3) to (6.center);
		\draw (7) to (1);
		\draw (7) to (2);
		\draw (7) to (3);
		\draw (8) to (9.center);
		\draw (8) to (10.center);
		\draw (8) to (11.center);
	\end{pgfonlayer}
\end{tikzpicture}

%% file: ypair.tikz
\begin{tikzpicture}
	\begin{pgfonlayer}{nodelayer}
		\node [style=y node] (0) at (0, 1.5) {};
		\node [style=y node] (1) at (0, -1.5) {};
		\node [style=none] (2) at (1.5, 1.5) {};
		\node [style=none] (3) at (1.5, -1.5) {};
		\node [style=none] (4) at (-1.5, -1.5) {};
		\node [style=none] (5) at (-1.5, 1.5) {};
	\end{pgfonlayer}
	\begin{pgfonlayer}{edgelayer}
		\draw (5.center) to (2.center);
		\draw (4.center) to (3.center);
		\draw (0) to (1);
	\end{pgfonlayer}
\end{tikzpicture}

%% file: xerror.tikz
\begin{tikzpicture}
	\begin{pgfonlayer}{nodelayer}
		\node [style=none] (0) at (-2, 0) {};
		\node [style=none] (1) at (2, 0) {};
		\node [style=z node] (2) at (0, 0) {$-1$};
	\end{pgfonlayer}
	\begin{pgfonlayer}{edgelayer}
		\draw (2) to (1.center);
		\draw (0.center) to (2);
	\end{pgfonlayer}
\end{tikzpicture}

%% file: zerror.tikz
\begin{tikzpicture}
	\begin{pgfonlayer}{nodelayer}
		\node [style=none] (0) at (-2, 0) {};
		\node [style=none] (1) at (2, 0) {};
		\node [style=x node] (2) at (0, 0) {$-1$};
	\end{pgfonlayer}
	\begin{pgfonlayer}{edgelayer}
		\draw (2) to (1.center);
		\draw (0.center) to (2);
	\end{pgfonlayer}
\end{tikzpicture}

%% file: pairerror.tikz
\begin{tikzpicture}
	\begin{pgfonlayer}{nodelayer}
		\node [style=z node] (0) at (0, 1.5) {};
		\node [style=z node] (1) at (0, -1.5) {};
		\node [style=none] (2) at (1, 1.5) {};
		\node [style=none] (3) at (1, -1.5) {};
		\node [style=none] (4) at (-1, -1.5) {};
		\node [style=none] (5) at (-1, 1.5) {};
		\node [style=none] (6) at (1.75, 0) {$\sim$};
		\node [style=z node] (7) at (3.75, 1.5) {};
		\node [style=z node] (8) at (3.75, -1.5) {};
		\node [style=none] (9) at (4.75, 1.5) {};
		\node [style=none] (10) at (4.75, -1.5) {};
		\node [style=none] (11) at (2.75, -1.5) {};
		\node [style=none] (12) at (2.75, 1.5) {};
		\node [style=x node] (13) at (2.75, 0) {};
		\node [style=x node] (14) at (3.75, 0) {};
		\node [style=x node] (15) at (5.25, 0) {$\mp1$};
		\node [style=x node] (16) at (0, 0) {$-1$};
	\end{pgfonlayer}
	\begin{pgfonlayer}{edgelayer}
		\draw (5.center) to (2.center);
		\draw (4.center) to (3.center);
		\draw (0) to (1);
		\draw (12.center) to (9.center);
		\draw (11.center) to (10.center);
		\draw (7) to (8);
		\draw (13) to (15);
	\end{pgfonlayer}
\end{tikzpicture}